\documentclass[indent=latex,review=false,theme=default,marginsize=2.5cm,fontsize=12pt]{acmhomework}

\usepackage{scalerel,stackengine}
\stackMath
\newcommand\reallywidehat[1]{%
\savestack{\tmpbox}{\stretchto{%
  \scaleto{%
    \scalerel*[\widthof{\ensuremath{#1}}]{\kern.1pt\mathchar"0362\kern.1pt}%
    {\rule{0ex}{\textheight}}
  }{\textheight}%
}{2.4ex}}%
\stackon[-6.9pt]{#1}{\tmpbox}%
}

\newcommand{\Diag}{\textsf{Diag}}
\newcommand{\Sq}{\textsf{Sq}}

\newcommand{\Cof}{\textsf{Cof}}
\newcommand{\txtis}{\textsf{is}}
\newcommand{\Path}{\textsf{Path}}
\newcommand{\via}{\textsf{via}}
\newcommand{\hetPath}{\textsf{hetPath}}
\newcommand{\hetvia}{\textsf{hetvia}}
\newcommand{\app}{\textsf{app}}
\newcommand{\tapp}{\widetilde{\textsf{app}}}
\renewcommand{\min}{\textsf{min}}
\renewcommand{\max}{\textsf{max}}
\newcommand{\pathres}{\textsf{pathres}}
\newcommand{\Fill}{\textsf{Fill}}

\renewcommand{\HIso}{\textsf{HIso}}
\newcommand{\GenHIso}{\textsf{GenHIso}}
\newcommand{\HIsoExt}{\textsf{HIsoExt}}

\newcommand{\El}{\textsf{El}}
\newcommand{\tEl}{\widetilde{\El}}

\renewcommand{\Unit}{\mathsf{Unit}}
\renewcommand{\Id}{\mathsf{Id}}

\newcommand{\src}{\textsf{src}}
\newcommand{\tgt}{\textsf{tgt}}
\renewcommand{\trv}{\textsf{trv}}

\newcommand{\Ext}{\textsf{Ext}}
\newcommand{\Glue}{\textsf{Glue}}
\newcommand{\glue}{\textsf{glue}}
\newcommand{\unglue}{\textsf{unglue}}
\newcommand{\Fib}{\textsf{Fib}}

\renewcommand{\ev}{\textsf{ev}}
\renewcommand{\const}{\textsf{const}}
\newcommand{\isConst}{\textsf{isConst}}
\renewcommand{\id}{\textsf{id}}

\newcommand{\hP}{\widehat{P}}

\newcommand\GenGlueProb{\textsf{GenGlueProb}}

\title{Yet another cubical type theory, but via a semantic approach}

\author{Krzysztof Kapulkin \and Yufeng Li}

\begin{abstract}
  We propose a new cubical type theory, termed (self-deprecatingly) the naive cubical type theory, and study its semantics using the universe category framework, which is similar to Uemura's categories with representable morphisms.
  In particular, we show that this new type theory admits an interpretation in a wide variety of settings, including simplicial sets and cartesian cubical sets.
\end{abstract}

\begin{document}

\maketitle

\section*{Introduction}

The univalence axiom, proposed by Voevodsky around 2009, was a major breakthrough in dependent type theory, as it proposed a completely different approach to foundations of mathematics, based on homotopy theory.
The resulting system, called homotopy type theory, has since attracted significant attention from both mathematical and computer science communities.
Proposing an axiom, however, meant that the computational content of the resulting type theory may no longer be taken for granted, and, indeed, Voevodsky's homotopy canonicity conjecture, asserting, loosely speaking, that the univalence axiom ``computes,'' became the central open problem in the field.
The first solution and a major breakthrough in this enterprise was the introduction of cubical type theory \cite{cchm15}, an extension of homotopy type theory with strong computational properties \cite{chs22} in which the univalence axiom is simply a theorem.

At present, cubical type theories, cubical proof assistants, and cubical libraries abound; examples include: \cite{cchm15, op16, ang+21, bir+16, ahh18, ch21, ch19, akgb24, ak18, vma19, achhs18, lops18}.
This list is, of course, far from exhaustive and we apologize to those whose contributions were omitted.
With a few notable exceptions, these developments are primarily driven by syntactic considerations, namely providing a computational interpretation of the univalence axiom.
In terms of semantics, these type theories can be interpreted in a variety of models based on cubical sets, including the original model in de Morgan cubical sets of \cite{cchm15}, and the cartesian cubical sets of \cite{ang+21}.
\citeauthor{op16} \cite{op16}, subsequently joined by Licata and Spitters \cite{lops18}, abstracted the key elements of the original model of cubical type theory in the category of de Morgan cubical sets to give a set of axioms required of a topos to admit an interpretation of cubical type theory.

One notable omission in all of these developments is an abstract definition of a model of cubical type theory, which we aim to rectify in the present work.
To this end, we propose yet another cubical type theory via the framework of universe categories, which closely resembles that of Uemura's categories with representable morphisms, and study its semantics.
We call it, self-deprecatingly, the naive cubical type theory.
At its core, the naive cubical type theory closely resembles the CCHM cubical type theory, but with three notable exceptions: first, it contains more cofibrations, second, our treatment of the filling operation is different, and third, we require a strong homotopy equivalence extension operation.

Regarding the first point, all sections of fibrations are cofibrations in our approach.
This can be justified semantically, since a large class of presheaf models, such as simplicial sets, satisfy this stronger condition.
As for our treatment of filling, it reverses the now-standard approach of first defining filling and then declaring the fibrations to be those maps that admit the filling operation.
In our framework, fibrations come first, and filling is a structure that needs to be selected afterwards.
Combined, these allow us to prove that path types satisfy the definitional computation rules, making them the identity types.
This differs from the treatment of \cite{op16,cchm15}, where path types can be used to construct identity types (by equipping with the data of a cofibration encoding where the path is constant), but are not themselves the identity types.

Finally, we address our strengthening of the equivalence extension property.
In \cite{cchm15} and other approaches, one typically extends a weak equivalence along a cofibration, whereas we require that homotopy equivalences (including all homotopy data) extend in this way.
There are numerous benefits to this approach.
First, it allows us to prove a corresponding strengthening of the univalence axiom, introduced in \cite{axm-univalence}, called \emph{pointed univalence}, whereas the framework of \cite{cchm15} only allows one to prove what is called there \emph{book univalence}.
In pointed univalence, one has that the inverse of the canonical map from the identity type to the type of equivalences takes the identity equivalence to the reflexivity term.
As explained in \cite{axm-univalence}, pointedness axiomatizes the fact that the map selecting the trivial equivalence is in the left class generated by fibrations.
And as the left class of fibrations generally axiomatizes pattern matching, this strengthening suggests the possibility of using pattern matching on equivalences, making it particularly user friendly.

We propose two versions of naive cubical type theory: one in which the filling operation satisfies the computation rule up to definitional equality, and one in which the rule holds only propositionally.
All of our models are models of the latter, but we do not know whether they also satisfy the stronger former condition.
That said, the lack of definitional computation for filling does not affect homotopy canonicity of the resulting type theory \cite{chs22}.

Our approach to models is based on the notion of a Quillen model category \cite{qui67,hov99}, and we provide a fairly general approach of turning Cisinski model structures \cite{cis06} on presheaf categories over elegant Reedy categories \cite{bergner-rezk} into models of naive cubical type theory.
We instantiate our approach with a model in simplicial sets and two cubical models, building on the work of Cisinski: minimal cubical sets and cubical sets with connections.
All of these have been thoroughly studied in classical homotopy theory \cite{gj99,kan:abstract-homotopy,kan:abstract-homotopy-ii,cis06,cis14,mal09,dkls24,carranza-kapulkin:homotopy-groups-of-cubical-sets}.
Finally, building on the work of Awodey, we show that the category of cartesian cubical sets \cite{awo23b} also admits a model of naive cubical type theory.

Our approach to semantics differs from the seminal work of \citeauthor{op16} \cite{op16}, and their collaborators \cite{lops18}, although many of our models are subsumed by theirs.
First off, while \cite{op16,lops18} work in the internal language of a topos, our work is more explicit in dealing with the categorical structure.
Second, we start off from a general definition of a model before turning our attention to toposes specifically.
Third, we more prominently rely on the use of Quillen model categories, which in \cite{op16,lops18} only feature in the background.
Lastly, the present paper owes a great deal to the work of \citeauthor{sat17} \cite{sat17} on the equivalence extension property, which we describe in \cref{subsec:wk-eqv-ext}.

For us, an important reason to consider such a semantic approach to cubical type theory is to extend the definition of Morita equivalence to cubical type theories.
Let us recall here informally that two dependent type theories are Morita equivalent if their categories of models are Quillen equivalent as left semi-model categories.
The notion of Morita equivalence has proven quite fruitful in investigating the expressive power of different type theories, e.g., the well-known result of Hofmann on conservativity of extensional type theory over intensional type theory with function extensionality and uniqueness of identity proofs can be strengthened to saying that the two are in fact Morita equivalent.
With that, it is natural to ask whether different cubical type theories are Morita equivalent, or, as posed in \cite[Section 8]{uem23}, whether they are Morita equivalent to homotopy type theory.
The first step in such an enterprise must inevitably be understanding what is an abstract model of cubical type theory and what structure the category of all such models possesses, with the former question being explicitly addressed in the present paper.

This paper is organized as follows.
We begin in \cref{sec:background}, by reviewing the necessary background on polynomial functors and universe models of dependent type theory.
Then, in \cref{sec:univ-axioms}, we extend that framework to provide a universe category axiomatization of cubical type theory (with the exception of the computation rules) up to definitional or propositional equality for the filling operation which is discussed in \cref{sec:filling-computation}.
We then turn our attention to the homotopy isomorphism extension operation (or the gluing operation, in the language of \cite{cchm15}), which we address in \cref{subsec:hiso-ext-op}.
In \cref{sec:naive-ctt}, we formulate the naive cubical type theory in the regime of universe categories and describe how it compares to the work of \cite{cchm15,op16}.
We briefly describe the equivalence extension property in \cref{sec:equiv-extension} before concluding, in \cref{sec:presheaf-models}, with a discussion of presheaf models.

%%% Local Variables:
%%% TeX-master: "./main.tex"
%%% TeX-engine: default
%%% End:

\tableofcontents

\section{Background}\label{sec:background}
In this paper we work in the framework of universe category models of type theory.
In particular, to model cubical type theory, we require three universal maps to model context extension by a type, a cofibration and an interval variable.
Modelling context extension by a dependent type in the framework of universe
categories is well established.
Therefore, we begin by recalling universe category axiomatisations of dependent
type theory and some of the background with polynomial functors required to work
with them.

\subsection{Polynomial Functors}
In order to model type constructors parameterised by dependent types, we use
polynomial functors.
\begin{definition}\label{def:pushforward}
  A map $p \colon E \to B$ in a category $\bC$ is \emph{exponentiable} if
  pullbacks along $p$ exist and the pullback functor
  $p^* \colon \sfrac{\bC}{B} \to \sfrac{\bC}{E}$ admits a right adjoint
  $p_* \colon \sfrac{\bC}{E} \to \sfrac{\bC}{B}$ which refer to as the
  \emph{pushforwards}.
\end{definition}

Using pushforwards, one can define the polynomial functor associated with an
exponentiable map.
\begin{definition}\label{def:expn-polynomial}
  The \emph{polynomial functor} associated with a map $p \colon E \to B$ in a
  locally cartesian closed category $\bC$, whose pushforward is denoted $p_*$,
  is the composite functor
  \begin{equation*}
    \bP_P \coloneqq \left(\bC \xrightarrow{E \times -} \sfrac{\bC}{E} \xrightarrow{p_*} \sfrac{\bC}{B} \xrightarrow{B_!} \bC\right)
  \end{equation*}
  where the last map $\sfrac{\bC}{B} \xrightarrow{B_!} \bC$ just forgets away
  the base.
\end{definition}

Using the polynomial functor associated with an exponential map, one constructs
the generic post-composite with an exponentiable map.
\begin{construction}\label{constr:gen-comp}
  Let $p \colon E \to B$ and $p' \colon E' \to B'$ be two maps in a locally
  cartesian closed category $\bC$ and let $p_*$ be a pushforward along $p$.
  Then, the map
  \begin{equation*}
   \GenComp(p',p) \colon \ev^*(E \times E') \to p^*p_*(E \times B') \to p_*(E \times B')
  \end{equation*}
  is constructed as follows as the composite of the two pullbacks of $E' \to B'$
  and $E \to B$ respectively.
  \begin{equation*}
    % https://q.uiver.app/#q=WzAsOSxbMiw0LCJCIl0sWzIsMywiRSJdLFsxLDIsIkUgXFx0aW1lcyBCJyJdLFsxLDEsIkUgXFx0aW1lcyBFJyJdLFsyLDIsIkInIl0sWzIsMSwiRSciXSxbMCwzLCJwXyooRSBcXHRpbWVzIEInKSJdLFswLDEsInBeKnBfKihFIFxcdGltZXMgQicpIl0sWzAsMCwiXFxldl4qKEUgXFx0aW1lcyBFJykiXSxbMywyXSxbMSwwLCJwIl0sWzIsMSwiXFxwcm9qXzEiLDEseyJsYWJlbF9wb3NpdGlvbiI6MzB9XSxbNSw0XSxbMiw0XSxbMyw1XSxbNiwwXSxbNywyLCJcXGV2IiwxXSxbNywxLCIiLDIseyJjdXJ2ZSI6M31dLFs3LDZdLFs4LDddLFs4LDNdLFs4LDIsIiIsMCx7InN0eWxlIjp7Im5hbWUiOiJjb3JuZXIifX1dLFs3LDAsIiIsMCx7InN0eWxlIjp7Im5hbWUiOiJjb3JuZXIifX1dLFszLDQsIiIsMCx7InN0eWxlIjp7Im5hbWUiOiJjb3JuZXIifX1dXQ==
    \begin{tikzcd}[cramped]
      {\ev^*(E \times E')} \\
      {p^*p_*(E \times B')} & {E \times E'} & {E'} \\
      & {E \times B'} & {B'} \\
      {p_*(E \times B')} && E \\
      && B
      \arrow[from=1-1, to=2-1]
      \arrow[from=1-1, to=2-2]
      \arrow["\ev"{description}, from=2-1, to=3-2]
      \arrow[from=2-1, to=4-1]
      \arrow[curve={height=18pt}, from=2-1, to=4-3]
      \arrow[from=2-2, to=2-3]
      \arrow[from=2-2, to=3-2]
      \arrow[from=2-3, to=3-3]
      \arrow[from=3-2, to=3-3]
      \arrow["{\proj_1}"{description, pos=0.3}, from=3-2, to=4-3]
      \arrow[from=4-1, to=5-3]
      \arrow["p", from=4-3, to=5-3]
      \arrow["\lrcorner"{anchor=center, pos=0.15, scale=1.5}, draw=none, from=1-1, to=3-2]
      \arrow["\lrcorner"{anchor=center, pos=0.15, scale=1.5}, draw=none, from=2-1, to=5-3]
      \arrow["\lrcorner"{anchor=center, pos=0.15, scale=1.5}, draw=none, from=2-2, to=3-3]
    \end{tikzcd}
  \end{equation*}
\end{construction}

The generic composite constructed in the above \Cref{constr:gen-comp} lives up
to its name.
\begin{lemma}[{\cite[Lemma 1.4]{axm-univalence}}]\label{lem:gen-comp}
  Let $p \colon E \to B \in \bC$ be an exponentiable map in a category $\bC$
  where all the products with $E$ exist.
  Take $p' \colon E' \to B' \in \bC$ to be just any map.

  Suppose that one has pullbacks as on the left.
  Then the composite $X_2 \to X_1 \to X_0$ arises as a pullback of
  $\GenComp(p,p') \colon \ev^*(E \times E') \to p^*p_*(E \times B') \to p_*(E
  \times B')$ as on the right.
  \begin{center}
    \begin{minipage}{0.45\linewidth}
      \begin{equation*}
        % https://q.uiver.app/#q=WzAsNyxbMCwxLCJFXzEiXSxbMCwyLCJFXzAiXSxbMSwzLCJcXE1jVSJdLFsxLDIsIlxcdE1jVSciXSxbMSwxLCJcXHRNY1UiXSxbMSwwLCJcXE1jVSciXSxbMCwwLCJFXzIiXSxbMSwyLCJcXGNlaWx7RV8xfSIsMl0sWzMsMl0sWzAsMSwiXFxwaV8xIiwyXSxbMCwzLCJcXHdpZGV0aWxkZXtcXGNlaWx7RV8xfX0iLDFdLFs1LDRdLFswLDQsIlxcY2VpbHtFXzJ9IiwxXSxbNiwwLCJcXHBpXzIiLDJdLFs2LDUsIlxcd2lkZXRpbGRle1xcY2VpbHtFXzJ9fSJdLFswLDIsIiIsMCx7InN0eWxlIjp7Im5hbWUiOiJjb3JuZXIifX1dLFs2LDQsIiIsMCx7InN0eWxlIjp7Im5hbWUiOiJjb3JuZXIifX1dXQ==
        \begin{tikzcd}[cramped]
          {X_2} & {E'} \\
          {X_1} & B' \\
          {X_0} & {E} \\
          & B
          \arrow["{\widetilde{\ceil{X_2}}}", from=1-1, to=1-2]
          \arrow["{q_2}"', from=1-1, to=2-1]
          \arrow[from=1-2, to=2-2]
          \arrow["{\ceil{X_2}}"{description}, from=2-1, to=2-2]
          \arrow["{q_1}"', from=2-1, to=3-1]
          \arrow["{\widetilde{\ceil{X_1}}}"{description}, from=2-1, to=3-2]
          \arrow["{\ceil{X_1}}"', from=3-1, to=4-2]
          \arrow[from=3-2, to=4-2]
          \arrow["\lrcorner"{anchor=center, pos=0.15, scale=1.5}, draw=none, from=1-1, to=2-2]
          \arrow["\lrcorner"{anchor=center, pos=0.15, scale=1.5}, draw=none, from=2-1, to=4-2]
        \end{tikzcd}
      \end{equation*}
    \end{minipage}
    \begin{minipage}{0.45\linewidth}
      \begin{equation*}
        % https://q.uiver.app/#q=WzAsNixbMCwwLCJFXzIiXSxbMCwxLCJFXzEiXSxbMCwyLCJFXzAiXSxbMSwyLCJcXHBpXyooXFx0TWNVIFxcdGltZXMgXFxNY1UnKSJdLFsxLDAsIlxcZXZeKihcXHRNY1UgXFx0aW1lcyBcXHRNY1UnKSJdLFsxLDEsIlxccGleKlxccGlfKihcXHRNY1UgXFx0aW1lcyBcXE1jVScpIl0sWzAsMSwiXFxwaV8yIiwyXSxbMSwyLCJcXHBpXzEiLDJdLFsyLDMsIlxcY2VpbHtFXzF9LlxcY2VpbHtFXzJ9IiwyLHsic3R5bGUiOnsiYm9keSI6eyJuYW1lIjoiZGFzaGVkIn19fV0sWzQsNV0sWzUsM10sWzAsNCwiIiwyLHsic3R5bGUiOnsiYm9keSI6eyJuYW1lIjoiZGFzaGVkIn19fV0sWzAsMywiIiwwLHsic3R5bGUiOnsibmFtZSI6ImNvcm5lciJ9fV0sWzEsNSwiIiwxLHsic3R5bGUiOnsiYm9keSI6eyJuYW1lIjoiZGFzaGVkIn19fV0sWzEsMywiIiwxLHsic3R5bGUiOnsibmFtZSI6ImNvcm5lciJ9fV1d
        \begin{tikzcd}[cramped]
          {X_2} & {\ev^*(E' \times E')} \\
          {X_1} & {p^*p_*(E \times B')} \\
          {X_0} & {p_*(E \times B')}
          \arrow[dashed, from=1-1, to=1-2]
          \arrow["{q_2}"', from=1-1, to=2-1]
          \arrow[from=1-2, to=2-2]
          \arrow[dashed, from=2-1, to=2-2]
          \arrow["{q_1}"', from=2-1, to=3-1]
          \arrow[from=2-2, to=3-2]
          \arrow["{\ceil{X_1}.\ceil{X_2}}"', dashed, from=3-1, to=3-2]
          \arrow["\lrcorner"{anchor=center, pos=0.05, scale=1.5}, draw=none, from=1-1, to=3-2]
          \arrow["\lrcorner"{anchor=center, pos=0.05, scale=1.5}, draw=none, from=2-1, to=3-2]
        \end{tikzcd}
      \end{equation*}
    \end{minipage}
  \end{center}
  such that
  $(X_0 \xrightarrow{\ceil{X_1}.\ceil{X_2}} p_*(E \times B) \to B)
  = (X_0 \xrightarrow{\ceil{X_1}} B)$
  and
  $(X_0 \xrightarrow{\ceil{X_1}.\ceil{X_2}} p_*(E \times B'))^\dagger
  = (X_1 \xrightarrow{(\widetilde{\ceil{X_1}}, \ceil{X_2})} E \times B')$.
  \def\endingmark{\qedsymbol}
\end{lemma}

Finally, we often speak of pullback-Homs, which are the internalisations of the
map taking a diagonal to the lifting problem it solves.
We fix some notation for this.
\begin{definition}\label{def:pullback-Hom}
  Let $\bC$ be a cartesian closed category.
  Fix maps $i$ and $p$ in $\bC$ as follows.
  \begin{equation*}
    % https://q.uiver.app/#q=WzAsNCxbMCwwLCJVIl0sWzAsMSwiViJdLFsxLDAsIkUiXSxbMSwxLCJCIl0sWzAsMSwiaSIsMl0sWzIsMywicCJdXQ==
    \begin{tikzcd}[cramped]
      U & E \\
      V & B
      \arrow["i"', from=1-1, to=2-1]
      \arrow["p", from=1-2, to=2-2]
    \end{tikzcd}
  \end{equation*}
  We define the map and objects
  \begin{equation*}
    \Diag_\bC(i,p) \eqqcolon [V,E]_\bC
    \xrightarrow{\reallywidehat{[i,p]}_\bC} \Sq_\bC(i,p) \in \bC
  \end{equation*}
  to be the dashed pullback-Hom map as below.
  \begin{equation*}
    % https://q.uiver.app/#q=WzAsNSxbMSwyLCJbVSxFXSJdLFsyLDEsIltWLEJdIl0sWzIsMiwiW1UsQl0iXSxbMSwxLCJcXFNxX1xcYkMoaSxwKSJdLFswLDAsIltWLEVdIl0sWzEsMiwiW2ksQl0iXSxbMCwyLCJbVSxwXSIsMl0sWzMsMF0sWzMsMV0sWzMsMiwiIiwxLHsic3R5bGUiOnsibmFtZSI6ImNvcm5lciJ9fV0sWzQsMCwiIiwxLHsiY3VydmUiOjJ9XSxbNCwxLCIiLDEseyJjdXJ2ZSI6LTJ9XSxbNCwzLCJcXHJlYWxseXdpZGVoYXR7W2kscF19IiwxLHsic3R5bGUiOnsiYm9keSI6eyJuYW1lIjoiZGFzaGVkIn19fV1d
    \begin{tikzcd}[cramped]
      {[V,E]_\bC} \\
      & {\Sq_\bC(i,p)} & {[V,B]} \\
      & {[U,E]} & {[U,B]}
      \arrow["{\reallywidehat{[i,p]}_\bC}"{description}, dashed, from=1-1, to=2-2]
      \arrow[curve={height=-12pt}, from=1-1, to=2-3]
      \arrow[curve={height=12pt}, from=1-1, to=3-2]
      \arrow[from=2-2, to=2-3]
      \arrow[from=2-2, to=3-2]
      \arrow["\lrcorner"{anchor=center, pos=0.15, scale=1.5}, draw=none, from=2-2, to=3-3]
      \arrow["{[i,B]}", from=2-3, to=3-3]
      \arrow["{[U,p]}"', from=3-2, to=3-3]
    \end{tikzcd}
  \end{equation*}

  If $\bC$ is locally cartesian closed and $C \in \bC$, then for
  $i \colon U \to V \in \sfrac{\bC}{C}$ and
  $p \colon E \to B \in \sfrac{\bC}{C}$, we denote by
  \begin{equation*}
    \Diag_C(i,p) \eqqcolon [V,E]_C \xrightarrow{\reallywidehat{[i,p]}_C} \Sq_C(i,p) \in \sfrac{\bC}{C}
  \end{equation*}
  the above definition applied to $\sfrac{\bC}{C}$.
\end{definition}

These internal commutative squares and diagonals are stable under pullback.

\begin{lemma}\label{lem:pullback-Hom-stable}
  Let $\bC$ be a locally cartesian closed category.
  Fix a map $f \colon D \to C \in \bC$ along with maps
  $i \colon U \to V \in \sfrac{\bC}{C}$ and
  $p \colon E \to B \in \sfrac{\bC}{C}$.
  Then, one obtains an iterated pullback square as follows.
  \begin{equation*}
    % https://q.uiver.app/#q=WzAsNixbMSwwLCJcXERpYWdfQyhpLHApIl0sWzEsMSwiXFxTcV9DKGkscCkiXSxbMSwyLCJDIl0sWzAsMiwiRCJdLFswLDEsIlxcU3FfRChmXippLGZeKnApIl0sWzAsMCwiXFxTcV9EKGZeKmksZl4qcCkiXSxbMywyLCJmIiwyXSxbMSwyXSxbMCwxXSxbNCwzXSxbNCwxXSxbNCwyLCIiLDIseyJzdHlsZSI6eyJuYW1lIjoiY29ybmVyIn19XSxbNSw0XSxbNSwwXSxbNSwxLCIiLDAseyJzdHlsZSI6eyJuYW1lIjoiY29ybmVyIn19XV0=
    \begin{tikzcd}[cramped, row sep=small, column sep=small]
      {\Sq_D(f^*i,f^*p)} & {\Diag_C(i,p)} \\
      {\Sq_D(f^*i,f^*p)} & {\Sq_C(i,p)} \\
      D & C
      \arrow[from=1-1, to=1-2]
      \arrow[from=1-1, to=2-1]
      \arrow[from=1-2, to=2-2]
      \arrow[from=2-1, to=2-2]
      \arrow[from=2-1, to=3-1]
      \arrow["\lrcorner"{anchor=center, pos=0.15, scale=1.5}, draw=none, from=1-1, to=2-2]
      \arrow["\lrcorner"{anchor=center, pos=0.15, scale=1.5}, draw=none, from=2-1, to=3-2]
      \arrow[from=2-2, to=3-2]
      \arrow["f"', from=3-1, to=3-2]
    \end{tikzcd}
  \end{equation*}
\end{lemma}
\begin{proof}
  This follows from the standard fact that pullbacks preserve pullbacks and
  internal-Homs.
\end{proof}

\subsection{Universe Category Models of Intensional Type Theory}\label{subsec:itt-types}
Universe category models \cite{voe15} of type theory are general high
power-to-weight ratio frameworks that can make models of dependent type theory
where context extension is modelled as certain pullback-stable projection maps
arise as special cases of them.
As previous mentioned, we will also use them to axiomatise cubical type theory,
which will be equipped with three universes.
Because cubical type theory shares numerous similar type formers with
intensional type theory, we first recall how various type formers are
axiomatised in universe categories.

\begin{definition}[{\cite[Definition 2.1]{voe15}}]\label{def:univ-struct}
  Let $\bC$ be a locally cartesian closed category.
  A \emph{universe structure} on a map $\pi \colon \tMcU \to \McU$ is a specific
  choice of a right adjoint $\var \colon \sfrac{\bC}{\McU} \to \sfrac{\bC}{\tMcU}$
  \begin{equation*}
    % https://q.uiver.app/#q=WzAsMixbMCwwLCJcXHNmcmFje1xcYkN9e1xcTWNVfSJdLFsyLDAsIlxcc2ZyYWN7XFxiQ317XFx0TWNVfSJdLFswLDEsIlxcdmFyIiwxLHsib2Zmc2V0IjotMiwic3R5bGUiOnsiYm9keSI6eyJuYW1lIjoiZGFzaGVkIn19fV0sWzEsMCwiXFxwaV8hIiwwLHsib2Zmc2V0IjotMn1dLFszLDIsIiIsMCx7ImxldmVsIjoxLCJzdHlsZSI6eyJuYW1lIjoiYWRqdW5jdGlvbiJ9fV1d
    \begin{tikzcd}[cramped]
      {\sfrac{\bC}{\McU}} && {\sfrac{\bC}{\tMcU}}
      \arrow[""{name=0, anchor=center, inner sep=0}, "\var"{}, shift left=2, dashed, from=1-1, to=1-3]
      \arrow[""{name=1, anchor=center, inner sep=0}, "{\pi_!}", shift left=2, from=1-3, to=1-1]
      \arrow["\dashv"{anchor=center, rotate=90}, draw=none, from=1, to=0]
    \end{tikzcd}
  \end{equation*}
  of $\pi_! \colon \sfrac{\bC}{\tMcU} \to \sfrac{\bC}{\McU}$ the
  post-composition functor (i.e. just a choice of pullbacks along $\pi$).

  For each $A \colon \Gamma \to \McU \in \sfrac{\bC}{\McU}$, we denote by
  $\pi_A$ the counit at $A$ and $\Gamma.A$ the domain of $\pi_A$
  so that one has a pullback square
  \begin{equation*}
    % https://q.uiver.app/#q=WzAsNCxbMCwxLCJcXEdhbW1hIl0sWzAsMCwiXFxHYW1tYS5BIl0sWzEsMCwiXFx0TWNVIl0sWzEsMSwiXFxNY1UiXSxbMiwzXSxbMSwwLCJcXHBpX0EiLDJdLFswLDMsIkEiLDJdLFsxLDIsIlxcdmFyX0EiXSxbMSwzLCIiLDEseyJzdHlsZSI6eyJuYW1lIjoiY29ybmVyIn19XV0=
    \begin{tikzcd}[cramped]
      {\Gamma.A} & \tMcU \\
      \Gamma & \McU
      \arrow["{\var_A}", from=1-1, to=1-2]
      \arrow["{\pi_A}"', from=1-1, to=2-1]
      \arrow["\lrcorner"{anchor=center, pos=0.15, scale=1.5}, draw=none, from=1-1, to=2-2]
      \arrow[from=1-2, to=2-2]
      \arrow["A"', from=2-1, to=2-2]
    \end{tikzcd}
  \end{equation*}
  The map $\pi_A$ is also the \emph{selected pullback} by the universe
  structure.

  The map $\pi \colon \tMcU \to \McU$ equipped with such a universe structure is
  called a \emph{universal map}.
  A \emph{universe structure} on a category $\bC$ consists of a choice of a
  distinguished universal map $\tMcU \to \McU$ and a universe structure on it.
\end{definition}

Given a universal map $\tMcU \to \McU$, one may wish to talk about arbitrary
pullbacks of it, as opposed to the specific pullback maps selected by its
universe structure.
We fix some vocabulary for this purpose.
\begin{definition}\label{def:univ-fibration}
  Let $\pi \colon \tMcU \to \McU$ be a universal map in a locally cartesian
  closed category $\bC$.

  The \emph{$\pi$-name} and \emph{$\pi$-point} of a map $E \to B$ are maps
  $\ceil{E}$ and $\widetilde{\ceil{E}}$ as below such that one has a pullback
  square as follows.
  \begin{equation*}
    % https://q.uiver.app/#q=WzAsNCxbMCwxLCJCIl0sWzAsMCwiRSJdLFsxLDEsIlxcTWNVIl0sWzEsMCwiXFx0TWNVIl0sWzMsMiwiIiwwLHsic3R5bGUiOnsiaGVhZCI6eyJuYW1lIjoiZXBpIn19fV0sWzAsMiwiXFxjZWlse0V9IiwyXSxbMSwwLCIiLDIseyJzdHlsZSI6eyJoZWFkIjp7Im5hbWUiOiJlcGkifX19XSxbMSwzLCJcXHdpZGV0aWxkZXtcXGNlaWx7RX19Il0sWzEsMiwiIiwxLHsic3R5bGUiOnsibmFtZSI6ImNvcm5lciJ9fV1d
    \begin{tikzcd}[cramped]
      E & \tMcU \\
      B & \McU
      \arrow["{\widetilde{\ceil{E}}}", from=1-1, to=1-2]
      \arrow[from=1-1, to=2-1]
      \arrow["\lrcorner"{anchor=center, pos=0.15, scale=1.5}, draw=none, from=1-1, to=2-2]
      \arrow[from=1-2, to=2-2]
      \arrow["{\ceil{E}}"', from=2-1, to=2-2]
    \end{tikzcd}
  \end{equation*}
  Together, the $\pi$-(name, point) pair $(\ceil{E}, \widetilde{\ceil{E}})$ forms
  a \emph{$\pi$-fibrancy structure} for $E \to B$.
  Maps that can be equipped with a $\pi$-fibrancy structure are called
  \emph{$\pi$-fibrations} and are denoted $E \twoheadrightarrow B$.
  As an object in the slice over $B$, we say
  $E \twoheadrightarrow B \in \sfrac{\bC}{B}$ is a $\pi$-fibrant object.

  Although $\pi$-fibrancy structures are not necessarily unique, given
  $A \colon \Gamma \to \McU$, the selected pullback $\Gamma.A \to \Gamma$ admits
  a \emph{canonical} $\pi$-fibrancy structure $(A,\var_A)$.
  Furthermore, given a $\pi$-fibration $E \twoheadrightarrow B$ with fibrancy
  structure $(\ceil{E},\widetilde{\ceil{E}})$, its \emph{canonical $\pi$-fibrant
    replacement} is the isomorphic map
  $\pi_{\ceil{E}} \colon B.\ceil{E} \cong E \to B$.
\end{definition}

We now recall how a universal map is equipped with various logical structures.

\subsubsection{Id-type Structure}
$\Id$-type structures are syntactical version of a generic fibred path object
for fibrations.
Their axiomatisation is described in \cite{voe15a}, which we now recall.
\begin{definition}[{\cite[Definition 2.7]{voe15a}}]\label{def:pre-id-type}
  Let $\pi \colon \tMcU \to \McU$ be a universal map in a locally cartesian
  closed category $\bC$.

  A \emph{pre-$\Id$-type structure} (also known as a
  \emph{$\textsf{J1}$-structure} in \cite[Definition 2.7]{voe15a}) on
  $\tMcU \to \McU$ consists of a pair of dashed maps $(\Id,\ceil{\refl})$ as
  below from the diagonal $\tMcU \to \tMcU \times_\McU \tMcU$ to
  $\tMcU \to \McU$.
  \begin{equation*}\label{eqn:pre-Id-def}\tag{\textsc{pre-$\Id$-def}}
    % https://q.uiver.app/#q=WzAsNCxbMCwwLCJcXHRNY1UiXSxbMCwxLCJcXHRNY1UgXFx0aW1lc19cXE1jVSBcXHRNY1UiXSxbMSwxLCJcXE1jVSJdLFsxLDAsIlxcdE1jVSJdLFszLDJdLFsxLDIsIlxcSWQiLDIseyJzdHlsZSI6eyJib2R5Ijp7Im5hbWUiOiJkYXNoZWQifX19XSxbMCwxXSxbMCwzLCJcXGNlaWx7XFxyZWZsfSIsMCx7InN0eWxlIjp7ImJvZHkiOnsibmFtZSI6ImRhc2hlZCJ9fX1dXQ==&macro_url=https%3A%2F%2Fgist.githubusercontent.com%2Flim495062%2F61b94af9ef95c1c7b0763c937de29c2b%2Fraw%2F254aa652b96f7d1c6dece2f3eba929431bb0adfe%2Facmhwmacros.sty
    \begin{tikzcd}[cramped]
      \tMcU & \tMcU \\
      {\tMcU \times_\McU \tMcU} & \McU
      \arrow["{\ceil{\refl}}", dashed, from=1-1, to=1-2]
      \arrow[from=1-1, to=2-1]
      \arrow[from=1-2, to=2-2]
      \arrow["\Id"', dashed, from=2-1, to=2-2]
    \end{tikzcd}
  \end{equation*}
\end{definition}

\begin{construction}[{\cite[Equation (7)]{voe15a}}]\label{constr:J-prob}
  Let $\pi \colon \tMcU \to \McU$ be a universal map in a locally cartesian
  closed category $\bC$ equipped with a pre-$\Id$-type structure
  $(\Id,\ceil{\refl})$ as in \Cref{eqn:pre-Id-def}.

  Denote by
  $\ev_\partial \colon \Id_\McU(\tMcU) \twoheadrightarrow \tMcU \times_\McU
  \tMcU$ the selected pullback of $\pi$ along
  $\Id \colon \tMcU \times_\McU \tMcU \to \McU$ by the universe structure so
  that there is a map $\refl \colon \tMcU \hookrightarrow \Id_\McU(\tMcU)$ into
  the pullback, as below.
  \begin{equation*}\label{eqn:Id-def}\tag{\textsc{$\Id$-def}}
    % https://q.uiver.app/#q=WzAsNSxbMiwxLCJcXHRNY1UiXSxbMiwyLCJcXE1jVSJdLFsxLDIsIlxcdE1jVSBcXHRpbWVzX1xcTWNVIFxcdE1jVSJdLFswLDAsIlxcdE1jVSJdLFsxLDEsIlxcSWReKlxcdE1jVSJdLFsyLDEsIlxcSWQiLDJdLFszLDIsIlxcRGVsdGEiLDIseyJjdXJ2ZSI6Mn1dLFszLDAsIlxcY2VpbHtcXHJlZmx9IiwwLHsiY3VydmUiOi0yfV0sWzAsMV0sWzQsMl0sWzQsMF0sWzQsMSwiIiwwLHsic3R5bGUiOnsibmFtZSI6ImNvcm5lciJ9fV0sWzMsNCwiXFxyZWZsIiwxLHsic3R5bGUiOnsiYm9keSI6eyJuYW1lIjoiZGFzaGVkIn19fV1d
    \begin{tikzcd}[cramped]
      \tMcU \\
      & {\Id_\McU(\tMcU)} & \tMcU \\
      & {\tMcU \times_\McU \tMcU} & \McU
      \arrow["\refl"{description}, hook, dashed, from=1-1, to=2-2]
      \arrow["{\ceil{\refl}}", curve={height=-12pt}, from=1-1, to=2-3,dashed]
      \arrow["\Delta"', curve={height=12pt}, from=1-1, to=3-2]
      \arrow["{\widetilde{\Id}}"{description}, from=2-2, to=2-3,dashed]
      \arrow["{\ev_\partial}"{description}, from=2-2, to=3-2, two heads]
      \arrow["\lrcorner"{anchor=center, pos=0.15, scale=1.5}, draw=none, from=2-2, to=3-3]
      \arrow[from=2-3, to=3-3, two heads]
      \arrow["\Id"', from=3-2, to=3-3,dashed]
    \end{tikzcd}
  \end{equation*}
  Then, over $\McU$, we have the pair of maps
  $\refl \colon \tMcU \hookrightarrow \Id_\McU(\tMcU)$ and the rebased map
  $\tMcU \times \McU \to \McU \times \McU$, as below.
  \begin{equation*}
    % https://q.uiver.app/#q=WzAsNSxbMCwwLCJcXHRNY1UiXSxbMCwxLCJcXElkX1xcTWNVKFxcdE1jVSkiXSxbMiwwLCJcXE1jVSBcXHRpbWVzIFxcdE1jVSJdLFsyLDEsIlxcTWNVIFxcdGltZXMgXFx0TWNVIl0sWzEsMiwiXFxNY1UiXSxbMCwxLCJcXHJlZmwiLDIseyJzdHlsZSI6eyJ0YWlsIjp7Im5hbWUiOiJob29rIiwic2lkZSI6InRvcCJ9fX1dLFsyLDMsIlxcTWNVIFxcdGltZXMgXFxwaSJdLFsxLDQsIlxcSWQgXFxjZG90IFxcZXZfXFxwYXJ0aWFsIiwyXSxbMyw0LCJcXHByb2pfMSJdXQ==&macro_url=https%3A%2F%2Fgist.githubusercontent.com%2Flim495062%2F61b94af9ef95c1c7b0763c937de29c2b%2Fraw%2F254aa652b96f7d1c6dece2f3eba929431bb0adfe%2Facmhwmacros.sty
    \begin{tikzcd}[cramped, column sep=small]
      \tMcU && {\McU \times \tMcU} \\
      {\Id_\McU(\tMcU)} && {\McU \times \tMcU} \\
      & \McU
      \arrow["\refl"', hook, from=1-1, to=2-1]
      \arrow["{\McU \times \pi}", from=1-3, to=2-3]
      \arrow["{\Id \cdot \ev_\partial}"', from=2-1, to=3-2]
      \arrow["{\proj_1}", from=2-3, to=3-2]
    \end{tikzcd}
  \end{equation*}

  We then apply \Cref{def:pullback-Hom} to obtain the pullback-Hom map of $\refl$ with $\tMcU \times \pi$ in the
  slice over $\McU$ as follows
  \begin{equation*}\label{eqn:J-prob}\tag{\textsc{$\textsf{J}$-prob}}\small
    % https://q.uiver.app/#q=WzAsMixbMCwwLCJcXHRleHRzZntEaWFnfV9cXE1jVShcXHJlZmwsIFxcdE1jVSBcXHRpbWVzIFxccGkpIl0sWzAsMSwiXFx0ZXh0c2Z7U3F9X1xcTWNVKFxccmVmbCwgXFx0TWNVIFxcdGltZXMgXFxwaSkiXSxbMCwxXV0=&macro_url=https%3A%2F%2Fgist.githubusercontent.com%2Flim495062%2F61b94af9ef95c1c7b0763c937de29c2b%2Fraw%2F254aa652b96f7d1c6dece2f3eba929431bb0adfe%2Facmhwmacros.sty
    \left(\begin{tikzcd}[cramped]
        {\Diag_\McU(\refl, \tMcU \times \pi)} \\
        {\Sq_\McU(\refl, \tMcU \times \pi)}
        \arrow[from=1-1, to=2-1]
      \end{tikzcd}\right)
    =
    \left(
      % https://q.uiver.app/#q=WzAsMixbMCwwLCJbXFxJZF9cXE1jVShcXHRNY1UpLCBcXE1jVSBcXHRpbWVzIFxcTWNVXV9cXE1jVSJdLFswLDEsIltcXHRNY1UsIFxcTWNVIFxcdGltZXMgXFx0TWNVXV9cXE1jVSBcXHRpbWVzX3tbXFx0TWNVLCBcXE1jVSBcXHRpbWVzIFxcdE1jVV1fXFxNY1V9IFtcXElkX1xcTWNVKFxcdE1jVSksIFxcTWNVIFxcdGltZXMgXFxNY1VdX1xcTWNVIl0sWzAsMV1d
      \begin{tikzcd}[cramped]
        {[\Id_\McU(\tMcU), \McU \times \McU]_\McU} \\
        {[\tMcU, \McU \times \tMcU]_\McU \times_{[\tMcU, \McU \times \tMcU]_\McU} [\Id_\McU(\tMcU), \McU \times \McU]_\McU}
        \arrow[from=1-1, to=2-1]
      \end{tikzcd}
    \right)
  \end{equation*}
  It encodes the generic $\MsJ$-elimination lifting problem.
\end{construction}

\begin{definition}[{\cite[Definition 2.8]{voe15a}}]\label{def:J-elim}
  Let $\pi \colon \tMcU \to \McU$ be a universal map in a locally cartesian
  closed category $\bC$ equipped with a pre-$\Id$-type structure
  $(\Id,\ceil{\refl})$ as in \Cref{eqn:pre-Id-def}.

  A \emph{J-elimination structure} (also known as a \emph{$\textsf{J2}$-structure} in
  \cite[Definition 2.8]{voe15a}) on the pre-$\Id$-type structure
  ($\textsf{J1}$-structure) is a section of the pullback-Hom map of
  \Cref{eqn:J-prob}.
  \begin{equation*}\label{eqn:J-def}\tag{\textsc{$\textsf{J}$-def}}
    % https://q.uiver.app/#q=WzAsMixbMCwxLCJcXHRleHRzZntTcX1fXFxNc1UoXFxyZWZsLCBcXHRNY1UgXFx0aW1lcyBcXHBpKSJdLFswLDAsIlxcdGV4dHNme0RpYWd9X1xcTXNVKFxccmVmbCwgXFx0TWNVIFxcdGltZXMgXFxwaSkiXSxbMSwwLCIiLDAseyJvZmZzZXQiOi0xfV0sWzAsMSwiXFxNc0oiLDAseyJvZmZzZXQiOi0xLCJzdHlsZSI6eyJ0YWlsIjp7Im5hbWUiOiJob29rIiwic2lkZSI6InRvcCJ9fX1dXQ==
    \begin{tikzcd}[cramped]
      {\Diag_\McU(\refl, \tMcU \times \pi)} \\
      {\Sq_\McU(\refl, \tMcU \times \pi)}
      \arrow[shift left, from=1-1, to=2-1]
      \arrow["\MsJ", shift left, hook, from=2-1, to=1-1]
    \end{tikzcd}
  \end{equation*}
\end{definition}

\begin{remark}
  The above definition of the $\MsJ$-elimination structure is equivalently
  expressed by saying that $\MsJ$ is a stable lifting structure, in the sense of
  \cite[Definitions 1.4 and 3.1]{struct-lift}, of
  $\refl \colon \tMcU \hookrightarrow \Id_\McU(\tMcU)$ against
  $\McU \times \tMcU \to \McU \times \McU$, in the slice $\sfrac{\bC}{\McU}$.
  \begin{equation*}
    \MsJ \in \left(
      \begin{tikzcd}[cramped]
        \tMcU \ar[d,hook,"{\refl}"'] \\ \Id_\McU(\tMcU)
      \end{tikzcd}
      \fracsquareslash{\McU}
      \begin{tikzcd}[cramped]
        \McU\times\tMcU \ar[d, two heads] \\ \McU\times\McU
      \end{tikzcd}
    \right)
  \end{equation*}
\end{remark}

\begin{definition}\label{def:Id-type}
  Let $\pi \colon \tMcU \to \McU$ be a universal map in a locally cartesian
  closed category $\bC$.

  An \emph{$\Id$-type} structure (also known as a \emph{full $\MsJ$-structure}
  in \cite{voe15a}) on $\tMcU \to \McU$ consists of a pre-$\Id$-structure
  $(\Id,\ceil{\refl})$ as in \Cref{eqn:Id-def} along with a $\MsJ$-elimination
  structure as in \Cref{eqn:J-def}.
  \begin{center}
    \begin{minipage}{0.45\linewidth}
      \begin{equation*}\tag{\ref{eqn:Id-def}}
        % https://q.uiver.app/#q=WzAsNSxbMiwxLCJcXHRNY1UiXSxbMiwyLCJcXE1jVSJdLFsxLDIsIlxcdE1jVSBcXHRpbWVzX1xcTWNVIFxcdE1jVSJdLFswLDAsIlxcdE1jVSJdLFsxLDEsIlxcSWReKlxcdE1jVSJdLFsyLDEsIlxcSWQiLDJdLFszLDIsIlxcRGVsdGEiLDIseyJjdXJ2ZSI6Mn1dLFszLDAsIlxcY2VpbHtcXHJlZmx9IiwwLHsiY3VydmUiOi0yfV0sWzAsMV0sWzQsMl0sWzQsMF0sWzQsMSwiIiwwLHsic3R5bGUiOnsibmFtZSI6ImNvcm5lciJ9fV0sWzMsNCwiXFxyZWZsIiwxLHsic3R5bGUiOnsiYm9keSI6eyJuYW1lIjoiZGFzaGVkIn19fV1d
        \begin{tikzcd}[cramped]
          \tMcU \\
          & {\Id_\McU(\tMcU)} & \tMcU \\
          & {\tMcU \times_\McU \tMcU} & \McU
          \arrow["\refl"{description}, hook, dashed, from=1-1, to=2-2]
          \arrow["{\ceil{\refl}}", curve={height=-12pt}, from=1-1, to=2-3,dashed]
          \arrow["\Delta"', curve={height=12pt}, from=1-1, to=3-2]
          \arrow["{\widetilde{\Id}}"{description}, from=2-2, to=2-3,dashed]
          \arrow["{\ev_\partial}"{description}, from=2-2, to=3-2, two heads]
          \arrow["\lrcorner"{anchor=center, pos=0.15, scale=1.5}, draw=none, from=2-2, to=3-3]
          \arrow[from=2-3, to=3-3, two heads]
          \arrow["\Id"', from=3-2, to=3-3,dashed]
        \end{tikzcd}
      \end{equation*}
    \end{minipage}
    \begin{minipage}{0.45\linewidth}
      \begin{equation*}\tag{\text{\ref{eqn:J-def}}}
        % https://q.uiver.app/#q=WzAsMixbMCwxLCJcXHRleHRzZntTcX1fXFxNc1UoXFxyZWZsLCBcXHRNY1UgXFx0aW1lcyBcXHBpKSJdLFswLDAsIlxcdGV4dHNme0RpYWd9X1xcTXNVKFxccmVmbCwgXFx0TWNVIFxcdGltZXMgXFxwaSkiXSxbMSwwLCIiLDAseyJvZmZzZXQiOi0xfV0sWzAsMSwiXFxNc0oiLDAseyJvZmZzZXQiOi0xLCJzdHlsZSI6eyJ0YWlsIjp7Im5hbWUiOiJob29rIiwic2lkZSI6InRvcCJ9fX1dXQ==
        \begin{tikzcd}[cramped]
          {\Diag_\McU(\refl, \tMcU \times \pi)} \\
          {\Sq_\McU(\refl, \tMcU \times \pi)}
          \arrow[shift left, from=1-1, to=2-1]
          \arrow["\MsJ", shift left, hook, from=2-1, to=1-1]
        \end{tikzcd}
      \end{equation*}
    \end{minipage}
  \end{center}
\end{definition}

\begin{remark}
  Let $\pi \colon \tMcU \to \McU$ be a universal map in a locally cartesian
  closed category $\bC$ equipped with a (pre-)$\Id$-type structure.

  For each $\pi$-fibrant $E \twoheadrightarrow B$ with $\pi$-fibrancy structure
  $(\ceil{E},\widetilde{\ceil{E}})$ as on the left, we also denote by
  $\Id_{B}(E)$ and $\Id_{\ceil{E}}(\widetilde{\ceil{E}})$ and $\refl_E$ the
  correspondingly labelled maps as on the right.
  \begin{center}
    \begin{minipage}{0.40\linewidth}
      \begin{equation*}
        %       https://q.uiver.app/#q=WzAsNCxbMCwwLCJCIl0sWzAsMSwiXFx1bHtCfSJdLFsxLDEsIlxcTWNVIl0sWzEsMCwiXFx0TWNVIl0sWzMsMl0sWzEsMiwiXFx1bHtifSIsMl0sWzAsMV0sWzAsMywiYiJdLFswLDIsIiIsMSx7InN0eWxlIjp7Im5hbWUiOiJjb3JuZXIifX1dXQ==
        \begin{tikzcd}
          E & \tMcU \\
          {B} & \McU
          \arrow["{\widetilde{\ceil{E}}}", from=1-1, to=1-2]
          \arrow[from=1-1, to=2-1, two heads]
          \arrow["\lrcorner"{anchor=center, pos=0.15, scale=1.5}, draw=none, from=1-1, to=2-2]
          \arrow[from=1-2, to=2-2, two heads]
          \arrow["{\ceil{E}}"', from=2-1, to=2-2]
        \end{tikzcd}
      \end{equation*}
    \end{minipage}
    \begin{minipage}{0.50\linewidth}
      \begin{equation*}
        %       https://q.uiver.app/#q=WzAsMTAsWzIsMSwiXFx0TWNVIl0sWzIsMiwiXFxNY1UiXSxbMSwyLCJcXHRNY1UgXFx0aW1lc19cXE1jVSBcXHRNY1UiXSxbMSwxLCJcXElkXipcXHRNY1UiXSxbMCwyLCJCIFxcdGltZXNfe1xcdWx7Qn19IEIiXSxbMCwxLCJcXElkX3tcXHVse0J9fShCKSJdLFsxLDAsIlxcdE1jVSJdLFswLDAsIkIiXSxbMSwzLCJcXE1jVSJdLFswLDMsIlxcdWx7Qn0iXSxbMiwxLCJcXElkIiwyXSxbMCwxXSxbMywyXSxbMywwXSxbMywxLCIiLDAseyJzdHlsZSI6eyJuYW1lIjoiY29ybmVyIn19XSxbNCwyLCJiIFxcdGltZXNfe1xcdWx7Yn19IGIiXSxbNSw0XSxbNSwzLCJcXElkX3tcXHVse2J9fShiKSJdLFs1LDIsIiIsMCx7InN0eWxlIjp7Im5hbWUiOiJjb3JuZXIifX1dLFs2LDMsIlxccmVmbCIsMV0sWzcsNSwiXFxyZWZsX0IiLDFdLFs3LDYsImIiXSxbNywzLCIiLDEseyJzdHlsZSI6eyJuYW1lIjoiY29ybmVyIn19XSxbMiw4XSxbNCw5XSxbOSw4LCJcXHVse2J9IiwyXSxbNCw4LCIiLDAseyJzdHlsZSI6eyJuYW1lIjoiY29ybmVyIn19XV0=
        \begin{tikzcd}[column sep=large]
          E & \tMcU \\
          {\Id_{B}(E)} & {\Id_\McU(\tMcU)} & \tMcU \\
          {E \times_{B} E} & {\tMcU \times_\McU \tMcU} & \McU \\
          {B} & \McU
          \arrow["\widetilde{\ceil{E}}", from=1-1, to=1-2]
          \arrow["{\refl_E}"{description}, from=1-1, to=2-1, hook]
          \arrow["\lrcorner"{anchor=center, pos=0.15, scale=1.5}, draw=none, from=1-1, to=2-2]
          \arrow["\refl"{description}, from=1-2, to=2-2, hook]
          \arrow["{\Id_{\ceil{E}}(\widetilde{\ceil{E}})}", from=2-1, to=2-2]
          \arrow["{\ev_\partial}"{description}, from=2-1, to=3-1, two heads]
          \arrow["\lrcorner"{anchor=center, pos=0.15, scale=1.5}, draw=none, from=2-1, to=3-2]
          \arrow[from=2-2, to=2-3]
          \arrow["{\ev_\partial}"{description}, from=2-2, to=3-2, two heads]
          \arrow["\lrcorner"{anchor=center, pos=0.15, scale=1.5}, draw=none, from=2-2, to=3-3]
          \arrow[from=2-3, to=3-3, two heads]
          \arrow["{\widetilde{\ceil{E}} \times_{\ceil{E}} \widetilde{\ceil{E}}}", from=3-1, to=3-2]
          \arrow[from=3-1, to=4-1]
          \arrow["\lrcorner"{anchor=center, pos=0.15, scale=1.5}, draw=none, from=3-1, to=4-2]
          \arrow["\Id"', from=3-2, to=3-3]
          \arrow[from=3-2, to=4-2]
          \arrow["{\ceil{E}}"', from=4-1, to=4-2]
        \end{tikzcd}
      \end{equation*}
    \end{minipage}
  \end{center}
  where $\Id_B(E) \twoheadrightarrow E \times_B E$ is chosen by the universe
  structure.
\end{remark}

\subsubsection{$\Sigma$-type Structure}
$\Sigma$-type structures amount to fibrations closed under composition.
\begin{definition}\label{def:sigma-type}
  An \emph{extensional $\Sigma$-type} structure on a universal map
  $\pi \colon \tMcU \to \McU$ in a locally cartesian closed category $\bC$ is a
  $\pi$-fibrancy structure $(\Sigma,\pair)$ on
  $\GenComp(\tMcU \to \McU, \tMcU \to \McU)$.
  \begin{equation*}
    % https://q.uiver.app/#q=WzAsMTEsWzIsNSwiXFxNY1UiXSxbMiw0LCJcXHRNY1UiXSxbMSwyLCJcXHRNY1UgXFx0aW1lcyBcXE1jVSJdLFsxLDEsIlxcdE1jVSBcXHRpbWVzIFxcdE1jVSJdLFsyLDIsIlxcTWNVIl0sWzIsMSwiXFx0TWNVIl0sWzAsMywiXFxwaV8qKFxcdE1jVSBcXHRpbWVzIFxcdE1jVSkiXSxbMCwxLCJcXHBpXipcXHBpXyooXFx0TWNVIFxcdGltZXMgXFx0TWNVKSJdLFswLDAsIlxcZXZeKihcXHRNY1UgXFx0aW1lcyBcXHRNY1UpIl0sWzMsMywiXFxNY1UiXSxbMywwLCJcXHRNY1UiXSxbMywyXSxbMSwwLCJcXHBpIl0sWzIsMSwiXFxwcm9qXzEiLDEseyJsYWJlbF9wb3NpdGlvbiI6MzB9XSxbNSw0XSxbMiw0XSxbMyw1XSxbNiwwXSxbNywyLCJcXGV2IiwxXSxbNywxLCIiLDIseyJjdXJ2ZSI6M31dLFs3LDZdLFs4LDddLFs4LDNdLFsxMCw5LCJcXHBpIl0sWzYsOSwiXFxTaWdtYSIsMix7ImxhYmVsX3Bvc2l0aW9uIjo4MCwic3R5bGUiOnsiYm9keSI6eyJuYW1lIjoiZGFzaGVkIn19fV0sWzgsMTAsIlxccGFpciIsMCx7InN0eWxlIjp7ImJvZHkiOnsibmFtZSI6ImRhc2hlZCJ9fX1dXQ==
    \begin{tikzcd}[row sep=small, column sep=small]
      {\ev^*(\tMcU \times \tMcU)} &&& \tMcU \\
      {\pi^*\pi_*(\tMcU \times \McU)} & {\tMcU \times \tMcU} & \tMcU \\
      & {\tMcU \times \McU} & \McU \\
      {\pi_*(\tMcU \times \McU)} &&& \McU \\
      && \tMcU \\
      && \McU
      \arrow["\pair", dashed, from=1-1, to=1-4]
      \arrow[from=1-1, to=2-1, two heads]
      \arrow[from=1-1, to=2-2]
      \arrow["\pi", from=1-4, to=4-4, two heads]
      \arrow["\ev\relax"{description}, from=2-1, to=3-2]
      \arrow[from=2-1, to=4-1, two heads]
      \arrow[from=2-2, to=2-3]
      \arrow[from=2-2, to=3-2, two heads]
      \arrow[from=2-3, to=3-3, two heads]
      \arrow[from=3-2, to=3-3]
      \arrow["\pi", from=5-3, to=6-3, two heads]
      \arrow[from=4-1, to=6-3]
      \arrow["\Sigma"'{pos=0.8}, dashed, from=4-1, to=4-4]
      \arrow[curve={height=18pt}, from=2-1, to=5-3, crossing over]
      \arrow["{\proj_1}"{description, pos=0.3}, from=3-2, to=5-3, crossing over]
    \end{tikzcd}
  \end{equation*}
\end{definition}

\subsubsection{$\Pi$-type Structure}
$\Pi$-type structures amount to fibrations closed under pushforwards along
fibrations.
\begin{definition}[{\cite[Definition 2.2]{voe17}}]\label{def:pi-type}
  An \emph{extensional $\Pi$-type} structure on a universal map
  $\pi \colon \tMcU \to \McU$ in a locally cartesian
  closed category $\bC$ is a pair of dashed maps $(\lam, \Pi)$ as below
  making the pushforward along $\tMcU \to \McU$ of
  $\tMcU \times \tMcU \to \tMcU \times \McU \in \sfrac{\bC}{\tMcU}$ a
  $\pi$-fibration.
  \begin{equation}\label{eqn:Pi-def}\tag{\textsc{$\Pi$-def}}
    % https://q.uiver.app/#q=WzAsMTIsWzIsMywiXFx0TWNVIl0sWzEsMiwiXFx0TWNVIFxcdGltZXMgXFxNY1UiXSxbMSwxLCJcXHRNY1UgXFx0aW1lcyBcXHRNY1UiXSxbMyw0LCJcXE1jVSJdLFszLDIsIlxccGlfKihcXHRNY1UgXFx0aW1lcyBcXE1jVSkiXSxbMywxLCJcXHBpXyooXFx0TWNVIFxcdGltZXMgXFx0TWNVKSJdLFs0LDIsIlxcTWNVIl0sWzQsMSwiXFx0TWNVIl0sWzIsMSwiXFxwaV4qXFxwaV8qKFxcdE1jVSBcXHRpbWVzIFxcTWNVKSJdLFsyLDAsIlxccGleKlxccGlfKihcXHRNY1UgXFx0aW1lcyBcXHRNY1UpIl0sWzAsMiwiXFxNY1UiXSxbMCwxLCJcXHRNY1UiXSxbMiwxLCIiLDAseyJjb2xvdXIiOlswLDYwLDYwXX1dLFsxLDBdLFswLDMsIlxccGkiLDJdLFs1LDRdLFs0LDNdLFs3LDYsIlxccGkiXSxbNCw2LCJcXFBpIiwyLHsic3R5bGUiOnsiYm9keSI6eyJuYW1lIjoiZGFzaGVkIn19fV0sWzUsNywiXFxsYW0iLDAseyJzdHlsZSI6eyJib2R5Ijp7Im5hbWUiOiJkYXNoZWQifX19XSxbNSw2LCIiLDAseyJzdHlsZSI6eyJuYW1lIjoiY29ybmVyIn19XSxbOSw4XSxbOCwwXSxbOCw0XSxbOSw1XSxbOCwxLCJcXGV2IiwxXSxbOSwyLCJcXGV2IiwxXSxbMSwxMF0sWzExLDEwXSxbMiwxMV0sWzgsMTAsIlxcYXBwIiwyLHsibGFiZWxfcG9zaXRpb24iOjcwfV0sWzksMTEsIlxcdGFwcCIsMl1d
    \begin{tikzcd}
      && {\pi^*\pi_*(\tMcU \times \tMcU)} \\
      \tMcU & {\tMcU \times \tMcU} & {\pi^*\pi_*(\tMcU \times \McU)} & {\pi_*(\tMcU \times \tMcU)} & \tMcU \\
      \McU & {\tMcU \times \McU} && {\pi_*(\tMcU \times \McU)} & \McU \\
      && \tMcU \\
      &&& \McU
      \arrow["\tapp"', from=1-3, to=2-1]
      \arrow["\ev"{description}, from=1-3, to=2-2]
      \arrow[from=1-3, to=2-3]
      \arrow[from=1-3, to=2-4]
      \arrow[from=2-1, to=3-1]
      \arrow[from=2-2, to=2-1]
      \arrow["\app"'{pos=0.7}, from=2-3, to=3-1]
      \arrow["\ev"{description}, from=2-3, to=3-2]
      \arrow[from=2-3, to=3-4]
      \arrow[from=2-3, to=4-3]
      \arrow["\lam", dashed, from=2-4, to=2-5]
      \arrow[from=2-4, to=3-4]
      \arrow["\lrcorner"{anchor=center, pos=0.125}, draw=none, from=2-4, to=3-5]
      \arrow["\pi", from=2-5, to=3-5]
      \arrow[from=3-2, to=3-1]
      \arrow[from=3-2, to=4-3]
      \arrow["\Pi"', dashed, from=3-4, to=3-5]
      \arrow[from=3-4, to=5-4]
      \arrow["\pi"', from=4-3, to=5-4]
      \arrow[crossing over, from=2-2, to=3-2]
    \end{tikzcd}
  \end{equation}

  We also refer to the maps
  $\app \colon \pi^*\pi_*(\tMcU \times \McU) \to \tMcU \times \McU \to \McU$ and
  $\tapp \colon \pi^*\pi_*(\tMcU \times \tMcU) \to \tMcU \times \McU \to \tMcU$
  obtained by composing the counit with the second projection as the
  \emph{application maps}.
\end{definition}

\subsubsection{Unit-type Structure}
$\Unit$-type structures amount to fibration structures on identity maps.
\begin{definition}\label{def:unit-type}
  A \emph{$\Unit$-type} structure on a universal map $\pi \colon \tMcU \to \McU$
  in a locally cartesian closed category $\bC$ is a pair of maps
  $(\star, \Unit)$ as below giving a $\pi$-fibration structure on the identity
  at the terminal object.
  \begin{equation*}
    % https://q.uiver.app/#q=WzAsNCxbMCwxLCIxIl0sWzAsMCwiMSJdLFsxLDEsIlxcTWNVIl0sWzEsMCwiXFx0TWNVIl0sWzMsMl0sWzAsMiwiXFxVbml0IiwyLHsic3R5bGUiOnsiYm9keSI6eyJuYW1lIjoiZGFzaGVkIn19fV0sWzEsMF0sWzEsMywiXFxzdGFyIiwwLHsic3R5bGUiOnsiYm9keSI6eyJuYW1lIjoiZGFzaGVkIn19fV0sWzEsMiwiIiwxLHsic3R5bGUiOnsibmFtZSI6ImNvcm5lciJ9fV1d
    \begin{tikzcd}[cramped]
      1 & \tMcU \\
      1 & \McU
      \arrow["\star", dashed, from=1-1, to=1-2]
      \arrow[from=1-1, to=2-1]
      \arrow["\lrcorner"{anchor=center, pos=0.05, scale=1.5}, draw=none, from=1-1, to=2-2]
      \arrow[from=1-2, to=2-2]
      \arrow["\Unit"', dashed, from=2-1, to=2-2]
    \end{tikzcd}
  \end{equation*}
\end{definition}

\subsubsection{Internal Universe Structure}
Given a map with a universal structure, all pullbacks of it also inherit its
universal structure.
Sometimes, we are interested in such a pulled back map with the inherited
universal structure, as in the case of internal universes.
\begin{definition}\label{def:int-univ}
  Let $\pi \colon \tMcU \to \McU$ be a universal map in a locally cartesian
  closed category $\bC$ with a choice of a terminal object 1.
  An \emph{internal (fibrant) universe} structure on $\tMcU \to \McU$ is a
  universal map $\pi_0 \colon \tMcU_0 \to \McU_0$ such that
  $\tMcU_0 \twoheadrightarrow \McU_0$ and $\McU_0 \twoheadrightarrow 1$ can be
  equipped with \emph{canonical} $\pi$-fibrancy structures $(\El,\tEl)$ and
  $(\ceil{\McU_0}, \widetilde{\ceil{\McU_0}})$.
  \begin{equation*}
    % https://q.uiver.app/#q=WzAsNyxbMCwyLCJcXE1jVV8wIl0sWzAsMywiMSJdLFswLDEsIlxcdE1jVV8wIl0sWzEsMCwiXFx0TWNVIl0sWzEsMSwiXFxNY1UiXSxbMSwyLCJcXHRNY1VfMCJdLFsxLDMsIlxcTWNVXzAiXSxbMCw0LCJcXEVsIiwxXSxbMyw0XSxbMiwwXSxbMiwzLCJcXHRFbCIsMV0sWzUsNl0sWzEsNiwiXFxjZWlse1xcTWNVXzB9IiwyXSxbMCwxXSxbMCw1LCJcXHdpZGV0aWxkZXtcXGNlaWx7XFxNY1VfMH19IiwxXSxbMiw0LCIiLDEseyJzdHlsZSI6eyJuYW1lIjoiY29ybmVyIn19XSxbMCw2LCIiLDEseyJzdHlsZSI6eyJuYW1lIjoiY29ybmVyIn19XV0=&macro_url=https%3A%2F%2Fgist.githubusercontent.com%2Flim495062%2F61b94af9ef95c1c7b0763c937de29c2b%2Fraw%2F456d405748eab3250184512b5240467b6083b2b4%2Facmhwmacros.sty
    \begin{tikzcd}[cramped]
      & \tMcU \\
      {\tMcU_0} & \McU \\
      {\McU_0} & {\tMcU} \\
      1 & {\McU}
      \arrow[from=1-2, to=2-2, two heads]
      \arrow["\tEl"{description}, from=2-1, to=1-2]
      \arrow[from=2-1, to=3-1, two heads]
      \arrow["\El"{description}, from=3-1, to=2-2]
      \arrow["{\widetilde{\ceil{\McU_0}}}"{description}, from=3-1, to=3-2]
      \arrow[from=3-1, to=4-1, two heads]
      \arrow[from=3-2, to=4-2, two heads]
      \arrow["{\ceil{\McU_0}}"', from=4-1, to=4-2]
      \arrow["\lrcorner"{anchor=center, pos=0.15, scale=1.5, rotate=45}, draw=none, from=2-1, to=2-2]
      \arrow["\lrcorner"{anchor=center, pos=0.15, scale=1.5}, draw=none, from=3-1, to=4-2]
    \end{tikzcd}
  \end{equation*}

  Such an internal universe structure is \emph{faithful} when
  $\El \colon \McU_0 \to \McU$ is a mono.
\end{definition}

The fact that the internal universe $\pi_0$ is $\pi$-fibrant automatically
implies, by universality of $\GenComp$ as in \Cref{lem:gen-comp}, that the
generic composable pair of $\pi_0$-fibrations is also a composable pair of
$\pi$-fibration.
This allows for the following construction, which is useful for defining
internal universes closed under various type-theoretic constructions.

\begin{construction}\label{constr:int-univ-gencomp}
  Let $\pi \colon \tMcU \to \McU$ be a universal map in a locally cartesian
  closed category $\bC$ with a choice of a terminal object 1.
  Fix a $\pi$-fibration $p \colon E \twoheadrightarrow B$.
  We construct connecting dashed maps as below so that the following is a
  diagram of pullback squares.
  \begin{equation*}
    % https://q.uiver.app/#q=WzAsNixbMSwyLCJcXHBpXyooXFx0TWNVIFxcdGltZXMgXFxNY1UpIl0sWzEsMSwiXFxwaV4qXFxwaV8qKFxcdE1jVSBcXHRpbWVzIFxcTWNVKSJdLFsxLDAsIlxcZXZeKihcXHRNY1UgXFx0aW1lcyBcXHRNY1UpIl0sWzAsMiwicF8qKEUgXFx0aW1lcyBCKSJdLFswLDEsInBeKnBfKihFIFxcdGltZXMgQikiXSxbMCwwLCJcXGV2XiooRSBcXHRpbWVzIEUpIl0sWzUsNF0sWzQsM10sWzEsMF0sWzIsMV0sWzMsMCwiIiwyLHsic3R5bGUiOnsiYm9keSI6eyJuYW1lIjoiZGFzaGVkIn19fV0sWzQsMSwiIiwxLHsic3R5bGUiOnsiYm9keSI6eyJuYW1lIjoiZGFzaGVkIn19fV0sWzUsMiwiIiwxLHsic3R5bGUiOnsiYm9keSI6eyJuYW1lIjoiZGFzaGVkIn19fV0sWzUsMSwiIiwxLHsic3R5bGUiOnsibmFtZSI6ImNvcm5lciJ9fV0sWzQsMCwiIiwxLHsic3R5bGUiOnsibmFtZSI6ImNvcm5lciJ9fV1d&macro_url=https%3A%2F%2Fgist.githubusercontent.com%2Flim495062%2F61b94af9ef95c1c7b0763c937de29c2b%2Fraw%2F254aa652b96f7d1c6dece2f3eba929431bb0adfe%2Facmhwmacros.sty
    \begin{tikzcd}[cramped]
      {\ev^*(E \times E)} & {\ev^*(\tMcU \times \tMcU)} \\
      {p^*p_*(E \times B)} & {\pi^*\pi_*(\tMcU \times \McU)} \\
      {p_*(E \times B)} & {\pi_*(\tMcU \times \McU)}
      \arrow[dashed, from=1-1, to=1-2]
      \arrow[from=1-1, to=2-1]
      \arrow[from=1-2, to=2-2]
      \arrow[dashed, from=2-1, to=2-2]
      \arrow[from=2-1, to=3-1]
      \arrow[from=2-2, to=3-2]
      \arrow[dashed, from=3-1, to=3-2]
      \arrow["\lrcorner"{anchor=center, pos=0.15, scale=1.5}, draw=none, from=1-1, to=2-2]
      \arrow["\lrcorner"{anchor=center, pos=0.15, scale=1.5}, draw=none, from=2-1, to=3-2]
    \end{tikzcd}
  \end{equation*}

  To do so, we note that by \Cref{constr:gen-comp}, the maps
  \begin{align*}
    \ev^*(E \times E) \to p^*p_*(E \times B)
    &&
    p^*p_*(E \times B) \to p_*(E \times B)
  \end{align*}
  are both $p$-fibrations with $p$-names
  \begin{align*}
    p^*p_*(E \times B) \xrightarrow{\ev} E \times B \to B
    &&
       p_*(E \times B) \to B
  \end{align*}
  Because $p$ has a $\pi$-name $\ceil{E} \colon B \to \McU$, it follows that
  post-composing with $\ceil{E}$ makes the $p$-names above $\pi$-names, as in
  the situation depicted below.
  \begin{equation*}
    % https://q.uiver.app/#q=WzAsMTMsWzIsNCwiQiJdLFsyLDMsIkUiXSxbMSwyLCJFIFxcdGltZXMgQiJdLFsxLDEsIkUgXFx0aW1lcyBFIl0sWzIsMiwiQiJdLFsyLDEsIkUiXSxbMCwzLCJwXyooRSBcXHRpbWVzIEIpIl0sWzAsMSwicF4qcF8qKEUgXFx0aW1lcyBCKSJdLFswLDAsIlxcZXZeKihFIFxcdGltZXMgRSkiXSxbMywzLCJcXHRNY1UiXSxbMyw0LCJcXE1jVSJdLFszLDIsIlxcTWNVIl0sWzMsMSwiXFx0TWNVIl0sWzMsMl0sWzEsMCwicCIsMl0sWzIsMSwiXFxwcm9qXzEiLDEseyJsYWJlbF9wb3NpdGlvbiI6MzB9XSxbNSw0XSxbMiw0XSxbMyw1XSxbNiwwXSxbNywyLCJcXGV2IiwxXSxbNywxLCIiLDIseyJjdXJ2ZSI6M31dLFs3LDZdLFs4LDddLFs4LDNdLFs4LDIsIiIsMCx7InN0eWxlIjp7Im5hbWUiOiJjb3JuZXIifX1dLFs3LDAsIiIsMCx7InN0eWxlIjp7Im5hbWUiOiJjb3JuZXIifX1dLFszLDQsIiIsMCx7InN0eWxlIjp7Im5hbWUiOiJjb3JuZXIifX1dLFs5LDEwLCJcXHBpIl0sWzAsMTBdLFsxLDldLFsxLDEwLCIiLDAseyJzdHlsZSI6eyJuYW1lIjoiY29ybmVyIn19XSxbMTIsMTFdLFs1LDEyXSxbNCwxMV0sWzUsMTEsIiIsMCx7InN0eWxlIjp7Im5hbWUiOiJjb3JuZXIifX1dXQ==
    \begin{tikzcd}[cramped]
      {\ev^*(E \times E)} \\
      {p^*p_*(E \times B)} & {E \times E} & E & \tMcU \\
      & {E \times B} & B & \McU \\
      {p_*(E \times B)} && E & \tMcU \\
      && B & \McU
      \arrow[from=1-1, to=2-1]
      \arrow[from=1-1, to=2-2]
      \arrow["\ev"{description}, from=2-1, to=3-2]
      \arrow[from=2-1, to=4-1]
      \arrow[curve={height=18pt}, from=2-1, to=4-3]
      \arrow[from=2-2, to=2-3]
      \arrow[from=2-2, to=3-2]
      \arrow[from=2-3, to=2-4]
      \arrow[from=2-3, to=3-3]
      \arrow[from=2-4, to=3-4]
      \arrow[from=3-2, to=3-3]
      \arrow["{\proj_1}"{description, pos=0.3}, from=3-2, to=4-3]
      \arrow[from=3-3, to=3-4]
      \arrow[from=4-1, to=5-3]
      \arrow[from=4-3, to=4-4]
      \arrow["p"', from=4-3, to=5-3]
      \arrow["\pi", from=4-4, to=5-4]
      \arrow[from=5-3, to=5-4]
      \arrow["\lrcorner"{anchor=center, pos=0.15, scale=1.5}, draw=none, from=1-1, to=3-2]
      \arrow["\lrcorner"{anchor=center, pos=0.15, scale=1.5}, draw=none, from=2-1, to=5-3]
      \arrow["\lrcorner"{anchor=center, pos=0.15, scale=1.5}, draw=none, from=2-2, to=3-3]
      \arrow["\lrcorner"{anchor=center, pos=0.15, scale=1.5}, draw=none, from=2-3, to=3-4]
      \arrow["\lrcorner"{anchor=center, pos=0.15, scale=1.5}, draw=none, from=4-3, to=5-4]
    \end{tikzcd}
  \end{equation*}
  By \Cref{lem:gen-comp}, $\GenComp(p,p)$ arises as a pullback of
  $\GenComp(\pi,\pi)$ along a unique map
  \begin{equation*}
    p_*(E \times B) \xrightarrow{\qquad} \pi_*(\tMcU \times \McU)
  \end{equation*}
  whose transpose under $\pi^* \dashv \pi_*$ is the map into $\tMcU \times \McU$
  is the pair whose first component is the $\pi$-point of
  $p^*p_*(E \times B) \to p_*(E \times B)$ and second component is the
  $\pi$-name of $\ev^*(E \times E) \to p^*p_*(E \times B)$.

  In particular, applying this construction to an internal universe
  $\tMcU_0 \to \McU_0$, we obtain a map
  \begin{equation*}
    (\pi_0)_*(\tMcU_0 \times \McU_0) \xrightarrow{\qquad} \pi_*(\tMcU \times \McU)
  \end{equation*}
  along which $\GenComp(\pi,\pi)$ pulls back to $\GenComp(\pi_0,\pi_0)$.
\end{construction}

\begin{remark}\label{rmk:int-univ-gencomp}
  In the context of \Cref{constr:int-univ-gencomp} above
  applied to an internal universe $\tMcU_0 \to \McU_0$ of $\tMcU \to \McU$ , we see that the map
  \begin{equation*}
    (\pi_0)_*(\tMcU_0 \times \McU_0)
    \xrightarrow{\qquad}
    \pi_*(\tMcU \times \McU)
  \end{equation*}
  as a map of presheaves is the family of functions
  \begin{equation*}
    \bC(\Gamma, (\pi_0)_*(\tMcU_0 \times \McU_0))
    \cong \left\{
      (A, B) \middle|
      % https://q.uiver.app/#q=WzAsNSxbMCwxLCJCIl0sWzEsMiwiXFxNY1VfMCJdLFsxLDEsIlxcdE1jVV8wIl0sWzAsMCwiRV8wIl0sWzEsMCwiXFxNY1VfMCJdLFswLDEsIlxcY2VpbHtFXzB9IiwyLHsic3R5bGUiOnsiYm9keSI6eyJuYW1lIjoiZGFzaGVkIn19fV0sWzMsMF0sWzMsMiwiIiwwLHsic3R5bGUiOnsiYm9keSI6eyJuYW1lIjoiZG90dGVkIn19fV0sWzMsMSwiIiwxLHsic3R5bGUiOnsibmFtZSI6ImNvcm5lciJ9fV0sWzMsNCwiXFxjZWlse0VfMCd9IiwwLHsic3R5bGUiOnsiYm9keSI6eyJuYW1lIjoiZG90dGVkIn19fV0sWzIsMV1d
      \begin{tikzcd}[cramped, row sep=small, column sep=small]
        {\Gamma.A} & {\McU_0} \\
        \Gamma & {\tMcU_0} \\
        & {\McU_0}
        \arrow["{B}", dashed, from=1-1, to=1-2]
        \arrow[from=1-1, to=2-1]
        \arrow[dotted, from=1-1, to=2-2]
        \arrow["\lrcorner"{anchor=center, pos=0.15, scale=1.5}, draw=none, from=1-1, to=3-2]
        \arrow["{A}"', dashed, from=2-1, to=3-2]
        \arrow[from=2-2, to=3-2]
      \end{tikzcd}
      \right\}
    \to
    \left\{
      (A, B) \middle|
      % https://q.uiver.app/#q=WzAsNSxbMCwxLCJCIl0sWzEsMiwiXFxNY1UiXSxbMSwxLCJcXHRNY1UiXSxbMCwwLCJFXzAiXSxbMSwwLCJcXE1jVSJdLFswLDEsIlxcY2VpbHtFfSIsMix7InN0eWxlIjp7ImJvZHkiOnsibmFtZSI6ImRhc2hlZCJ9fX1dLFszLDBdLFszLDIsIiIsMCx7InN0eWxlIjp7ImJvZHkiOnsibmFtZSI6ImRvdHRlZCJ9fX1dLFszLDEsIiIsMSx7InN0eWxlIjp7Im5hbWUiOiJjb3JuZXIifX1dLFszLDQsIlxcY2VpbHtFJ30iLDAseyJzdHlsZSI6eyJib2R5Ijp7Im5hbWUiOiJkb3R0ZWQifX19XSxbMiwxXV0=
      \begin{tikzcd}[cramped, row sep=small, column sep=small]
        {\Gamma.A} & \McU \\
        \Gamma & \tMcU \\
        & \McU
        \arrow["{B}", dashed, from=1-1, to=1-2]
        \arrow[from=1-1, to=2-1]
        \arrow[dotted, from=1-1, to=2-2]
        \arrow["\lrcorner"{anchor=center, pos=0.15, scale=1.5}, draw=none, from=1-1, to=3-2]
        \arrow["{A}"', dashed, from=2-1, to=3-2]
        \arrow[from=2-2, to=3-2]
      \end{tikzcd}
    \right\}
    \cong
    \bC(\Gamma, \pi_*(\tMcU \times \McU))
  \end{equation*}
  taking each $(A,B)$ to the postcomposite $(\El(A),\El(B))$ with
  $\El \colon \McU_0 \to \McU$ like so:
  \begin{equation*}
    % https://q.uiver.app/#q=WzAsOCxbMCwxLCJCIl0sWzEsMiwiXFxNY1VfMCJdLFsxLDEsIlxcdE1jVV8wIl0sWzAsMCwiRV8wIl0sWzEsMCwiXFxNY1VfMCJdLFsyLDIsIlxcTWNVIl0sWzIsMSwiXFx0TWNVIl0sWzIsMCwiXFxNY1UiXSxbMCwxLCJcXGNlaWx7RV8wfSIsMix7InN0eWxlIjp7ImJvZHkiOnsibmFtZSI6ImRhc2hlZCJ9fX1dLFszLDBdLFszLDIsIiIsMCx7InN0eWxlIjp7ImJvZHkiOnsibmFtZSI6ImRvdHRlZCJ9fX1dLFszLDEsIiIsMSx7InN0eWxlIjp7Im5hbWUiOiJjb3JuZXIifX1dLFszLDQsIlxcY2VpbHtFXzAnfSIsMCx7InN0eWxlIjp7ImJvZHkiOnsibmFtZSI6ImRhc2hlZCJ9fX1dLFsyLDFdLFsxLDVdLFsyLDZdLFs2LDVdLFsyLDUsIiIsMSx7InN0eWxlIjp7Im5hbWUiOiJjb3JuZXIifX1dLFs0LDddXQ==
    \begin{tikzcd}[cramped, row sep=small, column sep=small]
      {\Gamma.A} & {\McU_0} & \McU \\
      \Gamma & {\tMcU_0} & \tMcU \\
      & {\McU_0} & \McU
      \arrow["{B}", dashed, from=1-1, to=1-2]
      \arrow[from=1-1, to=2-1]
      \arrow[dotted, from=1-1, to=2-2]
      \arrow["\lrcorner"{anchor=center, pos=0.15, scale=1.5}, draw=none, from=1-1, to=3-2]
      \arrow[from=1-2, to=1-3]
      \arrow["{A}"', dashed, from=2-1, to=3-2]
      \arrow[from=2-2, to=2-3]
      \arrow[from=2-2, to=3-2]
      \arrow["\lrcorner"{anchor=center, pos=0.15, scale=1.5}, draw=none, from=2-2, to=3-3]
      \arrow[from=2-3, to=3-3]
      \arrow[from=3-2, to=3-3]
    \end{tikzcd}
  \end{equation*}
  In other words, the constructed map syntactically decodes the code for an
  internal type.
\end{remark}
\begin{remark}
  If our goal was to construct the usual map
  \begin{equation*}
    (\pi_0)_*(\tMcU_0 \times \McU_0) \xrightarrow{\qquad} \pi_*(\tMcU \times \McU)
  \end{equation*}
  to speak about closure of logical structures, we could have avoided
  entirely the discussion about $\GenComp$ in \Cref{constr:int-univ-gencomp}.
  However, because our formulation of cubical type theory includes a
  \emph{heterogeneous} $\Path$-type as in \Cref{def:htr-path-type}, whose
  definition requires $\GenComp$, axiomatising closure under heterogeneous
  $\Path$-types does require the full generality of
  \Cref{constr:int-univ-gencomp}.
\end{remark}

Applying \Cref{def:Id-type,def:sigma-type,def:pi-type,def:unit-type} to both an
external ambient universe and an internal universe, we are now ready to define
when the various type-theoretic structures in these two universes are compatible
with each other.

\begin{definition}\label{def:int-univ-sigma-pi-id}
  Let $\bC$ be a locally cartesian closed category with a choice of a terminal
  object 1.
  Fix $\Sigma,\Pi$-type structures
  \begin{align*}
    \Sigma,\Pi \colon \pi_*(\tMcU \times \tMcU) \rightrightarrows \McU
    &&
       \Sigma_0,\Pi_0 \colon (\pi_0)_*(\tMcU_0 \times \tMcU_0) \rightrightarrows \McU_0
  \end{align*}
  respectively on a universal map $\pi \colon \tMcU \to \McU$ and an internal
  universe $\tMcU_0 \to \McU_0$ on it with $\pi$-fibrancy structure
  $(\El,\tEl)$.
  The $\Sigma$- and $\Pi$-structures are \emph{compatible} when the diagram on
  the left below commutes.
  \begin{center}
    \begin{minipage}{0.45\linewidth}
      \begin{equation*}
        % https://q.uiver.app/#q=WzAsNCxbMSwwLCJcXHBpXyooXFx0TWNVIFxcdGltZXMgXFx0TWNVKSJdLFswLDAsIihcXHBpXzApXyooXFx0TWNVXzAgXFx0aW1lcyBcXHRNY1VfMCkiXSxbMSwxLCJcXHRNY1UiXSxbMCwxLCJcXHRNY1VfMCJdLFsxLDAsIlxcQ3JlZntjb25zdHI6aW50LXVuaXYtZ2VuY29tcH0iXSxbMCwyLCJcXFBpIiwyLHsib2Zmc2V0IjoxfV0sWzEsMywiXFxQaV8wIiwyLHsib2Zmc2V0IjoxfV0sWzMsMiwiXFxFbCIsMl0sWzEsMywiXFxTaWdtYV8wIiwwLHsib2Zmc2V0IjotMX1dLFswLDIsIlxcU2lnbWEiLDAseyJvZmZzZXQiOi0xfV1d
        \begin{tikzcd}[cramped, column sep=5em]
          {(\pi_0)_*(\tMcU_0 \times \tMcU_0)} & {\pi_*(\tMcU \times \tMcU)} \\
          {\tMcU_0} & \tMcU
          \arrow["{\text{\Cref{constr:int-univ-gencomp}}}", from=1-1, to=1-2]
          \arrow["{\Pi_0}"', shift right, from=1-1, to=2-1]
          \arrow["{\Sigma_0}", shift left, from=1-1, to=2-1]
          \arrow["\Pi"', shift right, from=1-2, to=2-2]
          \arrow["\Sigma", shift left, from=1-2, to=2-2]
          \arrow["\El"', from=2-1, to=2-2]
        \end{tikzcd}
      \end{equation*}
    \end{minipage}
    \begin{minipage}{0.45\linewidth}
      \begin{equation*}
        % https://q.uiver.app/#q=WzAsNCxbMSwwLCJcXHRNY1UgXFx0aW1lc197XFxNY1V9IFxcdE1jVSJdLFswLDAsIlxcdE1jVV8wIFxcdGltZXNfe1xcTWNVXzB9IFxcdE1jVV8wIl0sWzAsMSwiXFx0TWNVXzAiXSxbMSwxLCJcXHRNY1UiXSxbMSwwLCJcXHRFbCBcXHRpbWVzX1xcRWwgXFx0RWwiXSxbMSwyLCJcXElkXzAiLDIseyJvZmZzZXQiOjF9XSxbMiwzLCJcXEVsIiwyLHsibGFiZWxfcG9zaXRpb24iOjMwfV0sWzAsMywiXFxJZCIsMCx7Im9mZnNldCI6MX1dXQ==
        \begin{tikzcd}[cramped]
          {\tMcU_0 \times_{\McU_0} \tMcU_0} & {\tMcU \times_{\McU} \tMcU} \\
          {\tMcU_0} & \tMcU
          \arrow["{\tEl \times_\El \tEl}", from=1-1, to=1-2]
          \arrow["{\Id_0}"', shift right, from=1-1, to=2-1]
          \arrow["\Id", shift right, from=1-2, to=2-2]
          \arrow["\El"'{pos=0.5}, from=2-1, to=2-2]
        \end{tikzcd}
      \end{equation*}
    \end{minipage}
  \end{center}

  Likewise, pullbacks commutes with pullbacks, so if
  \begin{align*}
   \Id \colon \tMcU \times_\McU \tMcU \to \McU
    &&
       \Id_0 \colon \tMcU_0 \times_{\McU_0} \tMcU_0 \to \McU_0
  \end{align*}
  are respective $\Id$-type structures on $\tMcU \to \McU$ and
  $\tMcU_0 \to \McU_0$ then they are \emph{compatible} when the diagram on
  the right above commutes.

  Finally, if
  \begin{align*}
    \Unit \colon 1 \to \McU && \Unit_0 \colon 1 \to \McU_0
  \end{align*}
  are respective $\Unit$-type structures on $\tMcU \to \McU$ and
  $\tMcU_0 \to \McU_0$ then they are \emph{compatible} when
  the composites
  \begin{equation*}
    \El(\Unit_0) = \Unit
  \end{equation*}
  agree.
\end{definition}

%%% Local Variables:
%%% TeX-master: "./main.tex"
%%% TeX-engine: default
%%% End:

\section{Universe Axioms for Cubical Type Theory}\label{sec:univ-axioms}
With the background on universe category models of usual intensional type theory
recalled in the previous section, we next provide the universe category
axiomatisation of cubical type theory.

\begin{assumption}\label{asm:cubical-univs}
  From now, we work in a locally cartesian closed category $\bC$ with three
  universal maps, respectively the universe of fibrations, the universe of
  cofibrations and the interval object
  \begin{center}
    \begin{minipage}{0.3\linewidth}
      \begin{equation*}
        \begin{tikzcd}
          \tMcU \ar[d, "{\pi}"] \\ \McU
        \end{tikzcd}
      \end{equation*}
    \end{minipage}
    \begin{minipage}{0.3\linewidth}
      \begin{equation*}
        \begin{tikzcd}
          \partial\Cof \ar[d, "{\iota}", hook] \\ \Cof
        \end{tikzcd}
      \end{equation*}
    \end{minipage}
    \begin{minipage}{0.3\linewidth}
      \begin{equation*}
        \begin{tikzcd}
          \bI \ar[d, "{!}"] \\ 1
        \end{tikzcd}
      \end{equation*}
    \end{minipage}
  \end{center}
  where we further require $\partial\Cof \hookrightarrow \Cof$ to be a mono.
  % Furthermore, we refer to the (chosen) pullbacks of $\iota$ as the (chosen)
  % $\iota$-\emph{cofibrations}, as opposed to using the word fibrations, as in
  % \Cref{def:univ-fibration}.
  %
\end{assumption}

We rephrase some of the wording from \Cref{def:univ-fibration} when discussing
pullbacks of the universe of cofibrations $\partial\Cof \hookrightarrow \Cof$.

\begin{definition}\label{def:univ-cofibration}
  The \emph{$\iota$-name} and \emph{$\iota$-point} of a map $U \to V$ are maps
  $\ceil{U}$ and $\widetilde{\ceil{U}}$ as below such that one has a pullback
  square as follows.
  \begin{equation*}
    % https://q.uiver.app/#q=WzAsNCxbMCwxLCJCIl0sWzAsMCwiRSJdLFsxLDEsIlxcTWNVIl0sWzEsMCwiXFx0TWNVIl0sWzMsMiwiIiwwLHsic3R5bGUiOnsiaGVhZCI6eyJuYW1lIjoiZXBpIn19fV0sWzAsMiwiXFxjZWlse0V9IiwyXSxbMSwwLCIiLDIseyJzdHlsZSI6eyJoZWFkIjp7Im5hbWUiOiJlcGkifX19XSxbMSwzLCJcXHdpZGV0aWxkZXtcXGNlaWx7RX19Il0sWzEsMiwiIiwxLHsic3R5bGUiOnsibmFtZSI6ImNvcm5lciJ9fV1d
    \begin{tikzcd}[cramped]
      U & \partial\Cof \\
      V & \Cof
      \arrow["{\widetilde{\ceil{U}}}", from=1-1, to=1-2]
      \arrow[from=1-1, to=2-1, hook]
      \arrow["\lrcorner"{anchor=center, pos=0.15, scale=1.5}, draw=none, from=1-1, to=2-2]
      \arrow[from=1-2, to=2-2, hook]
      \arrow["{\ceil{U}}"', from=2-1, to=2-2]
    \end{tikzcd}
  \end{equation*}
  Together, the $\iota$-(name, point) pair $(\ceil{U}, \widetilde{\ceil{U}})$ forms
  a \emph{$\iota$-cofibrancy structure} for $U \to V$.
  Maps that can be equipped with a $\iota$-cofibrancy structure are called
  \emph{$\iota$-cofibrations}.
  As an object in the slice over $V$, we say
  $U \hookrightarrow V \in \sfrac{\bC}{V}$ is an $\iota$-cofibrant object.
\end{definition}

\subsection{Cofibration Structure}\label{subsec:cofib-struct}
We require the universe of cofibrations $\partial\Cof \hookrightarrow \Cof$ to
have the structure of a lattice of truth values relative to $\tMcU \to \McU$.
Briefly, it needs to be equipped with ``and'' and ``or'' operations as well as
``top'' and ``bottom'' elements, except the left concepts of ``or'' and
``bottom'' are relative to $\tMcU \to \McU$.

\subsubsection{And-Structure}\label{subsubsec:and-struct}
%
% We start with defining the ``and'' and ``or'' as well as the ``top'' and
% ``bottom'' structures required for the truth structure on
% $\partial\Cof \hookrightarrow \Cof$.

\begin{definition}\label{def:cof-and}
  An \emph{and-structure} on $\partial\Cof \hookrightarrow \Cof$ is an
  $\iota$-cofibration structure on the map
  $\partial\Cof \times \partial\Cof \hookrightarrow \Cof \times \Cof$ with
  $\iota$-name $-\wedge- \colon \Cof \times \Cof \to \Cof$.
  \begin{equation*}
    \begin{tikzcd}[cramped]
      {\partial\Cof \times \partial\Cof}
      \ar[r] \ar[d, hook]
      \ar[rd, "\lrcorner"{pos=0.15, scale=1.5}, draw=none]
      &
      {\partial\Cof}
      \ar[d, hook]
      \\
      {\Cof \times \Cof}
      \ar[r, "{- \wedge -}"']
      &
      \Cof
    \end{tikzcd}
  \end{equation*}
\end{definition}

\begin{remark}\label{rmk:cof-and-ptwise}
  We now unfold what it means to have an and-structure on
  $\partial\Cof \hookrightarrow \Cof$.

  To do so, we first note that $\partial\Cof \times \partial\Cof$ is the product
  of $\partial\Cof \times \Cof \to \Cof \times \Cof$ with
  $\Cof \times \partial\Cof \to \Cof \times \Cof$ in the slice over
  $\Cof \times \Cof$.
  And for each object $(\phi_1, \phi_2) \colon \Gamma \to \Cof \times \Cof$ over
  $\Cof \times \Cof$, one has pullbacks
  \begin{align*}
    (\phi_1,\phi_2)^*(\partial\Cof \times \Cof) \cong
    \phi_1^*\partial\Cof \cong \Gamma.\phi_1
    &&
       (\phi_1,\phi_2)^*(\Cof \times \partial\Cof) \cong
       \phi_2^*\partial\Cof \cong \Gamma.\phi_2
  \end{align*}
  as observed in the front and back slanted faces below.
  Because pullbacks commute with pullbacks, over each $\Gamma$, one sees that
  $\Gamma.(\phi_1 \wedge \phi_2) \eqqcolon (\phi_1 \wedge \phi_2)^*\partial\Cof$
  is the fibred product $\Gamma.\phi_1 \times_\Gamma \Gamma.\phi_2$ over
  $\Gamma$, as observed in the left face below.
  \begin{equation*}
    % https://q.uiver.app/#q=WzAsMTAsWzMsMiwiXFxwYXJ0aWFsXFxDb2YgXFx0aW1lcyBcXENvZiJdLFs1LDEsIlxcQ29mIFxcdGltZXMgXFxwYXJ0aWFsXFxDb2YiXSxbNCwzLCJcXENvZiBcXHRpbWVzIFxcQ29mIl0sWzQsMCwiXFxwYXJ0aWFsXFxDb2YgXFx0aW1lcyBcXHBhcnRpYWxcXENvZiJdLFs2LDAsIlxccGFydGlhbFxcQ29mIl0sWzYsMywiXFxDb2YiXSxbMSwzLCJcXEdhbW1hIl0sWzIsMSwiXFxHYW1tYS5cXHBoaV8yIl0sWzAsMiwiXFxHYW1tYS5cXHBoaV8xIl0sWzEsMCwiXFxHYW1tYS4oXFxwaGlfMSBcXHdlZGdlIFxccGhpXzIpIl0sWzMsMCwiIiwwLHsic3R5bGUiOnsidGFpbCI6eyJuYW1lIjoiaG9vayIsInNpZGUiOiJ0b3AifX19XSxbMCwyLCIiLDAseyJzdHlsZSI6eyJ0YWlsIjp7Im5hbWUiOiJob29rIiwic2lkZSI6InRvcCJ9fX1dLFsxLDIsIiIsMix7InN0eWxlIjp7InRhaWwiOnsibmFtZSI6Imhvb2siLCJzaWRlIjoidG9wIn19fV0sWzMsMSwiIiwyLHsic3R5bGUiOnsidGFpbCI6eyJuYW1lIjoiaG9vayIsInNpZGUiOiJ0b3AifX19XSxbMywyLCIiLDEseyJzdHlsZSI6eyJ0YWlsIjp7Im5hbWUiOiJob29rIiwic2lkZSI6InRvcCJ9fX1dLFszLDRdLFsyLDUsIi0gXFx3ZWRnZSAtIiwyXSxbNCw1LCIiLDIseyJzdHlsZSI6eyJ0YWlsIjp7Im5hbWUiOiJob29rIiwic2lkZSI6InRvcCJ9fX1dLFs2LDIsIihcXHBoaV8xLFxccGhpXzIpIiwyXSxbOCwwXSxbOCw2LCIiLDEseyJzdHlsZSI6eyJ0YWlsIjp7Im5hbWUiOiJob29rIiwic2lkZSI6InRvcCJ9fX1dLFs3LDYsIiIsMix7ImNvbG91ciI6WzAsNjAsNjBdLCJzdHlsZSI6eyJ0YWlsIjp7Im5hbWUiOiJob29rIiwic2lkZSI6InRvcCJ9fX1dLFs5LDgsIiIsMix7InN0eWxlIjp7InRhaWwiOnsibmFtZSI6Imhvb2siLCJzaWRlIjoidG9wIn19fV0sWzksNywiIiwxLHsic3R5bGUiOnsidGFpbCI6eyJuYW1lIjoiaG9vayIsInNpZGUiOiJ0b3AifX19XSxbOSwzXSxbOSw2LCIiLDEseyJjb2xvdXIiOlswLDYwLDYwXSwic3R5bGUiOnsidGFpbCI6eyJuYW1lIjoiaG9vayIsInNpZGUiOiJ0b3AifX19XSxbNywxLCIiLDEseyJjb2xvdXIiOlswLDYwLDYwXX1dXQ==
    \begin{tikzcd}[column sep=small, row sep=small]
      & {\Gamma.(\phi_1 \wedge \phi_2)} &&& {\partial\Cof \times \partial\Cof} && {\partial\Cof} \\
      && {\Gamma.\phi_2} &&& {\Cof \times \partial\Cof} \\
      {\Gamma.\phi_1} &&& {\partial\Cof \times \Cof} \\
      & \Gamma &&& {\Cof \times \Cof} && \Cof
      \arrow[from=1-2, to=1-5]
      \arrow[hook, from=1-2, to=2-3]
      \arrow[hook, from=1-2, to=3-1]
      \arrow[from=1-5, to=1-7]
      \arrow[hook, from=1-5, to=2-6]
      \arrow[hook, from=1-5, to=3-4]
      \arrow[hook, from=1-5, to=4-5]
      \arrow[hook, from=1-7, to=4-7]
      \arrow[hook, from=2-6, to=4-5]
      \arrow[from=3-1, to=3-4]
      \arrow[hook, from=3-1, to=4-2]
      \arrow[hook, from=3-4, to=4-5]
      \arrow["{(\phi_1,\phi_2)}"', from=4-2, to=4-5]
      \arrow["{- \wedge -}"', from=4-5, to=4-7]
      \arrow[crossing over, hook, from=1-2, to=4-2]
      \arrow[crossing over, from=2-3, to=2-6]
      \arrow[crossing over, hook, from=2-3, to=4-2]
    \end{tikzcd}
  \end{equation*}

  By requiring
  $\partial\Cof \times \partial\Cof \hookrightarrow \Cof \times \Cof$ to arise
  as a pullback of $\partial\Cof \hookrightarrow \Cof$, for all
  $f \colon \Delta \to \Gamma$, we further see that
  $\Gamma.(\phi_1 \wedge \phi_2) \to \Gamma$ pulls back along
  $f \colon \Delta \to \Gamma$ to the map $\Delta.(f^*\phi_1 \wedge f^*\phi_2)$
  obtained by taking the fibred product of the pulled back components
  $\Delta.f^*\phi_1,\Delta.f^*\phi_2 \rightrightarrows \Delta$.
  \begin{equation*}
    % https://q.uiver.app/#q=WzAsOCxbNCwzLCJcXEdhbW1hIl0sWzUsMSwiXFxHYW1tYS5cXHBoaV8yIl0sWzMsMiwiXFxHYW1tYS5cXHBoaV8xIl0sWzQsMCwiXFxHYW1tYS4oXFxwaGlfMSBcXHdlZGdlIFxccGhpXzIpIl0sWzEsMywiXFxEZWx0YSJdLFswLDIsIlxcRGVsdGEuZl4qXFxwaGlfMSJdLFsyLDEsIlxcRGVsdGEuZl4qXFxwaGlfMiJdLFsxLDAsIlxcRGVsdGEuZl4qKFxccGhpXzEgXFx3ZWRnZSBcXHBoaV8yKSJdLFsyLDAsIiIsMSx7InN0eWxlIjp7InRhaWwiOnsibmFtZSI6Imhvb2siLCJzaWRlIjoidG9wIn19fV0sWzEsMCwiIiwyLHsic3R5bGUiOnsidGFpbCI6eyJuYW1lIjoiaG9vayIsInNpZGUiOiJ0b3AifX19XSxbMywyLCIiLDIseyJzdHlsZSI6eyJ0YWlsIjp7Im5hbWUiOiJob29rIiwic2lkZSI6InRvcCJ9fX1dLFszLDEsIiIsMSx7InN0eWxlIjp7InRhaWwiOnsibmFtZSI6Imhvb2siLCJzaWRlIjoidG9wIn19fV0sWzMsMCwiIiwxLHsic3R5bGUiOnsidGFpbCI6eyJuYW1lIjoiaG9vayIsInNpZGUiOiJ0b3AifX19XSxbNCwwLCJmIiwyXSxbNiw0LCIiLDIseyJjb2xvdXIiOlswLDYwLDYwXSwic3R5bGUiOnsidGFpbCI6eyJuYW1lIjoiaG9vayIsInNpZGUiOiJ0b3AifX19XSxbNSw0LCIiLDIseyJzdHlsZSI6eyJ0YWlsIjp7Im5hbWUiOiJob29rIiwic2lkZSI6InRvcCJ9fX1dLFs3LDUsIiIsMix7InN0eWxlIjp7InRhaWwiOnsibmFtZSI6Imhvb2siLCJzaWRlIjoidG9wIn19fV0sWzcsNiwiIiwxLHsic3R5bGUiOnsidGFpbCI6eyJuYW1lIjoiaG9vayIsInNpZGUiOiJ0b3AifX19XSxbNiwxLCIiLDEseyJjb2xvdXIiOlswLDYwLDYwXX1dLFs1LDJdLFs3LDNdLFs3LDQsIiIsMSx7ImNvbG91ciI6WzAsNjAsNjBdLCJzdHlsZSI6eyJ0YWlsIjp7Im5hbWUiOiJob29rIiwic2lkZSI6InRvcCJ9fX1dXQ==
    \begin{tikzcd}[column sep=small, row sep=small]
      &
      {\begin{matrix}
        \Delta.f^*(\phi_1 \wedge \phi_2) \\
        \rotatebox{90}{=} \\
        \Delta.(f^*\phi_1 \wedge f^*\phi_2)
      \end{matrix}}
      &&& {\Gamma.(\phi_1 \wedge \phi_2)} \\
      && {\Delta.f^*\phi_2} &&& {\Gamma.\phi_2} \\
      {\Delta.f^*\phi_1} &&& {\Gamma.\phi_1} \\
      & \Delta &&& \Gamma
      \arrow[from=1-2, to=1-5]
      \arrow[hook, from=1-2, to=2-3]
      \arrow[hook, from=1-2, to=3-1]
      \arrow[hook, from=1-5, to=2-6]
      \arrow[hook, from=1-5, to=3-4]
      \arrow[hook, from=1-5, to=4-5]
      \arrow[hook, from=2-6, to=4-5]
      \arrow[from=3-1, to=3-4]
      \arrow[hook, from=3-1, to=4-2]
      \arrow[hook, from=3-4, to=4-5]
      \arrow["f"', from=4-2, to=4-5]
      \arrow[crossing over, hook, from=1-2, to=4-2]
      \arrow[crossing over, from=2-3, to=2-6]
      \arrow[crossing over, hook, from=2-3, to=4-2]
    \end{tikzcd}
  \end{equation*}
\end{remark}

\subsubsection{Or-Structure}\label{subsubsec:or-struct}
\begin{definition}\label{def:cof-or}
  A \emph{pre-or-structure} on $\partial\Cof \hookrightarrow \Cof$ consists of:
  \begin{enumerate}
    \item A map $\partial(\Cof \times \Cof) \hookrightarrow \Cof \times \Cof$
    \item Two maps
    $\inj_1 \colon \partial\Cof \times \Cof \hookrightarrow \partial(\Cof \times
    \Cof)$ and
    $\inj_2 \colon \Cof \times \partial\Cof \hookrightarrow \partial(\Cof \times
    \Cof)$
    \item An $\iota$-cofibration structure on
    $\partial(\Cof \times \Cof) \hookrightarrow \Cof \times \Cof$ with name
    denoted by $- \vee - \colon \Cof \times \Cof \to \Cof$.
    \begin{equation*}
      \begin{tikzcd}[cramped]
        {\partial(\Cof \times \Cof)}
        \ar[r] \ar[d, hook]
        \ar[rd, "\lrcorner"{pos=0.15,scale=1.5}, draw=none]
        &
        {\partial\Cof}
        \ar[d, hook]
        \\
        {\Cof \times \Cof}
        \ar[r, "{- \vee -}"']
        &
        \Cof
      \end{tikzcd}
    \end{equation*}
  \end{enumerate}
  Such that the following diagram commutes
  \begin{equation*}
    % https://q.uiver.app/#q=WzAsNSxbMCwwLCJcXHBhcnRpYWwgXFxDb2YgXFx0aW1lcyBcXHBhcnRpYWwgXFxDb2YiXSxbMCwxLCJcXHBhcnRpYWxcXENvZiBcXHRpbWVzIFxcQ29mIl0sWzEsMCwiXFxDb2YgXFx0aW1lcyBcXHBhcnRpYWxcXENvZiJdLFsxLDEsIlxccGFydGlhbChcXENvZiBcXHRpbWVzIFxcQ29mKSJdLFsyLDIsIlxcQ29mIFxcdGltZXMgXFxDb2YiXSxbMyw0LCIiLDAseyJzdHlsZSI6eyJib2R5Ijp7Im5hbWUiOiJkYXNoZWQifX19XSxbMiw0LCIiLDIseyJjdXJ2ZSI6LTJ9XSxbMSw0LCIiLDIseyJjdXJ2ZSI6Mn1dLFswLDFdLFswLDJdLFsyLDMsIlxcaW5qXzIiLDEseyJzdHlsZSI6eyJ0YWlsIjp7Im5hbWUiOiJob29rIiwic2lkZSI6InRvcCJ9LCJib2R5Ijp7Im5hbWUiOiJkYXNoZWQifX19XSxbMSwzLCJcXGlual8xIiwxLHsic3R5bGUiOnsidGFpbCI6eyJuYW1lIjoiaG9vayIsInNpZGUiOiJ0b3AifSwiYm9keSI6eyJuYW1lIjoiZGFzaGVkIn19fV1d
    \begin{tikzcd}
      {\partial \Cof \times \partial \Cof} & {\Cof \times \partial\Cof} \\
      {\partial\Cof \times \Cof} & {\partial(\Cof \times \Cof)} \\
      && {\Cof \times \Cof}
      \arrow[from=1-1, to=1-2, hook]
      \arrow[from=1-1, to=2-1, hook]
      \arrow["{\inj_2}"{description}, hook, from=1-2, to=2-2]
      \arrow[curve={height=-12pt}, from=1-2, to=3-3, hook]
      \arrow["{\inj_1}"{description}, hook, from=2-1, to=2-2]
      \arrow[curve={height=12pt}, from=2-1, to=3-3, hook]
      \arrow[from=2-2, to=3-3, hook]
    \end{tikzcd}
  \end{equation*}
  and that the following diagram is a pullback square
  \begin{equation}\label{eqn:or-cof}\tag{\textsc{or-cof}}
    \small
    % https://q.uiver.app/#q=WzAsNCxbMCwwLCJbXFxwYXJ0aWFsKFxcQ29mIFxcdGltZXMgXFxDb2YpLCBCIFxcdGltZXMgXFxDb2YgXFx0aW1lcyBcXENvZl0iXSxbMCwxLCJbXFxwYXJ0aWFsIFxcQ29mIFxcdGltZXMgXFxDb2YsIEIgXFx0aW1lcyBcXENvZiBcXHRpbWVzIFxcQ29mXSJdLFsxLDAsIltcXENvZiBcXHRpbWVzIFxccGFydGlhbFxcQ29mLCBCIFxcdGltZXMgXFxDb2ZdIl0sWzEsMSwiW1xccGFydGlhbFxcQ29mIFxcdGltZXMgXFxwYXJ0aWFsXFxDb2YsIEIgXFx0aW1lcyBcXENvZiBcXHRpbWVzIFxcQ29mXV97XFxDb2YgXFx0aW1lcyBcXENvZn0iXSxbMCwxXSxbMSwzXSxbMiwzXSxbMCwyXSxbMCwzLCIiLDEseyJzdHlsZSI6eyJuYW1lIjoiY29ybmVyIn19XV0=
    \begin{tikzcd}
      {[\Cof \times \partial(\Cof \times \Cof), \partial\Cof \times \Cof^2]_{\Cof^3}} & {[\Cof \times \Cof \times \partial\Cof, \partial\Cof \times \Cof^2]_{\Cof^3}} \\
      {[\Cof \times \partial \Cof \times \Cof, \partial\Cof \times \Cof^2]_{\Cof^3}} & {[\Cof \times \partial\Cof \times \partial\Cof, \partial\Cof \times \Cof^2]_{\Cof^3}}
      \arrow[from=1-1, to=1-2]
      \arrow[from=1-1, to=2-1]
      \arrow["\lrcorner"{anchor=center, pos=0.15, scale=1.5, rotate=0}, draw=none, from=1-1, to=2-2]
      \arrow[from=1-2, to=2-2]
      \arrow[from=2-1, to=2-2]
    \end{tikzcd}
  \end{equation}

  Such a pre-or-structure is an \emph{or-structure relative to} a map $X \to Y$
  when the squares induced by precomposition are pullback for each
  $B \in \set{X, Y}$.
  \begin{equation}\label{eqn:or-fib}\tag{\textsc{or-fib}}
    \small
    % https://q.uiver.app/#q=WzAsNCxbMCwwLCJbXFxwYXJ0aWFsKFxcQ29mIFxcdGltZXMgXFxDb2YpLCBCIFxcdGltZXMgXFxDb2YgXFx0aW1lcyBcXENvZl0iXSxbMCwxLCJbXFxwYXJ0aWFsIFxcQ29mIFxcdGltZXMgXFxDb2YsIEIgXFx0aW1lcyBcXENvZiBcXHRpbWVzIFxcQ29mXSJdLFsxLDAsIltcXENvZiBcXHRpbWVzIFxccGFydGlhbFxcQ29mLCBCIFxcdGltZXMgXFxDb2ZdIl0sWzEsMSwiW1xccGFydGlhbFxcQ29mIFxcdGltZXMgXFxwYXJ0aWFsXFxDb2YsIEIgXFx0aW1lcyBcXENvZiBcXHRpbWVzIFxcQ29mXV97XFxDb2YgXFx0aW1lcyBcXENvZn0iXSxbMCwxXSxbMSwzXSxbMiwzXSxbMCwyXSxbMCwzLCIiLDEseyJzdHlsZSI6eyJuYW1lIjoiY29ybmVyIn19XV0=
    \begin{tikzcd}
      {[\partial(\Cof \times \Cof), B \times \Cof^2]_{\Cof^2}} & {[\Cof \times \partial\Cof, B \times \Cof^2]_{\Cof^2}} \\
      {[\partial \Cof \times \Cof, B \times \Cof^2]_{\Cof^2}} & {[\partial\Cof \times \partial\Cof, B \times \Cof^2]_{\Cof^2}}
      \arrow[from=1-1, to=1-2]
      \arrow[from=1-1, to=2-1]
      \arrow["\lrcorner"{anchor=center, pos=0.15, scale=1.5, rotate=0}, draw=none, from=1-1, to=2-2]
      \arrow[from=1-2, to=2-2]
      \arrow[from=2-1, to=2-2]
    \end{tikzcd}
  \end{equation}
\end{definition}

\begin{remark}\label{rmk:or-psh-prod-approx}
  Using the vocabulary of \cite[Definition 4.4]{struct-lift}, we see in that
  with an or-structure relative to $\tMcU \to \McU$, in the slice over
  $\Cof \times \Cof$ the map
  $\partial(\Cof \times \Cof) \hookrightarrow \Cof \times \Cof$
  \emph{structurally approximates} the pushout-product
  \begin{equation*}
    (\Cof \times \partial\Cof \hookrightarrow \Cof \times \Cof) \ltimes_{\Cof
      \times \Cof} (\partial\Cof \times \Cof \hookrightarrow \Cof \times \Cof)
  \end{equation*}
  relative to $\tMcU \times \Cof \times \Cof \to \McU \times \Cof \times \Cof$.

  Likewise,
  $\Cof \times \partial(\Cof \times \Cof) \hookrightarrow \Cof \times \Cof^2$
  structurally approximates the pushout-product
  \begin{equation*}
    (\Cof \times \partial \Cof \times \Cof \hookrightarrow \Cof \times \Cof^2)
    \ltimes_{\Cof \times \Cof^2} (\Cof \times \Cof \times \partial\Cof
    \hookrightarrow \Cof \times \Cof^2) \ltimes_{\Cof \times \Cof^2}
  \end{equation*}
  relative to $\partial\Cof \times \Cof^2 \hookrightarrow \Cof \times \Cof^2$.
\end{remark}

\begin{remark}\label{rmk:cof-and-or-ptwise}
  We now unfold what it means to have both an and-structure and an or-structure
  on $\partial\Cof \hookrightarrow \Cof$ relative to $\tMcU \to \McU$.

  As in \Cref{rmk:cof-and-ptwise}, a pre-or-structures gives, for each
  $(\phi_1, \phi_2) \colon \Gamma \to \Cof \times \Cof$ the
  following diagram
  \begin{equation*}
    % https://q.uiver.app/#q=WzAsNSxbMCwwLCJcXEdhbW1hLihcXHBoaV8xIFxcd2VkZ2UgXFxwaGlfMikiXSxbMCwxLCJcXEdhbW1hLlxccGhpXzEiXSxbMSwwLCJcXEdhbW1hLlxccGhpXzIiXSxbMSwxLCJcXEdhbW1hLihcXHBoaV8xIFxcdmVlIFxccGhpXzIpIl0sWzIsMiwiXFxHYW1tYSJdLFswLDEsIiIsMCx7InN0eWxlIjp7InRhaWwiOnsibmFtZSI6Imhvb2siLCJzaWRlIjoidG9wIn19fV0sWzAsMiwiIiwyLHsic3R5bGUiOnsidGFpbCI6eyJuYW1lIjoiaG9vayIsInNpZGUiOiJ0b3AifX19XSxbMiwzLCIiLDIseyJzdHlsZSI6eyJ0YWlsIjp7Im5hbWUiOiJob29rIiwic2lkZSI6InRvcCJ9fX1dLFsxLDMsIiIsMCx7InN0eWxlIjp7InRhaWwiOnsibmFtZSI6Imhvb2siLCJzaWRlIjoidG9wIn19fV0sWzIsNCwiIiwwLHsiY3VydmUiOi0yLCJzdHlsZSI6eyJ0YWlsIjp7Im5hbWUiOiJob29rIiwic2lkZSI6InRvcCJ9fX1dLFsxLDQsIiIsMCx7ImN1cnZlIjoyLCJzdHlsZSI6eyJ0YWlsIjp7Im5hbWUiOiJob29rIiwic2lkZSI6InRvcCJ9fX1dLFszLDQsIiIsMCx7InN0eWxlIjp7InRhaWwiOnsibmFtZSI6Imhvb2siLCJzaWRlIjoidG9wIn19fV1d
    \begin{tikzcd}
      {\Gamma.(\phi_1 \wedge \phi_2)} & {\Gamma.\phi_2} \\
      {\Gamma.\phi_1} & {\Gamma.(\phi_1 \vee \phi_2)} \\
      && \Gamma
      \arrow[hook, from=1-1, to=1-2]
      \arrow[hook, from=1-1, to=2-1]
      \arrow[hook, from=1-2, to=2-2]
      \arrow[curve={height=-12pt}, hook, from=1-2, to=3-3]
      \arrow[hook, from=2-1, to=2-2]
      \arrow[curve={height=12pt}, hook, from=2-1, to=3-3]
      \arrow[hook, from=2-2, to=3-3]
    \end{tikzcd}
  \end{equation*}
  such that for each $f \colon \Delta \to \Gamma$, one has
  $\Delta.f^*(\phi_1 \vee \phi_2) = \Delta.(f^*\phi_1 \vee f^*\phi_2)$.

  When this is an or-structure relative to $\tMcU \to \McU$, one applies
  representability to observe a bijection
  \begin{equation*}
    \begin{matrix}
      \sfrac{\bC}{\Cof \times \Cof}(\Gamma.(\phi_1 \vee \phi_2), B \times \Cof \times \Cof)
      \\
      \cong
      \\
      \sfrac{\bC}{\Cof \times \Cof}(\Gamma.\phi_1, B \times \Cof \times \Cof)
      \times_{\sfrac{\bC}{\Cof \times \Cof}(\Gamma.(\phi_1 \wedge \phi_2), B \times \Cof \times \Cof)}
      \sfrac{\bC}{\Cof \times \Cof}(\Gamma.\phi_2, B \times \Cof \times \Cof)
    \end{matrix}
  \end{equation*}
  for $B \in \set{\tMcU, \McU}$.
  By the adjunction between pullback and postcomposition, this is equivalently
  \begin{align*}
    \bC(\Gamma.(\phi_1 \vee \phi_2), B) \cong
    \bC(\Gamma.\phi_1, B) \times_{\bC(\Gamma.(\phi_1 \wedge \phi_2), B)} \bC(\Gamma.\phi_2, B)
  \end{align*}

  In other words, as seen in the diagram on the left below, for each pair of
  maps $(A|_{\phi_i} \colon \Gamma.\phi_i \to \McU)_{i=1,2}$ such that the
  common restriction to $\Gamma.(\phi_1 \wedge \phi_2)$ agree, one obtains a
  common extension
  $(A|_{\phi_1}, A|_{\phi_2}) \colon \Gamma.(\phi_1 \vee \phi_2) \to \McU$.
  Conversely, for each $A \colon \Gamma.(\phi_1 \vee \phi_2) \to \McU$, taking
  the common extension $(A|_{\phi_1}, A|_{\phi_2})$ of their restrictions
  $A|_{\phi_i} \colon \Gamma.\phi_i \to \McU$ is precisely the original map $A$.
  Also as seen in the diagram on the right below, the same bijection also
  applies for pair of maps $(a_i \colon \Gamma.\phi_i \to \tMcU)_{i=1,2}$ and
  maps $a \colon \Gamma.(\phi_1 \vee \phi_2) \to \tMcU$.
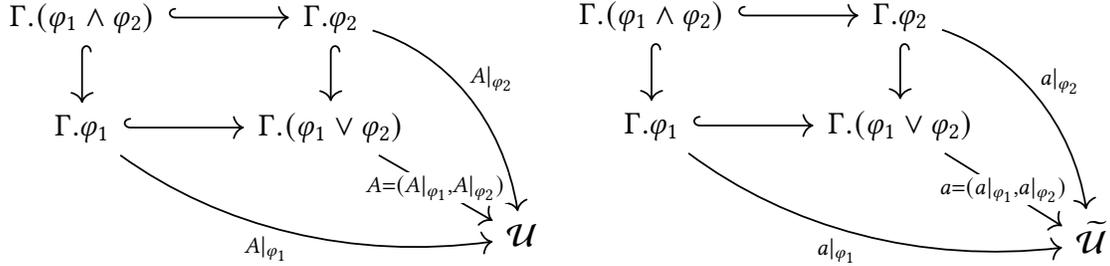
  \begin{center}
    \begin{minipage}{0.45\linewidth}
      \begin{equation*}
        % https://q.uiver.app/#q=WzAsNSxbMCwwLCJcXEdhbW1hLihcXHBoaV8xIFxcd2VkZ2UgXFxwaGlfMikiXSxbMCwxLCJcXEdhbW1hLlxccGhpXzEiXSxbMSwwLCJcXEdhbW1hLlxccGhpXzIiXSxbMSwxLCJcXEdhbW1hLihcXHBoaV8xIFxcdmVlIFxccGhpXzIpIl0sWzIsMiwiXFxNY1UiXSxbMCwxLCIiLDAseyJzdHlsZSI6eyJ0YWlsIjp7Im5hbWUiOiJob29rIiwic2lkZSI6InRvcCJ9fX1dLFswLDIsIiIsMix7InN0eWxlIjp7InRhaWwiOnsibmFtZSI6Imhvb2siLCJzaWRlIjoidG9wIn19fV0sWzIsMywiIiwyLHsic3R5bGUiOnsidGFpbCI6eyJuYW1lIjoiaG9vayIsInNpZGUiOiJ0b3AifX19XSxbMSwzLCIiLDAseyJzdHlsZSI6eyJ0YWlsIjp7Im5hbWUiOiJob29rIiwic2lkZSI6InRvcCJ9fX1dLFszLDQsIkEgPSAoQXxfe1xccGhpXzF9LCBBfF97XFxwaGlfMn0pIiwxXSxbMSw0LCJBfF97XFxwaGlfMX0iLDIseyJjdXJ2ZSI6M31dLFsyLDQsIkF8X3tcXHBoaV8yfSIsMCx7ImN1cnZlIjotM31dXQ==
        \begin{tikzcd}
          {\Gamma.(\phi_1 \wedge \phi_2)} & {\Gamma.\phi_2} \\
          {\Gamma.\phi_1} & {\Gamma.(\phi_1 \vee \phi_2)} \\
          && \McU
          \arrow[hook, from=1-1, to=1-2]
          \arrow[hook, from=1-1, to=2-1]
          \arrow[hook, from=1-2, to=2-2]
          \arrow["{A|_{\phi_2}}", curve={height=-24pt}, from=1-2, to=3-3]
          \arrow[hook, from=2-1, to=2-2]
          \arrow["{A|_{\phi_1}}"', curve={height=24pt}, from=2-1, to=3-3]
          \arrow["{A = (A|_{\phi_1}, A|_{\phi_2})}"{description}, from=2-2, to=3-3]
        \end{tikzcd}
      \end{equation*}
    \end{minipage}
    \begin{minipage}{0.45\linewidth}
      \begin{equation*}
        % https://q.uiver.app/#q=WzAsNSxbMCwwLCJcXEdhbW1hLihcXHBoaV8xIFxcd2VkZ2UgXFxwaGlfMikiXSxbMCwxLCJcXEdhbW1hLlxccGhpXzEiXSxbMSwwLCJcXEdhbW1hLlxccGhpXzIiXSxbMSwxLCJcXEdhbW1hLihcXHBoaV8xIFxcdmVlIFxccGhpXzIpIl0sWzIsMiwiXFxNY1UiXSxbMCwxLCIiLDAseyJzdHlsZSI6eyJ0YWlsIjp7Im5hbWUiOiJob29rIiwic2lkZSI6InRvcCJ9fX1dLFswLDIsIiIsMix7InN0eWxlIjp7InRhaWwiOnsibmFtZSI6Imhvb2siLCJzaWRlIjoidG9wIn19fV0sWzIsMywiIiwyLHsic3R5bGUiOnsidGFpbCI6eyJuYW1lIjoiaG9vayIsInNpZGUiOiJ0b3AifX19XSxbMSwzLCIiLDAseyJzdHlsZSI6eyJ0YWlsIjp7Im5hbWUiOiJob29rIiwic2lkZSI6InRvcCJ9fX1dLFszLDQsIkEgPSAoQXxfe1xccGhpXzF9LCBBfF97XFxwaGlfMn0pIiwxXSxbMSw0LCJBfF97XFxwaGlfMX0iLDIseyJjdXJ2ZSI6M31dLFsyLDQsIkF8X3tcXHBoaV8yfSIsMCx7ImN1cnZlIjotM31dXQ==
        \begin{tikzcd}
          {\Gamma.(\phi_1 \wedge \phi_2)} & {\Gamma.\phi_2} \\
          {\Gamma.\phi_1} & {\Gamma.(\phi_1 \vee \phi_2)} \\
          && \tMcU
          \arrow[hook, from=1-1, to=1-2]
          \arrow[hook, from=1-1, to=2-1]
          \arrow[hook, from=1-2, to=2-2]
          \arrow["{a|_{\phi_2}}", curve={height=-24pt}, from=1-2, to=3-3]
          \arrow[hook, from=2-1, to=2-2]
          \arrow["{a|_{\phi_1}}"', curve={height=24pt}, from=2-1, to=3-3]
          \arrow["{a = (a|_{\phi_1}, a|_{\phi_2})}"{description}, from=2-2, to=3-3]
        \end{tikzcd}
      \end{equation*}
    \end{minipage}
  \end{center}
  Furthermore, functoriality of the internal-Hom implies that restriction
  commutes with $\pi \colon \tMcU \to \McU$, in that
  $\pi(a|_{\phi_1}, a|_{\phi_2}) = (\pi a|_{\phi_1}, \pi a|_{\phi_2}) \colon
  \Gamma.(\phi_1 \vee \phi_2) \to \McU$.

  Also, as expected, the above data is stable under pullback, in that for any
  $f \colon \Delta \to \Gamma$, one has
  \begin{equation*}
    (f.(\phi_1 \vee \phi_2))^*(A|_{\phi_1}, A|_{\phi_2}) =
    ((f.\phi_1)^*A|_{\phi_1}, (f.\phi_2)^*A|_{\phi_2}) \colon \Delta.(f^*\phi_1
    \vee f^*\phi_2) \to \McU
  \end{equation*}
  and
  \begin{equation*}
    (f.(\phi_1 \vee \phi_2))^*(a|_{\phi_1}, a|_{\phi_2}) =
    ((f.\phi_1)^*a|_{\phi_1}, (f.\phi_2)^*a|_{\phi_2}) \colon \Delta.(f^*\phi_1
    \vee f^*\phi_2) \to \tMcU
  \end{equation*}
\end{remark}

By replacing the the colimiting concepts of pushout-product with their limiting
approximations, we see that or-structures transfer into internal universes.
\begin{lemma}\label{lem:or-struct-pb}
  Let $\partial\Cof \hookrightarrow \Cof$ be equipped with an or-structure
  relative to $\tMcU \to \McU$.
  Suppose that one has a $\pi$-fibration
  \begin{equation*}
    % https://q.uiver.app/#q=WzAsNCxbMSwwLCJcXHRNY1UiXSxbMSwxLCJcXE1jVSJdLFswLDEsIkIiXSxbMCwwLCJFIl0sWzAsMV0sWzIsMV0sWzMsMl0sWzMsMF0sWzMsMSwiIiwxLHsic3R5bGUiOnsibmFtZSI6ImNvcm5lciJ9fV1d
    \begin{tikzcd}[cramped]
      E & \tMcU \\
      B & \McU
      \arrow[from=1-1, to=1-2]
      \arrow[from=1-1, to=2-1, two heads]
      \arrow["\lrcorner"{anchor=center, pos=0.15, scale=1.5}, draw=none, from=1-1, to=2-2]
      \arrow[from=1-2, to=2-2, two heads]
      \arrow[from=2-1, to=2-2]
    \end{tikzcd}
  \end{equation*}
  where \Cref{eqn:or-fib} holds for $B$.
  \begin{enumerate}
    \item Then, \Cref{eqn:or-fib} holds for $E$.
    \item So, $\partial\Cof \hookrightarrow \Cof$ also has an or-structure relative to
    $E \to B$.
    \item In particular, if the terminal map $B \to 1$ is also a $\pi$-fibration
    $\tMcU \to \McU$ then $\partial\Cof \hookrightarrow \Cof$ has an
    or-structure relative to $E \to B$
  \end{enumerate}
\end{lemma}
\begin{proof}
  To prove the first part, one writes
  $C \colon \bC \to \sfrac{\bC}{\Cof\times\Cof}$ for the rebase map, and use the
  usual 3-by-3 argument, as described below.

  In the diagram below, the limits of each row and column in the grid are
  respectively denoted by the objects in the bottom-most row and right-most
  column.
  \begin{equation*}\scriptsize
    \begin{tikzcd}[cramped, column sep=tiny]
      {[\partial\Cof \times \Cof, C(B)]_{\Cof \times \Cof}} & {[\partial\Cof \times \Cof, C(\McU)]_{\Cof \times \Cof}} & {[\partial\Cof \times \Cof, C(\tMcU)]_{\Cof \times \Cof}} & {[\partial\Cof \times \Cof, C(E)]_{\Cof \times \Cof}} \\
      {[\partial\Cof \times \partial\Cof, C(B)]_{\Cof \times \Cof}} & {[\partial\Cof \times \partial\Cof, C(\McU)]_{\Cof \times \Cof}} & {[\partial\Cof \times \partial\Cof, C(\tMcU)]_{\Cof \times \Cof}} & {[\partial\Cof \times \partial\Cof, C(E)]_{\Cof \times \Cof}} \\
      {[\Cof \times \partial\Cof, C(B)]_{\Cof \times \Cof}} & {[\Cof \times \partial\Cof, C(\McU)]_{\Cof \times \Cof}} & {[\Cof \times \partial\Cof, C(\tMcU)]_{\Cof \times \Cof}} & {[\partial\Cof \times \Cof, C(E)]_{\Cof \times \Cof}} \\
      {[\partial(\Cof \times \Cof), C(B)]_{\Cof \times \Cof}} & {[\partial(\Cof \times \Cof), C(\McU)]_{\Cof \times \Cof}} & {[\partial(\Cof \times \Cof), C(\tMcU)]_{\Cof \times \Cof}} & {[\partial(\Cof \times \Cof), C(E)]_{\Cof \times \Cof}}
      \arrow[from=1-1, to=1-2]
      \arrow[from=1-1, to=2-1]
      \arrow[from=1-2, to=2-2]
      \arrow[from=1-3, to=1-2]
      \arrow[from=1-3, to=2-3]
      \arrow[from=1-4, to=2-4]
      \arrow[from=2-1, to=2-2]
      \arrow[from=2-3, to=2-2]
      \arrow[from=3-1, to=2-1]
      \arrow[from=3-1, to=3-2]
      \arrow["{\text{assumption}}"{description}, draw=none, from=3-1, to=4-1]
      \arrow[from=3-2, to=2-2]
      \arrow["{\text{or-structure}}"{description}, draw=none, from=3-2, to=4-2]
      \arrow[from=3-3, to=2-3]
      \arrow[from=3-3, to=3-2]
      \arrow["{\text{or-structure}}"{description}, draw=none, from=3-3, to=4-3]
      \arrow[from=3-4, to=2-4]
      \arrow[from=4-1, to=4-2]
      \arrow[from=4-3, to=4-2]
    \end{tikzcd}
  \end{equation*}
  But we note that the limit of the bottom-most row is
  $[\partial(\Cof \times \Cof), C(E)]_{\Cof \times \Cof}$ as the internal-Hom is
  continuous in the second argument.
  Therefore, $[\partial(\Cof \times \Cof), C(E)]_{\Cof \times \Cof}$ is the
  limit of the entire diagram.
  Limits commute with limits, so the limit of the entire diagram can also be
  equivalently computed by taking the limits of each row and then taking the
  limit of the resulting span.
  Hence, looking at the rightmost column above, one notes
  \begin{equation*}
    [\partial(\Cof \times \Cof), C(E)]_{\Cof \times \Cof}
    \cong
    [\partial\Cof \times \Cof, C(E)]_{\Cof \times \Cof}
    \times_\bullet
    [\Cof \times \partial\Cof, C(E)]_{\Cof \times \Cof}
  \end{equation*}
  as required for \Cref{def:cof-or}.

  With the first part proved, the second part follows by \Cref{def:cof-or}.
  For the third part part, we note that \Cref{eqn:or-fib} holds
  for the terminal object $1$ because $1 \times \Cof \times \Cof$ is the
  terminal object in $\sfrac{\bC}{\Cof \times \Cof}$.
  Then, the result follows by using the second part twice.
\end{proof}

\subsubsection{Top and Bottom Structures}\label{subsubsec:top-bot-strut}
\begin{definition}\label{def:cof-top-bot}
  A \emph{top structure} on the universe of cofibrations
  $\partial\Cof \hookrightarrow \Cof$ is a map $\ceil{\top} \colon 1 \to \Cof$
  along which the selected pullback of $\partial\Cof \hookrightarrow \Cof$ is an
  isomorphism.
  \begin{center}
    \begin{minipage}{0.45\linewidth}
      \begin{equation*}
        \begin{tikzcd}[cramped]
          {\top}
          \ar[r] \ar[d, "\cong"']
          \ar[rd, "\lrcorner"{anchor=center,pos=0.15,scale=1.5}, draw=none]
          &
          {\partial\Cof}
          \ar[d, hook]
          \\
          {1}
          \ar[r, "{\ceil{\top}}"']
          &
          \Cof
        \end{tikzcd}
      \end{equation*}
    \end{minipage}
    \begin{minipage}{0.45\linewidth}
      \begin{equation*}
        \begin{tikzcd}[cramped]
          {\bot}
          \ar[r] \ar[d, hook]
          \ar[rd, "\lrcorner"{anchor=center,pos=0.15,scale=1.5}, draw=none]
          &
          {\partial\Cof}
          \ar[d, hook]
          \\
          {1}
          \ar[r, "{\ceil{\bot}}"']
          &
          \Cof
        \end{tikzcd}
      \end{equation*}
    \end{minipage}
  \end{center}

  A \emph{pre-bottom structure} on the universe of cofibrations
  $\partial\Cof \hookrightarrow \Cof$ is a map $\ceil{\bot} \colon 1 \to \Cof$
  whose object $\bot$ selected by the universal structure satisfies
  \begin{equation*}
    [\bot,\partial\Cof] \cong [\bot,\Cof] \cong 1
  \end{equation*}

  The pre-bottom structure is a \emph{bottom structure relative to}
  a map $X \to Y$ when the object $\bot$ selected by the universe structure
  furthermore satisfies
  \begin{equation*}
    [\bot, X] \cong [\bot, Y] \cong 1
  \end{equation*}
\end{definition}

\begin{lemma}\label{lem:cof-bot}
  If $\partial\Cof \hookrightarrow \Cof$ has a bottom structure relative to
  $\tMcU \to \McU$ and $\phi \colon \Gamma \to \Cof$ factors through
  $\ceil{\bot} \colon 1 \to \Cof$ then
  $[\Gamma.\phi, \McU] \cong [\Gamma.\phi, \tMcU] \cong 1$
  and
  $[\Gamma.\phi, \Cof] \cong [\Gamma.\phi, \partial\Cof] \cong 1$.
\end{lemma}
\begin{proof}
  In this case, the pullback $\Gamma.\phi \cong \Gamma \times \bot$, so
  $[\Gamma \times \bot, \McU] \cong [\Gamma, [\bot, \McU]] \cong [\Gamma, 1]
  \cong 1$ and likewise for $\tMcU,\partial\Cof,\Cof$.
\end{proof}

Like or-structures in \Cref{lem:or-struct-pb}, bottom structures transfer into
internal universes.
\begin{lemma}\label{lem:bot-struct-pb}
  Let $\partial\Cof \hookrightarrow \Cof$ be equipped with a bottom structure $\bot$
  relative to $\tMcU \to \McU$.
  Suppose that one has a $\pi$-fibration
  \begin{equation*}
    % https://q.uiver.app/#q=WzAsNCxbMSwwLCJcXHRNY1UiXSxbMSwxLCJcXE1jVSJdLFswLDEsIkIiXSxbMCwwLCJFIl0sWzAsMV0sWzIsMV0sWzMsMl0sWzMsMF0sWzMsMSwiIiwxLHsic3R5bGUiOnsibmFtZSI6ImNvcm5lciJ9fV1d
    \begin{tikzcd}[cramped]
      E & \tMcU \\
      B & \McU
      \arrow[from=1-1, to=1-2]
      \arrow[from=1-1, to=2-1, two heads]
      \arrow["\lrcorner"{anchor=center, pos=0.15, scale=1.5}, draw=none, from=1-1, to=2-2]
      \arrow[from=1-2, to=2-2, two heads]
      \arrow[from=2-1, to=2-2]
    \end{tikzcd}
  \end{equation*}
  in which $[\bot, B] \cong 1$.
  \begin{enumerate}
    \item Then $\bot$ is also a bottom structure relative to $E \to B$.
    \item In particular, if $B \twoheadrightarrow 1$ is also $\pi$-fibrant then
    $\partial\Cof \hookrightarrow \Cof$ has a bottom structure relative to
    $E \twoheadrightarrow B$.
  \end{enumerate}
\end{lemma}
\begin{proof}
  The internal-Hom preserves pullbacks, so $[\bot, E] \to [\bot, B]$ is a
  pullback of $[\bot, \tMcU] \to [\bot, \McU]$.
  But $[\bot, B] \cong 1 \cong [\bot, \McU]$ by assumption, so the first part
  follows.

  The second part then follows by using the first part twice, and using the fact
  that internal-Hom preserves limits in the second variable.
\end{proof}

\subsubsection{Truth Structures}\label{subsubsec:truth-strut}
\begin{definition}\label{def:cof-truth-lattice}
  A \emph{pre-truth structure on $\partial\Cof \hookrightarrow \Cof$} is
  a tuple of (and, pre-or, top, pre-bottom) structures.
  Such a pre-truth structure is a truth structure \emph{relative to} a map
  $X \to Y$ when the pre-or and pre-bottom structures are or- and
  bottom-structures relative to $X \to Y$.
\end{definition}

We observe the subobject classifier $\Omega$ as a concrete example.
\begin{lemma}\label{lem:omega-truth}
  Suppose $\bC$ is a topos with truth map
  $1 \eqqcolon \partial\Omega \hookrightarrow \Omega$.
  Then $\partial\Omega \hookrightarrow \Omega$ has a truth-structure relative to
  any map $E \to B$.
\end{lemma}
\begin{proof}
  Isomorphisms are monos so the identity on the terminal object $1$ occurs as a
  pullback of the truth map, thus providing a top-structure.
  Toposes are cocomplete so there is a zero map $0 \hookrightarrow \Omega$, which
  is a mono and therefore a pullback of the truth map, thus providing a
  bottom-structure.
  The and-structure is immediate since
  $\partial\Omega \times \partial\Omega \hookrightarrow \Omega \times \Omega$ is
  a mono.
  Finally, the or-structure follows because toposes are adhesive, so
  $\partial\Omega \times \Omega \cup_{\partial\Omega \times \partial\Omega}
  \Omega \times \partial\Omega$, as a pushout of two subobjects of
  $\Omega \times \Omega$, is a subobject of $\Omega \times \Omega$.
\end{proof}

By the above results on or- and bottom-structures, truth structures transfer
into internal universes.
\begin{lemma}\label{lem:int-univ-truth}
  Suppose that $\partial\Cof \hookrightarrow \Cof$ has a truth structure
  relative to $\pi \colon \tMcU \twoheadrightarrow \McU$ and one has an internal
  universe whose data are given as below
  \begin{equation*}
    % https://q.uiver.app/#q=WzAsNyxbMCwyLCJcXE1jVV8wIl0sWzAsMywiMSJdLFswLDEsIlxcdE1jVV8wIl0sWzEsMCwiXFx0TWNVIl0sWzEsMSwiXFxNY1UiXSxbMSwyLCJcXHRNY1VfMCJdLFsxLDMsIlxcTWNVXzAiXSxbMCw0LCJcXEVsIiwxXSxbMyw0XSxbMiwwXSxbMiwzLCJcXHRFbCIsMV0sWzUsNl0sWzEsNiwiXFxjZWlse1xcTWNVXzB9IiwyXSxbMCwxXSxbMCw1LCJcXHdpZGV0aWxkZXtcXGNlaWx7XFxNY1VfMH19IiwxXSxbMiw0LCIiLDEseyJzdHlsZSI6eyJuYW1lIjoiY29ybmVyIn19XSxbMCw2LCIiLDEseyJzdHlsZSI6eyJuYW1lIjoiY29ybmVyIn19XV0=&macro_url=https%3A%2F%2Fgist.githubusercontent.com%2Flim495062%2F61b94af9ef95c1c7b0763c937de29c2b%2Fraw%2F456d405748eab3250184512b5240467b6083b2b4%2Facmhwmacros.sty
    \begin{tikzcd}[cramped]
      & \tMcU \\
      {\tMcU_0} & \McU \\
      {\McU_0} & {\tMcU} \\
      1 & {\McU}
      \arrow[from=1-2, to=2-2, two heads]
      \arrow["\tEl"{description}, from=2-1, to=1-2]
      \arrow[from=2-1, to=3-1, two heads]
      \arrow["\El"{description}, from=3-1, to=2-2]
      \arrow["{\widetilde{\ceil{\McU_0}}}"{description}, from=3-1, to=3-2]
      \arrow[from=3-1, to=4-1, two heads]
      \arrow[from=3-2, to=4-2, two heads]
      \arrow["{\ceil{\McU_0}}"', from=4-1, to=4-2]
      \arrow["\lrcorner"{anchor=center, pos=0.15, scale=1.5, rotate=45}, draw=none, from=2-1, to=2-2]
      \arrow["\lrcorner"{anchor=center, pos=0.15, scale=1.5}, draw=none, from=3-1, to=4-2]
    \end{tikzcd}
  \end{equation*}
  then $\partial\Cof \hookrightarrow \Cof$ has a truth structure relative to
  $\tMcU_0 \to \McU_0$.
\end{lemma}
\begin{proof}
  By \Cref{lem:or-struct-pb,lem:bot-struct-pb}.
\end{proof}

\subsubsection{Cofibrant Sections}\label{subsubsec:cofibrant-sect}
In order for the $\Path$-types subsequently defined in \Cref{def:path-type} to
be identity types, as we will see in \Cref{thm:Path-Id}, we require all sections
of fibrations to be equipped with cofibration structures.
We define the necessary axioms about $\partial\Cof \hookrightarrow \Cof$ for
this to happen.

We start with a generic classifier of sections of pullbacks of $\tMcU \to \McU$,
which follows from the following observation.
\begin{lemma}\label{lem:sect-classifier-master}
  Fix a map $p \colon E \to B$ and a pullback
  $p' \colon E' \to B'$.
  If $s \colon B' \to E'$ is a section of $p'$ as on the left then $s$
  occurs as a pullback of the diagonal as on the right.
  \begin{center}
    \begin{minipage}{0.45\linewidth}
      \begin{equation*}
        % https://q.uiver.app/#q=WzAsNCxbMCwwLCJFJyJdLFsxLDAsIkUiXSxbMCwxLCJCJyJdLFsxLDEsIkIiXSxbMCwxLCJlIl0sWzAsMiwicCciXSxbMSwzLCJwIl0sWzIsMCwicyIsMCx7Im9mZnNldCI6LTEsImN1cnZlIjotMX1dLFsyLDMsImIiLDJdLFswLDMsIiIsMCx7InN0eWxlIjp7Im5hbWUiOiJjb3JuZXIifX1dXQ==
        \begin{tikzcd}[cramped]
          {E'} & E \\
          {B'} & B
          \arrow["e", from=1-1, to=1-2]
          \arrow["{p'}", from=1-1, to=2-1]
          \arrow["\lrcorner"{anchor=center, pos=0.15, scale=1.5}, draw=none, from=1-1, to=2-2]
          \arrow["p", from=1-2, to=2-2]
          \arrow["s", shift left, curve={height=-6pt}, from=2-1, to=1-1]
          \arrow["b"', from=2-1, to=2-2]
        \end{tikzcd}
      \end{equation*}
    \end{minipage}
    \begin{minipage}{0.45\linewidth}
      % https://q.uiver.app/#q=WzAsNCxbMCwwLCJCJyJdLFsxLDAsIkUiXSxbMSwxLCJFIFxcdGltZXNfQiBFIl0sWzAsMSwiRSciXSxbMSwyLCJcXERlbHRhIl0sWzAsMSwiIiwwLHsic3R5bGUiOnsiYm9keSI6eyJuYW1lIjoiZGFzaGVkIn19fV0sWzAsMywicyIsMl0sWzMsMiwiIiwyLHsic3R5bGUiOnsiYm9keSI6eyJuYW1lIjoiZGFzaGVkIn19fV0sWzAsMiwiIiwxLHsic3R5bGUiOnsibmFtZSI6ImNvcm5lciJ9fV1d
      \begin{tikzcd}[cramped]
        {B'} & E \\
        {E'} & {E \times_B E}
        \arrow[dashed, from=1-1, to=1-2]
        \arrow["s"', from=1-1, to=2-1]
        \arrow["\lrcorner"{anchor=center, pos=0.15, scale=1.5}, draw=none, from=1-1, to=2-2]
        \arrow["\Delta", from=1-2, to=2-2]
        \arrow[dashed, from=2-1, to=2-2]
      \end{tikzcd}
    \end{minipage}
  \end{center}
\end{lemma}
\begin{proof}
  We construct a map $t \colon B' \to E$ as
  $t \coloneqq (B' \xrightarrow{s} E' \xrightarrow{e} E)$ so that the composite
  $(B' \xrightarrow{t} E \xrightarrow{p} B) = (B' \xrightarrow{s} E'
  \xrightarrow{e} E \xrightarrow{p} B) = (B' \xrightarrow{s} E' \xrightarrow{p'}
  B' \xrightarrow{b} B) = b$.
  Then, the required map $E' \to E \times_B E$ for which $s \colon B' \to E'$
  occurs as a pullback of is the pullback of $t \colon B' \to E$ along the
  projection $E \times_B E \to E$, which we observe in the diagram below.
  This is because $E \times_B E \to E$ pulls back along $t$ to $E' \to B'$ and
  the identity at $E$ pulls back to the identity at $B'$, as one can observe
  like so:
  \begin{equation*}
    % https://q.uiver.app/#q=WzAsOCxbMSwxLCJFIFxcdGltZXNfQiBFIl0sWzEsMiwiRSJdLFsyLDEsIkUiXSxbMiwyLCJCIl0sWzEsMCwiRSJdLFswLDIsIkInIl0sWzAsMSwiRSciXSxbMCwwLCJCJyJdLFswLDFdLFsyLDMsInAiXSxbMSwzLCJwIiwxXSxbMCwyXSxbNCwyLCIiLDAseyJjdXJ2ZSI6LTIsImxldmVsIjoyLCJzdHlsZSI6eyJoZWFkIjp7Im5hbWUiOiJub25lIn19fV0sWzQsMSwiIiwyLHsiY3VydmUiOjIsImxldmVsIjoyLCJzdHlsZSI6eyJoZWFkIjp7Im5hbWUiOiJub25lIn19fV0sWzQsMCwiXFxEZWx0YSJdLFs1LDEsInQiLDFdLFs1LDMsImIiLDIseyJjdXJ2ZSI6Mn1dLFs2LDUsInAnIiwxXSxbNiwwXSxbNiwxLCIiLDAseyJzdHlsZSI6eyJuYW1lIjoiY29ybmVyIn19XSxbMCwzLCIiLDAseyJzdHlsZSI6eyJuYW1lIjoiY29ybmVyIn19XSxbNyw2LCJzIiwxXSxbNyw0XSxbNyw1LCIiLDEseyJjdXJ2ZSI6MiwibGV2ZWwiOjJ9XSxbNywwLCIiLDIseyJzdHlsZSI6eyJuYW1lIjoiY29ybmVyIn19XV0=
    \begin{tikzcd}[cramped]
      {B'} & E \\
      {E'} & {E \times_B E} & E \\
      {B'} & E & B
      \arrow[from=1-1, to=1-2]
      \arrow["s"{description}, from=1-1, to=2-1]
      \arrow[curve={height=12pt}, equals, from=1-1, to=3-1]
      \arrow["\Delta", from=1-2, to=2-2]
      \arrow[curve={height=-12pt}, equals, from=1-2, to=2-3]
      \arrow[curve={height=12pt}, equals, from=1-2, to=3-2]
      \arrow[from=2-1, to=2-2]
      \arrow["{p'}"{description}, from=2-1, to=3-1]
      \arrow[from=2-2, to=2-3]
      \arrow[from=2-2, to=3-2]
      \arrow["p", from=2-3, to=3-3]
      \arrow["t"{description}, from=3-1, to=3-2]
      \arrow["b"', curve={height=12pt}, from=3-1, to=3-3]
      \arrow["p"{description}, from=3-2, to=3-3]
      \arrow["\lrcorner"{anchor=center, pos=0.15, scale=1.5}, draw=none, from=1-1, to=2-2]
      \arrow["\lrcorner"{anchor=center, pos=0.15, scale=1.5}, draw=none, from=2-1, to=3-2]
      \arrow["\lrcorner"{anchor=center, pos=0.15, scale=1.5}, draw=none, from=2-2, to=3-3]
    \end{tikzcd}
  \end{equation*}
  So the result follows by the pullback lemma.
\end{proof}

Thus, the following definition implements the specification that sections of
pullbacks to $\tMcU \to \McU$ are pullbacks of
$\partial\Cof \hookrightarrow \Cof$.
\begin{definition}\label{def:cof-sect}
  A \emph{$(\partial\Cof \hookrightarrow \Cof$)-cofibrant structure on
    $(\tMcU \twoheadrightarrow \McU)$-sections} is a pair of maps
  $(\mathsf{eq}, \mathsf{isEq})$ making the following square a pullback.
  \begin{equation*}
    % https://q.uiver.app/#q=WzAsNCxbMCwwLCJcXHRNY1UiXSxbMCwxLCJcXHRNY1UgXFx0aW1lc19cXE1jVSBcXHRNY1UiXSxbMSwwLCJcXHBhcnRpYWxcXENvZiJdLFsxLDEsIlxcQ29mIl0sWzIsMywiIiwwLHsic3R5bGUiOnsidGFpbCI6eyJuYW1lIjoiaG9vayIsInNpZGUiOiJ0b3AifX19XSxbMSwzLCJcXG1hdGhzZntpc0VxfSIsMix7InN0eWxlIjp7ImJvZHkiOnsibmFtZSI6ImRhc2hlZCJ9fX1dLFswLDEsIlxcRGVsdGEiLDIseyJzdHlsZSI6eyJ0YWlsIjp7Im5hbWUiOiJob29rIiwic2lkZSI6InRvcCJ9fX1dLFswLDIsIlxcbWF0aHNme2VxfSIsMCx7InN0eWxlIjp7ImJvZHkiOnsibmFtZSI6ImRhc2hlZCJ9fX1dLFswLDMsIiIsMSx7InN0eWxlIjp7Im5hbWUiOiJjb3JuZXIifX1dXQ==
    \begin{tikzcd}
      \tMcU & {\partial\Cof} \\
      {\tMcU \times_\McU \tMcU} & \Cof
      \arrow["{\mathsf{eq}}", dashed, from=1-1, to=1-2]
      \arrow["\Delta"', hook, from=1-1, to=2-1]
      \arrow["\lrcorner"{anchor=center, pos=0.15, scale=1.5}, draw=none, from=1-1, to=2-2]
      \arrow[hook, from=1-2, to=2-2]
      \arrow["{\mathsf{isEq}}"', dashed, from=2-1, to=2-2]
    \end{tikzcd}
  \end{equation*}
\end{definition}
\begin{corollary}\label{cor:cof-sect}
  Whenever there is an $\iota$-cofibrant structure on $\pi$-sections, if
  $E \twoheadrightarrow B$ is a $\pi$-fibration and $s \colon B \to E$ is a
  section of $p$ then $s$ is an $\iota$-cofibration.
\end{corollary}
\begin{proof}
  Directly by \Cref{lem:sect-classifier-master}.
\end{proof}

\begin{lemma}\label{lem:int-univ-cof-sect}
  Let $\pi_0 \colon \tMcU_0 \to \McU_0$ be an internal universe of
  $\pi \colon \tMcU \to \McU$.
  Then $\iota$-cofibrant structures on $\pi$-sections also restrict to
  $\iota$-cofibrant sections on $\pi_0$-sections.
\end{lemma}
\begin{proof}
  The internal universe $\pi_0 \colon \tMcU_0 \to \McU_0$ is $\pi$-fibrant so
  the diagonal $\tMcU_0 \to \tMcU_0 \times_{\McU_0} \tMcU_0$ is a section to a
  $\pi$-fibration.
  The result now follows by \Cref{lem:sect-classifier-master}.
\end{proof}

%%% Local Variables:
%%% TeX-master: "./main.tex"
%%% TeX-engine: default
%%% End:

\subsection{Interval Structure}\label{subsec:interval-struct}
In this section, we define the necessary axioms that make $\bI \to 1$ into a
relative interval object in the sense of \cite[Definition 5.10]{struct-lift}.
Briefly, if $\partial\Cof \hookrightarrow \Cof$ has a truth structure relative
to $\pi$ then $\bI$ has a pre-interval structure when has two cofibrant points.
This becomes an interval structure when the names for these two cofibrant points
are disjoint, in the sense of factoring through the relative bottom map
$\ceil{\bot} \colon 1 \to \Cof$.

\begin{definition}\label{def:interval-disjoint-endpt}
  An \emph{$\iota$-cofibrant endpoint} structure on $\bI$ is a choice of two
  points $\set{\mbbo{0}}, \set{\mbbo{1}} \colon 1 \rightrightarrows \bI$ of $\bI$ along with
  maps $\txtis_{\delta} \colon \bI \rightrightarrows \Cof$, for
  $\delta = \mbbo{0}, \mbbo{1}$, such that the universe structure on
  $\partial\Cof \hookrightarrow \Cof$ selects the pullback
  \begin{equation*}
    \begin{tikzcd}
      \set{\delta} \ar[r] \ar[d, hook] \ar[rd, "\lrcorner"{pos=0.15, scale=1.5}, draw=none]
      &
      \partial\Cof \ar[d, hook] \\
      \bI \ar[r, "{\txtis_{\delta}}"']
      &
      \Cof
    \end{tikzcd}
  \end{equation*}

  When $\partial\Cof \hookrightarrow \Cof$ has a pre-bottom structure along with
  an and-structure, a cofibrant endpoint structure on $\bI$ is \emph{disjoint}
  when $\txtis_{\mbbo{0}} \wedge \txtis_{\mbbo{1}}$ factors through
  $\ceil{\bot} \colon 1 \to \Cof$, like so:
  \begin{equation*}
    \begin{tikzcd}
      \bI \ar[r, "{(\txtis_{\mbbo{0}}, \txtis_{\mbbo{1}})}"] \ar[rd, "!"']
      &
      \Cof \times \Cof \ar[r, "{-\wedge-}"]
      &
      \Cof
      \\
      & 1 \ar[ur, "{\ceil{\bot}}"']
    \end{tikzcd}
  \end{equation*}

  If furthermore $\partial\Cof \hookrightarrow \Cof$ has a pre-or-structure then
  we put $\partial\bI \coloneqq \bI.\txtis_\partial \hookrightarrow \bI$ as its
  boundary, where
  $\txtis_\partial \coloneqq \txtis_{\mbbo{0}} \vee \txtis_{\mbbo{1}} \colon \bI
  \to \Cof \times \Cof \to \Cof$.
\end{definition}
\begin{example}\label{ex:omega-bipointed-int}
  If $\partial\Cof \hookrightarrow \Cof$ is the subobject classifier equipped
  with the truth-structure from \Cref{lem:omega-truth} in a topos then any
  bipointed object has a disjoint cofibrant endpoint structure.
\end{example}

We next note that when $\partial\Cof \hookrightarrow \Cof$ has a relative truth
structure then a disjoint cofibrant endpoint structure on $\bI$ induces
a relative interval structure in the sense of \cite[Definition 5.10]{struct-lift}.
\begin{lemma}\label{lem:int-int-obj}
  Suppose $\partial\Cof \hookrightarrow \Cof$ has a truth-structure relative to
  $\tMcU \to \McU$.
  Let $\mbbo{0}, \mbbo{1} \colon 1 \rightrightarrows \bI$ be two
  $\iota$-cofibrant endpoints of $\bI$ whose $\iota$-names are
  $\txtis_{\mbbo{0}}, \txtis_{\mbbo{1}} \colon \bI \rightrightarrows \Cof$.
  \begin{enumerate}
    \item Then, one may put a pre-interval structure on $\bI$ whose two points
    are $\set{\mbbo{0}}, \set{\mbbo{1}} \rightrightarrows \bI$ and whose
    boundary is
    $\partial I = \bI.(\txtis_{\mbbo{0}} \vee \txtis_{\mbbo{1}}) \hookrightarrow
    \bI$.
    \item When the $\iota$-cofibrant endpoints
    $\set{\mbbo{0}}, \set{\mbbo{1}} \colon 1 \rightrightarrows \bI$ are
    disjoint,
    $(\set{\mbbo{0}}, \set{\mbbo{1}} \rightrightarrows \bI, \partial I =
    \bI.(\txtis_{\mbbo{0}} \vee \txtis_{\mbbo{1}}) \hookrightarrow \bI)$ is in
    fact an interval structure relative to $\set{\tMcU, \McU}$, in the sense of
    \cite[Definition 5.10]{struct-lift}.
    \item In particular, for each $X \in \bC$, the object
    $X \times \bI \to X \in \sfrac{\bC}{X}$ is an interval structure relative to
    $\set{X \times \tMcU, X \times \McU}$.
  \end{enumerate}
\end{lemma}
\begin{proof}
  The first part follows by
  \Cref{rmk:cof-and-ptwise,rmk:cof-and-or-ptwise}.

  For second part, pulling back \Cref{eqn:or-fib} along
  $(\txtis_{\mbbo{0}}, \txtis_{\mbbo{1}}) \colon \bI \to \Cof \times \Cof$ gives
  \begin{equation*}
    % https://q.uiver.app/#q=WzAsNCxbMCwwLCJbXFxiSS4oXFx0eHRpc197XFxtYmJvezB9fSBcXHZlZSBcXHR4dGlzX3tcXG1iYm97MH19KSwgQiBcXHRpbWVzIFxcYkldX1xcYkkiXSxbMCwxLCJbXFxzZXR7XFxtYmJvezB9fSwgQiBcXHRpbWVzIFxcYkldX1xcYkkiXSxbMSwwLCJbXFxzZXR7XFxtYmJvezF9fSwgQiBcXHRpbWVzIFxcYkldX1xcYkkiXSxbMSwxLCJbXFxiSS4oXFx0eHRpc197XFxtYmJvezB9fSBcXHdlZGdlIFxcdHh0aXNfe1xcbWJib3swfX0pLCBCIFxcdGltZXMgXFxiSV1fXFxiSSJdLFsyLDNdLFsxLDNdLFswLDJdLFswLDFdLFswLDMsIiIsMSx7InN0eWxlIjp7Im5hbWUiOiJjb3JuZXIifX1dXQ==
    \begin{tikzcd}
      {[\bI.(\txtis_{\mbbo{0}} \vee \txtis_{\mbbo{0}}), B \times \bI]_\bI} & {[\set{\mbbo{1}}, B \times \bI]_\bI} \\
      {[\set{\mbbo{0}}, B \times \bI]_\bI} & {[\bI.(\txtis_{\mbbo{0}} \wedge \txtis_{\mbbo{0}}), B \times \bI]_\bI}
      \arrow[from=1-1, to=1-2]
      \arrow[from=1-1, to=2-1]
      \arrow["\lrcorner"{anchor=center, pos=0.15, scale=1.5, rotate=0}, draw=none, from=1-1, to=2-2]
      \arrow[from=1-2, to=2-2]
      \arrow[from=2-1, to=2-2]
    \end{tikzcd}
  \end{equation*}
  in $\sfrac{\bC}{\bI}$ for $B \in \set{\tMcU, \McU}$.
  Applying pushforward by $\bI \to 1$ and using \cite[Lemma 3.3]{struct-lift},
  one obtains the pullback square
  \begin{equation*}
    % https://q.uiver.app/#q=WzAsNCxbMCwwLCJbXFxiSS4oXFx0eHRpc197XFxtYmJvezB9fSBcXHZlZSBcXHR4dGlzX3tcXG1iYm97MH19KSwgQl0iXSxbMCwxLCJbXFxzZXR7XFxtYmJvezB9fSwgQl0iXSxbMSwwLCJbXFxzZXR7XFxtYmJvezF9fSwgQl0iXSxbMSwxLCJbXFxiSS4oXFx0eHRpc197XFxtYmJvezB9fSBcXHdlZGdlIFxcdHh0aXNfe1xcbWJib3swfX0pLCBCXSJdLFsyLDNdLFsxLDNdLFswLDJdLFswLDFdLFswLDMsIiIsMSx7InN0eWxlIjp7Im5hbWUiOiJjb3JuZXIifX1dXQ==
    \begin{tikzcd}
      {[\bI.(\txtis_{\mbbo{0}} \vee \txtis_{\mbbo{0}}), B]} & {[\set{\mbbo{1}}, B]} \\
      {[\set{\mbbo{0}}, B]} & {[\bI.(\txtis_{\mbbo{0}} \wedge \txtis_{\mbbo{0}}), B]}
      \arrow[from=1-1, to=1-2]
      \arrow[from=1-1, to=2-1]
      \arrow["\lrcorner"{anchor=center, pos=0.15, scale=1.5, rotate=0}, draw=none, from=1-1, to=2-2]
      \arrow[from=1-2, to=2-2]
      \arrow[from=2-1, to=2-2]
    \end{tikzcd}
  \end{equation*}
  Now if $\txtis_{\mbbo{0}} \wedge \txtis_{\mbbo{1}}$ factors through
  $\ceil{\bot} \colon 1 \to \Cof$ then by \Cref{lem:cof-bot}, it follows that
  $[\bI.(\txtis_{\mbbo{0}} \wedge \txtis_{\mbbo{0}}), B] \cong 1$, proving the
  interval structure as required by \cite[Definition 5.10]{struct-lift}.

  The final part about slices $\sfrac{\bC}{X}$ then follows by
  \cite[Lemma 5.12]{struct-lift}.
\end{proof}

In preparation for the filling operation, we also prove a result about
approximations of pushout-products of cofibrations.
\begin{lemma}\label{lem:cof-int-pushout-product}
  Suppose $\partial\Cof \hookrightarrow \Cof$ admits an or-structure relative to
  a map $E \to B$ and $\bI$ has a disjoint $\iota$-cofibrant endpoint structure.

  Fix a pullback $\bI.\phi \hookrightarrow \bI$ of
  $\partial\Cof \hookrightarrow \Cof$ with name
  $\phi \colon \bI \to \Cof$.
  \begin{equation*}
    \begin{tikzcd}
      \bI.\phi \ar[r] \ar[d, hook] \ar[rd, "\lrcorner"{pos=0.15, scale=1.5}, draw=none]
      &
      \partial\Cof \ar[d, hook] \\
      \bI \ar[r, "{\phi}"']
      &
      \Cof
    \end{tikzcd}
  \end{equation*}
  Then the map
  $(\Cof \vee \phi)^*\partial\Cof \hookrightarrow \Cof \times \bI$
  over $\Cof$
  constructed as the iterated pullback as below
  \begin{equation*}
    % https://q.uiver.app/#q=WzAsNyxbMCwwLCIoXFxDb2YgXFx2ZWUgXFxjZWlse1xccGhpfSleKlxccGFydGlhbFxcQ29mIl0sWzAsMSwiXFxDb2YgXFx0aW1lcyBcXGJJIl0sWzAsMiwiXFxDb2YiXSxbMSwxLCJcXENvZiBcXHRpbWVzIFxcQ29mIl0sWzIsMSwiXFxDb2YiXSxbMiwwLCJcXHBhcnRpYWxcXENvZiJdLFsxLDAsIlxccGFydGlhbChcXENvZiBcXHRpbWVzIFxcQ29mKSJdLFsxLDIsIlxccHJvaiIsMl0sWzAsMSwiIiwyLHsic3R5bGUiOnsidGFpbCI6eyJuYW1lIjoiaG9vayIsInNpZGUiOiJ0b3AifX19XSxbMSwzLCJcXENvZiBcXHRpbWVzIFxccGhpIiwyXSxbMyw0LCJcXHZlZSIsMl0sWzUsNCwiIiwwLHsic3R5bGUiOnsidGFpbCI6eyJuYW1lIjoiaG9vayIsInNpZGUiOiJ0b3AifX19XSxbNiwzLCIiLDAseyJzdHlsZSI6eyJ0YWlsIjp7Im5hbWUiOiJob29rIiwic2lkZSI6InRvcCJ9fX1dLFs2LDUsIiIsMSx7InN0eWxlIjp7InRhaWwiOnsibmFtZSI6Imhvb2siLCJzaWRlIjoidG9wIn19fV0sWzAsNiwiIiwxLHsic3R5bGUiOnsidGFpbCI6eyJuYW1lIjoiaG9vayIsInNpZGUiOiJ0b3AifX19XSxbNiw0LCIiLDEseyJzdHlsZSI6eyJuYW1lIjoiY29ybmVyIn19XSxbMCwzLCIiLDEseyJzdHlsZSI6eyJuYW1lIjoiY29ybmVyIn19XV0=
    \begin{tikzcd}
      {(\Cof \vee \phi)^*\partial\Cof} & {\partial(\Cof \times \Cof)} & {\partial\Cof} \\
      {\Cof \times \bI} & {\Cof \times \Cof} & \Cof \\
      \Cof
      \arrow[hook, from=1-1, to=1-2]
      \arrow[hook, from=1-1, to=2-1]
      \arrow["\lrcorner"{anchor=center, pos=0.15, scale=1.5, rotate=0}, draw=none, from=1-1, to=2-2]
      \arrow[hook, from=1-2, to=1-3]
      \arrow[hook, from=1-2, to=2-2]
      \arrow["\lrcorner"{anchor=center, pos=0.15, scale=1.5}, draw=none, from=1-2, to=2-3]
      \arrow[hook, from=1-3, to=2-3]
      \arrow["{\Cof \times \phi}"', from=2-1, to=2-2]
      \arrow["\proj"', from=2-1, to=3-1]
      \arrow["-\vee-"', from=2-2, to=2-3]
    \end{tikzcd}
  \end{equation*}
  approximates the pushout-product
  $(\Cof \times \bI.\phi \hookrightarrow \Cof \times \bI) \ltimes_\Cof
  (\partial\Cof \hookrightarrow \Cof)$ relative to
  $E \times \Cof \to B \times \Cof$ in $\sfrac{\bC}{\Cof}$.
\end{lemma}
\begin{proof}
  This follows from the definition of a relative approximation of
  pushout-product as in \cite[Definition 4.4]{struct-lift} as well as base
  change properties of approximations of pushout-products as in \cite[Lemmas 4.5
  and 4.6]{struct-lift}.

  As observed in \Cref{rmk:or-psh-prod-approx}, the or-structure entails that
  $\partial(\Cof \times \Cof) \hookrightarrow \Cof \times \Cof$ approximates the
  pushout-product
  \begin{align*}
    (\Cof \times \partial\Cof \hookrightarrow \Cof \times \Cof) \ltimes_{\Cof
    \times \Cof} (\partial\Cof \times \Cof \hookrightarrow \Cof \times \Cof)
  \end{align*}
  relative to $E \times \Cof \times \Cof \to B \times \Cof \times \Cof$
  in $\sfrac{\bC}{\Cof \times \Cof}$.

  Pulling back by $\Cof \times \phi$ and using
  \cite[Lemma 4.5]{struct-lift}, we see that
  $(\Cof \vee \phi)^*\partial\Cof \hookrightarrow \Cof \times \bI$
  approximates the pushout-product
  \begin{align*}
    (\Cof \times \bI.\phi \hookrightarrow \Cof \times \bI) \ltimes_{\Cof \times \bI}
  (\partial\Cof \times \bI \hookrightarrow \Cof \times \bI)
  \end{align*}
  relative to $E \times \Cof \times \bI \to B \times \Cof \times \bI$ in
  $\sfrac{\bC}{\Cof \times \bI}$.

  The result then follows by \cite[Lemma 4.6]{struct-lift} and by noting that
  $\proj^*(\partial\Cof \hookrightarrow \Cof) = (\partial\Cof \times \bI
  \hookrightarrow \Cof \times \bI)$ and likewise
  $\proj^*(E \times \Cof \to B \times \Cof) = (E \times \Cof \times \bI \to B
  \times \Cof \times \bI)$
\end{proof}

\begin{corollary}\label{cor:cof-int-bd}
  Suppose $\partial\Cof \hookrightarrow \Cof$ has a truth structure relative to
  $\tMcU \to \McU$ and that $\bI$ has a disjoint $\iota$-cofibrant endpoint
  structure.
  Then,
  \begin{enumerate}
    \item
    $(\Cof \vee \txtis_{\delta})^*\partial\Cof \hookrightarrow \Cof \times \bI$
    approximates the pushout-product
    $(\Cof \times \set{\delta} \hookrightarrow \Cof \times \bI) \ltimes_\Cof
    (\partial\Cof \hookrightarrow \Cof)$ relative to
    $\Cof \times \tMcU \to \Cof \times \McU$ in $\sfrac{\bC}{\Cof}$ for
    $\delta = \mbbo{0}, \mbbo{1}$.
    \item
    $(\Cof \vee \txtis_{\partial})^*\partial\Cof \hookrightarrow \Cof \times
    \bI$ approximates the pushout-product
    $(\Cof \times \partial\bI \hookrightarrow \Cof \times \bI) \ltimes_\Cof
    (\partial\Cof \hookrightarrow \Cof)$ relative to
    $\Cof \times \tMcU \to \Cof \times \McU$ in $\sfrac{\bC}{\Cof}$.
    \item If $\partial A \hookrightarrow A$ occurs as a pullback of
    $\partial\Cof \hookrightarrow \Cof$ along $\phi \colon A \to \Cof$ and
    $E \twoheadrightarrow B$ is a $\pi$-fibration where
    \Cref{eqn:or-fib} holds for $B$ then
    $(\rho \vee \txtis_{\partial})^*\partial\Cof \hookrightarrow A \times \bI$
    approximates the pushout-product
    $(A \times \partial\bI \hookrightarrow A \times \bI) \ltimes_A (\partial A
    \hookrightarrow A)$ relative to $A \times E \to A \times B$ in
    $\sfrac{\bC}{A}$.
  \end{enumerate}
\end{corollary}
\begin{proof}
  Immediate by \Cref{lem:cof-int-pushout-product,lem:or-struct-pb}.
\end{proof}

\begin{example}\label{ex:omega-pushout-product}
  If $\partial\Cof \hookrightarrow \Cof$ is the subobject classifier in a topos
  then the map $\partial(\Cof\times\Cof) \hookrightarrow \Cof\times\Cof$ from
  the or-structure in \Cref{lem:omega-truth} is precisely the pushout-product
  $\partial\Cof \times \Cof \cup \Cof \times \partial\Cof \hookrightarrow \Cof
  \times \Cof$ and $\bC$ is adhesive so the approximations to pushout-products
  in \Cref{lem:cof-int-pushout-product,cor:cof-int-bd} are in fact exactly the
  pushout-products.
\end{example}

Finally, we introduce the connection operation for intervals.

\begin{definition}\label{def:interval-min-max}
  Suppose $\bI$ has an $\iota$-cofibrant endpoint structure relative to
  $\partial\Cof \hookrightarrow \Cof$.
  A \emph{min}-structure on $\bI$ is an internal symmetric binary operation
  $\min\relax \colon \bI \times \bI \to \bI$ with sections
  $\set{\mbbo{1}} \times \bI, \bI \times \set{\mbbo{1}} \rightrightarrows \bI \times \bI$
  that is constant $\mbbo{0}$ when either of its arguments are $\mbbo{0}$.
  \begin{center}
    \begin{minipage}{0.45\linewidth}
      \begin{equation*}
        % https://q.uiver.app/#q=WzAsMyxbMCwwLCJcXGJJXzEgXFx0aW1lcyBcXGJJXzIiXSxbMSwxLCJcXGJJIl0sWzIsMCwiXFxiSV8yIFxcdGltZXMgXFxiSV8xIl0sWzAsMiwiXFxjb25nIiwwLHsic3R5bGUiOnsidGFpbCI6eyJuYW1lIjoiYXJyb3doZWFkIn19fV0sWzAsMSwiXFxtYXgiLDJdLFsyLDEsIlxcbWF4Il1d
        \begin{tikzcd}[cramped]
          {\bI_1 \times \bI_2} && {\bI_2 \times \bI_1} \\
          & \bI
          \arrow["\cong", tail reversed, from=1-1, to=1-3]
          \arrow["\min"', from=1-1, to=2-2]
          \arrow["\min", from=1-3, to=2-2]
        \end{tikzcd}
      \end{equation*}
    \end{minipage}
    \begin{minipage}{0.45\linewidth}
      \begin{equation*}
        % https://q.uiver.app/#q=WzAsNCxbMSwwLCJcXGJJIFxcdGltZXMgXFxiSSJdLFswLDAsIlxcc2V0e1xcbWJib3sxfX0gXFx0aW1lcyBcXGJJIl0sWzIsMCwiXFxiSSBcXHRpbWVzIFxcc2V0e1xcbWJib3sxfX0iXSxbMSwxLCJcXGJJIl0sWzEsMF0sWzIsMF0sWzAsMywiXFxtYXgiLDFdLFsyLDMsIj0iLDFdLFsxLDMsIj0iLDFdXQ==
        \begin{tikzcd}[cramped]
          {\set{\mbbo{1}} \times \bI} & {\bI \times \bI} & {\bI \times \set{\mbbo{1}}} \\
          & \bI
          \arrow[from=1-1, to=1-2]
          \arrow["{=}"{description}, from=1-1, to=2-2]
          \arrow["\min"{description}, from=1-2, to=2-2]
          \arrow[from=1-3, to=1-2]
          \arrow["{=}"{description}, from=1-3, to=2-2]
        \end{tikzcd}
      \end{equation*}
    \end{minipage}
    \begin{equation*}
      % https://q.uiver.app/#q=WzAsNixbMCwwLCJcXGJJIFxcdGltZXMgXFxzZXR7XFxtYmJvezB9fSJdLFsxLDAsIlxcYkkgXFx0aW1lcyBcXGJJIl0sWzIsMCwiIFxcc2V0e1xcbWJib3swfX0gXFx0aW1lcyBcXGJJIl0sWzAsMSwiIFxcc2V0e1xcbWJib3swfX0iXSxbMiwxLCIgXFxzZXR7XFxtYmJvezB9fSJdLFsxLDEsIlxcYkkiXSxbMSw1LCJcXG1pbiIsMV0sWzMsNSwiIiwxLHsic3R5bGUiOnsidGFpbCI6eyJuYW1lIjoiaG9vayIsInNpZGUiOiJ0b3AifX19XSxbNCw1LCIiLDEseyJzdHlsZSI6eyJ0YWlsIjp7Im5hbWUiOiJob29rIiwic2lkZSI6ImJvdHRvbSJ9fX1dLFsyLDQsIiEiXSxbMCwzLCIhIiwyXSxbMCwxLCIiLDAseyJzdHlsZSI6eyJ0YWlsIjp7Im5hbWUiOiJob29rIiwic2lkZSI6InRvcCJ9fX1dLFsyLDEsIiIsMCx7InN0eWxlIjp7InRhaWwiOnsibmFtZSI6Imhvb2siLCJzaWRlIjoiYm90dG9tIn19fV1d
      \begin{tikzcd}[cramped]
        {\bI \times \set{\mbbo{0}}} & {\bI \times \bI} & { \set{\mbbo{0}} \times \bI} \\
        { \set{\mbbo{0}}} & \bI & { \set{\mbbo{0}}}
        \arrow[hook, from=1-1, to=1-2]
        \arrow["{!}"', from=1-1, to=2-1]
        \arrow["\min"{description}, from=1-2, to=2-2]
        \arrow[hook', from=1-3, to=1-2]
        \arrow["{!}", from=1-3, to=2-3]
        \arrow[hook, from=2-1, to=2-2]
        \arrow[hook', from=2-3, to=2-2]
      \end{tikzcd}
    \end{equation*}
  \end{center}

  Likewise, a \emph{max} structure on $\bI$ as an internal symmetric binary
  operation $\max\relax \colon \bI \times \bI \to \bI$ with sections
  $\set{\mbbo{0}} \times \bI, \bI \times \set{\mbbo{0}} \rightrightarrows \bI
  \times \bI$ that is constant $\mbbo{1}$ when either of its arguments are
  $\mbbo{1}$.
\end{definition}

%%% Local Variables:
%%% TeX-master: "./main.tex"
%%% TeX-engine: default
%%% End:

\subsection{Fibration Structure}\label{subsec:fibration-struct}
In addition to the usual logical structures from intensional type theory like
from \Cref{subsec:itt-types}, cubical type theory also requires some extra
logical structures on the universe of fibrations $\tMcU \to \McU$.
Namely, these are the $\Path$-type structures, which we define in this part.
In particular, we will define a generalisation of the $\Path$-type structure and
also show that that under some conditions $\Path$-types induce the usual
$\Id$-types.

\subsubsection{Path-Type Structure}
The $\Path$-type structure states that the endpoint evaluation map
$P^{\bI}_{\McU}(\tMcU) \to \partial P^\bI_\McU(\tMcU)$ of the fibred $\bI$-path
object from \cite[Definition 5.6]{struct-lift} is a $\pi$-fibration.

% \begin{lemma}\label{lem:path-type-path-obj}
%   %
%   Assume that $\partial\Cof \hookrightarrow \Cof$ has a truth structure relative
%   to $\tMcU \to \McU$ under which $\bI$ has a disjoint cofibrant endpoint
%   structure.
%   %
%   Then, $P_{\McU}^\bI(\tMcU) = [\bI, \tMcU] \times_{[I, \McU]} \McU$ is the
%   fibred path object, in the sense of \Cref{constr:fibred-path} of
%   $\tMcU \to \McU$ constructed from $\bI$, which has endpoint evaluation map
%   $\ev_\partial \colon P_\McU^\bI(\tMcU) \to \tMcU \times_\McU \tMcU$.

%   And for each $X \in \bC$, the fibred path object of
%   $X \times \tMcU \to X \times \McU$ is given by $X \times P_\McU^\bI(\tMcU)$,
%   with endpoint evaluation map
%   $X \times \ev_\partial \colon X \times P_\McU^\bI(\tMcU) \to X \times (\tMcU
%   \times_\McU \tMcU)$.
%   %
% \end{lemma}
% %
% \begin{proof}
%   %
%   Immediate by \Cref{lem:interval-obj-fibred-path-pb,lem:int-int-obj}.
%   %
% \end{proof}

\begin{definition}\label{def:path-type}
  Assume that $\partial\Cof \hookrightarrow \Cof$ has a truth structure relative
  to $\tMcU \to \McU$ under which $\bI$ has a disjoint cofibrant endpoint
  structure.

  A \emph{$\Path$-type} structure on $\tMcU \to \McU$ is
  a pair of dashed maps $(\via, \Path)$ giving the boundary evaluation map
  \begin{equation*}
    P_{\McU}^\bI(\tMcU)
    \cong [\bI, \tMcU] \times_{[\bI, \McU]} \McU
    \xrightarrow{\ev_\partial}
    \tMcU \times_\McU \tMcU
    \cong \partial P_{\McU}^\bI(\tMcU)
  \end{equation*}
  from \cite[Definition 5.6]{struct-lift}
  a $\pi$-fibration structure.
  \begin{equation}\label{eqn:Path-def}\tag{\textsc{Path-def}}
    % https://q.uiver.app/#q=WzAsNSxbMSwxLCJQX3tcXE1jVX1eXFxiSShcXHRNY1UpIl0sWzEsMiwiXFx0TWNVIFxcdGltZXNfXFxNY1UgXFx0TWNVIl0sWzIsMiwiXFxNY1UiXSxbMiwxLCJcXHRNY1UiXSxbMCwwLCJcXHRNY1UiXSxbMCwxXSxbMywyXSxbMSwyLCJcXFBhdGgiLDIseyJzdHlsZSI6eyJib2R5Ijp7Im5hbWUiOiJkYXNoZWQifX19XSxbMCwzLCJcXHZpYSIsMCx7InN0eWxlIjp7ImJvZHkiOnsibmFtZSI6ImRhc2hlZCJ9fX1dLFswLDIsIiIsMSx7InN0eWxlIjp7Im5hbWUiOiJjb3JuZXIifX1dLFs0LDEsIlxcRGVsdGEiLDIseyJjdXJ2ZSI6Mn1dLFs0LDAsIlxccmVmbCJdXQ==
    \begin{tikzcd}[cramped]
      \tMcU \\
      & {P_{\McU}^\bI(\tMcU)} & \tMcU \\
      & {\tMcU \times_\McU \tMcU} & \McU
      \arrow["{(\const,\pi)}", from=1-1, to=2-2]
      \arrow["\Delta"', curve={height=12pt}, from=1-1, to=3-2]
      \arrow["\via", dashed, from=2-2, to=2-3]
      \arrow[from=2-2, to=3-2, "{\ev_\partial}"{description}]
      \arrow["\lrcorner"{anchor=center, pos=0.15, scale=1.5}, draw=none, from=2-2, to=3-3]
      \arrow[from=2-3, to=3-3]
      \arrow["\Path"', dashed, from=3-2, to=3-3]
    \end{tikzcd}
  \end{equation}

  If $\pi_0 \colon \tMcU_0 \to \McU_0$ is an internal universe of
  $\tMcU \to \McU$ with $\pi$-name $\El$ and is also equipped with a
  $\Path$-type structure $\Path_0$ then they are \emph{compatible} when the
  following diagram commutes, just like in the $\Id$-type case
  \Cref{def:int-univ-sigma-pi-id}.
  \begin{equation*}
    % https://q.uiver.app/#q=WzAsNCxbMSwwLCJcXHRNY1UgXFx0aW1lc197XFxNY1V9IFxcdE1jVSJdLFswLDAsIlxcdE1jVV8wIFxcdGltZXNfe1xcTWNVXzB9IFxcdE1jVV8wIl0sWzAsMSwiXFx0TWNVXzAiXSxbMSwxLCJcXHRNY1UiXSxbMSwwLCJcXHRFbCBcXHRpbWVzX1xcRWwgXFx0RWwiXSxbMSwyLCJcXElkXzAiLDIseyJvZmZzZXQiOjF9XSxbMiwzLCJcXEVsIiwyLHsibGFiZWxfcG9zaXRpb24iOjMwfV0sWzAsMywiXFxJZCIsMCx7Im9mZnNldCI6MX1dXQ==
    \begin{tikzcd}[cramped]
      {\tMcU_0 \times_{\McU_0} \tMcU_0} & {\tMcU \times_{\McU} \tMcU} \\
      {\tMcU_0} & \tMcU
      \arrow["{\tEl \times_\El \tEl}", from=1-1, to=1-2]
      \arrow["{\Path_0}"', shift right, from=1-1, to=2-1]
      \arrow["{\Path}", shift right, from=1-2, to=2-2]
      \arrow["\El"'{pos=0.5}, from=2-1, to=2-2]
    \end{tikzcd}
  \end{equation*}
\end{definition}

\subsubsection{Heterogeneous Path-Type Structure}
An unfortunate feature of usual $\Path$-type structure is that their fibrancy
over the product is not preserved under formation of inverse diagrams, as
illustrated in the following example.

\begin{example}
  Let $\bC$ be a bicomplete locally cartesian closed category with a set of
  pullback-stable fibrations and let $I \in \bC$ be a equipped with an interval
  structure $\partial I \coloneqq 1 \sqcup 1 \to I$ (relative to all objects).
  Assume the endpoint evaluation map $P_B^I(E) \twoheadrightarrow E \times_B E$
  is a fibration whenever $E \twoheadrightarrow B$ is a fibration so that the
  interval object gives rise to fibred path objects.

  Now, consider the arrow category $\bC^{1 \to 0}$ where the fibrations are the
  Reedy fibrations and the interval object is chosen to be constant at $I$.
  Then in general $\bC^{1 \to 0}$ might not have the property that the constant
  interval object gives rise to fibred path objects.
  To see this, fix a Reedy fibration $E \to B \in \bC^{1 \to 0}$ so that the
  local exponential $[E \times I, B]_B \in \sfrac{\bC^{1 \to 0}}{B_1 \to B_0}$
  is given by
  $[B_1 \times I, E_1]_{B_1} \to [B_1 \times I, E_0^*]_{B_1} \to [B_0 \times I,
  E_0]_{B_0}$ calculated diagrammatically as below.

  \begin{equation*}\small
    % https://q.uiver.app/#q=WzAsOSxbMywyLCJCXzEiXSxbNSwyLCJCXzAiXSxbNiwxLCJFXzAgXFx0aW1lc197Ql8wfSBFXzAiXSxbNSwwLCJbQl8wIFxcdGltZXMgSSwgRV8wXV97Ql8wfSJdLFszLDAsIltCXzEgXFx0aW1lcyBJLCBFXzBeKl1fe0JfMX0iXSxbNCwxLCJFXzBeKiBcXHRpbWVzX3tCXzF9IEVfMF4qIl0sWzIsMSwiRV8xIFxcdGltZXNfe0JfMX0gRV8xIl0sWzAsMCwiW0JfMSBcXHRpbWVzIEksIEVfMV1fe0JfMX0iXSxbMSwwLCJcXGJ1bGxldCJdLFszLDFdLFsyLDFdLFszLDIsIiIsMix7InN0eWxlIjp7ImhlYWQiOnsibmFtZSI6ImVwaSJ9fX1dLFs0LDBdLFswLDFdLFs0LDNdLFs0LDVdLFs1LDBdLFs1LDJdLFs2LDBdLFs2LDVdLFs3LDZdLFs4LDZdLFs3LDhdLFs4LDRdLFs4LDUsIiIsMSx7InN0eWxlIjp7Im5hbWUiOiJjb3JuZXIifX1dLFs1LDEsIiIsMSx7InN0eWxlIjp7Im5hbWUiOiJjb3JuZXIifX1dLFs0LDEsIiIsMSx7InN0eWxlIjp7Im5hbWUiOiJjb3JuZXIifX1dXQ==
    \begin{tikzcd}[cramped, column sep=small]
      {[B_1 \times I, E_1]_{B_1}} & \bullet && {[B_1 \times I, E_0^*]_{B_1}} && {[B_0 \times I, E_0]_{B_0}} \\
      && {E_1 \times_{B_1} E_1} && {E_0^* \times_{B_1} E_0^*} && {E_0 \times_{B_0} E_0} \\
      &&& {B_1} && {B_0}
      \arrow[from=1-1, to=1-2, dashed, two heads]
      \arrow[from=1-1, to=2-3]
      \arrow[from=1-2, to=1-4]
      \arrow[from=1-2, to=2-3]
      \arrow[from=1-4, to=1-6]
      \arrow[from=1-4, to=2-5]
      \arrow[from=1-4, to=3-4]
      \arrow[two heads, from=1-6, to=2-7]
      \arrow[from=1-6, to=3-6]
      \arrow[from=2-3, to=3-4]
      \arrow[from=2-5, to=3-4]
      \arrow[from=2-7, to=3-6]
      \arrow[from=3-4, to=3-6]
      \arrow[from=2-3, to=2-5, crossing over]
      \arrow[from=2-5, to=2-7, crossing over]
      \arrow["\lrcorner"{anchor=center, pos=0.15, scale=1.5, rotate=45}, draw=none, from=1-2, to=2-5]
      \arrow["\lrcorner"{anchor=center, pos=0.15, scale=1.5}, draw=none, from=1-4, to=3-6]
      \arrow["\lrcorner"{anchor=center, pos=0.15, scale=1.5}, draw=none, from=2-5, to=3-6]
    \end{tikzcd}
  \end{equation*}

  Limits are computed pointwise, so to require that the map
  $[B \times I, B]_B \to E \times_B E$ is a Reedy fibration is by definition to
  require the relative matching map
  \begin{equation*}
    [B_1 \times I, E_1]_{B_1}
    \xrightarrow{\quad}
    [B_0 \times I, E_0]_{B_0} \times_{(E_0 \times_{B_0} E_0)} (E_1 \times_{B_1} E_1)
  \end{equation*}
  to be a fibration.
  But one realises that the codomain object, written as the pullback of the
  entire top face of the diagram above, is also just the pullback of the top
  left face above.
  \begin{equation*}
    [B_0 \times I, E_0]_{B_0} \times_{(E_0 \times_{B_0} E_0)} (E_1 \times_{B_1} E_1) \cong
    [B_1 \times I, E_0^*]_{B_1} \times_{(E_0^* \times_{B_1} E_0^*)} (E_1 \times_{B_1} E_1)
  \end{equation*}
  So, by adhesiveness, the relative matching map is in fact the pullback-Hom of
  $\partial I \times B_1 \to I \times B_1$ with $E_1 \twoheadrightarrow E_0^*$
  in the slice over $B_1$.
  The latter of these maps is a fibration because $E \to B$ is a Reedy
  fibration.

  Therefore, by \cite[Lemma 5.7]{struct-lift}, it seems that to ensure fibred
  path objects are constructed by mapping out of the constant interval object in
  inverse diagrams on $\bC$, one requires that the pullback-Hom of a
  (cofibration, fibration)-pair is still a fibration in $\bC$ .
\end{example}

This lack of closure under inverse diagrams appears problematic for replicating
the path object construction in the category of models from \cite{kl18}.
Informed by the above example, the following more general formulation of
$\Path$-type structures, in the heterogeneous sense, which captures the axiom
that the pullback-Hom of the endpoint inclusion cofibration and a fibration
remains a fibration, does satisfy the closure property.

To state the formulation, we first construct the classifying heterogeneous
$\Path$-type map.
\begin{construction}\label{def:htr-path-square}
  Assume that $\partial\Cof \hookrightarrow \Cof$ has a truth structure relative
  to $\tMcU \to \McU$ under which $\bI$ has a disjoint $\iota$-cofibrant
  endpoint structure.

  In the following, we denote by $\txtis_\partial^*$ and $\ev^*(\tMcU \times \pi)$ to be respectively the map
  $\partial\bI \to \bI$ rebased over $\pi_*(\tMcU \times \McU)$ and the map
  $\GenComp(\tMcU\twoheadrightarrow\McU, \tMcU\twoheadrightarrow\McU)$ viewed as a map over $\pi_*(\tMcU \times \McU)$.
  \begin{equation*}
    % https://q.uiver.app/#q=WzAsNSxbMSwyLCJcXHBpXyooXFx0TWNVIFxcdGltZXMgXFxNY1UpIl0sWzAsMSwiXFxwaV8qKFxcdE1jVSBcXHRpbWVzIFxcTWNVKSBcXHRpbWVzIFxcYkkiXSxbMCwwLCJcXHBpXyooXFx0TWNVIFxcdGltZXMgXFxNY1UpIFxcdGltZXMgXFxwYXJ0aWFsXFxiSSJdLFsyLDEsIlxccGleKlxccGlfKihcXHRNY1UgXFx0aW1lcyBcXE1jVSkiXSxbMiwwLCJcXGV2XiooXFx0TWNVIFxcdGltZXMgXFx0TWNVKSJdLFsyLDEsIlxccGlfKihcXHRNY1UgXFx0aW1lcyBcXE1jVSkgXFx0aW1lcyBcXHR4dGlzX3tcXHBhcnRpYWx9IiwyXSxbNCwzLCJcXGV2XiooXFx0TWNVIFxcdGltZXMgXFxwaSkiXSxbMSwwXSxbMywwXV0=&macro_url=https%3A%2F%2Fgist.githubusercontent.com%2Flim495062%2F61b94af9ef95c1c7b0763c937de29c2b%2Fraw%2F254aa652b96f7d1c6dece2f3eba929431bb0adfe%2Facmhwmacros.sty
    \begin{tikzcd}[cramped]
      {\pi_*(\tMcU \times \McU) \times \partial\bI} && {\ev^*(\tMcU \times \tMcU)} \\
      {\pi_*(\tMcU \times \McU) \times \bI} && {\pi^*\pi_*(\tMcU \times \McU)} \\
      & {\pi_*(\tMcU \times \McU)}
      \arrow["{\txtis_{\partial}^*}"', from=1-1, to=2-1]
      \arrow["{\ev^*(\tMcU \times \pi)}", from=1-3, to=2-3]
      \arrow[from=2-1, to=3-2]
      \arrow[from=2-3, to=3-2]
    \end{tikzcd} \in \sfrac{\bC}{\pi_*(\tMcU \times \McU)}
  \end{equation*}
  Then, applying \Cref{def:pullback-Hom} on the above pair of maps over
  $\pi_*(\tMcU \times \McU)$ gives the following map on the left which one
  recognises to be the same as the map
  $\hP_{\pi_*(\tMcU\times\McU)}^\bI(\ev^*(\tMcU \times \pi))$ obtained from
  \cite[Definition 5.6]{struct-lift} on the right.
  \begin{equation*}
    % https://q.uiver.app/#q=WzAsMixbMCwwLCJcXERpYWdfe1xccGlfKihcXHRNY1UgXFx0aW1lcyBcXE1jVSl9KFxccGlfKihcXHRNY1UgXFx0aW1lcyBcXE1jVSkgXFx0aW1lcyBcXHR4dGlzX3tcXHBhcnRpYWx9LCBcXGV2XiooXFx0TWNVIFxcdGltZXMgXFxwaSkpIl0sWzAsMSwiXFxTcV97XFxwaV8qKFxcdE1jVSBcXHRpbWVzIFxcTWNVKX0oXFxwaV8qKFxcdE1jVSBcXHRpbWVzIFxcTWNVKSBcXHRpbWVzIFxcdHh0aXNfe1xccGFydGlhbH0sIFxcZXZeKihcXHRNY1UgXFx0aW1lcyBcXHBpKSkiXSxbMCwxLCJcXGhQX3tcXHBpXyooXFx0TWNVXFx0aW1lc1xcTWNVKX1eXFxiSShcXGV2XiooXFx0TWNVIFxcdGltZXMgXFx0TWNVKSBcXHRvICAgXFxwaV4qXFxwaV8qKFxcdE1jVSBcXHRpbWVzIFxcTWNVKSkiXV0=&macro_url=https%3A%2F%2Fgist.githubusercontent.com%2Flim495062%2F61b94af9ef95c1c7b0763c937de29c2b%2Fraw%2F254aa652b96f7d1c6dece2f3eba929431bb0adfe%2Facmhwmacros.sty
    \left(\begin{tikzcd}[cramped, row sep=large]
        {\Diag_{\pi_*(\tMcU \times \McU)}(\txtis_{\partial}^*, \ev^*(\tMcU \times \pi))} \\
        {\Sq_{\pi_*(\tMcU \times \McU)}(\txtis_{\partial}^*, \ev^*(\tMcU \times \pi))}
        \arrow["", from=1-1, to=2-1]
      \end{tikzcd}\right)
    =
    \left(\begin{tikzcd}[cramped, row sep=large]
        {P_{\pi_*(\tMcU\times\McU)}^\bI(\ev^*(\tMcU\times\tMcU))} \\
        {\Sq_{\pi_*(\tMcU \times \McU)}(\txtis_{\partial}^*, \ev^*(\tMcU \times \pi))}
        \arrow["{\hP_{\pi_*(\tMcU\times\McU)}^\bI(\ev^*(\tMcU \times \pi))}"{description}, from=1-1, to=2-1]
      \end{tikzcd}\right)
  \end{equation*}
\end{construction}

We also state a construction that allows one to express closure of an internal
universe by heterogeneous $\Path$-types.
\begin{construction}\label{def:htr-path-square-pb}
  Assume that $\partial\Cof \hookrightarrow \Cof$ has a truth structure under
  which $\bI$ has a disjoint $\iota$-cofibrant endpoint structure.

  If $\pi_0 \colon \tMcU_0 \to \McU_0$ is an internal universe of
  $\tMcU \to \McU$ then by \Cref{constr:int-univ-gencomp},
  $\GenComp(\pi_0,\pi_0)$ occurs as a pullback of $\GenComp(\pi,\pi)$, so that
  \Cref{lem:pullback-Hom-stable} gives rise to iterated pullback squares as
  follows
  \begin{equation*}
    % https://q.uiver.app/#q=WzAsNixbMSwwLCJQX3tcXHBpXyooXFx0TWNVXFx0aW1lc1xcTWNVKX1eXFxiSShcXGV2XiooXFx0TWNVXFx0aW1lc1xcdE1jVSkpIl0sWzEsMSwiXFxTcV97XFxwaV8qKFxcdE1jVSBcXHRpbWVzIFxcTWNVKX0oXFx0eHRpc197XFxwYXJ0aWFsfV4qLCBcXGV2XiooXFx0TWNVIFxcdGltZXMgXFxwaSkpIl0sWzAsMSwiXFxTcV97KFxccGlfMClfKihcXHRNY1VfMCBcXHRpbWVzIFxcTWNVXzApfShcXHR4dGlzX3tcXHBhcnRpYWx9XiosIFxcZXZeKihcXHRNY1VfMCBcXHRpbWVzIFxccGlfMCkpIl0sWzAsMCwiUF97KFxccGlfMClfKihcXHRNY1VfMFxcdGltZXNcXE1jVV8wKX1eXFxiSShcXGV2XiooXFx0TWNVXzBcXHRpbWVzXFx0TWNVXzApKSJdLFsxLDIsIlxccGlfKihcXHRNY1UgXFx0aW1lcyBcXE1jVSkiXSxbMCwyLCIoXFxwaV8wKV8qKFxcdE1jVV8wIFxcdGltZXMgXFxNY1VfMCkiXSxbMCwxXSxbMywyXSxbMSw0XSxbMiw1XSxbNSw0XSxbMiwxXSxbMywwXSxbMiw0LCIiLDEseyJzdHlsZSI6eyJuYW1lIjoiY29ybmVyIn19XSxbMywxLCIiLDEseyJzdHlsZSI6eyJuYW1lIjoiY29ybmVyIn19XV0=&macro_url=https%3A%2F%2Fgist.githubusercontent.com%2Flim495062%2F61b94af9ef95c1c7b0763c937de29c2b%2Fraw%2F254aa652b96f7d1c6dece2f3eba929431bb0adfe%2Facmhwmacros.sty
    \begin{tikzcd}[cramped]
      {P_{(\pi_0)_*(\tMcU_0\times\McU_0)}^\bI(\ev^*(\tMcU_0\times\tMcU_0))} & {P_{\pi_*(\tMcU\times\McU)}^\bI(\ev^*(\tMcU\times\tMcU))} \\
      {\Sq_{(\pi_0)_*(\tMcU_0 \times \McU_0)}(\txtis_{\partial}^*, \ev^*(\tMcU_0 \times \pi_0))} & {\Sq_{\pi_*(\tMcU \times \McU)}(\txtis_{\partial}^*, \ev^*(\tMcU \times \pi))} \\
      {(\pi_0)_*(\tMcU_0 \times \McU_0)} & {\pi_*(\tMcU \times \McU)}
      \arrow[from=1-1, to=1-2]
      \arrow[from=1-1, to=2-1]
      \arrow[from=1-2, to=2-2]
      \arrow[from=2-1, to=2-2]
      \arrow[from=2-1, to=3-1]
      \arrow[from=2-2, to=3-2]
      \arrow[from=3-1, to=3-2]
      \arrow["\lrcorner"{anchor=center, pos=0.15, scale=1.5}, draw=none, from=1-1, to=2-2]
      \arrow["\lrcorner"{anchor=center, pos=0.15, scale=1.5}, draw=none, from=2-1, to=3-2]
    \end{tikzcd}
  \end{equation*}
\end{construction}

We can now formulate the heterogeneous $\Path$-type by requiring a fibration
structure on the corresponding map.
\begin{definition}\label{def:htr-path-type}
  Assume that $\partial\Cof \hookrightarrow \Cof$ has a truth structure relative
  to $\tMcU \to \McU$ which $\bI$ has a disjoint $\iota$-cofibrant endpoint
  structure.

  A \emph{heterogeneous $\Path$-type structure} on $\tMcU \to \McU$ is a pair of
  dashed maps $(\hetvia, \hetPath)$ giving a $\pi$-fibration structure on the
  pullback-Hom $\hP_{\pi_*(\tMcU\times\McU)}^\bI(\ev^*(\tMcU \times \pi))$
  construction, as from \Cref{def:htr-path-square}.
  That is, it is a pullback square as follows.
  \begin{equation*}
    % https://q.uiver.app/#q=WzAsNCxbMCwwLCJQX3tcXHBpXyooXFx0TWNVIFxcdGltZXMgXFxNY1UpfV5cXGJJKFxcZXZeKihcXHRNY1VcXHRpbWVzXFx0TWNVKSkiXSxbMCwxLCJcXFNxX3tcXHBpXyooXFx0TWNVIFxcdGltZXMgXFxNY1UpfShcXHR4dGlzX1xccGFydGlhbF4qLFxcZXZeKihcXHRNY1VcXHRpbWVzXFxwaSkpIl0sWzEsMCwiXFx0TWNVIl0sWzEsMSwiXFxNY1UiXSxbMiwzXSxbMSwzLCJcXGhldFBhdGgiLDIseyJzdHlsZSI6eyJib2R5Ijp7Im5hbWUiOiJkYXNoZWQifX19XSxbMCwyLCJcXGhldHZpYSIsMCx7InN0eWxlIjp7ImJvZHkiOnsibmFtZSI6ImRhc2hlZCJ9fX1dLFswLDEsIlxcdGV4dHtcXENyZWZ7ZGVmOmh0ci1wYXRoLXNxdWFyZX19IiwyXSxbMCwzLCIiLDEseyJzdHlsZSI6eyJuYW1lIjoiY29ybmVyIn19XV0=&macro_url=https%3A%2F%2Fgist.githubusercontent.com%2Flim495062%2F61b94af9ef95c1c7b0763c937de29c2b%2Fraw%2F254aa652b96f7d1c6dece2f3eba929431bb0adfe%2Facmhwmacros.sty
    \begin{tikzcd}[cramped]
      {P_{\pi_*(\tMcU \times \McU)}^\bI(\ev^*(\tMcU\times\tMcU))} & \tMcU \\
      {\Sq_{\pi_*(\tMcU \times \McU)}(\txtis_\partial^*,\ev^*(\tMcU\times\pi))} & \McU
      \arrow["\hetvia", dashed, from=1-1, to=1-2]
      \arrow["{\text{\Cref{def:htr-path-square}}}"', from=1-1, to=2-1]
      \arrow["\lrcorner"{anchor=center, pos=0.05, scale=1.5}, draw=none, from=1-1, to=2-2]
      \arrow[from=1-2, to=2-2]
      \arrow["\hetPath"', dashed, from=2-1, to=2-2]
    \end{tikzcd}
  \end{equation*}

  Now suppose that $\pi_0 \colon \tMcU_0 \to \McU_0$ is an internal universe of
  $\tMcU \to \McU$ with $\pi$-name $\El \colon \McU_0 \to \McU$.
  Then, by the middle layer map of \Cref{def:htr-path-square-pb}, one obtains
  \begin{equation*}
    \Sq_{(\pi_0)_*(\tMcU_0 \times \McU_0)}(\txtis_{\partial}^*, \ev^*(\tMcU_0 \times \pi_0))
    \xrightarrow{\text{\Cref{def:htr-path-square-pb}}}
    \Sq_{\pi_*(\tMcU \times \McU)}(\txtis_{\partial}^*, \ev^*(\tMcU \times \pi))
  \end{equation*}
  If the internal universe is also equipped with a heterogeneous $\Path$-type
  structure $\hetPath_0$ then it is \emph{compatible} with the heterogeneous
  $\Path$-type structure of the external ambient universe when the following
  diagram commutes.
  \begin{equation*}
    % https://q.uiver.app/#q=WzAsNCxbMCwxLCJcXFNxX3tcXHBpXyooXFx0TWNVIFxcdGltZXMgXFxNY1UpfShcXHR4dGlzX1xccGFydGlhbF4qLFxcZXZeKihcXHRNY1VcXHRpbWVzXFxwaSkpIl0sWzEsMSwiXFxNY1UiXSxbMCwwLCJcXFNxX3soXFxwaV8wKV8qKFxcdE1jVV8wIFxcdGltZXMgXFxNY1VfMCl9KFxcdHh0aXNfe1xccGFydGlhbH1eKiwgXFxldl4qKFxcdE1jVV8wIFxcdGltZXMgXFxwaV8wKSkiXSxbMSwwLCJcXE1jVV8wIl0sWzAsMSwiXFxoZXRQYXRoIiwyXSxbMiwwLCJcXHRleHR7XFxDcmVme2RlZjpodHItcGF0aC1zcXVhcmUtcGJ9fSIsMl0sWzMsMSwiXFxFbCJdLFsyLDMsIlxcaGV0UGF0aF8wIl1d&macro_url=https%3A%2F%2Fgist.githubusercontent.com%2Flim495062%2F61b94af9ef95c1c7b0763c937de29c2b%2Fraw%2F254aa652b96f7d1c6dece2f3eba929431bb0adfe%2Facmhwmacros.sty
    \begin{tikzcd}[cramped]
      {\Sq_{(\pi_0)_*(\tMcU_0 \times \McU_0)}(\txtis_{\partial}^*, \ev^*(\tMcU_0 \times \pi_0))} & {\McU_0} \\
      {\Sq_{\pi_*(\tMcU \times \McU)}(\txtis_\partial^*,\ev^*(\tMcU\times\pi))} & \McU
      \arrow["{\hetPath_0}", from=1-1, to=1-2]
      \arrow["{\text{\Cref{def:htr-path-square-pb}}}"', from=1-1, to=2-1]
      \arrow["\El", from=1-2, to=2-2]
      \arrow["\hetPath"', from=2-1, to=2-2]
    \end{tikzcd}
  \end{equation*}
\end{definition}

Heterogeneous $\Path$-types are indeed generalisations of $\Path$-types under
the existence of $\Unit$ types.
This is because
$\ev^*(\tMcU \times \tMcU) \to \pi^*\pi_*(\tMcU \times \McU) \to \pi_*(\tMcU \to
\McU)$ is the generic for maps obtained by composing two composable pullbacks of
$\tMcU \to \McU$ and \cite[Lemma 5.12]{struct-lift} shows that
pullback-Homs of fibred path objects are pullback-stable.
We formally record this fact.

\begin{proposition}\label{prop:htr-path-type-path}
  Assume that $\partial\Cof \hookrightarrow \Cof$ has a truth structure relative
  to $\tMcU \to \McU$ and $\bI$ has a disjoint $\iota$-cofibrant endpoint
  structure.

  If $\tMcU \to \McU$ admits a $\Unit$-type structure then every heterogeneous
  $\Path$-type structure on $\tMcU \to \McU$ induces a $\Path$-type structure on
  $\tMcU \to \McU$.
\end{proposition}
\begin{proof}
  Because $\bI$ has a disjoint cofibrant endpoint structure, it follows that
  \begin{equation*}
    P_\McU^\bI(\McU) \times_{\partial P_\McU^\bI(\McU)} \partial P_\McU(\tMcU)
    \cong
    [\McU \times \bI, \McU]_{\McU} \times_{[\McU \times \partial \bI, \McU]_{\McU}}
    [\McU \times \partial \bI, \tMcU]_{\McU}
    \cong
    \tMcU \times_\McU \tMcU
  \end{equation*}
  by \Cref{lem:int-int-obj}.
  Therefore, $P_\McU^\bI(\McU) \to \tMcU \times_\McU \tMcU$ is the pullback-Hom
  $\hP_\McU^\bI(\tMcU \to \McU)$ applied on
  $\tMcU \to \McU \in \sfrac{\bC}{\McU}$.

  Thus, by \cite[Lemma 5.12]{struct-lift}, it suffices to exhibit
  $\tMcU \to \McU \xrightarrow{=} \McU$ as a pullback of
  $\ev^*(\tMcU \times \tMcU) \to \pi^*\pi_*(\tMcU \times \McU) \to \pi_*(\tMcU
  \to \McU)$.
  But then by the $\Unit$-structure, all identity maps are $\pi$-fibrations, so
  the result follows by the universality of
  $\GenComp(\tMcU\twoheadrightarrow\McU, \tMcU\twoheadrightarrow\McU)$ given by
  \Cref{lem:gen-comp}.
\end{proof}

%%% Local Variables:
%%% TeX-master: "./main.tex"
%%% TeX-engine: default
%%% End:

\subsection{Filling Operation}\label{subsec:filling-operation}
In this section, we axiomatise the biased filling operation and its various
computation rules.
We also use it to derive $\Id$-types and $\MsJ$-elimination.

\begin{definition}\label{def:biased-filling}
  Suppose $\partial\Cof \hookrightarrow \Cof$ has a truth structure relative
  to $\tMcU \to \McU$ under which $\bI$ has a disjoint cofibrant endpoint
  structure.

  A \emph{$\delta$-biased filling structure} on these data is then a lifting
  structure, for $\delta = \mbbo{0},\mbbo{1}$ in $\sfrac{\bC}{\Cof}$, as defined
  in \cite[Definitions 1.4 and 3.1]{struct-lift}, as follows.
  \begin{equation*}
    \Fill_\delta \in \left(
      \begin{tikzcd}[cramped]
        (\Cof \vee \txtis_{\delta})^*\partial\Cof \ar[d,hook] \\ \Cof \times \bI
      \end{tikzcd}
      \fracsquareslash{\Cof}
      \begin{tikzcd}[cramped]
        \Cof \times \tMcU \ar[d] \\ \Cof \times \McU
      \end{tikzcd}
    \right)
  \end{equation*}
\end{definition}

\newsavebox{\dCofI}
\begin{lrbox}{\dCofI}{\scriptsize\begin{tikzcd}[column sep=small, row sep=small, cramped]
      {(\Cof \vee \txtis_{\delta})^*\partial\Cof} \ar[d,hook] \\ {\Cof \times \bI}\end{tikzcd}}
\end{lrbox}
\newsavebox{\tU}
\begin{lrbox}{\tU}{\scriptsize\begin{tikzcd}[column sep=small, row sep=small, cramped]
        {\Cof \times \tMcU} \ar[d] \\ {\Cof \times \McU}\end{tikzcd}}
\end{lrbox}

\begin{remark}\label{rmk:biased-filling}
  Suppose $\partial\Cof \hookrightarrow \Cof$ has a truth structure respective
  to $\tMcU \to \McU$ under which $\bI$ has a disjoint cofibrant endpoint
  structure.
  We explain syntactically what it means to have a $\delta$-biased filling
  structure.

  For an object $\Gamma \in \bC$, we write $\Gamma.\bI$ to mean
  $\Gamma \times \bI$.
  Then, for a map $\phi \colon \Gamma \to \Cof$, we can construct the iterated
  pullback $\Gamma.\bI.(\phi^* \vee \txtis_\delta) \hookrightarrow \Gamma.\bI$
  as follows.
  \begin{equation*}
    % https://q.uiver.app/#q=WzAsOCxbMSwxLCJcXENvZiBcXHRpbWVzIFxcYkkiXSxbMiwxLCJcXENvZiBcXHRpbWVzIFxcQ29mIl0sWzIsMCwiXFxwYXJ0aWFsKFxcQ29mIFxcdmVlIFxcQ29mKSJdLFsxLDAsIihcXENvZiBcXHZlZSBcXHR4dGlzX1xcZGVsdGEpXipcXHBhcnRpYWxcXENvZiJdLFswLDEsIlxcR2FtbWEuXFxiSSJdLFswLDAsIlxcR2FtbWEuXFxiSS4oXFxwaGleKiBcXHZlZSBcXHR4dGlzX3tcXGRlbHRhfSkiXSxbMywxLCJcXENvZiJdLFszLDAsIlxccGFydGlhbFxcQ29mIl0sWzIsMSwiIiwwLHsic3R5bGUiOnsidGFpbCI6eyJuYW1lIjoiaG9vayIsInNpZGUiOiJ0b3AifX19XSxbMywwLCIiLDAseyJzdHlsZSI6eyJ0YWlsIjp7Im5hbWUiOiJob29rIiwic2lkZSI6InRvcCJ9fX1dLFswLDEsIlxcQ29mIFxcdGltZXMgXFx0eHRpc197XFxkZWx0YX0iLDJdLFszLDJdLFs0LDAsIlxccGhpXiogXFx0aW1lcyBcXGJJIiwyXSxbMSw2LCItXFx2ZWUtIiwyXSxbNyw2LCIiLDEseyJzdHlsZSI6eyJ0YWlsIjp7Im5hbWUiOiJob29rIiwic2lkZSI6InRvcCJ9fX1dLFsyLDddLFs1LDQsIiIsMSx7InN0eWxlIjp7InRhaWwiOnsibmFtZSI6Imhvb2siLCJzaWRlIjoidG9wIn19fV0sWzUsM10sWzUsMCwiIiwyLHsic3R5bGUiOnsibmFtZSI6ImNvcm5lciJ9fV0sWzMsMSwiIiwyLHsic3R5bGUiOnsibmFtZSI6ImNvcm5lciJ9fV0sWzIsNiwiIiwyLHsic3R5bGUiOnsibmFtZSI6ImNvcm5lciJ9fV1d
    \begin{tikzcd}[cramped]
      {\Gamma.\bI.(\phi^* \vee \txtis_{\delta})} & {(\Cof \vee \txtis_\delta)^*\partial\Cof} & {\partial(\Cof \vee \Cof)} & {\partial\Cof} \\
      {\Gamma.\bI} & {\Cof \times \bI} & {\Cof \times \Cof} & \Cof
      \arrow[from=1-1, to=1-2]
      \arrow[hook, from=1-1, to=2-1]
      \arrow[from=1-2, to=1-3]
      \arrow[hook, from=1-2, to=2-2]
      \arrow[from=1-3, to=1-4]
      \arrow[hook, from=1-3, to=2-3]
      \arrow[hook, from=1-4, to=2-4]
      \arrow["{\phi \times \bI}"', from=2-1, to=2-2]
      \arrow["{\Cof \times \txtis_{\delta}}"', from=2-2, to=2-3]
      \arrow["{-\vee-}"', from=2-3, to=2-4]
      \arrow["\lrcorner"{anchor=center, pos=0.15, scale=1.5}, draw=none, from=1-1, to=2-2]
      \arrow["\lrcorner"{anchor=center, pos=0.15, scale=1.5}, draw=none, from=1-2, to=2-3]
      \arrow["\lrcorner"{anchor=center, pos=0.15, scale=1.5}, draw=none, from=1-3, to=2-4]
    \end{tikzcd}
  \end{equation*}

  Then, by \cite[Definitions 1.4 and 3.1]{struct-lift}, $\delta$-biased filling
  structure provides a \textcolor{yellow0}{solution} to each of the following
  \textcolor{red0}{lifting problems}.
  \begin{equation*}
    % https://q.uiver.app/#q=WzAsNCxbMCwwLCJcXEdhbW1hLlxcYkkuKFxccGhpXiogXFx2ZWUgXFx0eHRpc197XFxkZWx0YX0pIl0sWzAsMSwiXFxHYW1tYS5cXGJJIl0sWzEsMCwiXFx0TWNVIl0sWzEsMSwiXFxNY1UiXSxbMCwxLCIiLDAseyJzdHlsZSI6eyJ0YWlsIjp7Im5hbWUiOiJob29rIiwic2lkZSI6InRvcCJ9fX1dLFsyLDNdLFsxLDMsIiIsMSx7ImNvbG91ciI6Wy0yOCw5NCw2N10sInN0eWxlIjp7ImJvZHkiOnsibmFtZSI6ImRhc2hlZCJ9fX1dLFswLDIsIiIsMSx7ImNvbG91ciI6Wy0yOCw5NCw2N10sInN0eWxlIjp7ImJvZHkiOnsibmFtZSI6ImRhc2hlZCJ9fX1dLFsxLDIsIiIsMSx7ImNvbG91ciI6WzE2LDkzLDcxXSwic3R5bGUiOnsiYm9keSI6eyJuYW1lIjoiZGFzaGVkIn19fV1d&macro_url=https%3A%2F%2Fgist.githubusercontent.com%2Flim495062%2F61b94af9ef95c1c7b0763c937de29c2b%2Fraw%2F254aa652b96f7d1c6dece2f3eba929431bb0adfe%2Facmhwmacros.sty
    \begin{tikzcd}[cramped]
      {\Gamma.\bI.(\phi^* \vee \txtis_{\delta})} & \tMcU \\
      {\Gamma.\bI} & \McU
      \arrow[color=red0, dashed, from=1-1, to=1-2]
      \arrow[hook, from=1-1, to=2-1]
      \arrow[from=1-2, to=2-2]
      \arrow[color=yellow0, dotted, from=2-1, to=1-2]
      \arrow[color=red0, dashed, from=2-1, to=2-2]
    \end{tikzcd}
  \end{equation*}

  By the or-structure, the left map
  $\Gamma.\bI(\phi^* \vee \txtis_\delta) \hookrightarrow \Gamma.\bI$, viewed as
  an object over $\Gamma.\bI$, is a cone under
  $\Gamma.\bI.\phi^* \hookleftarrow \Gamma.\bI.(\phi^*\wedge\txtis_\delta)
  \hookrightarrow \Gamma.\bI.\txtis_\delta$.
  Furthermore,
  $\Gamma.\bI.(\phi^* \vee \txtis_{\delta}) \mathbin{\textcolor{red0}{\to}}
  \tMcU$ as on the top row above is entirely determined by its restrictions to
  $\Gamma.\bI.\phi^*$ and $\Gamma.\bI.\txtis_{\delta}$.
  Therefore, we may equivalently replace the top row of the lifting problem with
  a pair of maps $\Gamma.\bI.\phi^* \mathbin{\textcolor{magenta0}{\to}} \tMcU$
  and $\Gamma.\bI.\txtis_{\delta} \mathbin{\textcolor{magenta0}{\to}} \tMcU$
  whose common restriction to $\Gamma.\bI.(\phi^* \wedge \txtis_\delta)$ agree, as
  follows.
  \begin{equation*}
    % https://q.uiver.app/#q=WzAsNyxbMywyLCJcXEdhbW1hLlxcYkkiXSxbMywxLCJcXEdhbW1hLlxcYkkuKFxccGhpXiogXFx2ZWUgXFx0eHRpc197XFxkZWx0YX0pIl0sWzUsMSwiXFx0TWNVIl0sWzUsMiwiXFxNY1UiXSxbMSwxLCJcXEdhbW1hLlxcYkkuXFxwaGleKiJdLFsyLDAsIlxcR2FtbWEuXFxiSS5cXHR4dGlzX3tcXGRlbHRhfSJdLFswLDAsIlxcR2FtbWEuXFxiSS4oXFxwaGleKiBcXHdlZGdlIFxcdHh0aXNfXFxkZWx0YSkiXSxbMSwwLCIiLDAseyJzdHlsZSI6eyJ0YWlsIjp7Im5hbWUiOiJob29rIiwic2lkZSI6InRvcCJ9fX1dLFsyLDNdLFswLDMsIiIsMSx7ImNvbG91ciI6Wy0yOCw5NCw2N10sInN0eWxlIjp7ImJvZHkiOnsibmFtZSI6ImRhc2hlZCJ9fX1dLFsxLDIsIiIsMSx7ImNvbG91ciI6Wy0yOCw5NCw2N10sInN0eWxlIjp7ImJvZHkiOnsibmFtZSI6ImRhc2hlZCJ9fX1dLFswLDIsIiIsMSx7ImNvbG91ciI6WzE2LDkzLDcxXSwic3R5bGUiOnsiYm9keSI6eyJuYW1lIjoiZG90dGVkIn19fV0sWzQsMSwiIiwxLHsic3R5bGUiOnsidGFpbCI6eyJuYW1lIjoiaG9vayIsInNpZGUiOiJ0b3AifX19XSxbNSwxLCIiLDEseyJzdHlsZSI6eyJ0YWlsIjp7Im5hbWUiOiJob29rIiwic2lkZSI6InRvcCJ9fX1dLFs0LDAsIiIsMSx7InN0eWxlIjp7InRhaWwiOnsibmFtZSI6Imhvb2siLCJzaWRlIjoidG9wIn19fV0sWzUsMCwiIiwxLHsic3R5bGUiOnsidGFpbCI6eyJuYW1lIjoiaG9vayIsInNpZGUiOiJ0b3AifX19XSxbNCwyLCIiLDEseyJjdXJ2ZSI6LTIsImNvbG91ciI6WzI3OCw4Miw3Nl0sInN0eWxlIjp7ImJvZHkiOnsibmFtZSI6ImRhc2hlZCJ9fX1dLFs1LDIsIiIsMSx7ImN1cnZlIjotMSwiY29sb3VyIjpbMjc4LDgyLDc2XSwic3R5bGUiOnsiYm9keSI6eyJuYW1lIjoiZGFzaGVkIn19fV0sWzYsNCwiIiwxLHsic3R5bGUiOnsidGFpbCI6eyJuYW1lIjoiaG9vayIsInNpZGUiOiJ0b3AifX19XSxbNiw1LCIiLDEseyJzdHlsZSI6eyJ0YWlsIjp7Im5hbWUiOiJob29rIiwic2lkZSI6InRvcCJ9fX1dXQ==&macro_url=https%3A%2F%2Fgist.githubusercontent.com%2Flim495062%2F61b94af9ef95c1c7b0763c937de29c2b%2Fraw%2F456d405748eab3250184512b5240467b6083b2b4%2Facmhwmacros.sty
    \begin{tikzcd}[cramped, column sep=small]
      {\Gamma.\bI.(\phi^* \wedge \txtis_\delta)} && {\Gamma.\bI.\txtis_{\delta}} \\
      & {\Gamma.\bI.\phi^*} && {\Gamma.\bI.(\phi^* \vee \txtis_{\delta})} && \tMcU \\
      &&& {\Gamma.\bI} && \McU
      \arrow[hook, from=1-1, to=1-3]
      \arrow[hook, from=1-1, to=2-2]
      \arrow[hook, from=1-3, to=2-4]
      \arrow[hook, from=1-3, to=3-4]
      \arrow[hook, from=2-2, to=3-4]
      \arrow[color=red0, dashed, from=2-4, to=2-6]
      \arrow[hook, from=2-4, to=3-4]
      \arrow[from=2-6, to=3-6]
      \arrow[color=yellow0, dotted, from=3-4, to=2-6]
      \arrow[color=red0, dashed, from=3-4, to=3-6]
      \arrow[crossing over, hook, from=2-2, to=2-4]
      \arrow[color=magenta0, curve={height=-6pt}, dashed, from=1-3, to=2-6]
      \arrow[crossing over, color=magenta0, curve={height=-12pt}, dashed, from=2-2, to=2-6]
    \end{tikzcd}
  \end{equation*}
  Therefore, deleting the map
  $\Gamma.\bI.(\phi^* \vee \txtis_{\delta}) \hookrightarrow \Gamma.\bI$, such a
  filling structure is equivalently an operation of ``extending lifts'' or
  ``lifting entire squares'' as follows.
  \begin{equation*}
    % https://q.uiver.app/#q=WzAsNixbMSwxLCJcXEdhbW1hLlxcYkkiXSxbMiwwLCJcXHRNY1UiXSxbMiwxLCJcXE1jVSJdLFswLDEsIlxcR2FtbWEuXFxiSS5cXHBoaV4qIl0sWzEsMCwiXFxHYW1tYS5cXGJJLlxcdHh0aXNfe1xcZGVsdGF9Il0sWzAsMCwiXFxHYW1tYS5cXGJJLihcXHBoaV4qIFxcd2VkZ2UgXFx0eHRpc19cXGRlbHRhKSJdLFsxLDJdLFswLDIsIiIsMSx7ImNvbG91ciI6Wy0yOCw5NCw2N10sInN0eWxlIjp7ImJvZHkiOnsibmFtZSI6ImRhc2hlZCJ9fX1dLFswLDEsIiIsMSx7ImNvbG91ciI6WzE2LDkzLDcxXSwic3R5bGUiOnsiYm9keSI6eyJuYW1lIjoiZG90dGVkIn19fV0sWzMsMCwiIiwxLHsic3R5bGUiOnsidGFpbCI6eyJuYW1lIjoiaG9vayIsInNpZGUiOiJ0b3AifX19XSxbNCwwLCIiLDEseyJzdHlsZSI6eyJ0YWlsIjp7Im5hbWUiOiJob29rIiwic2lkZSI6InRvcCJ9fX1dLFszLDEsIiIsMSx7ImNvbG91ciI6WzI3OCw4Miw3Nl0sInN0eWxlIjp7ImJvZHkiOnsibmFtZSI6ImRhc2hlZCJ9fX1dLFs0LDEsIiIsMSx7ImNvbG91ciI6WzI3OCw4Miw3Nl0sInN0eWxlIjp7ImJvZHkiOnsibmFtZSI6ImRhc2hlZCJ9fX1dLFs1LDMsIiIsMSx7InN0eWxlIjp7InRhaWwiOnsibmFtZSI6Imhvb2siLCJzaWRlIjoidG9wIn19fV0sWzUsNCwiIiwxLHsic3R5bGUiOnsidGFpbCI6eyJuYW1lIjoiaG9vayIsInNpZGUiOiJ0b3AifX19XV0=&macro_url=https%3A%2F%2Fgist.githubusercontent.com%2Flim495062%2F61b94af9ef95c1c7b0763c937de29c2b%2Fraw%2F456d405748eab3250184512b5240467b6083b2b4%2Facmhwmacros.sty
    \begin{tikzcd}[cramped, column sep=small]
      {\Gamma.\bI.(\phi^* \wedge \txtis_\delta)} & {\Gamma.\bI.\txtis_{\delta}} & \tMcU \\
      {\Gamma.\bI.\phi^*} & {\Gamma.\bI} & \McU
      \arrow[hook, from=1-1, to=1-2]
      \arrow[hook, from=1-1, to=2-1]
      \arrow[color=magenta0, dashed, from=1-2, to=1-3]
      \arrow[hook, from=1-2, to=2-2]
      \arrow[from=1-3, to=2-3]
      \arrow[hook, from=2-1, to=2-2]
      \arrow[color=yellow0, dotted, from=2-2, to=1-3]
      \arrow[color=red0, dashed, from=2-2, to=2-3]
      \arrow[crossing over, color=magenta0, dashed, from=2-1, to=1-3]
    \end{tikzcd}
  \end{equation*}
\end{remark}

\subsubsection*{Deriving Id-types}\label{subsubsec:filling-id}
In this part, we derive in \Cref{thm:Path-Id} $\Id$-type and $\MsJ$-elimination
using $\Path$-type structures under the presence of a $\mbbo{0}$-biased filling
operation, a $\min$-structure on $\bI$ and a cofibrant structure on fibration
sections.

As from \Cref{eqn:Path-def,eqn:Id-def}, when there are $\Path$-types, one would
like to take $\Id \coloneqq \Path \colon \tMcU \times_\McU \tMcU \to \McU$.
To do so, we must show that as in \Cref{eqn:J-def} the trivial path
$\const \colon \tMcU \to P_\McU^\bI(\tMcU)$ lifts uniformly on the left against
$\tMcU \to \McU$ so that we obtain the $\MsJ$-elimination operation.
The strategy is to slowly exhibit in \Cref{prop:refl-retract}, that under some
conditions, the candidate $\const$ for the $\refl$ map as a retract of a
$\mbbo{0}$-biased boundary.
Then, it lifts against $\tMcU \to \McU$ by the definition of a $\mbbo{0}$-biased
filling operation.

We first construct the $\mbbo{0}$-biased boundary which we intend to exhibit
$\refl$ as a retract of.
\begin{construction}\label{constr:isConst-or-0}
  Assume that $\partial\Cof \hookrightarrow \Cof$ has a truth structure relative
  to $\tMcU \to \McU$ under which $\bI$ has a disjoint cofibrant endpoint
  structure.
  Let $\isConst \colon P^\bI_\McU(\tMcU) \to \Cof$ be an $\iota$-name
  for the constant map $\const \colon \tMcU \hookrightarrow P_\McU^\bI(\tMcU)$
  as follows:
  \begin{equation*}
    % https://q.uiver.app/#q=WzAsNCxbMCwwLCJcXHRNY1UiXSxbMCwxLCJQX1xcTWNVXlxcYkkoXFxNY1UpIl0sWzEsMSwiXFxDb2YiXSxbMSwwLCJcXHBhcnRpYWxcXENvZiJdLFswLDEsIihcXHBpLFxcY29uc3QpIiwyLHsic3R5bGUiOnsidGFpbCI6eyJuYW1lIjoiaG9vayIsInNpZGUiOiJ0b3AifX19XSxbMywyLCIiLDAseyJzdHlsZSI6eyJ0YWlsIjp7Im5hbWUiOiJob29rIiwic2lkZSI6InRvcCJ9fX1dLFsxLDIsIlxccmhvIiwyXSxbMCwzXSxbMCwyLCIiLDEseyJzdHlsZSI6eyJuYW1lIjoiY29ybmVyIn19XV0=
    \begin{tikzcd}[cramped]
      \tMcU & {\partial\Cof} \\
      {P_\McU^\bI(\tMcU)} & \Cof
      \arrow[from=1-1, to=1-2]
      \arrow["{\const}"', hook, from=1-1, to=2-1]
      \arrow["\lrcorner"{anchor=center, pos=0.15, scale=1.5}, draw=none, from=1-1, to=2-2]
      \arrow[hook, from=1-2, to=2-2]
      \arrow["\isConst"', from=2-1, to=2-2]
    \end{tikzcd}
  \end{equation*}
  We construct the $\iota$-cofibration
  \begin{equation*}
    \begin{tikzcd}[cramped]
      (\isConst \vee \txtis_{\mbbo{0}})^*\partial\Cof
      \ar[r,hook]
      &
      P_\McU^\bI(\tMcU) \times \bI
    \end{tikzcd}
  \end{equation*}
  as the iterated pullback as follows:
  \begin{equation*}
    % https://q.uiver.app/#q=WzAsMTAsWzAsMSwiUF9cXE1jVV5cXGJJKFxcdE1jVSkgXFx0aW1lcyBcXGJJIl0sWzEsMSwiXFxDb2YgXFx0aW1lcyBcXGJJIl0sWzEsMiwiXFxDb2YiXSxbMCwyLCJQX1xcTWNVXlxcYkkoXFx0TWNVKSJdLFswLDAsIihcXHJobyBcXHZlZSBcXHR4dGlzX3tcXG1iYm97MH19KV4qXFxwYXJ0aWFsXFxDb2YiXSxbMiwxLCJcXENvZiBcXHRpbWVzIFxcQ29mIl0sWzMsMSwiXFxDb2YiXSxbMywwLCJcXHBhcnRpYWxcXENvZiJdLFsyLDAsIlxccGFydGlhbChcXENvZiBcXHRpbWVzIFxcQ29mKSJdLFsxLDAsIihcXENvZiBcXHZlZSBcXHR4dGlzX3tcXG1iYm97MH19KV4qXFxwYXJ0aWFsXFxDb2YiXSxbMywyLCJcXHJobyIsMl0sWzEsMl0sWzAsM10sWzAsMSwiXFxyaG8gXFx0aW1lc1xcIGJJIiwxXSxbNSw2LCItXFx2ZWUtIiwyXSxbMSw1LCJcXENvZiBcXHRpbWVzIFxcdHh0aXNfe1xcbWJib3swfX0iLDJdLFs3LDYsIiIsMCx7InN0eWxlIjp7InRhaWwiOnsibmFtZSI6Imhvb2siLCJzaWRlIjoidG9wIn19fV0sWzgsNSwiIiwwLHsic3R5bGUiOnsidGFpbCI6eyJuYW1lIjoiaG9vayIsInNpZGUiOiJ0b3AifX19XSxbOCw3XSxbOCw2LCIiLDEseyJzdHlsZSI6eyJuYW1lIjoiY29ybmVyIn19XSxbNCwwLCIiLDEseyJzdHlsZSI6eyJ0YWlsIjp7Im5hbWUiOiJob29rIiwic2lkZSI6InRvcCJ9fX1dLFs0LDEsIiIsMSx7InN0eWxlIjp7Im5hbWUiOiJjb3JuZXIifX1dLFs5LDhdLFs5LDEsIiIsMSx7InN0eWxlIjp7InRhaWwiOnsibmFtZSI6Imhvb2siLCJzaWRlIjoidG9wIn19fV0sWzQsOV0sWzksNSwiIiwxLHsic3R5bGUiOnsibmFtZSI6ImNvcm5lciJ9fV1d
    \begin{tikzcd}[cramped]
      {(\isConst \vee \txtis_{\mbbo{0}})^*\partial\Cof} & {(\Cof \vee \txtis_{\mbbo{0}})^*\partial\Cof} & {\partial(\Cof \times \Cof)} & {\partial\Cof} \\
      {P_\McU^\bI(\tMcU) \times \bI} & {\Cof \times \bI} & {\Cof \times \Cof} & \Cof
      \arrow[from=1-1, to=1-2]
      \arrow[hook, from=1-1, to=2-1]
      \arrow["\lrcorner"{anchor=center, pos=0.15, scale=1.5, rotate=0}, draw=none, from=1-1, to=2-2]
      \arrow[from=1-2, to=1-3]
      \arrow[hook, from=1-2, to=2-2]
      \arrow["\lrcorner"{anchor=center, pos=0.15, scale=1.5, rotate=0}, draw=none, from=1-2, to=2-3]
      \arrow[from=1-3, to=1-4]
      \arrow[hook, from=1-3, to=2-3]
      \arrow["\lrcorner"{anchor=center, pos=0.15, scale=1.5}, draw=none, from=1-3, to=2-4]
      \arrow[hook, from=1-4, to=2-4]
      \arrow["{\isConst \times \bI}"', from=2-1, to=2-2]
      \arrow["{\Cof \times \txtis_{\mbbo{0}}}"', from=2-2, to=2-3]
      \arrow["{-\vee-}"', from=2-3, to=2-4]
    \end{tikzcd}
  \end{equation*}
\end{construction}

With \Cref{constr:isConst-or-0}, our goal is now to put it in the middle and the
candidate $\const$ for the $\refl$ map on the two sides of the following diagram
so that one can exhibit a retract
\begin{equation*}\label{eqn:const-sr}\tag{\textsc{const-sr}}
  % https://q.uiver.app/#q=WzAsNixbMCwwLCJcXHRNY1UiXSxbMCwxLCJQX3tcXE1jVX1eXFxiSShcXHRNY1UpIl0sWzEsMSwiUF97XFxNY1V9XlxcYkkoXFx0TWNVKSBcXHRpbWVzIFxcYkkiXSxbMSwwLCIoXFxyaG8gXFx2ZWUgXFx0eHRpc197XFxtYmJvezB9fSleKlxccGFydGlhbFxcQ29mIl0sWzIsMCwiXFx0TWNVIl0sWzIsMSwiUF9cXE1jVV5cXGJJKFxcdE1jVSkiXSxbMCwxLCJcXGNvbnN0IiwyLHsic3R5bGUiOnsidGFpbCI6eyJuYW1lIjoiaG9vayIsInNpZGUiOiJ0b3AifX19XSxbMywyLCIiLDEseyJzdHlsZSI6eyJ0YWlsIjp7Im5hbWUiOiJob29rIiwic2lkZSI6InRvcCJ9fX1dLFs0LDUsIlxcY29uc3QiLDAseyJzdHlsZSI6eyJ0YWlsIjp7Im5hbWUiOiJob29rIiwic2lkZSI6InRvcCJ9fX1dLFsyLDUsIiIsMix7InN0eWxlIjp7ImJvZHkiOnsibmFtZSI6ImRhc2hlZCJ9fX1dLFszLDQsIiIsMCx7InN0eWxlIjp7ImJvZHkiOnsibmFtZSI6ImRhc2hlZCJ9fX1dLFswLDMsIiIsMCx7InN0eWxlIjp7ImJvZHkiOnsibmFtZSI6ImRhc2hlZCJ9fX1dLFsxLDIsIiIsMCx7InN0eWxlIjp7ImJvZHkiOnsibmFtZSI6ImRhc2hlZCJ9fX1dXQ==
  \begin{tikzcd}[cramped, column sep=small]
    \tMcU & {(\isConst \vee \txtis_{\mbbo{0}})^*\partial\Cof} & \tMcU \\
    {P_{\McU}^\bI(\tMcU)} & {P_{\McU}^\bI(\tMcU) \times \bI} & {P_\McU^\bI(\tMcU)}
    \arrow[dashed, from=1-1, to=1-2]
    \arrow["\const"', hook, from=1-1, to=2-1]
    \arrow[dashed, from=1-2, to=1-3]
    \arrow[hook, from=1-2, to=2-2]
    \arrow["\const", hook, from=1-3, to=2-3]
    \arrow[dashed, from=2-1, to=2-2]
    \arrow[dashed, from=2-2, to=2-3]
  \end{tikzcd}
\end{equation*}

In the right side square of \Cref{eqn:const-sr} we map out of
$(\isConst \vee \txtis_{\mbbo{0}})^*\partial\Cof \to \tMcU$.
This object is constructed by way of the or-structure and has $\iota$-name
\begin{equation*}
  P_\McU^\bI(\tMcU) \times \bI \xrightarrow{\isConst \vee \txtis_{\mbbo{0}}} \Cof
\end{equation*}
Therefore, one has the following square over $P_\McU^\bI(\tMcU) \times \bI$:
\begin{equation*}
  % https://q.uiver.app/#q=WzAsNSxbMSwwLCJQX1xcTWNVXlxcYkkoXFx0TWNVKSBcXHRpbWVzIFxcc2V0XFxtYmJvezB9fSJdLFswLDEsIlxcdE1jVSBcXHRpbWVzIFxcYkkiXSxbMiwyLCJQX1xcTWNVXlxcYkkoXFx0TWNVKSBcXHRpbWVzIFxcYkkiXSxbMSwxLCIoXFxyaG8gXFx2ZWUgXFx0eHRpc197XFxtYmJvezB9fSleKlxccGFydGlhbFxcQ29mIl0sWzAsMCwiXFx0TWNVIFxcdGltZXMgXFxzZXR7XFxtYmJvezB9fSJdLFs0LDEsIiIsMCx7InN0eWxlIjp7InRhaWwiOnsibmFtZSI6Imhvb2siLCJzaWRlIjoidG9wIn19fV0sWzQsMCwiIiwyLHsic3R5bGUiOnsidGFpbCI6eyJuYW1lIjoiaG9vayIsInNpZGUiOiJ0b3AifX19XSxbMCwzLCIiLDIseyJzdHlsZSI6eyJ0YWlsIjp7Im5hbWUiOiJob29rIiwic2lkZSI6InRvcCJ9fX1dLFsxLDMsIiIsMCx7InN0eWxlIjp7InRhaWwiOnsibmFtZSI6Imhvb2siLCJzaWRlIjoidG9wIn19fV0sWzMsMiwiIiwwLHsic3R5bGUiOnsidGFpbCI6eyJuYW1lIjoiaG9vayIsInNpZGUiOiJ0b3AifX19XSxbMCwyLCIiLDAseyJjdXJ2ZSI6LTIsInN0eWxlIjp7InRhaWwiOnsibmFtZSI6Imhvb2siLCJzaWRlIjoidG9wIn19fV0sWzEsMiwiIiwwLHsiY3VydmUiOjIsInN0eWxlIjp7InRhaWwiOnsibmFtZSI6Imhvb2siLCJzaWRlIjoidG9wIn19fV1d
  \begin{tikzcd}[cramped, row sep=small, column sep=small]
    {\tMcU \times \set{\mbbo{0}}} & {P_\McU^\bI(\tMcU) \times \set{\mbbo{0}}} \\
    {\tMcU \times \bI} & {(\isConst \vee \txtis_{\mbbo{0}})^*\partial\Cof} \\
    && {P_\McU^\bI(\tMcU) \times \bI}
    \arrow[hook, from=1-1, to=1-2]
    \arrow[hook, from=1-1, to=2-1]
    \arrow[hook, from=1-2, to=2-2]
    \arrow[curve={height=-12pt}, hook, from=1-2, to=3-3]
    \arrow[hook, from=2-1, to=2-2]
    \arrow[curve={height=12pt}, hook, from=2-1, to=3-3]
    \arrow[hook, from=2-2, to=3-3]
  \end{tikzcd}
\end{equation*}
In the right side square of \Cref{eqn:const-sr} the top row maps into $\tMcU$
while the entire composite maps into $P_\McU^\bI(\tMcU)$, so it would be
convenient to know that maps out of
$(\isConst \vee \txtis_{\mbbo{0}})^*\partial\Cof$ into either $\tMcU$ or
$P_\McU^\bI(\tMcU)$ are determined uniquely by their restrictions to
$\tMcU \times \bI$ and $P_\McU^\bI(\tMcU) \times \set{\mbbo{0}}$.
This is what the following result would like to make an attempt to try to
illustrate.

\begin{lemma}\label{lem:const-or-0-res}
  Assume that $\partial\Cof \hookrightarrow \Cof$ has a truth structure relative
  to $\tMcU \to \McU$ under which $\bI$ has a disjoint cofibrant endpoint
  structure.
  Suppose that $\tMcU \to \McU$ has a $\Path$-type structure and let
  $\isConst \colon P^\bI_\McU(\tMcU) \to \Cof$ be an $\iota$-name for the
  constant map $\const \colon \tMcU \hookrightarrow P_\McU^\bI(\tMcU)$ like in
  \Cref{constr:isConst-or-0}.

  Then, for the objects $X \in \set{\tMcU, P_\McU^\bI(\tMcU)}$, one has a
  pullback square of Hom-sets as on the left below.
  \begin{center}
    \begin{minipage}{0.45\linewidth}
      \begin{equation*}
        % https://q.uiver.app/#q=WzAsNCxbMSwwLCJcXGJDKFxcdE1jVSBcXHRpbWVzIFxcYkksIFgpIl0sWzAsMSwiXFxiQygoUF9cXE1jVV5cXGJJKFxcdE1jVSkgXFx0aW1lcyBcXHNldHtcXG1iYm97MH19LCBYKSJdLFsxLDEsIlxcYkMoXFx0TWNVIFxcdGltZXMgXFxzZXR7XFxtYmJvezB9fSwgWCkiXSxbMCwwLCJcXGJDKChcXHJobyBcXHZlZSBcXHR4dGlzX3tcXG1iYm97MH19KV4qXFxwYXJ0aWFsXFxDb2YsIFgpIl0sWzMsMV0sWzAsMl0sWzEsMl0sWzMsMF0sWzMsMiwiIiwxLHsic3R5bGUiOnsibmFtZSI6ImNvcm5lciJ9fV1d
        \begin{tikzcd}[cramped, row sep=small, column sep=small]
          {\bC((\isConst \vee \txtis_{\mbbo{0}})^*\partial\Cof, X)} & {\bC(\tMcU \times \bI, X)} \\
          {\bC(P_\McU^\bI(\tMcU) \times \set{\mbbo{0}}, X)} & {\bC(\tMcU \times \set{\mbbo{0}}, X)}
          \arrow[from=1-1, to=1-2]
          \arrow[from=1-1, to=2-1]
          \arrow["\lrcorner"{anchor=center, pos=0.05, scale=1.5, rotate=0}, draw=none, from=1-1, to=2-2]
          \arrow[from=1-2, to=2-2]
          \arrow[from=2-1, to=2-2]
        \end{tikzcd}
      \end{equation*}
    \end{minipage}
    \begin{minipage}{0.45\linewidth}
      \begin{equation*}
        % https://q.uiver.app/#q=WzAsNSxbMSwwLCJQX1xcTWNVXlxcYkkoXFx0TWNVKSBcXHRpbWVzIFxcc2V0e1xcbWJib3swfX0iXSxbMCwxLCJcXHRNY1UgXFx0aW1lcyBcXGJJIl0sWzIsMiwiWCJdLFsxLDEsIihcXHJobyBcXHZlZSBcXHR4dGlzX3tcXG1iYm97MH19KV4qXFxwYXJ0aWFsXFxDb2YiXSxbMCwwLCJcXHRNY1UgXFx0aW1lcyBcXHNldHtcXG1iYm97MH19Il0sWzQsMSwiIiwwLHsic3R5bGUiOnsidGFpbCI6eyJuYW1lIjoiaG9vayIsInNpZGUiOiJ0b3AifX19XSxbNCwwLCIiLDIseyJzdHlsZSI6eyJ0YWlsIjp7Im5hbWUiOiJob29rIiwic2lkZSI6InRvcCJ9fX1dLFswLDMsIiIsMix7InN0eWxlIjp7InRhaWwiOnsibmFtZSI6Imhvb2siLCJzaWRlIjoidG9wIn19fV0sWzEsMywiIiwwLHsic3R5bGUiOnsidGFpbCI6eyJuYW1lIjoiaG9vayIsInNpZGUiOiJ0b3AifX19XSxbMywyLCIiLDAseyJzdHlsZSI6eyJib2R5Ijp7Im5hbWUiOiJkYXNoZWQifX19XSxbMCwyLCIiLDAseyJjdXJ2ZSI6LTIsInN0eWxlIjp7ImJvZHkiOnsibmFtZSI6ImRvdHRlZCJ9fX1dLFsxLDIsIiIsMCx7ImN1cnZlIjoyLCJzdHlsZSI6eyJib2R5Ijp7Im5hbWUiOiJkb3R0ZWQifX19XV0=
        \begin{tikzcd}[cramped, row sep=small, column sep=small]
          {\tMcU \times \set{\mbbo{0}}} & {P_\McU^\bI(\tMcU) \times \set{\mbbo{0}}} \\
          {\tMcU \times \bI} & {(\rho \vee \txtis_{\mbbo{0}})^*\partial\Cof} \\
          && X
          \arrow[hook, from=1-1, to=1-2]
          \arrow[hook, from=1-1, to=2-1]
          \arrow[hook, from=1-2, to=2-2]
          \arrow[curve={height=-12pt}, dotted, from=1-2, to=3-3]
          \arrow[hook, from=2-1, to=2-2]
          \arrow[curve={height=12pt}, dotted, from=2-1, to=3-3]
          \arrow[dashed, from=2-2, to=3-3]
        \end{tikzcd}
      \end{equation*}
    \end{minipage}
  \end{center}

  In other words, like on the square on the right above, maps
  $(\isConst \vee \txtis_{\mbbo{0}})^*\partial\Cof \to X$ are in the bijective
  correspondence given by way of restriction with pairs of maps
  $(\tMcU \times \bI \to X, P_\McU^\bI(\tMcU) \times \set{\mbbo{0}} \to X)$
  whose restrictions to $\tMcU \times \set{\mbbo{0}}$ agree.
\end{lemma}
\begin{proof}
  We first take any map $E \to B$ relative to which the pre-or-structure of
  $\iota$ is an or-structure and establish the result in the statement for
  $X \coloneqq E$.
  By \Cref{lem:cof-int-pushout-product}, it follows that
  \begin{equation*}
    (\Cof \vee \txtis_{\mbbo{0}})^*\partial\Cof \hookrightarrow \Cof \times \bI \in \sfrac{\bC}{\Cof}
  \end{equation*}
  approximates the pushout-product
  \begin{equation*}
    (\Cof \times \set{\mbbo{0}} \hookrightarrow \Cof \times \bI)
    \ltimes_\Cof
    (\partial\Cof \times \bI \hookrightarrow \Cof \times \bI)
  \end{equation*}
  relative to $E \times \Cof \to B \times \Cof$ in the slice over $\Cof$.

  Because the map
  $(\isConst \vee \txtis_{\mbbo{0}})^*\partial\Cof \in
  \sfrac{\bC}{P_\McU^\bI(\tMcU)}$ is the pullback of
  $(\Cof \vee \txtis_{\mbbo{0}})^*\partial\Cof \hookrightarrow \Cof \times \bI
  \in \sfrac{\bC}{\Cof}$ under $\isConst \colon P_\McU^\bI(\tMcU) \to \Cof$, one
  applies \cite[Lemma 4.5]{struct-lift} to observe that
  \begin{equation*}
    (\const \vee \txtis_{\mbbo{0}})^*\partial\Cof \hookrightarrow P_\McU^\bI(\tMcU) \times \bI \in \sfrac{\bC}{P_\McU^\bI(\tMcU)}
  \end{equation*}
  approximates the pushout-product
  \begin{equation*}
    \left( P_\McU^\bI(\tMcU) \times \set{\mbbo{0}} \hookrightarrow P_\McU^\bI(\tMcU) \times \bI \right)
    \ltimes_{P_\McU^\bI(\tMcU)}
    \left( \tMcU \times \bI \hookrightarrow P_\McU^\bI(\tMcU) \times \bI \right)
  \end{equation*}
  relative to $E \times P_\McU^\bI(\tMcU) \to B \times P_\McU^\bI(\tMcU)$ in the
  slice over $\Cof$.
  By the definition of an approximate pushout-product as in \cite[Definition
  4.4]{struct-lift}, it then reveals the following pullback square
  \begin{equation*}
    % https://q.uiver.app/#q=WzAsNCxbMCwwLCJcXHNmcmFje1xcYkN9e1BeXFxiSV9cXE1jVShcXHRNY1UpfSgoXFxjb25zdCBcXHZlZSBcXHR4dGlzX3tcXG1iYm97MH19KV4qXFxwYXJ0aWFsXFxDb2YsIEUgXFx0aW1lcyBQXlxcYklfXFxNY1UoXFx0TWNVKSkiXSxbMCwxLCJcXHNmcmFje1xcYkN9e1BeXFxiSV9cXE1jVShcXHRNY1UpfShQX1xcTWNVXlxcYkkoXFx0TWNVKSBcXHRpbWVzIFxcc2V0e1xcbWJib3swfX0sIEUgXFx0aW1lcyBQXlxcYklfXFxNY1UoXFx0TWNVKSkiXSxbMSwwLCJcXHNmcmFje1xcYkN9e1BeXFxiSV9cXE1jVShcXHRNY1UpfShcXHRNY1UgXFx0aW1lcyBcXGJJLCBFIFxcdGltZXMgUF5cXGJJX1xcTWNVKFxcdE1jVSkpIl0sWzEsMSwiXFxzZnJhY3tcXGJDfXtQXlxcYklfXFxNY1UoXFx0TWNVKX0oXFx0TWNVIFxcdGltZXMgXFxzZXR7XFxtYmJvezB9fSwgRSBcXHRpbWVzIFBeXFxiSV9cXE1jVShcXHRNY1UpKSJdLFsyLDNdLFsxLDNdLFswLDFdLFswLDJdLFswLDMsIiIsMSx7InN0eWxlIjp7Im5hbWUiOiJjb3JuZXIifX1dXQ==
    \begin{tikzcd}[cramped]
      {\sfrac{\bC}{P^\bI_\McU(\tMcU)}((\const \vee \txtis_{\mbbo{0}})^*\partial\Cof, E \times P^\bI_\McU(\tMcU))} & {\sfrac{\bC}{P^\bI_\McU(\tMcU)}(\tMcU \times \bI, E \times P^\bI_\McU(\tMcU))} \\
      {\sfrac{\bC}{P^\bI_\McU(\tMcU)}(P_\McU^\bI(\tMcU) \times \set{\mbbo{0}}, E \times P^\bI_\McU(\tMcU))} & {\sfrac{\bC}{P^\bI_\McU(\tMcU)}(\tMcU \times \set{\mbbo{0}}, E \times P^\bI_\McU(\tMcU))}
      \arrow[from=1-1, to=1-2]
      \arrow[from=1-1, to=2-1]
      \arrow["\lrcorner"{anchor=center, pos=0.05, scale=1.5}, draw=none, from=1-1, to=2-2]
      \arrow[from=1-2, to=2-2]
      \arrow[from=2-1, to=2-2]
    \end{tikzcd}
  \end{equation*}
  By the base-forgetting $\dashv$ product adjunction, one observes
  \begin{equation*}
    % https://q.uiver.app/#q=WzAsNCxbMSwwLCJcXGJDKFxcdE1jVSBcXHRpbWVzIFxcYkksIFgpIl0sWzAsMSwiXFxiQygoUF9cXE1jVV5cXGJJKFxcdE1jVSkgXFx0aW1lcyBcXHNldHtcXG1iYm97MH19LCBYKSJdLFsxLDEsIlxcYkMoXFx0TWNVIFxcdGltZXMgXFxzZXR7XFxtYmJvezB9fSwgWCkiXSxbMCwwLCJcXGJDKChcXHJobyBcXHZlZSBcXHR4dGlzX3tcXG1iYm97MH19KV4qXFxwYXJ0aWFsXFxDb2YsIFgpIl0sWzMsMV0sWzAsMl0sWzEsMl0sWzMsMF0sWzMsMiwiIiwxLHsic3R5bGUiOnsibmFtZSI6ImNvcm5lciJ9fV1d
    \begin{tikzcd}[cramped, row sep=small, column sep=small]
      {\bC((\isConst \vee \txtis_{\mbbo{0}})^*\partial\Cof, E)} & {\bC(\tMcU \times \bI, E)} \\
      {\bC(P_\McU^\bI(\tMcU) \times \set{\mbbo{0}}, E)} & {\bC(\tMcU \times \set{\mbbo{0}}, E)}
      \arrow[from=1-1, to=1-2]
      \arrow[from=1-1, to=2-1]
      \arrow["\lrcorner"{anchor=center, pos=0.05, scale=1.5, rotate=0}, draw=none, from=1-1, to=2-2]
      \arrow[from=1-2, to=2-2]
      \arrow[from=2-1, to=2-2]
    \end{tikzcd}
  \end{equation*}

  One must now instantiate the argument above with $E \to B$ with the maps
  \begin{align*}
    \tMcU \to \McU && P_\McU^\bI(\tMcU) \to \tMcU \times_\McU \tMcU
  \end{align*}
  To do so, the pre-or structure of $\iota$ relative to them must be
  or-structures.
  We have this by assumption for $\tMcU \to \McU$, so it just remains to show this for
  $P_\McU^\bI(\tMcU) \to \tMcU \times_\McU \tMcU$.
  The $\Path$-type structure on $\tMcU \to \McU$ gives a $\pi$-fibration
  structure on it.
  Hence, \Cref{lem:or-struct-pb}, it is sufficient to show that
  \Cref{eqn:or-fib} holds for $\tMcU \times_\McU \tMcU$.
  But his is once again due to \Cref{lem:or-struct-pb}, because
  $\tMcU \times_\McU \tMcU \to \tMcU$ is by construction a $\pi$-fibration and
  \Cref{eqn:or-fib} holds for $\tMcU$.
\end{proof}

We have now prepared the ingredient to map out of
$(\isConst \vee \txtis_{\mbbo{0}})^*\partial\Cof$, which helps with the top row
of the retract part of \Cref{eqn:const-sr}.
We now start to construct the bottom row of the retract part of \Cref{eqn:const-sr}.
To do so, one defines first a path restriction operation
$P_{B}^\bI(E) \times \bI \to P_{B}^\bI(E)$ over $B$ for each $E \to B$.
This operation takes a path and a point in the interval and restricts the
furthest the path can go according to the point of the interval.

\begin{construction}\label{constr:path-restrict}
  Suppose $\bI$ has a $\min$-structure.
  Given a map $E \to B$, recalling that
  $P_{B}^\bI(E) = [B \times \bI, E]_{B}$ is the local
  exponential, we define the map
  \begin{align*}
    P_{B}^\bI(E) \times \bI \xrightarrow{\pathres_{B}(E)} P_{B}^\bI(E)
  \end{align*}
  as the exponential transpose of the map
  \begin{equation*}
    \left(
      (P_{B}^\bI(E) \times \bI) \times_{B} (B \times \bI)
      \cong
      P_{B}^\bI(E) \times \bI \times \bI
      \xrightarrow{P_{B}^\bI(E) \times \min}
      P_{B}^\bI(E) \times \bI
      \cong
      P_{B}^\bI(E) \times_{B} (B \times \bI)
      \xrightarrow{\ev}
      E
    \right)^\dagger \in \sfrac{\bC}{B}
  \end{equation*}
\end{construction}

We first observe a few useful properties of the path restriction operation.
First, restricting a constant path gives back the constant path.

\begin{lemma}\label{lem:pathres-const}
  Suppose $\bI$ has a $\min$-structure.
  Then, for each $E \to B \in \bC$, the following diagram commutes.
  \begin{equation*}
    % https://q.uiver.app/#q=WzAsNCxbMCwwLCJFIFxcdGltZXMgXFxiSSJdLFsxLDAsIkUiXSxbMSwxLCJQX3tcXHVse0V9fV5cXGJJKEUpIl0sWzAsMSwiUF97XFx1bHtFfX1eXFxiSShFKSBcXHRpbWVzIFxcYkkiXSxbMSwyLCJcXGNvbnN0Il0sWzAsMSwiXFxwcm9qIl0sWzMsMiwiXFxwYXRocmVzIiwyXSxbMCwzLCJcXGNvbnN0IFxcdGltZXMgXFxiSSIsMl1d
    \begin{tikzcd}[cramped, row sep=small]
      {E \times \bI} & E \\
      {P_{B}^\bI(E) \times \bI} & {P_{B}^\bI(E)}
      \arrow["\proj", from=1-1, to=1-2]
      \arrow["{\const \times \bI}"', from=1-1, to=2-1]
      \arrow["\const", from=1-2, to=2-2]
      \arrow["\pathres"', from=2-1, to=2-2]
    \end{tikzcd}
  \end{equation*}
\end{lemma}
\begin{proof}
  We show that the transposes of the two maps agree.

  Taking the exponential transpose of the top-right path in the diagram of the
  statement is the same as first taking the product of
  $E \times \bI \to E \in \sfrac{\bC}{B}$ with $B \times \bI$ over $B$ and then
  composing with the transpose of $\const \colon E \to P_B^\bI(E)$.
  The transpose of $\const \colon E \to P_B^\bI(E)$ is exactly the projection
  map $E \times \bI \to E$, so we obtain the top-right path as below.
  By similar computation for the transpose of the bottom-left path in the
  diagram of the statement and the fact that $\pathres$ was constructed to have
  transpose $\ev \cdot (P_B^\bI(E) \times \min)$, one obtains the bottom-left
  path as below.
  \begin{equation*}
    % https://q.uiver.app/#q=WzAsNixbMCwwLCJFIFxcdGltZXMgXFxiSSBcXHRpbWVzIFxcYkkiXSxbMCwyLCJQX3tcXHVse0V9fV5cXGJJKEUpIFxcdGltZXMgXFxiSSBcXHRpbWVzIFxcYkkiXSxbMiwyLCJQX3tcXHVse0V9fV5cXGJJKEUpIFxcdGltZXMgXFxiSSJdLFszLDIsIkUiXSxbMiwxLCJFIFxcdGltZXMgXFxiSSJdLFszLDAsIkUgXFx0aW1lcyBcXGJJIl0sWzAsMSwiXFxjb25zdCBcXHRpbWVzIFxcYkkgXFx0aW1lcyBcXGJJIiwyXSxbMSwyLCJQX3tcXHVse0V9fV5cXGJJKEUpIFxcdGltZXMgXFxtaW4iLDJdLFsyLDMsIlxcZXYiLDJdLFswLDQsIkUgXFx0aW1lcyBcXG1pbiIsMV0sWzQsMywiXFxwcm9qIiwxXSxbNCwyLCJcXGNvbnN0IFxcdGltZXMgXFxiSSIsMV0sWzAsNSwiXFxwcm9qIiwxXSxbNSwzLCJcXHByb2oiLDFdXQ==
    \begin{tikzcd}[cramped, row sep=small]
      {E \times \bI \times \bI} &&& {E \times \bI} \\
      && {E \times \bI} \\
      {P_{B}^\bI(E) \times \bI \times \bI} && {P_{B}^\bI(E) \times \bI} & E
      \arrow["\proj"{description}, from=1-1, to=1-4]
      \arrow["{E \times \min}"{description}, from=1-1, to=2-3]
      \arrow["{\const \times \bI \times \bI}"', from=1-1, to=3-1]
      \arrow["\proj"{description}, from=1-4, to=3-4]
      \arrow["{\const \times \bI}"'{}, from=2-3, to=3-3]
      \arrow["\proj"{description}, from=2-3, to=3-4]
      \arrow["{P_{B}^\bI(E) \times \min}"', from=3-1, to=3-3]
      \arrow["\ev"', from=3-3, to=3-4]
    \end{tikzcd}
  \end{equation*}
  One is then able to check that these two transposes are both the diagonal
  $\proj \cdot (E \times \min)$.
\end{proof}

Next, restricting a path at $\mbbo{0}$ gives the constant path at $\mbbo{0}$.
\begin{lemma}\label{lem:pathres-const-0}
  Suppose $\bI$ has a $\min$-structure.
  Then, for each $E \to B \in \bC$, the following diagram commutes.
  \begin{equation*}
    % https://q.uiver.app/#q=WzAsNCxbMCwwLCJQX3tcXHVse0V9fVxcYkkoRSkgXFx0aW1lcyBcXHNldHtcXG1iYm97MH19Il0sWzAsMSwiUF97XFx1bHtFfX1eXFxiSShFKSBcXHRpbWVzIFxcYkkiXSxbMSwxLCJQX3tcXHVse0V9fV5cXGJJKEUpIl0sWzEsMCwiRSJdLFswLDEsIiIsMix7InN0eWxlIjp7InRhaWwiOnsibmFtZSI6Imhvb2siLCJzaWRlIjoidG9wIn19fV0sWzEsMiwiXFxwYXRocmVzIiwyXSxbMywyLCJcXGNvbnN0Il0sWzAsMywiXFxldl97XFxtYmJvezB9fSJdXQ==
    \begin{tikzcd}[cramped, row sep=small]
      {P_{B}^\bI(E) \times \set{\mbbo{0}}} & E \\
      {P_{B}^\bI(E) \times \bI} & {P_{B}^\bI(E)}
      \arrow["{\ev_{\mbbo{0}}}", from=1-1, to=1-2]
      \arrow[hook, from=1-1, to=2-1]
      \arrow["\const", from=1-2, to=2-2]
      \arrow["\pathres"', from=2-1, to=2-2]
    \end{tikzcd}
  \end{equation*}
\end{lemma}
\begin{proof}
  Once again, we transpose both maps and use the fact that $\min$ gives
  $\mbbo{0}$ when either of its arguments is $\mbbo{0}$.

  For the top-right path in the diagram of the statement, the transpose is
  computed by first taking the image of the top map
  $\ev_{\mbbo{0}} \colon P_B^\bI(E) \times \set{\mbbo{0}} \to E \in
  \sfrac{\bC}{B}$ under the fibred product functor with $B \times \bI$ before
  composing with the transpose of $\const$, which is
  $\proj \colon E \times \bI \to E$.
  But by definition the top map $\ev_{\mbbo{0}}$ in the diagram of the statement
  is
  $P_B^\bI \times \set{\mbbo{0}} \hookrightarrow P_B^\bI \times \bI
  \xrightarrow{\ev} E \times \bI$, so one arrives at the top-right path in the
  following diagram.
  \begin{equation*}
    % https://q.uiver.app/#q=WzAsNyxbMCwwLCJQX3tcXHVse0V9fV5cXGJJKEUpIFxcdGltZXMgXFxzZXR7XFxtYmJvezB9fSBcXHRpbWVzIFxcYkkiXSxbMCwyLCJQX3tcXHVse0V9fV5cXGJJKEUpIFxcdGltZXMgXFxiSSBcXHRpbWVzIFxcYkkiXSxbMiwyLCJQX3tcXHVse0V9fV5cXGJJKEUpIFxcdGltZXMgXFxiSSJdLFszLDIsIkUiXSxbMywwLCJFIFxcdGltZXMgXFxiSSJdLFsyLDAsIlBfe1xcdWx7RX19XlxcYkkoRSkgXFx0aW1lcyBcXGJJIFxcdGltZXMgXFxiSSJdLFsxLDEsIlBfe1xcdWx7RX19XlxcYkkoRSkgXFx0aW1lcyBcXHNldHtcXG1iYm97MH19Il0sWzAsMSwiIiwyLHsic3R5bGUiOnsidGFpbCI6eyJuYW1lIjoiaG9vayIsInNpZGUiOiJ0b3AifX19XSxbMSwyLCJQX3tcXHVse0V9fV5cXGJJKEUpIFxcdGltZXMgXFxtaW4iLDJdLFsyLDMsIlxcZXYiLDJdLFs0LDMsIlxccHJvaiJdLFs1LDIsIlxccHJvaiIsMV0sWzUsNCwiXFxldl8wIFxcdGltZXMgXFxiSSJdLFswLDUsIiIsMSx7InN0eWxlIjp7InRhaWwiOnsibmFtZSI6Imhvb2siLCJzaWRlIjoidG9wIn19fV0sWzYsMiwiIiwwLHsic3R5bGUiOnsidGFpbCI6eyJuYW1lIjoiaG9vayIsInNpZGUiOiJ0b3AifX19XSxbMCw2LCJcXHByb2oiLDFdXQ==
    \begin{tikzcd}[cramped, row sep=small]
      {P_{B}^\bI(E) \times \set{\mbbo{0}} \times \bI} && {P_{B}^\bI(E) \times \bI \times \bI} & {E \times \bI} \\
      & {P_{B}^\bI(E) \times \set{\mbbo{0}}} \\
      {P_{B}^\bI(E) \times \bI \times \bI} && {P_{B}^\bI(E) \times \bI} & E
      \arrow[hook, from=1-1, to=1-3]
      \arrow["\proj"{description}, from=1-1, to=2-2]
      \arrow[hook, from=1-1, to=3-1]
      \arrow["{\ev \times \bI}", from=1-3, to=1-4]
      \arrow["\proj"{description}, from=1-3, to=3-3]
      \arrow["\proj", from=1-4, to=3-4]
      \arrow[hook, from=2-2, to=3-3]
      \arrow["{P_{B}^\bI(E) \times \min}"', from=3-1, to=3-3]
      \arrow["\ev"', from=3-3, to=3-4]
    \end{tikzcd}
  \end{equation*}
  Computing the transpose for the bottom-left path of diagram in the statement
  and using \Cref{constr:path-restrict}, one arrives at the bottom-left path of
  the diagram above.

  Then, the right side square in the diagram above is by naturality of the
  counit and the diagonal in the left side square is by the fact that $\min$
  gives $\mbbo{0}$ when either of its arguments is $\mbbo{0}$.
\end{proof}

Finally, restricting a path to go farthest to $\mbbo{1}$ does nothing to the path.

\begin{lemma}\label{lem:pathres-sect}
  Suppose $\bI$ has a $\min$ structure.
  Then, for each $E \to B$, one has the following section-retraction
  pair.
  \begin{equation*}
    % https://q.uiver.app/#q=WzAsMyxbMSwwLCJQX3tcXHVuZGVybGluZXtFfX1eXFxiSShFKSBcXHRpbWVzIFxcYkkiXSxbMiwwLCJQX3tcXHVuZGVybGluZXtFfX1eXFxiSShFKSJdLFswLDAsIlBfe1xcdW5kZXJsaW5le0V9fV5cXGJJKEUpIFxcdGltZXMgXFxzZXR7XFxtYmJvezF9fSJdLFswLDEsIlxccGF0aHJlcyJdLFsyLDAsIiIsMCx7InN0eWxlIjp7InRhaWwiOnsibmFtZSI6Imhvb2siLCJzaWRlIjoidG9wIn19fV1d
    \begin{tikzcd}[cramped]
      {P_{B}^\bI(E) \times \set{\mbbo{1}}} & {P_{B}^\bI(E) \times \bI} & {P_{B}^\bI(E)}
      \arrow[hook, from=1-1, to=1-2]
      \arrow["\pathres", from=1-2, to=1-3]
    \end{tikzcd}
  \end{equation*}
\end{lemma}
\begin{proof}
  The $\min$-structure says that
  $\set{\mbbo{1}} \times \bI \hookrightarrow \bI \times \bI \xrightarrow{\min}
  \bI$ is a section-retraction pair.
  Thus, calculating the transpose, we see
  \begin{equation*}
    % https://q.uiver.app/#q=WzAsNCxbMCwwLCJQX3tcXHVse0V9fV5cXGJJKEUpIFxcdGltZXMgXFxzZXR7XFxtYmJvezF9fSBcXHRpbWVzIFxcYkkiXSxbMCwxLCJQX3tcXHVse0V9fV5cXGJJKEUpIFxcdGltZXMgXFxiSSBcXHRpbWVzIFxcYkkiXSxbMSwxLCJQX3tcXHVse0V9fV5cXGJJKEUpIFxcdGltZXMgXFxiSSJdLFsyLDEsIkUiXSxbMCwxLCIiLDIseyJzdHlsZSI6eyJ0YWlsIjp7Im5hbWUiOiJob29rIiwic2lkZSI6InRvcCJ9fX1dLFsxLDIsIlBfe1xcdWx7RX19XlxcYkkoRSkgXFx0aW1lcyBcXG1pbiIsMl0sWzIsMywiXFxldiIsMl0sWzAsMiwiXFxjb25nIl1d
    \begin{tikzcd}[cramped]
      {P_{B}^\bI(E) \times \set{\mbbo{1}} \times \bI} \\
      {P_{B}^\bI(E) \times \bI \times \bI} & {P_{B}^\bI(E) \times \bI} & E
      \arrow[hook, from=1-1, to=2-1]
      \arrow["\cong", from=1-1, to=2-2]
      \arrow["{P_{B}^\bI(E) \times \min}"', from=2-1, to=2-2]
      \arrow["\ev"', from=2-2, to=2-3]
    \end{tikzcd}
  \end{equation*}
\end{proof}

We are now ready to show that if the constant path is a cofibration then the
constant path is also a retract of a pullback of the left class in the filling
operation.

\begin{proposition}\label{prop:refl-retract}
  Suppose
  \begin{enumerate}
    \item $\partial\Cof \hookrightarrow \Cof$ has a truth structure relative to
    $\tMcU \to \McU$
    \item $\bI$ has a disjoint $\iota$-cofibrant endpoint structure
    \item $\bI$ has a $\min \colon \bI \times \bI \to \bI$ structure
    \item $\tMcU \to \McU$ has a $\Path$-type structure
    \item The map
    $\tMcU \xrightarrow{\const} P_\McU^\bI(\tMcU)$ arises as a pullback of
    $\partial\Cof \hookrightarrow \Cof$ along $\isConst$.
  \end{enumerate}
  Then, $\const \colon\tMcU \hookrightarrow P_\McU^\bI(\tMcU)$ is a retract of
  $(\isConst \vee \txtis_{\mbbo{0}})^*\partial\Cof \hookrightarrow
  P_\McU^\bI(\tMcU) \times \bI$ as in \Cref{eqn:const-sr}, which one recalls is
  the following diagram.
  \begin{equation*}\tag{\text{\ref{eqn:const-sr}}}
    % https://q.uiver.app/#q=WzAsNixbMCwwLCJcXHRNY1UiXSxbMCwxLCJQX3tcXE1jVX1eXFxiSShcXHRNY1UpIl0sWzEsMSwiUF97XFxNY1V9XlxcYkkoXFx0TWNVKSBcXHRpbWVzIFxcYkkiXSxbMSwwLCIoXFxyaG8gXFx2ZWUgXFx0eHRpc197XFxtYmJvezB9fSleKlxccGFydGlhbFxcQ29mIl0sWzIsMCwiXFx0TWNVIl0sWzIsMSwiUF9cXE1jVV5cXGJJKFxcdE1jVSkiXSxbMCwxLCJcXGNvbnN0IiwyLHsic3R5bGUiOnsidGFpbCI6eyJuYW1lIjoiaG9vayIsInNpZGUiOiJ0b3AifX19XSxbMywyLCIiLDEseyJzdHlsZSI6eyJ0YWlsIjp7Im5hbWUiOiJob29rIiwic2lkZSI6InRvcCJ9fX1dLFs0LDUsIlxcY29uc3QiLDAseyJzdHlsZSI6eyJ0YWlsIjp7Im5hbWUiOiJob29rIiwic2lkZSI6InRvcCJ9fX1dLFsyLDUsIiIsMix7InN0eWxlIjp7ImJvZHkiOnsibmFtZSI6ImRhc2hlZCJ9fX1dLFszLDQsIiIsMCx7InN0eWxlIjp7ImJvZHkiOnsibmFtZSI6ImRhc2hlZCJ9fX1dLFswLDMsIiIsMCx7InN0eWxlIjp7ImJvZHkiOnsibmFtZSI6ImRhc2hlZCJ9fX1dLFsxLDIsIiIsMCx7InN0eWxlIjp7ImJvZHkiOnsibmFtZSI6ImRhc2hlZCJ9fX1dXQ==
    \begin{tikzcd}[cramped, column sep=small]
      \tMcU & {(\isConst \vee \txtis_{\mbbo{0}})^*\partial\Cof} & \tMcU \\
      {P_{\McU}^\bI(\tMcU)} & {P_{\McU}^\bI(\tMcU) \times \bI} & {P_\McU^\bI(\tMcU)}
      \arrow[dashed, from=1-1, to=1-2]
      \arrow["\const"', hook, from=1-1, to=2-1]
      \arrow[dashed, from=1-2, to=1-3]
      \arrow[hook, from=1-2, to=2-2]
      \arrow["\const", hook, from=1-3, to=2-3]
      \arrow[dashed, from=2-1, to=2-2]
      \arrow[dashed, from=2-2, to=2-3]
    \end{tikzcd}
  \end{equation*}
\end{proposition}
\begin{proof}
  We first define the retraction part of the section-retraction pair by
  constructing a map for the top row of the right side square of
  \Cref{eqn:const-sr} as follows.
  \begin{equation*}
    (\proj, \ev_{\mbbo{0}}) \colon (\isConst \vee \txtis_{\mbbo{0}})^*\partial\Cof
    \to \tMcU
  \end{equation*}
  Then, one must check that the following square for the right side of the
  square of \Cref{eqn:const-sr} commutes.
  \begin{equation}\label{eqn:const-retract}\tag{\textsc{const-retract}}
    % https://q.uiver.app/#q=WzAsNCxbMCwxLCJQX1xcTWNVXlxcYkkoXFx0TWNVKSBcXHRpbWVzIFxcYkkiXSxbMCwwLCIoXFxyaG8gXFx2ZWUgXFx0eHRpc197XFxtYmJvezB9fSleKlxccGFydGlhbFxcQ29mIl0sWzEsMCwiXFx0TWNVIl0sWzEsMSwiUF9cXE1jVV5cXGJJKFxcdE1jVSkiXSxbMSwwLCIiLDEseyJzdHlsZSI6eyJ0YWlsIjp7Im5hbWUiOiJob29rIiwic2lkZSI6InRvcCJ9fX1dLFsxLDIsIihcXHByb2osIFxcZXZfe1xcbWJib3swfX0pIl0sWzIsM10sWzAsMywiXFxwYXRocmVzIiwyXV0=
    \begin{tikzcd}[cramped]
      {(\isConst \vee \txtis_{\mbbo{0}})^*\partial\Cof} & \tMcU \\
      {P_\McU^\bI(\tMcU) \times \bI} & {P_\McU^\bI(\tMcU)}
      \arrow["{(\proj, \ev_{\mbbo{0}})}", from=1-1, to=1-2]
      \arrow[hook, from=1-1, to=2-1]
      \arrow[from=1-2, to=2-2, hook, "{\const}"]
      \arrow["\pathres"', from=2-1, to=2-2]
    \end{tikzcd}
  \end{equation}

  By \Cref{lem:const-or-0-res}, such a map
  $(\proj, \ev_{\mbbo{0}}) \colon (\isConst \vee
  \txtis_{\mbbo{0}})^*\partial\Cof \to \tMcU$ is determined uniquely by two maps
  $P_\McU^\bI(\tMcU) \times \set{\mbbo{0}} \to \tMcU$ and
  $\tMcU \times \bI \to \tMcU$ whose restriction to
  $\tMcU \times \set{\mbbo{0}}$ are the same.
  In other words, one obtains such a map by the following commutative square,
  where the maps $P_\McU^\bI(\tMcU) \times \set{\mbbo{0}} \to \tMcU$ and
  $\tMcU \times \bI \to \tMcU$ over $\tMcU$ are respectively $\ev_{\mbbo{0}}$
  and $\proj$.
  \begin{equation*}
    % https://q.uiver.app/#q=WzAsNCxbMSwxLCJcXHRNY1UiXSxbMSwwLCJcXHRNY1UgXFx0aW1lcyBcXGJJIl0sWzAsMSwiUF9cXE1jVV5cXGJJKFxcdE1jVSkgXFx0aW1lcyBcXHNldHtcXG1iYm97MH19Il0sWzAsMCwiXFx0TWNVIFxcdGltZXMgXFxzZXR7XFxtYmJvezB9fSJdLFsxLDAsIlxccHJvaiJdLFszLDEsIiIsMCx7InN0eWxlIjp7InRhaWwiOnsibmFtZSI6Imhvb2siLCJzaWRlIjoidG9wIn19fV0sWzMsMiwiXFxjb25zdCBcXHRpbWVzIFxcc2V0e1xcbWJib3swfX0iLDJdLFsyLDAsIlxcZXZfe1xcbWJib3swfX0iLDJdXQ==
    \begin{tikzcd}[cramped]
      {\tMcU \times \set{\mbbo{0}}} \ar[rd, "\cong"{description}]& {\tMcU \times \bI} \\
      {P_\McU^\bI(\tMcU) \times \set{\mbbo{0}}} & \tMcU
      \arrow[hook, from=1-1, to=1-2]
      \arrow["{\const \times \set{\mbbo{0}}}"', from=1-1, to=2-1]
      \arrow["\proj", from=1-2, to=2-2]
      \arrow["{\ev_{\mbbo{0}}}"', from=2-1, to=2-2]
    \end{tikzcd}
  \end{equation*}

  We now check \Cref{eqn:const-retract} commutes.
  Once again by \Cref{lem:const-or-0-res}, it suffices to check separately
  commutativity when restricted to the components $\tMcU \times \bI$ and
  $P_{\McU}^\bI(\tMcU) \times \set{\mbbo{0}}$ respectively.
  In other words, one must check that the following squares commute.
  \begin{center}
    \begin{minipage}{0.45\linewidth}
      \begin{equation*}
        % https://q.uiver.app/#q=WzAsNCxbMCwwLCJcXHRNY1UgXFx0aW1lcyBcXGJJIl0sWzAsMSwiUF9cXE1jVV5cXGJJKFxcdE1jVSkgXFx0aW1lcyBcXGJJIl0sWzEsMSwiUF9cXE1jVV5cXGJJKFxcdE1jVSkiXSxbMSwwLCJcXHRNY1UiXSxbMCwxLCJcXHJlZmwgXFx0aW1lcyBcXGJJIiwyXSxbMSwyLCJcXHBhdGhyZXMiLDJdLFszLDIsIlxccmVmbCJdLFswLDMsIlxccHJvaiJdXQ==
        \begin{tikzcd}[cramped]
          {\tMcU \times \bI} & \tMcU \\
          {P_\McU^\bI(\tMcU) \times \bI} & {P_\McU^\bI(\tMcU)}
          \arrow["\proj", from=1-1, to=1-2]
          \arrow["{\const \times \bI}"', from=1-1, to=2-1, hook]
          \arrow["{\const}", from=1-2, to=2-2, hook]
          \arrow["\pathres"', from=2-1, to=2-2]
        \end{tikzcd}
      \end{equation*}
    \end{minipage}
    \begin{minipage}{0.45\linewidth}
      \begin{equation*}
        % https://q.uiver.app/#q=WzAsNCxbMCwwLCJQX1xcTWNVXlxcYkkoXFx0TWNVKSBcXHRpbWVzIFxcc2V0e1xcbWJib3swfX0iXSxbMCwxLCJQX1xcTWNVXlxcYkkoXFx0TWNVKSBcXHRpbWVzIFxcYkkiXSxbMSwxLCJQX1xcTWNVXlxcYkkoXFx0TWNVKSJdLFsxLDAsIlxcdE1jVSJdLFswLDEsIiIsMix7InN0eWxlIjp7InRhaWwiOnsibmFtZSI6Imhvb2siLCJzaWRlIjoidG9wIn19fV0sWzEsMiwiXFxwYXRocmVzIiwyXSxbMywyLCJcXHJlZmwiLDAseyJzdHlsZSI6eyJ0YWlsIjp7Im5hbWUiOiJob29rIiwic2lkZSI6InRvcCJ9fX1dLFswLDMsIlxcZXZfe1xcbWJib3swfX0iXV0=
        \begin{tikzcd}[cramped]
          {P_\McU^\bI(\tMcU) \times \set{\mbbo{0}}} & \tMcU \\
          {P_\McU^\bI(\tMcU) \times \bI} & {P_\McU^\bI(\tMcU)}
          \arrow["{\ev_{\mbbo{0}}}", from=1-1, to=1-2]
          \arrow[hook, from=1-1, to=2-1]
          \arrow["{\const}", hook, from=1-2, to=2-2]
          \arrow["\pathres"', from=2-1, to=2-2]
        \end{tikzcd}
      \end{equation*}
    \end{minipage}
  \end{center}
  This is indeed the case due to \Cref{lem:pathres-const,lem:pathres-const-0}.

  We now construct the section part of \Cref{eqn:const-sr} by taking the top map
  $\tMcU \to (\isConst \vee \txtis_{\mbbo{0}})^*\partial\Cof$ as the composite
  \begin{equation*}
    \begin{tikzcd}
      \tMcU \cong \tMcU \times \set{\mbbo{1}}
      \ar[r, hook]
      &
      \tMcU \times \bI
      \ar[r, hook, "{\inj_1}"]
      &
      (\isConst \vee \txtis_{\mbbo{0}})^*\partial\Cof
    \end{tikzcd}
  \end{equation*}
  and the bottom map
  \begin{equation*}
    P_\McU^\bI(\tMcU) \cong P_\McU^\bI(\tMcU) \times \set{\mbbo{1}} \hookrightarrow P_\McU^\bI(\tMcU) \times \bI
  \end{equation*}
  as the inclusion.
  Because the composite
  $\tMcU \times \bI \hookrightarrow (\isConst \vee
  \txtis_{\mbbo{0}})^*\partial\Cof \hookrightarrow P_\McU^\bI(\tMcU) \times \bI$
  is exactly the inclusion
  $\tMcU \times \bI \hookrightarrow P_\McU^\bI(\tMcU) \times \bI$, the left part
  square of \Cref{eqn:const-sr} commutes, which is observed as follows.
  \begin{equation*}
    % https://q.uiver.app/#q=WzAsNSxbMCwwLCJcXHRNY1UgXFx0aW1lcyBcXHNldHtcXG1iYm97MX19Il0sWzAsMSwiUF97XFxNY1V9XlxcYkkoXFx0TWNVKSBcXHRpbWVzIFxcc2V0e1xcbWJib3sxfX0iXSxbMSwwLCJcXHRNY1UgXFx0aW1lcyBcXGJJIl0sWzIsMSwiUF97XFxNY1V9XlxcYkkoXFx0TWNVKSBcXHRpbWVzIFxcYkkiXSxbMiwwLCIoXFxyaG8gXFx2ZWUgXFx0eHRpc197XFxtYmJvezB9fSleKlxccGFydGlhbFxcQ29mIl0sWzAsMSwiXFxjb25zdCBcXHRpbWVzIFxcc2V0e1xcbWJib3sxfX0iLDIseyJzdHlsZSI6eyJ0YWlsIjp7Im5hbWUiOiJob29rIiwic2lkZSI6InRvcCJ9fX1dLFswLDIsIiIsMCx7InN0eWxlIjp7InRhaWwiOnsibmFtZSI6Imhvb2siLCJzaWRlIjoidG9wIn19fV0sWzIsMywiXFxjb25zdCBcXHRpbWVzIFxcYkkiLDFdLFs0LDMsIiIsMSx7InN0eWxlIjp7InRhaWwiOnsibmFtZSI6Imhvb2siLCJzaWRlIjoidG9wIn19fV0sWzIsNCwiXFxpbmpfMSIsMCx7InN0eWxlIjp7InRhaWwiOnsibmFtZSI6Imhvb2siLCJzaWRlIjoidG9wIn19fV0sWzEsMywiIiwwLHsic3R5bGUiOnsidGFpbCI6eyJuYW1lIjoiaG9vayIsInNpZGUiOiJ0b3AifX19XV0=
    \begin{tikzcd}[cramped]
      {\tMcU \times \set{\mbbo{1}}} & {\tMcU \times \bI} & {(\rho \vee \txtis_{\mbbo{0}})^*\partial\Cof} \\
      {P_{\McU}^\bI(\tMcU) \times \set{\mbbo{1}}} && {P_{\McU}^\bI(\tMcU) \times \bI}
      \arrow[hook, from=1-1, to=1-2]
      \arrow["{\const \times \set{\mbbo{1}}}"', hook, from=1-1, to=2-1]
      \arrow["{\inj_1}", hook, from=1-2, to=1-3]
      \arrow["{\const \times \bI}"{description}, from=1-2, to=2-3]
      \arrow[hook, from=1-3, to=2-3]
      \arrow[hook, from=2-1, to=2-3]
    \end{tikzcd}
  \end{equation*}

  Then, \Cref{eqn:const-sr} becomes two commutative squares like so
  \begin{equation*}
    % https://q.uiver.app/#q=WzAsNyxbMCwwLCJcXHRNY1UgXFx0aW1lcyBcXHNldHtcXG1iYm97MX19Il0sWzAsMSwiUF97XFxNY1V9XlxcYkkoXFx0TWNVKSBcXHRpbWVzIFxcc2V0e1xcbWJib3sxfX0iXSxbMSwwLCJcXHRNY1UgXFx0aW1lcyBcXGJJIl0sWzIsMSwiUF97XFxNY1V9XlxcYkkoXFx0TWNVKSBcXHRpbWVzIFxcYkkiXSxbMiwwLCIoXFxyaG8gXFx2ZWUgXFx0eHRpc197XFxtYmJvezB9fSleKlxccGFydGlhbFxcQ29mIl0sWzMsMCwiXFx0TWNVIl0sWzMsMSwiUF9cXE1jVV5cXGJJKFxcdE1jVSkiXSxbMCwxLCJcXGNvbnN0IFxcdGltZXMgXFxzZXR7XFxtYmJvezF9fSIsMix7InN0eWxlIjp7InRhaWwiOnsibmFtZSI6Imhvb2siLCJzaWRlIjoidG9wIn19fV0sWzAsMiwiIiwwLHsic3R5bGUiOnsidGFpbCI6eyJuYW1lIjoiaG9vayIsInNpZGUiOiJ0b3AifX19XSxbNCwzLCIiLDEseyJzdHlsZSI6eyJ0YWlsIjp7Im5hbWUiOiJob29rIiwic2lkZSI6InRvcCJ9fX1dLFsyLDQsIlxcaW5qXzEiLDAseyJzdHlsZSI6eyJ0YWlsIjp7Im5hbWUiOiJob29rIiwic2lkZSI6InRvcCJ9fX1dLFsxLDMsIiIsMCx7InN0eWxlIjp7InRhaWwiOnsibmFtZSI6Imhvb2siLCJzaWRlIjoidG9wIn19fV0sWzUsNiwiXFxjb25zdCIsMCx7InN0eWxlIjp7InRhaWwiOnsibmFtZSI6Imhvb2siLCJzaWRlIjoidG9wIn19fV0sWzMsNiwiXFxwYXRocmVzIiwyXSxbNCw1LCIoXFxwcm9qLCBcXGV2X3tcXG1iYm97MH19KSJdXQ==
    \begin{tikzcd}[cramped]
      {\tMcU \times \set{\mbbo{1}}} & {\tMcU \times \bI} & {(\rho \vee \txtis_{\mbbo{0}})^*\partial\Cof} & \tMcU \\
      {P_{\McU}^\bI(\tMcU) \times \set{\mbbo{1}}} && {P_{\McU}^\bI(\tMcU) \times \bI} & {P_\McU^\bI(\tMcU)}
      \arrow[hook, from=1-1, to=1-2]
      \arrow["{\const \times \set{\mbbo{1}}}"', hook, from=1-1, to=2-1]
      \arrow["{\inj_1}", hook, from=1-2, to=1-3]
      \arrow["{(\proj, \ev_{\mbbo{0}})}", from=1-3, to=1-4]
      \arrow[hook, from=1-3, to=2-3]
      \arrow["\const", hook, from=1-4, to=2-4]
      \arrow[hook, from=2-1, to=2-3]
      \arrow["\pathres"', from=2-3, to=2-4]
    \end{tikzcd}
  \end{equation*}
  The composite of the entire top row is the identity because by the definition
  of $(\proj,\ev_{\mbbo{0}})$ from \Cref{eqn:const-retract}, its restriction to
  $\tMcU \times \bI$ is the projection map $\tMcU \times \bI \to \tMcU$.
  The composite of the entire bottom row is the identity because of
  \Cref{lem:pathres-sect}.
  Hence, this exhibits $\const$ as the retract of a pullback of
  $(\Cof \vee \txtis_{\mbbo{0}})^*\partial\Cof \hookrightarrow \Cof \times \bI$.
\end{proof}

In view of \Cref{prop:refl-retract}, to actually exhibit
$\tMcU \hookrightarrow P_\McU^\bI(\tMcU)$ as a retract of a pullback of
$(\Cof \vee \txtis_{\mbbo{0}})^*\partial\Cof \hookrightarrow \Cof \times \bI$,
we need to exhibit $\tMcU \hookrightarrow P_\McU^\bI(\tMcU)$ as an
$\iota$-cofibration.
This is ensured by the following result.
\begin{lemma}\label{lem:refl-cof}
  Assume that $\partial\Cof \hookrightarrow \Cof$ has a truth structure relative
  to $\tMcU \to \McU$ and that $\bI$ has a disjoint $\iota$-cofibrant endpoint
  structure.
  Further let $\tMcU \to \McU$ be equipped with a $\Sigma$- and $\Path$-type
  structure and that $\pi$-fibration sections are equipped with a cofibrant
  structure.

  Then, one can find a map $\isConst \colon P_\McU^\bI(\tMcU) \to \Cof$ giving
  \begin{align*}
    \const \colon \tMcU \hookrightarrow P_\McU^\bI(\tMcU)
  \end{align*}
  an $\iota$-name.
\end{lemma}
\begin{proof}
  We first note that $\const \colon \tMcU \hookrightarrow P_\McU^\bI(\tMcU)$ is a
  section of the composite
  \begin{equation*}
    P_\McU^\bI(\tMcU) \xrightarrow{\ev_\partial} \tMcU \times_\McU \tMcU
    \xrightarrow{\proj_1} \tMcU
  \end{equation*}
  By definition, $\tMcU \times_\McU \tMcU \xrightarrow{\proj_1} \tMcU$ is a
  $\pi$-fibration and under the presence of a $\Path$-type structure,
  $P_\McU^\bI(\tMcU) \xrightarrow{\ev_\partial} \tMcU \times_\McU \tMcU$ is also
  a $\pi$-fibration.
  Therefore, when there is a $\Sigma$-type structure, by \cite[Proposition
  1.14]{axm-univalence}, the composite
  $P_\McU^\bI(\tMcU) \xrightarrow{\ev_\partial} \tMcU \times_\McU \tMcU
  \xrightarrow{\proj_1} \tMcU$ is a $\pi$-fibration as well.
  Hence, we have exhibited
  $\const \colon \tMcU \hookrightarrow P_\McU^\bI(\tMcU)$ as a section of a
  $\pi$-fibration
  The result now follows by \Cref{cor:cof-sect}.
\end{proof}

Thus, putting \Cref{prop:refl-retract,lem:refl-cof} together, we can obtain the
following construction.
\begin{construction}\label{constr:0-biased-fill-J}
  Assume that $\partial\Cof \hookrightarrow \Cof$ has a truth structure relative
  to $\McU \to \McU$ under which $\bI$ has a disjoint cofibrant endpoint
  structure and a $\min$-structure.
  Further let $\tMcU \to \McU$ be equipped with a $\Sigma$- and $\Path$-type
  structure and that sections of $\tMcU \to \McU$ are equipped with a cofibrant
  structure.

  Then, one has a map
  \newsavebox{\zeroHorn}
  \begin{lrbox}{\zeroHorn}\small
    \begin{tikzcd}[cramped, row sep=small]
      (\Cof \vee \txtis_{\mbbo{0}})^*\partial\Cof \ar[d, hook] \\
      \Cof \times \bI
    \end{tikzcd}
  \end{lrbox}
  \newsavebox{\zeroHornRho}
  \begin{lrbox}{\zeroHornRho}\small
    \begin{tikzcd}[cramped, row sep=small]
      (\isConst \vee \txtis_{\mbbo{0}})^*\partial\Cof \ar[d, hook] \\
      P_\McU^\bI(\tMcU) \times \bI
    \end{tikzcd}
  \end{lrbox}
  \newsavebox{\reflMap}
  \begin{lrbox}{\reflMap}\small
    \begin{tikzcd}[cramped, row sep=small]
      \tMcU \ar[d, "{\const}"', hook] \\ P_\McU^\bI(\tMcU)
    \end{tikzcd}
  \end{lrbox}
  \newsavebox{\UCof}
  \begin{lrbox}{\UCof}\small
    \begin{tikzcd}[cramped, row sep=small]
      \Cof \times \tMcU \ar[d] \\ \Cof \times \McU
    \end{tikzcd}
  \end{lrbox}
  \newsavebox{\UU}
  \begin{lrbox}{\UU}\small
    \begin{tikzcd}[cramped, row sep=small]
      \McU \times \tMcU \ar[d] \\ \McU \times \McU
    \end{tikzcd}
  \end{lrbox}
  \newsavebox{\PU}
  \begin{lrbox}{\PU}\small
    \begin{tikzcd}[cramped, row sep=small]
      P_\McU^\bI(\tMcU) \times \tMcU \ar[d] \\ P_\McU^\bI(\tMcU) \times \McU
    \end{tikzcd}
  \end{lrbox}
  \begin{equation*}
    \left(
      \usebox{\zeroHorn}
      \fracsquareslash{\Cof}
      \usebox{\UCof}
    \right)
    \xrightarrow{\qquad}
    \left(
      \usebox{\reflMap}
      \fracsquareslash{\McU}
      \usebox{\UU}
    \right)
  \end{equation*}
  defined as the composite
  \begin{equation*}\small
    \begin{tikzcd}[column sep=5em]
      {\left(
          \usebox{\zeroHorn}
          \fracsquareslash{\Cof}
          \usebox{\UCof}
        \right)}
      \ar[r, "{\text{\Cref{lem:refl-cof}}}", "\text{\cite[Con. 3.2]{struct-lift}}"']
      &
      {\left(
          \usebox{\zeroHornRho}
          \fracsquareslash{P_\McU^\bI(\tMcU)}
          \usebox{\PU}
        \right)}
      \ar[d, "{\text{\Cref{prop:refl-retract}}}", "\text{\cite[Con. 2.6]{struct-lift}}"']
      \\
      {\left(
          \usebox{\reflMap}
          \fracsquareslash{\McU}
          \usebox{\UU}
        \right)}
      &
      {\left(
          \usebox{\reflMap}
          \fracsquareslash{P^\bI_\McU(\tMcU)}
          \usebox{\PU}
        \right)}
      \ar[l,"\text{\cite[Con. 3.4]{struct-lift}}"']
    \end{tikzcd}
  \end{equation*}
\end{construction}

Putting everything together, we can induce an $\Id$-type structure from a
$\Path$-type structure.

\begin{theorem}\label{thm:Path-Id}
  Suppose
  \begin{enumerate}
    \item $\partial\Cof \hookrightarrow \Cof$ has a truth structure relative to
    $\tMcU \to \McU$
    \item $\bI$ has a disjoint cofibrant endpoint structure relative to
    $\partial\Cof \hookrightarrow \Cof$
    \item $\bI$ has a $\min \colon \bI \times \bI \to \bI$ structure
    \item $\tMcU \to \McU$ has is equipped with a $\Path$- and a $\Sigma$-type
    structure
    \item Sections to $\pi$-fibrations are equipped with an $\iota$-cofibration
    structure.
  \end{enumerate}

  Then, a $\mbbo{0}$-biased filling structure on these data makes the
  $\Path$-type structure an $\Id$-type structure by the expected choice
  $(\refl,\Id) \coloneqq (\via \cdot \const,\Path)$ so that
  \begin{align*}
    \Id_{\McU}(\tMcU) = P^\bI_{\McU}(\tMcU)
    &&
       \tMcU  \xrightarrow{\refl = \const} \Id_\McU(\tMcU)
  \end{align*}
\end{theorem}
\begin{proof}
  It suffices to show there exists a uniform lifting structure of $\const$
  against $\McU \times \tMcU \to \McU \times \tMcU$ in the slice over $\McU$.
  This is because of the following map
  \begin{equation*}
    \left(
      \begin{tikzcd}[cramped, row sep=small]
        (\Cof \vee \txtis_{\mbbo{0}})^*\partial\Cof \ar[d, hook] \\
        \Cof \times \bI
      \end{tikzcd}
      \fracsquareslash{\Cof}
      \begin{tikzcd}[cramped, row sep=small]
        \Cof \times \tMcU \ar[d] \\ \Cof \times \McU
      \end{tikzcd}
    \right)
    \xrightarrow{\text{\Cref{constr:0-biased-fill-J}}}
    \left(
      \begin{tikzcd}[cramped, row sep=small]
        \tMcU \ar[d, hook] \\ P_\McU^\bI(\tMcU)
      \end{tikzcd}
      \fracsquareslash{\McU}
      \begin{tikzcd}[cramped, row sep=small]
        \McU \times \tMcU \ar[d] \\ \McU \times \McU
      \end{tikzcd}
    \right)
  \end{equation*}
  and the fact that $\Fill_{\mbbo{0}}$ is an element of the domain, so its
  codomain, which is the set of lifting structures we are seeking, is non-empty.
\end{proof}

%%% Local Variables:
%%% TeX-master: "./main.tex"
%%% TeX-engine: default
%%% End:

%%% Local Variables:
%%% TeX-master: "./main.tex"
%%% TeX-engine: default
%%% End:

\section{Computation Rules for Filling}\label{sec:filling-computation}
In the formulation of cubical type theory from \cite{cchm15}, the composition
operation is equipped with various computation rules.
We port over those computation rules to our setting, which uses the filing
operation.

\subsection{Computation Rule for Filling $\Sigma$-types}\label{subsec:filling-sigma}
In this part we will derive the computation rule for the $\delta$-biased filling
operation for $\Sigma$-types.
The computation rule for $\Sigma$-types says that the choice of the lifting
structure chosen for the $\delta$-biased filling operation is invariant under
vertical composition of maps.

We will now make these notions precise, first by starting with the following
construction.
\newsavebox{\sigmaTop}
\begin{lrbox}{\sigmaTop}{\scriptsize\begin{tikzcd}[column sep=small, row sep=small]
      {\Cof \times \ev^*(\tMcU \times \tMcU)} \ar[d] \\ {\Cof \times \pi^*\pi_*(\McU \times \McU)}\end{tikzcd}}
\end{lrbox}
\newsavebox{\sigmaBot}
\begin{lrbox}{\sigmaBot}{\scriptsize\begin{tikzcd}[column sep=small, row sep=small]
      {\Cof \times \pi^*\pi_*(\McU \times \McU)} \ar[d] \\ {\Cof \times \pi_*(\McU \times \McU)}\end{tikzcd}}
\end{lrbox}
\newsavebox{\sigmaComp}
\begin{lrbox}{\sigmaComp}{\scriptsize\begin{tikzcd}[column sep=small, row sep=small]
      {\Cof \times \pi^*\pi_*(\McU \times \McU)} \ar[d] \\ {\Cof \times \pi^*\pi_*(\McU \times \McU)} \ar[d] \\ {\Cof \times \pi_*(\McU \times \McU)}\end{tikzcd}}
\end{lrbox}
\newsavebox{\sigmaRes}
\begin{lrbox}{\sigmaRes}{\scriptsize\begin{tikzcd}[cramped, row sep=small]
      & \Cof \times \tMcU \ar[d] \\ \Cof \times \pi_*(\tMcU \times \McU) \ar[r, "{\Cof\times\Sigma}"'] & \Cof \times \McU
    \end{tikzcd}}
\end{lrbox}
\begin{construction}\label{constr:Sigma-fill}
  Suppose $\partial\Cof \hookrightarrow \Cof$ has a truth structure relative
  to $\tMcU \to \McU$ under which $\bI$ has a disjoint cofibrant endpoint
  structure.

  Further assume that $\tMcU \to \McU$ has a $\Sigma$-type structure.
  Then, one can define a map
  \begin{equation*}\small
    \left(
      \begin{tikzcd}[cramped, row sep=small]
        (\Cof \vee \txtis_{\delta})^*\partial\Cof \ar[d,hook] \\ \Cof \times \bI
      \end{tikzcd}
      \fracsquareslash{\Cof}
      \begin{tikzcd}[cramped, row sep=small]
        \Cof \times \tMcU \ar[d] \\ \Cof \times \McU
      \end{tikzcd}
    \right)
    \xrightarrow{\qquad}
    \left(
      \begin{tikzcd}[cramped, row sep=small]
        (\Cof \vee \txtis_{\delta})^*\partial\Cof \ar[d,hook] \\ \Cof \times \bI
      \end{tikzcd}
      \fracsquareslash{\Cof}
      \begin{tikzcd}[cramped, column sep=small, row sep=small]
        & \Cof \times \tMcU \ar[d] \\ \Cof \times \pi_*(\tMcU \times \McU) \ar[r, "{\Sigma}"'] & \Cof \times \McU
      \end{tikzcd}
    \right)
  \end{equation*}
  as the following composite using the $\Sigma$-type structure along with the
  restriction and composition of maps fromm \cite[Constructions 2.1, 2.2 and
  2.3]{struct-lift}.
  \begin{equation*}
    \begin{tikzcd}
      \Fill_\delta \in \left( \usebox{\dCofI} \fracsquareslash{\Cof} \usebox{\tU} \right)
      \ar[d, "{\text{\cite[Constructions 2.1 and 2.2]{struct-lift}}}"]
      \\
      \left( \usebox{\dCofI} \fracsquareslash{\Cof} \usebox{\sigmaTop} \right) \times
      \left( \usebox{\dCofI} \fracsquareslash{\Cof} \usebox{\sigmaBot} \right)
      \ar[d, "{\text{\cite[Construction 2.3]{struct-lift}}}"]
      \\
      \left( \usebox{\dCofI} \fracsquareslash{\Cof} \usebox{\sigmaComp} \right)
      \ar[d, "{\text{\cite[Construction 2.2]{struct-lift}}}", "\cong"']
      \\
      \pair(\Fill_\delta) \in \left( \usebox{\dCofI} \fracsquareslash{\Cof} \usebox{\sigmaRes} \right)
    \end{tikzcd}
  \end{equation*}
  When $\tMcU \to \McU$ has furthermore a $\delta$-biased filling structure
  $\Fill_\delta$, denote by $\pair(\Fill_\delta)$ the image of $\Fill_\delta$
  under the above map.

  We also denote by $\Fill_\delta(\pair)$ the image of $\Fill_\delta$ under the
  following restriction map directly by \cite[Construction 2.1]{struct-lift}.
  \begin{equation*}
    \begin{tikzcd}
      \Fill_\delta \in \left( \usebox{\dCofI} \fracsquareslash{\Cof} \usebox{\tU} \right)
      \ar[d, "{\text{\cite[Construction 2.1]{struct-lift}}}"]
      \\
      \Fill_\delta(\pair) \in
      \left( \usebox{\dCofI} \fracsquareslash{\Cof} \usebox{\sigmaRes} \right)
    \end{tikzcd}
  \end{equation*}
\end{construction}

With the above construction, we are now ready to define the definitional and
propositional versions of computation rules for filling $\Sigma$-types.
\begin{definition}\label{def:Sigma-fill-comp}
  Suppose $\partial\Cof \hookrightarrow \Cof$ has a truth structure relative
  to $\tMcU \to \McU$ under which $\bI$ has a disjoint cofibrant endpoint
  structure.
  Further assume that $\tMcU \to \McU$ has a $\Sigma$-type structure and that
  $\tMcU \to \McU$ has a $\delta$-biased filling structure $\Fill_\delta$.

  We say that the $\delta$-biased filling structure $\Fill_\delta$
  \emph{computes definitionally} for this $\Sigma$-type structure if
  \begin{equation*}
    \Fill_\delta(\pair) = \pair(\Fill_\delta)
  \end{equation*}
  are precisely equal as from \Cref{constr:Sigma-fill}.

  And a \emph{stable propositional computation structure} relative to an
  $\Id$-type structure $\Id \colon \tMcU \times_\McU \tMcU \to \McU$ for this
  biased $\delta$-filling structure $\Fill_\delta$ with respect to this
  $\Sigma$-structure is a witness $H$ that
  \begin{equation*}
    \Fill_\delta(\pair), \pair(\Fill_\delta) \in
    \left( \usebox{\dCofI} \fracsquareslash{\Cof} \usebox{\sigmaRes} \right)
  \end{equation*}
  are structurally related, in the sense of \cite[Definition 6.3]{struct-lift}, by
  $\Cof \times \Id_\McU(\tMcU) \to \Cof \times (\tMcU \times_{\McU} \tMcU)$.
\end{definition}

\subsection{Computation Rule for Filling $\Pi$-types}\label{subsec:filling-pi}
In this part, we will derive the computation rule for the biased filling
operation for $\Pi$-types.
However, unlike in the in the $\Sigma$-type case, where we induced two different
lifting structures against the same map on the right, we instead require a lift
to be sent to another lift, in the sense of \cite[Definition 6.3]{struct-lift}
by the application maps $(\widetilde{\app}, \app)$ as from \Cref{eqn:Pi-def}.

As before, we start with a construction to allow this to be stated precisely.
\begin{construction}\label{constr:Pi-fill}
  Suppose $\partial\Cof \hookrightarrow \Cof$ has a truth structure relative to
  $\tMcU \to \McU$ under which $\bI$ has a disjoint cofibrant endpoint
  structure.

  Then, by \Cref{eqn:Pi-def} one can define a map
  \begin{equation*}\small
    \left(
      \begin{tikzcd}[cramped]
        (\Cof \vee \txtis_{\delta})^*\partial\Cof \ar[d,hook] \\ \Cof \times \bI
      \end{tikzcd}
      \fracsquareslash{\Cof}
      \begin{tikzcd}[cramped]
        \Cof \times \tMcU \ar[d] \\ \Cof \times \McU
      \end{tikzcd}
    \right)
    \xrightarrow{\text{\cite[Constructions 2.1 and 2.2]{struct-lift}}}
    \left(
      \begin{tikzcd}[cramped]
        (\Cof \vee \txtis_{\delta})^*\partial\Cof \ar[d,hook] \\ \Cof \times \bI
      \end{tikzcd}
      \fracsquareslash{\Cof}
      \begin{tikzcd}[cramped, column sep=small]
        \Cof \times \pi^*\pi_*(\tMcU \times \tMcU) \ar[d] \\ \Cof \times \pi^*\pi_*(\tMcU \times \McU)
      \end{tikzcd}
    \right)
  \end{equation*}
  When $\tMcU \to \McU$ has furthermore a $\delta$-biased filling structure
  $\Fill_\delta$, denote by $\Pi(\Fill_\delta)$ the image of $\Fill_\delta$
  under the above map.
\end{construction}

With the above construction, we are now ready to define the definitional and
propositional versions of computation rules for filling $\Pi$-types.

\begin{definition}\label{def:Pi-fill-comp}
  Suppose $\partial\Cof \hookrightarrow \Cof$ has a truth structure respective
  to $\tMcU \to \McU$ under which $\bI$ has a disjoint cofibrant endpoint
  structure.
  Further assume that $\tMcU \to \McU$ has a $\Pi$-type structure and that
  $\tMcU \to \McU$ has a $\delta$-biased filling structure $\Fill_\delta$.

  We say that the $\delta$-biased filling structure $\Fill_\delta$
  \emph{computes definitionally} for this $\Pi$-type structure if the pair of application maps
  \begin{equation*}
    % https://q.uiver.app/#q=WzAsNCxbMCwwLCJcXENvZiBcXHRpbWVzIFxccGleKlxccGlfKihcXHRNY1UgXFx0aW1lcyBcXHRNY1UpIl0sWzAsMSwiXFxDb2YgXFx0aW1lcyBcXHBpXipcXHBpXyooXFx0TWNVIFxcdGltZXMgXFxNY1UpIl0sWzEsMCwiXFxDb2YgXFx0aW1lcyBcXHRNY1UiXSxbMSwxLCJcXENvZiBcXHRpbWVzIFxcTWNVIl0sWzIsM10sWzAsMV0sWzEsMywiXFxDb2YgXFx0aW1lcyBcXGFwcCIsMl0sWzAsMiwiXFxDb2YgXFx0aW1lcyBcXHdpZGV0aWxkZXtcXGFwcH0iXV0=
    \begin{tikzcd}[cramped, row sep=small]
      {\Cof \times \pi^*\pi_*(\tMcU \times \tMcU)} & {\Cof \times \tMcU} \\
      {\Cof \times \pi^*\pi_*(\tMcU \times \McU)} & {\Cof \times \McU}
      \arrow["{\Cof \times \widetilde{\app}}", from=1-1, to=1-2]
      \arrow[from=1-1, to=2-1]
      \arrow[from=1-2, to=2-2]
      \arrow["{\Cof \times \app}"', from=2-1, to=2-2]
    \end{tikzcd}
  \end{equation*}
  from \Cref{eqn:Pi-def} precisely send $\Pi(\Fill_\delta)$ back to
  $\Fill_\delta$, in the sense of \cite[Definition 6.3]{struct-lift}.

  A \emph{stable propositional computation structure} for biased
  $\delta$-filling structure $\Fill_\delta$ with respect to this $\Pi$-structure
  relative to an $\Id$-type structure
  $\Id \colon \tMcU \times_\McU \tMcU \to \McU$ is a structured
  $(\Cof \times \Id_\McU(\tMcU) \to \Cof \times (\tMcU \times_\McU
  \tMcU))$-witness $H$ from $\Pi(\Fill_\delta)$ back to $\Fill_\delta$ via the
  pair of application maps $(\Cof \times \widetilde{\app}, \Cof \times \app)$,
  in the sense of \cite[Definition 6.3]{struct-lift}.
\end{definition}

\subsection{Computation Rule for Filling Path-types}\label{subsec:filling-path}
We now define the computation rule for the biased filling operation for
$\Path$-types.
In short, this says that the filling operation is compatible with Leibniz
transpose.

To make this precise, we need to start with an observation.

\begin{lemma}\label{lem:biased-bd-pushout-product-approx}
  Suppose $\partial\Cof \hookrightarrow \Cof$ admits an truth-structure relative to
  $\tMcU \to \McU$ under which $\bI$ has a disjoint cofibrant endpoint
  structure.

  Then, the map
  $((\Cof \vee \txtis_\partial) \vee \txtis_\delta)^*\partial\Cof
  \hookrightarrow \Cof \times \bI \times \bI$, for $\delta = \mbbo{0},\mbbo{1}$,
  constructed as the iterated pullback in the top left vertical map as follows
  \begin{equation*}
    % https://q.uiver.app/#q=WzAsMTEsWzAsMywiXFxDb2YiXSxbMCwyLCJcXENvZiBcXHRpbWVzIFxcYkkiXSxbMCwxLCJcXENvZiBcXHRpbWVzIFxcYkkgXFx0aW1lcyBcXGJJIl0sWzAsMCwiKChcXENvZiBcXHZlZSBcXHR4dGlzX1xccGFydGlhbClcXHZlZVxcdHh0aXNfXFxkZWx0YSleKlxccGFydGlhbFxcQ29mIl0sWzEsMiwiXFxDb2YiXSxbMSwxLCJcXENvZiBcXHRpbWVzIFxcYkkiXSxbMSwwLCIoXFxDb2YgXFx2ZWUgXFx0eHRpc19cXGRlbHRhKV4qXFxwYXJ0aWFsXFxDb2YiXSxbMiwxLCJcXENvZiBcXHRpbWVzIFxcQ29mIl0sWzIsMCwiXFxwYXJ0aWFsKFxcQ29mIFxcdGltZXMgXFxDb2YpIl0sWzMsMSwiXFxDb2YiXSxbMywwLCJcXHBhcnRpYWxcXENvZiJdLFsyLDEsIlxccHJvaiIsMl0sWzEsMCwiXFxwcm9qIiwyXSxbMywyLCIiLDAseyJzdHlsZSI6eyJ0YWlsIjp7Im5hbWUiOiJob29rIiwic2lkZSI6InRvcCJ9fX1dLFs2LDUsIiIsMCx7InN0eWxlIjp7InRhaWwiOnsibmFtZSI6Imhvb2siLCJzaWRlIjoidG9wIn19fV0sWzUsNF0sWzIsNSwiKFxcQ29mIFxcdmVlIFxcdHh0aXNfXFxwYXJ0aWFsKVxcdGltZXMgXFxiSSIsMV0sWzEsNCwiXFxDb2YgXFx2ZWUgXFx0eHRpc19cXHBhcnRpYWwiLDJdLFszLDZdLFszLDUsIiIsMSx7InN0eWxlIjp7Im5hbWUiOiJjb3JuZXIifX1dLFsyLDQsIiIsMSx7InN0eWxlIjp7Im5hbWUiOiJjb3JuZXIifX1dLFs1LDcsIlxcQ29mIFxcdGltZXMgXFx0eHRpc19cXGRlbHRhIiwyXSxbOCw3LCIiLDEseyJzdHlsZSI6eyJ0YWlsIjp7Im5hbWUiOiJob29rIiwic2lkZSI6InRvcCJ9fX1dLFs2LDhdLFs2LDcsIiIsMSx7InN0eWxlIjp7Im5hbWUiOiJjb3JuZXIifX1dLFsxMCw5LCIiLDIseyJzdHlsZSI6eyJ0YWlsIjp7Im5hbWUiOiJob29rIiwic2lkZSI6InRvcCJ9fX1dLFs3LDksIi1cXHZlZS0iLDJdLFs4LDEwXSxbOCw5LCIiLDIseyJzdHlsZSI6eyJuYW1lIjoiY29ybmVyIn19XV0=&macro_url=https%3A%2F%2Fgist.githubusercontent.com%2Flim495062%2F61b94af9ef95c1c7b0763c937de29c2b%2Fraw%2F456d405748eab3250184512b5240467b6083b2b4%2Facmhwmacros.sty
    \begin{tikzcd}[cramped]
      {((\Cof \vee \txtis_\partial)\vee\txtis_\delta)^*\partial\Cof} & {(\Cof \vee \txtis_\delta)^*\partial\Cof} & {\partial(\Cof \times \Cof)} & {\partial\Cof} \\
      {\Cof \times \bI \times \bI} & {\Cof \times \bI} & {\Cof \times \Cof} & \Cof \\
      {\Cof \times \bI} & \Cof \\
      \Cof
      \arrow[from=1-1, to=1-2]
      \arrow[hook, from=1-1, to=2-1]
      \arrow[from=1-2, to=1-3]
      \arrow[hook, from=1-2, to=2-2]
      \arrow[from=1-3, to=1-4]
      \arrow[hook, from=1-3, to=2-3]
      \arrow[hook, from=1-4, to=2-4]
      \arrow["{(\Cof \vee \txtis_\partial)\times \bI}"{description}, from=2-1, to=2-2]
      \arrow["\proj"', from=2-1, to=3-1]
      \arrow["{\Cof \times \txtis_\delta}"', from=2-2, to=2-3]
      \arrow[from=2-2, to=3-2]
      \arrow["{-\vee-}"', from=2-3, to=2-4]
      \arrow["{\Cof \vee \txtis_\partial}"', from=3-1, to=3-2]
      \arrow["\proj"', from=3-1, to=4-1]
      \arrow["\lrcorner"{anchor=center, pos=0.15, scale=1.5, rotate=0}, draw=none, from=1-1, to=2-2]
      \arrow["\lrcorner"{anchor=center, pos=0.15, scale=1.5, rotate=0}, draw=none, from=1-2, to=2-3]
      \arrow["\lrcorner"{anchor=center, pos=0.15, scale=1.5}, draw=none, from=1-3, to=2-4]
      \arrow["\lrcorner"{anchor=center, pos=0.15, scale=1.5, rotate=0}, draw=none, from=2-1, to=3-2]
    \end{tikzcd}
  \end{equation*}
  structurally approximates the pushout-product
  \begin{equation*}
    (\Cof \times \partial\bI \hookrightarrow \Cof \times \bI) \ltimes_\Cof
    ((\Cof \vee \txtis_\delta)^*\partial\Cof \hookrightarrow \Cof \times \bI)
  \end{equation*}
  relative to $\tMcU \times \Cof \to \McU \times \Cof$ in the slice over $\Cof$.
\end{lemma}
\begin{proof}
  We have shown in \Cref{cor:cof-int-bd} that
  $(\Cof \vee \txtis_{\delta})^*\partial\Cof \hookrightarrow \Cof \times \bI$
  given by the top row vertical map in the second column in the diagram of the
  statement approximates the pushout-product
  \begin{align*}
   (\Cof \times \set{\delta} \hookrightarrow \Cof \times \bI)
    \ltimes_\Cof
    (\partial\Cof \hookrightarrow \Cof)
  \end{align*}
  relative to $\Cof \times \tMcU \to \Cof \times \McU$ in $\sfrac{\bC}{\Cof}$.
  Thus, further pulling back along
  $\Cof \vee \txtis_\partial \colon \Cof \times \bI \to \Cof$, which is the
  third row horizontal map in the diagram of the statement, we observe by
  \cite[Lemma 4.5]{struct-lift} that
  $(\Cof \vee \txtis_\partial)^*(\Cof \vee \txtis_\delta)^*\partial\Cof \cong
  ((\Cof \vee \txtis_\partial)\vee\txtis_\delta)^*\partial\Cof \hookrightarrow
  \Cof \times \bI \times \bI$, which is the top left vertical map in the diagram
  of the statement we are interested in, structurally approximates the
  pushout-product
  \begin{align*}
   (\Cof \times \set{\delta} \times \bI \hookrightarrow \Cof \times \bI \times \bI)
    \ltimes_{\Cof\times\bI}
    ((\Cof \vee \txtis_\partial)^*\partial\Cof \hookrightarrow \Cof \times \bI)
  \end{align*}
  relative to $\Cof \times \tMcU \times \bI \to \Cof \times \McU \times \bI$ in
  $\sfrac{\bC}{\Cof \times \bI}$.

  Now
  $\Cof \times \set{\delta} \times \bI \hookrightarrow \Cof \times \bI \times
  \bI$ and $\Cof \times \tMcU \times \bI \to \Cof \times \McU \times \bI$ are
  the pullbacks of $\Cof \times \set{\delta} \hookrightarrow \Cof \times \bI$
  and $\Cof \times \tMcU \to \Cof \times \McU$ under
  $\proj \colon \Cof \times \bI \to \Cof$.
  Thus, by \cite[Lemma 4.6]{struct-lift}, the map
  $((\Cof
  \vee \txtis_\partial)\vee\txtis_\delta)^*\partial\Cof \hookrightarrow \Cof
  \times \bI \times \bI$ as a
  map over $\Cof$ structurally approximates the pushout-product
  \begin{align*}
   (\Cof \times \set{\delta} \hookrightarrow \Cof \times \bI)
    \ltimes_{\Cof}
    ((\Cof \vee \txtis_\partial)^*\partial\Cof \hookrightarrow \Cof \times \bI)
  \end{align*}
  relative to $\Cof \times \tMcU \to \Cof \times \McU $ in $\sfrac{\bC}{\Cof}$.
  But then $(\Cof \vee \txtis_\partial)^*\partial\Cof \hookrightarrow \Cof \times \bI$
  structurally approximates the pushout-product
  \begin{align*}
    (\partial\Cof \hookrightarrow \Cof) \ltimes
    (\Cof \times \partial\bI \hookrightarrow \Cof \times \bI)
  \end{align*}
  relative to $\Cof \times \tMcU \to \Cof \times \McU$ in $\sfrac{\bC}{\Cof}$ by
  \Cref{cor:cof-int-bd}.
  So by associativity of approximations of pushout-product as from \cite[Theorem
  4.7]{struct-lift}, the result follows.
\end{proof}

The above result now facilitates the following definition.

\begin{construction}\label{constr:Path-fill}
  Suppose $\partial\Cof \hookrightarrow \Cof$ admits an truth-structure relative to
  $\tMcU \to \McU$ under which $\bI$ has a disjoint cofibrant endpoint
  structure.
  Further assume that $\tMcU \to \McU$ has a $\Path$-type structure.

  We define a map
  \begin{equation*}
    \scriptsize
    \left(
      \begin{tikzcd}[cramped]
        (\Cof \vee \txtis_{\delta})^*\partial\Cof \ar[d,hook] \\ \Cof \times \bI
      \end{tikzcd}
      \fracsquareslash{\Cof}
      \begin{tikzcd}[cramped]
        \Cof \times \tMcU \ar[d] \\ \Cof \times \McU
      \end{tikzcd}
    \right)
    \xrightarrow{\qquad}
    \left(
      \begin{tikzcd}[cramped]
        (\Cof \vee \txtis_{\delta})^*\partial\Cof \ar[d,hook] \\ \Cof \times \bI
      \end{tikzcd}
      \fracsquareslash{\Cof}
      \begin{tikzcd}[cramped]
        & \Cof \times \tMcU \ar[d] \\ \Cof \times (\tMcU \times_\McU \tMcU) \ar[r, "{\Cof\times\Path}"'] & \Cof \times \McU
      \end{tikzcd}
    \right)
  \end{equation*}
  as the composite using the following numbered constructions from \cite{struct-lift}.
  \newsavebox{\ddCofI}
  \begin{lrbox}{\ddCofI}{\scriptsize\begin{tikzcd}[column sep=small, row sep=small, cramped]
        {((\Cof \vee \txtis_\partial) \vee \txtis_{\delta})^*\partial\Cof} \ar[d,hook] \\ {\Cof \times \bI \times \bI}\end{tikzcd}}
  \end{lrbox}
  \newsavebox{\ddCofIRes}
  \begin{lrbox}{\ddCofIRes}{\scriptsize\begin{tikzcd}[column sep=small, row sep=small, cramped]
        {((\Cof \vee \txtis_\partial) \vee \txtis_{\delta})^*\partial\Cof} \ar[d,hook] \\
        {\Cof \times \bI \times \bI} \ar[r] & {\Cof \times \bI}
\end{tikzcd}}
  \end{lrbox}
  \newsavebox{\tCU}
  \begin{lrbox}{\tCU}{\scriptsize\begin{tikzcd}[column sep=small, row sep=small, cramped]
        {\Cof \times \tMcU \times \bI} \ar[d] \\ {\Cof \times \McU \times \bI}\end{tikzcd}}
  \end{lrbox}
  \newsavebox{\pathU}
  \begin{lrbox}{\pathU}{\scriptsize\begin{tikzcd}[cramped, row sep=small, column sep=small, cramped]
        \Cof \times P_\McU^\bI(\tMcU) \ar[d] \\ \Cof \times (\tMcU \times_\McU \tMcU)
      \end{tikzcd}}
  \end{lrbox}
  \newsavebox{\ptU}
  \begin{lrbox}{\ptU}{\scriptsize\begin{tikzcd}[cramped, row sep=small, column sep=small, cramped]
        & \Cof \times \tMcU \ar[d] \\ \Cof \times (\tMcU \times_\McU \tMcU) \ar[r, "{\Cof\times\Path}"'] & \Cof \times \McU
      \end{tikzcd}}
  \end{lrbox}
  \begin{equation*}
    \begin{tikzcd}[column sep=2.6em]
      {\left(\usebox{\dCofI} \fracsquareslash{\Cof} \usebox{\tU} \right)}
      \ar[r, "{\text{\cite[Con. 3.2]{struct-lift}}}"]
      &
      {\left(\usebox{\ddCofI} \fracsquareslash{\Cof \times \bI} \usebox{\tCU} \right)}
      \ar[d, "{\text{\cite[Con. 3.4]{struct-lift}}}"']
      \\
      % {\left(\usebox{\dCofI} \fracsquareslash{\Cof} \usebox{\pathU} \right)}
      % \ar[d, "{\text{\Cref{constr:struct-lift-right-pb}}}"]
      {\left(\usebox{\ddCofIRes} \fracsquareslash{\Cof} \usebox{\tU} \right)}
      \ar[d, tail reversed, "{\begin{matrix}
        \text{\cite[Cor. 5.11]{struct-lift}} \\[-.5em]
        \text{\Cref{lem:biased-bd-pushout-product-approx}}\end{matrix}}", "\cong"']
      &
      {\left(\usebox{\ddCofI} \fracsquareslash{\Cof} \usebox{\tU} \right)}
      \ar[l,  "{\text{\cite[Con. 2.4]{struct-lift}}}"']
      % \ar[l, "{\begin{matrix}
      %   \text{\Cref{cor:struct-ltrans-path}} \\[-.5em]
      %   \text{\Cref{lem:biased-bd-pushout-product-approx}}
      \\
      {\left(\usebox{\dCofI} \fracsquareslash{\Cof} \usebox{\pathU} \right)}
      \ar[r, "{\text{\cite[Con. 2.3]{struct-lift}}}"', tail reversed, "\cong"]
      &
      {\left(\usebox{\dCofI} \fracsquareslash{\Cof} \usebox{\ptU} \right)}
    \end{tikzcd}
  \end{equation*}

  By \cite[Construction 2.1]{struct-lift}, we also have a direct map
  \begin{equation*}
    \scriptsize
    \left(
      \begin{tikzcd}[cramped]
        (\Cof \vee \txtis_{\delta})^*\partial\Cof \ar[d,hook] \\ \Cof \times \bI
      \end{tikzcd}
      \fracsquareslash{\Cof}
      \begin{tikzcd}[cramped]
        \Cof \times \tMcU \ar[d] \\ \Cof \times \McU
      \end{tikzcd}
    \right)
    \xrightarrow{\text{\cite[Con. 2.1]{struct-lift}}}
    \left(
      \begin{tikzcd}[cramped]
        (\Cof \vee \txtis_{\delta})^*\partial\Cof \ar[d,hook] \\ \Cof \times \bI
      \end{tikzcd}
      \fracsquareslash{\Cof}
      \begin{tikzcd}[cramped]
        & \Cof \times \tMcU \ar[d] \\ \Cof \times (\tMcU \times_\McU \tMcU) \ar[r, "{\Cof\times\Path}"'] & \Cof \times \McU
      \end{tikzcd}
    \right)
  \end{equation*}

  When $\tMcU \to \McU$ has furthermore a $\delta$-biased filling structure
  $\Fill_\delta$ we denote by $\Path(\Fill_\delta)$ and $\Fill_\delta(\Path)$
  the images of $\Fill_\delta$ under the two maps above.
\end{construction}

With the above construction, we are now ready to define the definitional and
propositional versions of computation rules for filling $\Path$-types.

\begin{definition}\label{def:Path-fill-comp}
  Suppose $\partial\Cof \hookrightarrow \Cof$ admits a truth-structure relative
  to $\tMcU \to \McU$ under which $\bI$ has a disjoint cofibrant endpoint
  structure.
  Further assume that $\tMcU \to \McU$ has a $\Path$-type structure and a
  $\delta$-biased filling structure $\Fill_\delta$.

  We say that the $\delta$-biased filling structure $\Fill_\delta$
  \emph{computes definitionally} for this $\Path$-type structure if
  $\Fill_\delta(\Path) = \Fill_\delta$ are precisely equal as from
  \Cref{constr:Path-fill}.

  And a \emph{stable propositional computation structure} for this biased
  $\delta$-filling structure $\Fill_\delta$ with respect to this $\Path$-type
  structure relative to an $\Id$-type structure
  $\Id \colon \tMcU \times_\McU \tMcU \to \McU$, is a witness $H$ that
  \begin{equation*}
    \Fill_\delta(\Path), \Path(\Fill_\delta) \in
    \left(
      \begin{tikzcd}[cramped, row sep=small]
        (\Cof \vee \txtis_{\delta})^*\partial\Cof \ar[d,hook] \\ \Cof \times \bI
      \end{tikzcd}
      \fracsquareslash{\Cof}
      \begin{tikzcd}[cramped, row sep=small]
        & \Cof \times \tMcU \ar[d] \\ \Cof \times (\tMcU \times_\McU \tMcU) \ar[r, "{\Cof\times\Path}"'] & \Cof \times \McU
      \end{tikzcd}
    \right)
  \end{equation*}
  are structurally related, in the sense of \cite[Definition 6.3]{struct-lift},
  by
  $\Cof \times \Id_\McU(\tMcU) \to \Cof \times (\tMcU \times_{\McU} \tMcU)$.
\end{definition}

%%% Local Variables:
%%% TeX-master: "./main.tex"
%%% TeX-engine: default
%%% End:

\section{Homotopy Isomorphism Extension Operations}\label{subsec:hiso-ext-op}
We now axiomatise various forms of the homotopy isomorphism extension operation,
corresponding to the gluing operation of \cite{cchm15}, so that one may derive
univalence.

To begin, we first recall the generic object of homotopy isomorphisms from
\cite[Construction 2.7]{axm-univalence}.

\begin{construction}[{\cite[Construction 2.7]{axm-univalence}}]\label{constr:hiso}
  Suppose that $\tMcU \to \McU$ has an $\Id$-type structure
  $\Id \colon \tMcU \times_\McU \tMcU \to \McU$.
  Then, the object and parallel maps as below is constructed to be
  \begin{equation*}
    \src, \tgt \colon \HIso_\McU^\Id(\tMcU) \rightrightarrows \McU
  \end{equation*}
  representing the presheaf and natural transformations
  \newsavebox{\HRslice}
  \begin{lrbox}{\HRslice}\scriptsize
    % https://q.uiver.app/#q=WzAsMyxbMCwwLCJcXEdhbW1hLkEiXSxbMiwwLCJQX1xcR2FtbWEoXFxHYW1tYS5BKSJdLFsxLDEsIlxcR2FtbWEuQSBcXHRpbWVzX1xcR2FtbWEgXFxHYW1tYS5BIl0sWzAsMiwiKGZnLCBcXGlkKSIsMV0sWzAsMSwiSCIsMV0sWzEsMiwiXFxldl9cXHBhcnRpYWwiLDFdXQ==
    \begin{tikzcd}[cramped, column sep=tiny]
      {\Gamma.A} && {\Id_\Gamma(\Gamma.A)} \\
      & {\Gamma.A \times_\Gamma \Gamma.A}
      \arrow["H_s"{description}, from=1-1, to=1-3]
      \arrow["{(fs, \id)}"{description}, from=1-1, to=2-2]
      \arrow["{\ev_\partial}"{description}, from=1-3, to=2-2]
    \end{tikzcd}
  \end{lrbox}
  \newsavebox{\HLslice}
  \begin{lrbox}{\HLslice}\scriptsize
    % https://q.uiver.app/#q=WzAsMyxbMCwwLCJcXEdhbW1hLkEiXSxbMiwwLCJQX1xcR2FtbWEoXFxHYW1tYS5BKSJdLFsxLDEsIlxcR2FtbWEuQSBcXHRpbWVzX1xcR2FtbWEgXFxHYW1tYS5BIl0sWzAsMiwiKGZnLCBcXGlkKSIsMV0sWzAsMSwiSCIsMV0sWzEsMiwiXFxldl9cXHBhcnRpYWwiLDFdXQ==
    \begin{tikzcd}[cramped, column sep=tiny]
      {\Gamma.B} && {\Id_\Gamma(\Gamma.B)} \\
      & {\Gamma.B \times_\Gamma \Gamma.B}
      \arrow["{H_r}"{description}, from=1-1, to=1-3]
      \arrow["{(rf, \id)}"{description}, from=1-1, to=2-2]
      \arrow["{\ev_\partial}"{description}, from=1-3, to=2-2]
    \end{tikzcd}
  \end{lrbox}
  \begin{equation*}\scriptsize
    \begin{tikzcd}
      {\displaystyle
        \left(
          \coprod_{\substack{A \colon \Gamma \to \McU \\ B \colon \Gamma \to \McU}}
          \coprod_{\substack{f \in \sfrac{\bC}{\Gamma}(\Gamma.A,\Gamma.B) \\ s,r \in \sfrac{\bC}{\Gamma}(\Gamma.B,\Gamma.A)}}
          \left\{\scriptsize
            \begin{pmatrix}
              H_s \in \sfrac{\bC}{\Gamma}(\Gamma.A, \Id_\Gamma(\Gamma.A)) \\
              H_r \in \sfrac{\bC}{\Gamma}(\Gamma.B, \Id_\Gamma(\Gamma.B)) \\
            \end{pmatrix}
            ~\middle|~
            \usebox{\HRslice} \text{\scriptsize and}
            \usebox{\HLslice}
          \right\}
        \right)_{\Gamma \in \bC}
      }
      \ar[d, shift left=2, "{(A,B,f,s,r,H_s,H_r) \mapsto B}"]
      \ar[d, shift left=-2, "{(A,B,f,s,r,H_s,H_r) \mapsto A}"']
      \\
      {\scriptsize (\bC(\Gamma,\McU))_{\Gamma \in \bC}}
    \end{tikzcd}
  \end{equation*}
  They both also share a common section
  \begin{equation*}
    \trv \colon \McU \hookrightarrow \HIso_\McU^\Id(\tMcU)
  \end{equation*}
  defined representably by selecting the homotopy isomorphisms where the
  underlying map is the identity map equipped with the reflexivity homotopies.
\end{construction}

We also recall axiomatic formulations of univalence given by
\cite{axm-univalence}.
\begin{definition}[{\cite[Definition 3.1]{axm-univalence}}]
  A \emph{lower-half lift} of a map $K \to L$ against a map $E \to B$ for a
  specified lifting is a filler only making the bottom triangle commute.
  \begin{equation*}
    % https://q.uiver.app/#q=WzAsNCxbMCwwLCJLIl0sWzAsMSwiTCJdLFsxLDAsIkUiXSxbMSwxLCJCIl0sWzAsMV0sWzIsM10sWzAsMiwiIiwxLHsic3R5bGUiOnsiYm9keSI6eyJuYW1lIjoiZGFzaGVkIn19fV0sWzEsMywiIiwxLHsic3R5bGUiOnsiYm9keSI6eyJuYW1lIjoiZGFzaGVkIn19fV0sWzEsMiwiIiwxLHsic3R5bGUiOnsiYm9keSI6eyJuYW1lIjoiZG90dGVkIn19fV0sWzAsOCwiXFxjYW5jZWx7XFxjaXJjbGVhcnJvd2xlZnR9IiwxLHsic2hvcnRlbiI6eyJ0YXJnZXQiOjIwfSwic3R5bGUiOnsiYm9keSI6eyJuYW1lIjoibm9uZSJ9LCJoZWFkIjp7Im5hbWUiOiJub25lIn19fV1d
    \begin{tikzcd}[cramped]
      K & E \\
      L & B
      \arrow[dashed, from=1-1, to=1-2]
      \arrow[from=1-1, to=2-1]
      \arrow[from=1-2, to=2-2]
      \arrow[""{name=0, anchor=center, inner sep=0}, dotted, from=2-1, to=1-2]
      \arrow[dashed, from=2-1, to=2-2]
      \arrow["{?}"{description}, draw=none, from=1-1, to=0]
    \end{tikzcd}
  \end{equation*}
\end{definition}

\begin{definition}[{\cite[Definition 3.2]{axm-univalence}}]\label{def:axm-univalence}
  Let $\pi_0 \colon \tMcU_0 \to \McU_0$ and $\pi \colon \tMcU \to \McU$ be two
  universal maps such that $\pi_0$ has a pre-$\Id$-type structure
  $\Id_0 \colon \tMcU_0 \times_{\McU_0} \tMcU_0 \to \McU_0$.

  A \emph{pointed} (respectively, \emph{book}) \emph{$\pi$-univalence} structure
  on the (universe, pre-$\Id$-type)-pair $(\pi_0,\Id_0)$ is a $\pi$-fibrancy
  structure on the map $(\src, \tgt)$ along with a choice of (respectively,
  lower-half) lifts of $\trv$ on the left against $\pi$ on the right
  \begin{equation*}
    \begin{tikzcd}
      \McU_0 & \tMcU \\
      {\HIso_{\McU_0}^{\Id_0}(\tMcU_0)} & \McU
      \arrow[from=1-1, to=1-2, dashed]
      \arrow["{\trv}"', hook, from=1-1, to=2-1]
      \arrow[from=1-2, to=2-2, two heads]
      \arrow[""{name=0, anchor=center, inner sep=0}, dotted, from=2-1, to=1-2]
      \arrow[from=2-1, to=2-2, dashed]
      \arrow["{(?)}"{description}, draw=none, from=1-1, to=0]
    \end{tikzcd}
  \end{equation*}
  where $(\src, \tgt)$ and $\trv$ are from the diagonal factorisation in
  \Cref{constr:hiso} via $\HIso_{\McU_0}^{\Id_0}(\tMcU_0)$
  \begin{equation*}
    % https://q.uiver.app/#q=WzAsMyxbMCwwLCJcXE1jVV8wIl0sWzEsMCwiXFxISXNvX3tcXE1jVV8wfV57XFxJZF8wfShcXHRNY1VfMCkiXSxbMiwwLCJcXE1jVV8wIFxcdGltZXMgXFxNY1VfMCJdLFsxLDIsIlxcc3JjIFxcdGltZXMgXFxkZXN0IiwwLHsic3R5bGUiOnsiaGVhZCI6eyJuYW1lIjoiZXBpIn19fV0sWzAsMSwiIiwwLHsic3R5bGUiOnsidGFpbCI6eyJuYW1lIjoiaG9vayIsInNpZGUiOiJ0b3AifX19XV0=
    \begin{tikzcd}[cramped]
      {\McU_0} & {\HIso_{\McU_0}^{\Id_0}(\tMcU_0)} &[3em] {\McU_0 \times \McU_0}
      \arrow["{\trv}", hook, from=1-1, to=1-2]
      \arrow["{(\src, \tgt)}", two heads, from=1-2, to=1-3]
    \end{tikzcd}
  \end{equation*}
  The (universe, pre-$\Id$-type)-pair $(\pi_0,\Id_0)$ is \emph{pointed}
  (respectively \emph{book}) \emph{$\pi$-univalent} when it admits a pointed
  (respectively book) $\pi$-univalence structure.
\end{definition}
\begin{remark}
  Clearly, pointed univalence implies book univalence as in \Cref{def:axm-univalence}.
  It is also shown in \cite[Theorem 3.15]{axm-univalence} that under $\Pi$- and
  $\Sigma$-type structures, the above formulation of book univalence implies the
  usual type-theoretic formulation as from the HoTT book \cite{hottbook}.
\end{remark}

\subsection{Strong Homotopy Isomorphism Extension Operation}
We can now define the strong version of the homotopy isomorphism extension
operation.
\begin{definition}\label{def:strong-hiso-ext-op}
  Assume $\tMcU \to \McU$ has an $\Id$-type structure
  $\Id \colon \tMcU \times_\McU \tMcU \to \McU$.

  A \emph{(strong) $\Id$-homotopy isomorphism extension structure} on
  $\tMcU \to \McU$ is a structured lift as in the sense of \cite[Definitions 1.4
  and 3.1]{struct-lift} in the slice over $\sfrac{\bC}{\Cof}$
  \begin{equation*}
    \HIsoExt \in
    \left(
      \begin{tikzcd}[cramped]
        \partial\Cof \ar[d, hook] \\ \Cof
      \end{tikzcd}
      \fracsquareslash{\Cof}
      \begin{tikzcd}[cramped]
        \Cof \times \HIso^\Id_\McU(\tMcU) \ar[d, "{\Cof \times \tgt}"] \\ \Cof \times \McU
      \end{tikzcd}
    \right)
  \end{equation*}
  where the object $\HIso_\McU^\Id(\tMcU)$ and map
  $\tgt \colon \HIso_\McU^\Id(\tMcU) \to \McU$ are respectively from are from
  \Cref{constr:hiso}.
\end{definition}

Using the homotopy retract argument of \cite{axm-univalence}, we can now show
that the strong homotopy isomorphism extension structure gives rise to pointed
univalence.
To do so, we would like to like to use \cite[Theorem 3.12]{axm-univalence},
which requires the construction of a lifting structure of the cofibrations
against the path space of the generic object of homotopy equivalences.
\begin{construction}\label{constr:hiso-endpt-ext}
  Suppose $\partial\Cof \hookrightarrow \Cof$ admits a truth-structure relative
  to $\tMcU \to \McU$ under which $\bI$ has a disjoint cofibrant endpoint
  structure.
  Further assume that $\tMcU \to \McU$ has a $\Path$-type structure as well as
  an internal universe $\tMcU_0 \to \McU_0$ equipped with an $\Id$-type
  structure $\Id_0 \colon \tMcU_0 \times_{\McU_0} \tMcU_0 \to \McU_0$.

  One forms the fibred path object
  \begin{equation*}
    P_{\McU_0}^\bI(\HIso_{\McU_0}^{\Id_0}(\tMcU_0))
    \xrightarrow{\qquad}
    \HIso_{\McU_0}^{\Id_0}(\tMcU_0) \times_{\McU_0} \HIso_{\McU_0}^{\Id_0}(\tMcU_0)
  \end{equation*}
  with the $\tgt \colon \HIso_{\McU_0}^{\Id_0}(\tMcU_0) \to \McU_0$ map and
  then constructs a map
  \begin{equation*}
    \small
    \left(
      \begin{tikzcd}[cramped]
        \partial\Cof \ar[d, hook] \\ \Cof
      \end{tikzcd}
      \fracsquareslash{\Cof}
      \begin{tikzcd}[cramped]
        \Cof \times \HIso_{\McU_0}^{\Id_0}(\tMcU_0) \ar[d, "{\Cof \times \tgt}"] \\ \Cof \times \McU_0
      \end{tikzcd}
    \right)
    \xrightarrow{\qquad}
    \left(
      \begin{tikzcd}[cramped]
        \partial\Cof \ar[d, hook] \\ \Cof
      \end{tikzcd}
      \fracsquareslash{\Cof}
      \begin{tikzcd}[cramped]
        \Cof \times P_{\McU_0}^\bI(\HIso_{\McU_0}^{\Id_0}(\tMcU_0)) \ar[d, "{\Cof \times \ev_\partial}"] \\
        \Cof \times (\HIso_{\McU_0}^{\Id_0}(\tMcU_0) \times_{\McU_0} \HIso_{\McU_0}^{\Id_0}(\tMcU_0))
      \end{tikzcd}
    \right)
  \end{equation*}
  as the following composite
  \newsavebox{\dCof}
  \begin{lrbox}{\dCof}\scriptsize
    \begin{tikzcd}[cramped, row sep=small, column sep=small]
      \partial\Cof \ar[d, hook] \\ \Cof
    \end{tikzcd}
  \end{lrbox}
  \newsavebox{\cHiso}
  \begin{lrbox}{\cHiso}\scriptsize
    \begin{tikzcd}[cramped, row sep=small, column sep=small]
      \Cof \times \HIso_{\McU_0}^{\Id_0}(\tMcU_0) \ar[d, "{\Cof \times \tgt}"] \\ \Cof \times \McU_0
    \end{tikzcd}
  \end{lrbox}
  \newsavebox{\dICof}
  \begin{lrbox}{\dICof}\scriptsize
    \begin{tikzcd}[cramped, row sep=small, column sep=small]
      (\Cof \vee \txtis_\partial)^*\partial\Cof \ar[d, hook] \\ \Cof \times \bI
    \end{tikzcd}
  \end{lrbox}
  \newsavebox{\dICofRes}
  \begin{lrbox}{\dICofRes}\scriptsize
    \begin{tikzcd}[cramped, row sep=small, column sep=small]
      (\Cof \vee \txtis_\partial)^*\partial\Cof \ar[d, hook] \\ \Cof \times \bI \ar[r] & \Cof
    \end{tikzcd}
  \end{lrbox}
  \newsavebox{\cIHiso}
  \begin{lrbox}{\cIHiso}\scriptsize
    \begin{tikzcd}[cramped, row sep=small, column sep=small]
      (\Cof \times \bI) \times \HIso_{\McU_0}^{\Id_0}(\tMcU_0) \ar[d, "{(\Cof \times \bI) \times \tgt}"] \\ (\Cof \times \bI) \times \McU_0
    \end{tikzcd}
  \end{lrbox}
  \newsavebox{\PHiso}
  \begin{lrbox}{\PHiso}\scriptsize
    \begin{tikzcd}[cramped, row sep=small, column sep=small]
      P_{\Cof \times \McU_0}^{\Cof\times\bI}(\Cof \times \HIso_{\McU_0}^{\Id_0}(\tMcU_0)) \ar[d] \\
      (\Cof \times \HIso_{\McU_0}^{\Id_0}(\tMcU_0)) \times_{\Cof \times \McU_0} (\Cof \times \HIso_{\McU_0}^{\Id_0}(\tMcU_0))
    \end{tikzcd}
  \end{lrbox}
  \newsavebox{\cPHiso}
  \begin{lrbox}{\cPHiso}\scriptsize
    \begin{tikzcd}[cramped, row sep=small, column sep=small]
      \Cof \times P_{\McU_0}^\bI(\HIso_{\McU_0}^{\Id_0}(\tMcU_0)) \ar[d, "{\Cof \times \ev_\partial}"] \\
      \Cof \times (\HIso_{\McU_0}^{\Id_0}(\tMcU_0) \times_{\McU_0} \HIso_{\McU_0}^{\Id_0}(\tMcU_0))
    \end{tikzcd}
  \end{lrbox}
  \begin{equation*}
    \begin{tikzcd}[column sep=3em]
      {\left(
          \usebox{\dCof} \fracsquareslash{\Cof} \usebox{\cHiso}
        \right)}
      \ar[r, "{\text{\cite[Con. 3.2]{struct-lift}}}"]
      &
      {\left(
          \usebox{\dICof} \fracsquareslash{\Cof \times \bI} \usebox{\cIHiso}
        \right)}
      \ar[d, "{\text{\cite[Con. 3.4]{struct-lift}}}"]
      \\
      {\left(
          \usebox{\dICofRes} \fracsquareslash{\Cof} \usebox{\cHiso}
        \right)}
      \ar[d, tail reversed, "{\cong}"', "{\text{\cite[Cor. 5.11]{struct-lift}}}"]
      &
      {\left(
          \usebox{\dICof} \fracsquareslash{\Cof} \usebox{\cHiso}
        \right)}
      \ar[l, "{\text{\cite[Con. 2.4]{struct-lift}}}"']
      \\
      {\left(
          \usebox{\dCof} \fracsquareslash{\Cof} \usebox{\cPHiso}
        \right)}
    \end{tikzcd}
  \end{equation*}
\end{construction}

We are now able to show univalence from the homotopy isomorphism extension
operation.
\begin{theorem}\label{thm:ptd-cubical-univalence}
  Suppose $\partial\Cof \hookrightarrow \Cof$ admits a truth-structure relative
  to $\tMcU \to \McU$ under which $\bI$ has a disjoint cofibrant endpoint
  structure.
  Further assume
  \begin{enumerate}
    \item $\bI$ is equipped with a $\min \colon \bI \times \bI \to \bI$
    structure.
    \item Sections to $\pi$-fibrations are equipped with an $\iota$-cofibration
    structure.
    \item There is a $\mbbo{0}$-biased filling structure on the above data.
    \item The ambient universe $\tMcU \to \McU$ is equipped with $\Path$-,
    $\Sigma$-, $\Pi$-type structures.
    \item The internal universe $\tMcU_0 \to \McU_0$ is equipped with $\Id$-,
    $\Sigma$-, $\Pi$-type structures with the $\Id$-type structure denoted
    $\Id_0 \colon \tMcU_0 \times_{\McU_0} \tMcU_0 \to \McU_0$.
    \item The internal universe $\tMcU_0 \to \McU_0$ is equipped with an
    $\Id_0$-homotopy isomorphism extension structure.
  \end{enumerate}

  Then, the (internal universe, $\Id$-type)-pair $(\pi_0,\Id_0)$ is
  pointed $\pi$-univalent.
\end{theorem}
\begin{proof}
  By \Cref{thm:Path-Id}, the $\Path$-type structure on the ambient universe
  $\tMcU \to \McU$ is also an $\Id$-type structure.
  Therefore, by \cite[Theorem 3.12(2)]{axm-univalence}, it suffices to show that
  \begin{equation*}
    \trv \colon \McU_0 \hookrightarrow \HIso_{\McU_0}^{\Id_0}(\tMcU_0)
  \end{equation*}
  lifts on the left against
  \begin{equation*}
    P_{\McU_0}^\bI(\HIso_{\McU_0}^{\Id_0}(\tMcU_0)) \xrightarrow{\ev_\partial}
    {\HIso_{\McU_0}^{\Id_0}(\tMcU_0) \times_{\McU_0} \HIso_{\McU_0}^{\Id_0}(\tMcU_0)}
  \end{equation*}
  This follows from \Cref{constr:hiso-endpt-ext} and the homotopy isomorphism
  extension operation on the internal universe.
\end{proof}

\subsection{Weak Homotopy Isomorphism Extension Operation}
The strong homotopy isomorphism extension operation of
\Cref{def:strong-hiso-ext-op}, which extends all the homotopy data (i.e. not
only the homotopy isomorphism itself but also its two inverses and the
homotopies to identities), was shown in \Cref{thm:ptd-cubical-univalence} to
give pointed univalence of \Cref{def:axm-univalence}, which is something
stronger than book univalence.
By requiring to only extend the underlying map of the homotopy isomorphism
instead of all the homotopy data, we are still able to show book univalence
using \cite[Theorem 3.12(1)]{axm-univalence} via an identical argument used in
\cite[Section 7.2]{cchm15}.

\begin{definition}\label{def:weak-hiso-ext-op}
  Assume $\tMcU \to \McU$ has a pre-$\Id$-type structure
  $\Id \colon \tMcU \times_\McU \tMcU \to \McU$.

  A \emph{weak $\Id$-homotopy isomorphism extension structure} on
  $\tMcU \to \McU$ is an assignment of a lift factoring through the map
  $\HIso_\McU^\Id(\tMcU) \to [\tMcU \times \McU, \McU \times \tMcU]_{\McU \times
    \McU}$ forgetting the homotopy data of a homotopy isomorphism, as in the
  dotted map below, for each $\iota$-cofibration $\partial A \hookrightarrow A$,
  as on the left below, and lifting problem against the internal codomain map
  $\tgt \colon [\tMcU \times \McU, \McU \times \tMcU]_{\McU \times \McU} \to
  \McU$ factoring through the map
  $\HIso_\McU^\Id(\tMcU) \to [\tMcU \times \McU, \McU \times \tMcU]_{\McU \times
    \McU}$, as in the dashed maps below.
  \begin{equation*}
    % https://q.uiver.app/#q=WzAsNixbMiwwLCJcXEhJc29fXFxNY1VeXFxJZChcXHRNY1UpIl0sWzMsMCwiW1xcdE1jVSBcXHRpbWVzIFxcTWNVLCBcXE1jVSBcXHRpbWVzIFxcdE1jVV1fe1xcTWNVIFxcdGltZXMgXFxNY1V9Il0sWzIsMSwiXFxISXNvX1xcTWNVXlxcSWQoXFx0TWNVKSJdLFszLDIsIlxcTWNVIl0sWzAsMCwiXFxwYXJ0aWFsIEEiXSxbMCwyLCJBIl0sWzEsMywiXFx0Z3QiXSxbMCwxXSxbMiwxXSxbNCw1XSxbNSwzLCIiLDEseyJzdHlsZSI6eyJib2R5Ijp7Im5hbWUiOiJkYXNoZWQifX19XSxbNCwwLCIiLDAseyJzdHlsZSI6eyJib2R5Ijp7Im5hbWUiOiJkYXNoZWQifX19XSxbNSwyLCIiLDEseyJzdHlsZSI6eyJib2R5Ijp7Im5hbWUiOiJkb3R0ZWQifX19XV0=&macro_url=https%3A%2F%2Fgist.githubusercontent.com%2Flim495062%2F61b94af9ef95c1c7b0763c937de29c2b%2Fraw%2F254aa652b96f7d1c6dece2f3eba929431bb0adfe%2Facmhwmacros.sty
    \begin{tikzcd}[cramped, row sep=small, column sep=small]
      {\partial A} && {\HIso_\McU^\Id(\tMcU)} & {[\tMcU \times \McU, \McU \times \tMcU]_{\McU \times \McU}} \\
      && {\HIso_\McU^\Id(\tMcU)} \\
      A &&& \McU
      \arrow[dashed, from=1-1, to=1-3]
      \arrow[from=1-1, to=3-1]
      \arrow[from=1-3, to=1-4]
      \arrow["\tgt", from=1-4, to=3-4]
      \arrow[from=2-3, to=1-4]
      \arrow[dotted, from=3-1, to=2-3]
      \arrow[dashed, from=3-1, to=3-4]
    \end{tikzcd}
  \end{equation*}
  This choice of lift must also be stable under pullback of the cofibration.
\end{definition}

By taking the pre-$\Id$-structure of a weak homotopy isomorphism extension
structure to be the path type, one can reproduce the proof of univalence used in
\cite[Section 7.2]{cchm15}.
\begin{theorem}\label{thm:weak-cubical-univalence}
  Suppose $\partial\Cof \hookrightarrow \Cof$ admits a truth-structure relative
  to $\tMcU \to \McU$ under which $\bI$ has a disjoint cofibrant endpoint structure.
  Further assume
  \begin{enumerate}
    \item $\bI$ is equipped with a $\min$-structure.
    \item The above data is equipped with a $\mbbo{0}$-biased filling structure.
    \item The ambient universe $\tMcU \to \McU$ is equipped with
    $\Path$-,$\Sigma$,-$\Pi$-type structures respectively denoted
    $\Path,\Sigma,\Pi$.
    \item The internal universe $\tMcU_0 \to \McU_0$ is closed under
    $\Path$-,$\Sigma$,-$\Path$-type structures respectively denoted
    $\Path_0$,$\Sigma_0$,$\Pi_0$.
    \item The internal universe $\tMcU_0 \to \McU_0$ is equipped with a weak
    $\Path$-homotopy isomorphism extension structure.
  \end{enumerate}
  Then, the (internal universe, pre-$\Id$-type)-pair $(\pi_0,\Path_0)$ is book
  $\pi$-univalent.
\end{theorem}
\begin{proof}
  The weak $\Path$-homotopy isomorphism extension structure assumption provides
  \cite[Theorem 9]{cchm15}.
  Thus, the proof of \cite[Corollary 10]{cchm15} that
  $\tgt \colon \HIso_{\McU_0}^{\Path_0}(\tMcU_0) \twoheadrightarrow \McU_0$ is
  contractible can be repeated.
  Contractibility, along with transport, then results in book univalence by
  using either \cite[Corollary 11]{cchm15} or \cite[Proposition
  3.11(1)]{axm-univalence}.
\end{proof}

\subsection{Gluing Operation}
In the \cite{cchm15}, the weak $\Path$-homotopy isomorphism extension structure
(\Cref{def:weak-hiso-ext-op}), which leads to univalence as in \cite[Corollary
11]{cchm15}, was obtained as a consequence of the \emph{Glue-unglue} operation,
which we now recall by presenting it in the framework of universe categories.
In short, it is the type-theoretic counterpart of the upcoming
\Cref{constr:glue-op}, which was also used in \cite[Theorem 3.4.1]{kl21} of the
simplicial model of univalence.

\begin{construction}\label{constr:glue-op}
  Let there be a mono $i \colon \partial B \hookrightarrow B$ along with an
  object $\underline{E} \to B \in \sfrac{\bC}{B}$ whose pullback under
  $i$ is $\partial\underline{E} \to \partial B$.
  Functoriality of the pushforward
  $i_* \colon \sfrac{\bC}{\partial B} \to \sfrac{\bC}{B}$ restricts to a map
  $i_* \colon \sfrac{\bC}{\partial\underline{E}} \to
  \sfrac{\bC}{i_*(\partial\underline{E})}$.
  Denote by
  \begin{equation*}
    \unglue_{i,\underline{E}} \colon
   \sfrac{\bC}{\partial\underline{E}} \to \sfrac{\bC}{\underline{E}}
  \end{equation*}
  the composite
  \begin{equation*}
    \sfrac{\bC}{\partial\underline{E}} \xrightarrow{i_*}
    \sfrac{\bC}{i_*(\partial\underline{E})} \xrightarrow{\eta^*}
    \sfrac{\bC}{\underline{E}}
  \end{equation*}
  where
  $i_* \colon \sfrac{\bC}{\partial\underline{E}} \to
  \sfrac{\bC}{i_*(\partial\underline{E})}$ is the restricted pushforward and
  $\eta \colon \underline{E} \to i_*(\partial\underline{E}) \cong
  i_*i^*\underline{E}$ is the unit of the adjunction $i^* \dashv i_*$.
  We also write $\Glue_{i,\ul{E}}$ for the domain of
  $\unglue_{i,\ul{E}}$.

  In other words given $f \colon \partial E \to \partial\ul{E}$, the map
  $\unglue_{i,\ul{E}}(f) \colon \Glue_{i,\ul{E}}(f) \to \underline{E}$
  is constructed by the following diagram.
  \begin{equation*}
    % https://q.uiver.app/#q=WzAsOCxbMSwyLCJcXHBhcnRpYWwgQiJdLFszLDIsIkIiXSxbMiwxLCJcXHBhcnRpYWwgXFx1bmRlcmxpbmUgRSJdLFswLDAsIlxccGFydGlhbCBFIl0sWzQsMSwiXFx1bmRlcmxpbmV7RX0iXSxbNSwxLCJpXyooXFxwYXJ0aWFsXFx1bmRlcmxpbmV7RX0pIl0sWzQsMCwiaV8qXFxwYXJ0aWFsIEUiXSxbMywwLCJcXEV4dChcXHBhcnRpYWwgRSkiXSxbMCwxLCJpIiwyLHsic3R5bGUiOnsidGFpbCI6eyJuYW1lIjoiaG9vayIsInNpZGUiOiJ0b3AifX19XSxbNCwxLCJwIiwxXSxbMiwwLCJcXHBhcnRpYWwgcCIsMV0sWzIsNCwiIiwyLHsiY29sb3VyIjpbMCw2MCw2MF0sInN0eWxlIjp7InRhaWwiOnsibmFtZSI6Imhvb2siLCJzaWRlIjoidG9wIn19fV0sWzIsMSwiIiwxLHsic3R5bGUiOnsibmFtZSI6ImNvcm5lciJ9fV0sWzMsMl0sWzMsMF0sWzQsNSwiXFxldGEiLDFdLFs1LDFdLFs2LDVdLFs2LDFdLFs3LDQsIiIsMSx7ImNvbG91ciI6WzAsNjAsNjBdLCJzdHlsZSI6eyJib2R5Ijp7Im5hbWUiOiJkYXNoZWQifX19XSxbNyw2XSxbNyw1LCIiLDEseyJzdHlsZSI6eyJuYW1lIjoiY29ybmVyIn19XSxbNywxLCIiLDEseyJzdHlsZSI6eyJib2R5Ijp7Im5hbWUiOiJkYXNoZWQifX19XSxbMyw3LCIiLDEseyJzdHlsZSI6eyJib2R5Ijp7Im5hbWUiOiJkYXNoZWQifX19XV0=
    \begin{tikzcd}[cramped]
      {\partial E} &&& {\Glue_{i,\underline{E}}(f)} & {i_*(\partial E)} \\
      && {\partial \underline E} && {\underline{E}} & {i_*(\partial\underline{E})} \\
      & {\partial B} && B
      %\arrow[dashed, from=1-1, to=1-4]
      \arrow[from=1-1, to=2-3]
      \arrow[from=1-1, to=3-2]
      \arrow[from=1-4, to=1-5]
      \arrow[dashed, from=1-4, to=3-4]
      \arrow[from=1-5, to=2-6]
      \arrow[from=1-5, to=3-4]
      \arrow[from=2-3, to=3-2]
      \arrow["\eta"{description}, from=2-5, to=2-6]
      \arrow[from=2-5, to=3-4]
      \arrow[from=2-6, to=3-4]
      \arrow["i"', hook, from=3-2, to=3-4]
      \arrow[crossing over, dashed, from=1-4, to=2-5]
      \arrow[crossing over, hook, from=2-3, to=2-5]
      \arrow["\lrcorner"{anchor=center, pos=0.15, scale=1.5, rotate=45}, draw=none, from=1-4, to=2-6]
      \arrow["\lrcorner"{anchor=center, pos=0.15, scale=1.5, rotate=0}, draw=none, from=2-3, to=3-4]
    \end{tikzcd}
  \end{equation*}
\end{construction}

We then recall a few standard properties about the above $\unglue$ construction.

\begin{lemma}\label{lem:mono-pushfoward-counit-iso}
  Take $\bC$ to be any category.
  Let $i \colon \partial B \to B \in \bC$ be a monomorphism where the
  pullback-pushforward adjunction $i^* \dashv i_*$ along it exists.
  Then, the counit $\ev\relax \colon i^*i_* \to \id_{\bC/\partial B}$ is an
  isomorphism and the pullback of the unit
  $\eta\relax \colon \id_{\bC/B} \to i_*i^*$ along $i$ is an isomorphism.
\end{lemma}
\begin{proof}
  It suffices to show that the counit
  $\ev\relax \colon i^*i_* \to \id_{\bC/\partial B}$ is an isomorphism as this
  implies the pullback of the unit $\eta\relax \colon \id_{\bC/B} \to i_*i^*$
  is an isomorphism by the triangle laws.
  To show that $\ev\relax \colon i^*i_* \to \id_{\bC/\partial B}$ is an
  isomorphism it suffices to show that the right adjoint $i_*$ is fully
  faithful.

  To do so, we first observe that because $i$ is a mono, the postcomposition
  left adjoint $i_!$ to $i^*$ is fully faithful, so this is the same as saying
  the unit $i^*i_! \to \id_{\bC/\partial B}$ is an isomorphism.
  Hence, for each $\partial Y \to \partial B \in \sfrac{\bC}{\partial B}$, one
  has the following series of isomorphisms, naturally in
  $\partial X \to \partial B \in \sfrac{\bC}{\partial B}$, by applying the
  transpose $i_! \dashv i^* \dashv i_*$ sequentially:
  $\sfrac{\bC}{\partial B}(\partial X, \partial Y) \cong
  \sfrac{\bC}{\partial B}(i^*i_!(\partial X), \partial Y) \cong
  \sfrac{\bC}{\partial B}(\partial X, i^*i_*(\partial Y))$.
  By representability, this shows $\partial Y \cong i^*i_*(\partial Y)$.
  Thus, $i_*$ is fully faithful because for any
  $\partial Y, \partial Z \rightrightarrows \partial B \in \sfrac{\bC}{\partial
    A}$, one has
  $\sfrac{\bC}{\partial B}(\partial Y, \partial Z) \cong \sfrac{\bC}{\partial
    A}(i^*i_*(\partial Y), \partial Z) \cong \sfrac{\bC}{A}(i_*(\partial Y),
  i_*(\partial Z))$.
\end{proof}

\begin{lemma}\label{lem:ext-retract}
  When $i \colon \partial B \hookrightarrow B$ is a mono and one has
  $\underline{E} \to B \in \sfrac{\bC}{B}$ along with
  $\partial\underline{E} \to \partial B \in \sfrac{\bC}{B}$ as in
  \Cref{constr:glue-op}, the functor
  $\unglue_{i,\underline{E}} \colon \sfrac{\bC}{\partial\underline{E}} \to
  \sfrac{\bC}{ \underline{E}}$ is a retraction of the pullback functor
  $\sfrac{\bC}{B} \to \sfrac{\bC}{\partial B}$ restricted to
  $i^* \colon \sfrac{\bC}{\underline{E}} \to \sfrac{\bC}{\partial
    \underline{E}}$.
  That is, for
  $f \colon \partial E \to \partial\underline{E} \in
  \sfrac{\bC}{\partial\underline{E}}$, one has
  $i^*\unglue_{i,\underline{E}}(f) \cong f \in
  \sfrac{\bC}{\partial\underline{E}}$.
\end{lemma}
\begin{proof}
  Fix an object $f \colon \partial E \to \partial\underline{E}$.
  Then, by \Cref{lem:mono-pushfoward-counit-iso}, pulling back along $i$ the
  square over $B$ whose bottom edge is the unit $\eta$ of the adjunction
  $i^* \dashv i_*$ yields a square over $\partial B$ whose bottom edge is an
  isomorphism as follows.
  \begin{equation*}
    % https://q.uiver.app/#q=WzAsMTAsWzEsMiwiXFxwYXJ0aWFsIEIiXSxbNCwyLCJCIl0sWzIsMSwiXFxwYXJ0aWFsIFxcdW5kZXJsaW5lIEUgPSBpXipcXHVuZGVybGluZXtFfSJdLFswLDAsImleKlxcRXh0KFxccGFydGlhbCBFKSJdLFs1LDEsIlxcdW5kZXJsaW5le0V9Il0sWzYsMSwiaV8qKFxccGFydGlhbFxcdW5kZXJsaW5le0V9KSJdLFs1LDAsImlfKihcXHBhcnRpYWwgRSkiXSxbNCwwLCJcXEV4dChcXHBhcnRpYWwgRSkiXSxbMywxLCJpXippXyooXFxwYXJ0aWFsXFx1bmRlcmxpbmV7RX0pIl0sWzIsMCwiaV4qaV8qKFxccGFydGlhbCBFKSJdLFswLDEsImkiLDIseyJzdHlsZSI6eyJ0YWlsIjp7Im5hbWUiOiJob29rIiwic2lkZSI6InRvcCJ9fX1dLFs0LDEsInAiLDFdLFsyLDAsIlxccGFydGlhbCBwIiwxXSxbMywyLCIiLDEseyJjb2xvdXIiOlswLDYwLDYwXSwic3R5bGUiOnsiYm9keSI6eyJuYW1lIjoiZGFzaGVkIn19fV0sWzMsMCwiIiwxLHsic3R5bGUiOnsiYm9keSI6eyJuYW1lIjoiZGFzaGVkIn19fV0sWzQsNSwiXFxldGEiLDFdLFs1LDFdLFs2LDVdLFs2LDFdLFs3LDQsIiIsMSx7ImNvbG91ciI6WzAsNjAsNjBdLCJzdHlsZSI6eyJib2R5Ijp7Im5hbWUiOiJkYXNoZWQifX19XSxbNyw2XSxbNyw1LCIiLDEseyJzdHlsZSI6eyJuYW1lIjoiY29ybmVyIn19XSxbMiw4LCJpXipcXGV0YSIsMCx7InN0eWxlIjp7InRhaWwiOnsibmFtZSI6ImFycm93aGVhZCJ9fX1dLFs5LDhdLFszLDksIlxcY29uZyIsMSx7InN0eWxlIjp7InRhaWwiOnsibmFtZSI6ImFycm93aGVhZCJ9LCJib2R5Ijp7Im5hbWUiOiJkYXNoZWQifX19XSxbNywxLCIiLDEseyJzdHlsZSI6eyJib2R5Ijp7Im5hbWUiOiJkYXNoZWQifX19XSxbMyw4LCIiLDEseyJzdHlsZSI6eyJuYW1lIjoiY29ybmVyIn19XSxbOSwwXSxbOCwwXV0=
    \begin{tikzcd}[cramped]
      {i^*\Glue_{i,\underline{E}}(f)} && {i^*i_*(\partial E)} && {\Glue_{i,\ul{E}}(f)} & {i_*(\partial E)} \\
      && {\partial \underline E = i^*\underline{E}} & {i^*i_*(\partial\underline{E})} && {\underline{E}} & {i_*(\partial\underline{E})} \\
      & {\partial B} &&& B
      \arrow["\cong"{description}, dashed, tail reversed, from=1-1, to=1-3]
      \arrow[dashed, from=1-1, to=3-2]
      \arrow[from=1-3, to=2-4]
      \arrow[from=1-3, to=3-2]
      \arrow[from=1-5, to=1-6]
      \arrow[dashed, from=1-5, to=3-5]
      \arrow[from=1-6, to=2-7]
      \arrow[from=1-6, to=3-5]
      \arrow["{i^*\eta}", tail reversed, from=2-3, to=2-4]
      \arrow[from=2-3, to=3-2]
      \arrow[from=2-4, to=3-2]
      \arrow["\eta"{description}, from=2-6, to=2-7]
      \arrow[from=2-6, to=3-5]
      \arrow[from=2-7, to=3-5]
      \arrow["i"', hook, from=3-2, to=3-5]
      \arrow[crossing over, dashed, from=1-1, to=2-3]
      \arrow[crossing over, dashed, from=1-5, to=2-6]
      \arrow["\lrcorner"{anchor=center, pos=0.15, scale=1.5, rotate=45}, draw=none, from=1-1, to=2-4]
      \arrow["\lrcorner"{anchor=center, pos=0.15, scale=1.5, rotate=45}, draw=none, from=1-5, to=2-7]
    \end{tikzcd}
  \end{equation*}
  Thus, $i^*\unglue_{i,\underline{E}}(f) \cong i^*i_*f$ and
  $i^*i_*(\partial E) \cong \partial E$ by
  \Cref{lem:mono-pushfoward-counit-iso} again.
\end{proof}

In \cite{cchm15}, the gluing operation is only defined for those
$f \colon \partial E \to \partial\ul{E}$ that are $\Path$-homotopy isomorphisms.
Therefore, in order to axiomatise the gluing operation in the framework of
universe categories, we must produce the generic extension problem of
\Cref{constr:glue-op} where the map $f \colon \partial E \to \partial\ul{E}$ is
additionally a $\Path$-homotopy isomorphism.

We do so in several steps, first by producing the generic $\Id$-homotopy
isomorphism between two fibrant objects in a slice.
\begin{construction}\label{constr:gen-path-htpy}
  Let $\Id \colon \tMcU \times_\McU \tMcU \to \McU$ be a pre-$\Id$-type
  structure on a universal map $\pi \colon \tMcU \to \McU$.
  Then, by \Cref{constr:hiso}, one obtains the maps
  \begin{equation*}
    \src, \tgt \colon \HIso_\McU^\Id(\tMcU) \rightrightarrows \McU
  \end{equation*}
  Taking the pullback of $\tMcU \to \McU$ along $\src,\tgt$ then yields the
  $\pi$-fibrant objects
  $\src^*\tMcU, \tgt^*\tMcU \rightrightarrows \HIso_\McU^\Id(\tMcU)$.
  One constructs the following horizontal map $\GenHIso_\McU^\Id(\tMcU)$
  \begin{equation*}
    % https://q.uiver.app/#q=WzAsMyxbMSwxLCJcXEhJc29fXFxNY1VeXFxQYXRoKFxcdE1jVSkiXSxbMiwwLCJcXHRndF4qXFx0TWNVIl0sWzAsMCwiXFxzcmNeKlxcdE1jVSJdLFsyLDBdLFsxLDBdLFsyLDEsIlxcR2VuSElzb19cXE1jVV5cXElkKFxcdE1jVSkiXV0=
    \begin{tikzcd}[cramped, column sep=small, row sep=small]
      {\src^*\tMcU} && {\tgt^*\tMcU} \\
      & {\HIso_\McU^\Id(\tMcU)}
      \arrow["{\GenHIso_\McU^\Id(\tMcU)}", from=1-1, to=1-3]
      \arrow[from=1-1, to=2-2]
      \arrow[from=1-3, to=2-2]
    \end{tikzcd}
  \end{equation*}
  as the unique map whose image under the Yoneda embedding yields the following
  map of presheaves as below that post-composes the map
  $f \colon \Gamma.A \to \Gamma.B$ with the section
  $a \colon \Gamma \to \Gamma.A$
  \begin{equation*}
    \begin{tikzcd}
      {\displaystyle
        \left(
          \coprod_{\substack{A \colon \Gamma \to \McU \\ B \colon \Gamma \to \McU}}
          \coprod_{a \in \sfrac{\bC}{\Gamma}(\Gamma,\Gamma.A)}
          \coprod_{\substack{f \in \sfrac{\bC}{\Gamma}(\Gamma.A,\Gamma.B) \\ s,r \in \sfrac{\bC}{\Gamma}(\Gamma.B,\Gamma.A)}}
          \left\{\scriptsize
            \begin{pmatrix}
              H_s \in \sfrac{\bC}{\Gamma}(\Gamma.A, \Id_\Gamma(\Gamma.A)) \\
              H_r \in \sfrac{\bC}{\Gamma}(\Gamma.B, \Id_\Gamma(\Gamma.B)) \\
            \end{pmatrix}
            ~\middle|~
            H_s \colon fs \simeq \id
            \text{ and }
            H_r \colon rf \simeq \id
          \right\}
        \right)_{\Gamma \in \bC}}
      \ar[d, "{(A,B,a,f,s,r,H_s,H_r) \mapsto (A,B,fa,f,s,r,H_s,H_r)}"]
      \\
      {\displaystyle
        \left(
          \coprod_{\substack{A \colon \Gamma \to \McU \\ B \colon \Gamma \to \McU}}
          \coprod_{b \in \sfrac{\bC}{\Gamma}(\Gamma,\Gamma.B)}
          \coprod_{\substack{f \in \sfrac{\bC}{\Gamma}(\Gamma.A,\Gamma.B) \\ s,r \in \sfrac{\bC}{\Gamma}(\Gamma.B,\Gamma.A)}}
          \left\{\scriptsize
            \begin{pmatrix}
              H_s \in \sfrac{\bC}{\Gamma}(\Gamma.A, \Id_\Gamma(\Gamma.A)) \\
              H_r \in \sfrac{\bC}{\Gamma}(\Gamma.B, \Id_\Gamma(\Gamma.B)) \\
            \end{pmatrix}
            ~\middle|~
            H_s \colon fs \simeq \id
            \text{ and }
            H_r \colon rf \simeq \id
          \right\}
        \right)_{\Gamma \in \bC}}
    \end{tikzcd}
  \end{equation*}
\end{construction}

One observes the following universal property of \Cref{constr:gen-path-htpy}.
\begin{lemma}\label{lem:gen-hiso}
  Let $\Id \colon \tMcU \times_\McU \tMcU \to \McU$ be a pre-$\Id$-type
  structure on a universal map $\pi \colon \tMcU \to \McU$.
  In the context of \Cref{constr:gen-path-htpy},
  \begin{enumerate}
    \item The map $\GenHIso_\McU^\Id(\tMcU) \colon \src^*\tMcU \to \tgt^*\tMcU$
    is an $\Id$-homotopy isomorphism
    \item Any other $\Id$-homotopy isomorphism
    $f \colon E_1 \to E_2 \in \sfrac{\bC}{B}$ between a pair of $\pi$-fibrant
    objects $E_1, E_2 \rightrightarrows B$ over an object $B$ occurs as a
    pullback of $\GenHIso_\McU^\Id(\tMcU) \colon \src^*\tMcU \to \tgt^*\tMcU$
    along a map $B \to \HIso_\McU^\Id(\tMcU)$.
    \begin{equation*}
      % https://q.uiver.app/#q=WzAsNixbMiwyLCJcXEhJc29fXFxNY1VeXFxQYXRoKFxcdE1jVSkiXSxbMywxLCJcXHRndF4qXFx0TWNVIl0sWzIsMCwiXFxzcmNeKlxcdE1jVSJdLFsxLDEsIkVfMiJdLFswLDIsIkIiXSxbMCwwLCJFXzEiXSxbMiwwLCIiLDAseyJzdHlsZSI6eyJoZWFkIjp7Im5hbWUiOiJlcGkifX19XSxbMSwwLCIiLDIseyJzdHlsZSI6eyJoZWFkIjp7Im5hbWUiOiJlcGkifX19XSxbMiwxLCJcXEdlbkhJc29fXFxNY1VeXFxJZChcXHRNY1UpIl0sWzUsNCwiIiwwLHsic3R5bGUiOnsiaGVhZCI6eyJuYW1lIjoiZXBpIn19fV0sWzMsNCwiIiwyLHsic3R5bGUiOnsiaGVhZCI6eyJuYW1lIjoiZXBpIn19fV0sWzUsMywiZiJdLFs0LDAsIiIsMix7InN0eWxlIjp7ImJvZHkiOnsibmFtZSI6ImRhc2hlZCJ9fX1dLFszLDEsIiIsMSx7InN0eWxlIjp7ImJvZHkiOnsibmFtZSI6ImRhc2hlZCJ9fX1dLFs1LDIsIiIsMSx7InN0eWxlIjp7ImJvZHkiOnsibmFtZSI6ImRhc2hlZCJ9fX1dLFs1LDAsIiIsMCx7InN0eWxlIjp7Im5hbWUiOiJjb3JuZXIifX1dLFszLDAsIiIsMCx7InN0eWxlIjp7Im5hbWUiOiJjb3JuZXIifX1dXQ==
      \begin{tikzcd}[cramped, row sep=small, column sep=small]
        {E_1} && {\src^*\tMcU} \\
        & {E_2} && {\tgt^*\tMcU} \\
        B && {\HIso_\McU^\Id(\tMcU)}
        \arrow["f"', from=1-1, to=2-2]
        \arrow[two heads, from=1-1, to=3-1]
        \arrow["{\GenHIso_\McU^\Id(\tMcU)}", from=1-3, to=2-4]
        \arrow[two heads, from=1-3, to=3-3]
        \arrow[two heads, from=2-2, to=3-1]
        \arrow[two heads, from=2-4, to=3-3]
        \arrow[dashed, from=1-1, to=1-3, crossing over]
        \arrow[dashed, from=2-2, to=2-4, crossing over]
        \arrow[dashed, from=3-1, to=3-3, crossing over]
        \arrow["\lrcorner"{anchor=center, pos=0.15, scale=1.5}, draw=none, from=1-1, to=3-3]
        \arrow["\lrcorner"{anchor=center, pos=0.15, scale=1.5}, draw=none, from=2-2, to=3-3]
      \end{tikzcd}
    \end{equation*}
  \end{enumerate}
\end{lemma}
\begin{proof}
  Both parts follow by standard representability arguments.
  We show the first part as illustration.

  The homotopy retracts and sections $\tgt^*\tMcU \rightrightarrows \src^*\tMcU$
  of $\GenHIso_\McU^\Id(\tMcU) \colon \src^*\tMcU \to \tgt^*\tMcU$ are defined
  as the unique maps whose image under the Yoneda embedding are the natural
  transformations that respectively pre-compose the section
  $b \colon \Gamma \to \Gamma.B$ with
  $s,r \colon \Gamma.B \rightrightarrows \Gamma.A$
  \begin{equation*}
    \begin{tikzcd}
      {\displaystyle
        \left(
          \coprod_{\substack{A \colon \Gamma \to \McU \\ B \colon \Gamma \to \McU}}
          \coprod_{a \in \sfrac{\bC}{\Gamma}(\Gamma,\Gamma.A)}
          \coprod_{\substack{f \in \sfrac{\bC}{\Gamma}(\Gamma.A,\Gamma.B) \\ s,r \in \sfrac{\bC}{\Gamma}(\Gamma.B,\Gamma.A)}}
          \left\{\scriptsize
            \begin{pmatrix}
              H_s \in \sfrac{\bC}{\Gamma}(\Gamma.A, \Id_\Gamma(\Gamma.A)) \\
              H_r \in \sfrac{\bC}{\Gamma}(\Gamma.B, \Id_\Gamma(\Gamma.B)) \\
            \end{pmatrix}
            ~\middle|~
            H_s \colon fs \simeq \id
            \text{ and }
            H_r \colon rf \simeq \id
          \right\}
        \right)_{\Gamma \in \bC}}
      \\
      {\displaystyle
        \left(
          \coprod_{\substack{A \colon \Gamma \to \McU \\ B \colon \Gamma \to \McU}}
          \coprod_{b \in \sfrac{\bC}{\Gamma}(\Gamma,\Gamma.B)}
          \coprod_{\substack{f \in \sfrac{\bC}{\Gamma}(\Gamma.A,\Gamma.B) \\ s,r \in \sfrac{\bC}{\Gamma}(\Gamma.B,\Gamma.A)}}
          \left\{\scriptsize
            \begin{pmatrix}
              H_s \in \sfrac{\bC}{\Gamma}(\Gamma.A, \Id_\Gamma(\Gamma.A)) \\
              H_r \in \sfrac{\bC}{\Gamma}(\Gamma.B, \Id_\Gamma(\Gamma.B)) \\
            \end{pmatrix}
            ~\middle|~
            H_s \colon fs \simeq \id
            \text{ and }
            H_r \colon rf \simeq \id
          \right\}
        \right)_{\Gamma \in \bC}}
      \arrow["{(A,B,b,f,s,r,H_s,H_r) \mapsto (A,B,sb,f,s,r,H_s,H_r)}", shift left, from=2-1, to=1-1]
      \arrow["{(A,B,b,f,s,r,H_s,H_r) \mapsto (A,B,rb,f,s,r,H_s,H_r)}"', shift right, from=2-1, to=1-1]
    \end{tikzcd}
  \end{equation*}

  The required $\Id$-homotopies for the $\Id$-homotopy section and retraction
  are produced in a similar way.
  We take the $\Id$-homotopy for the section case as an example.
  By stability of the $\Id$-type under pullback, one sees that
  $\Id_{\HIso_\McU^\Id(\tMcU)}(\src^*\tMcU)$ is given by the following
  iterated pullbacks.
  \begin{equation*}
    % https://q.uiver.app/#q=WzAsOCxbMCwyLCJcXEhJc29fXFxNY1VeXFxQYXRoKFxcdE1jVSkiXSxbMCwwLCJcXElkX3tcXEhJc29fXFxNY1VeXFxQYXRoKFxcdE1jVSl9KFxcc3JjXipcXHRNY1UpIl0sWzEsMiwiXFxNY1UiXSxbMSwwLCJcXElkX1xcTWNVKFxcdE1jVSkiXSxbMSwxLCJcXHRNY1UgXFx0aW1lc19cXE1jVSBcXHRNY1UiXSxbMCwxLCJcXHNyY14qXFx0TWNVIFxcdGltZXNfe1xcSElzb19cXE1jVV5cXFBhdGgoXFx0TWNVKX0gXFxzcmNeKlxcdE1jVSJdLFsyLDEsIlxcTWNVIl0sWzIsMCwiXFx0TWNVIl0sWzAsMiwiXFxzcmMiLDJdLFsxLDNdLFsxLDIsIiIsMSx7InN0eWxlIjp7Im5hbWUiOiJjb3JuZXIifX1dLFszLDRdLFs0LDJdLFsxLDVdLFs1LDBdLFs1LDRdLFs1LDIsIiIsMSx7InN0eWxlIjp7Im5hbWUiOiJjb3JuZXIifX1dLFszLDddLFs3LDZdLFs0LDZdLFszLDYsIiIsMCx7InN0eWxlIjp7Im5hbWUiOiJjb3JuZXIifX1dXQ==
    \begin{tikzcd}[cramped, row sep=small, column sep=small]
      {\Id_{\HIso_\McU^\Id(\tMcU)}(\src^*\tMcU)} & {\Id_\McU(\tMcU)} & \tMcU \\
      {\src^*\tMcU \times_{\HIso_\McU^\Id(\tMcU)} \src^*\tMcU} & {\tMcU \times_\McU \tMcU} & \McU \\
      {\HIso_\McU^\Id(\tMcU)} & \McU
      \arrow[from=1-1, to=1-2]
      \arrow[from=1-1, to=2-1]
      \arrow[from=1-2, to=1-3]
      \arrow[from=1-2, to=2-2]
      \arrow[from=1-3, to=2-3]
      \arrow[from=2-1, to=2-2]
      \arrow[from=2-1, to=3-1]
      \arrow[from=2-2, to=2-3]
      \arrow[from=2-2, to=3-2]
      \arrow["\src"', from=3-1, to=3-2]
      \arrow["\Id"', from=2-2, to=2-3]
      \arrow["\lrcorner"{anchor=center, pos=0.05, scale=1.5}, draw=none, from=1-1, to=3-2]
      \arrow["\lrcorner"{anchor=center, pos=0.05, scale=1.5}, draw=none, from=1-2, to=2-3]
      \arrow["\lrcorner"{anchor=center, pos=0.05, scale=1.5}, draw=none, from=2-1, to=3-2]
    \end{tikzcd}
  \end{equation*}
  By the two vertical squares on the left above, this means that
  $\Id_{\HIso_\McU^\Id(\tMcU)}(\src^*\tMcU) \cong \src^*\Id_{\McU}(\tMcU)$ and
  so it represents the presheaf
  \begin{equation*}
    {\displaystyle
      \left(
        \coprod_{\substack{A \colon \Gamma \to \McU \\ B \colon \Gamma \to \McU}}
        \coprod_{\substack{(f,s,r,H_s,H_r) \colon \Gamma.A \simeq \Gamma.B}}
        \coprod_{a_0,a_1 \in \sfrac{\bC}{\Gamma}(\Gamma,\Gamma.A)}
        \left\{
            H \in \sfrac{\bC}{\Gamma}(\Gamma, \Id_\Gamma(\Gamma.A)) \\
          ~\middle|~
          \ev_\partial \cdot H = (a_0,a_1)
        \right\}
      \right)_{\Gamma \in \bC}}
  \end{equation*}
  Hence, one obtains a map
  $\src^*\tMcU \to \Id_{\HIso_\McU^\Id(\tMcU)}(\src^*\tMcU)$ serving as the
  $\Id$-homotopy section defined by representability as follows
  \begin{equation*}
    (A,B,a,(f,s,r,H_s,H_r)) \mapsto (A,B,(f,s,r,H_s,H_r),fsa,a,H_sa)
  \end{equation*}
\end{proof}

With the generic $\Id$-homotopy isomorphism constructed, the next step to
defining the generic homotopy isomorphism extension problem to serve as input to
\Cref{constr:glue-op} is to define the object representing such gluing problems.

\begin{construction}\label{constr:gen-glue-prob}
  Let $\Id \colon \tMcU \times_\McU \tMcU \to \McU$ be a pre-$\Id$-type
  structure on a universal map $\pi \colon \tMcU \to \McU$.
  Denote by $\bP_\iota \colon \bC \to \sfrac{\bC}{\Cof}$ the polynomial functor
  associated with $\iota \colon \partial\Cof \hookrightarrow \Cof$.
  That is, $\bP_\iota(X) = \iota_*(\partial\Cof \times X) \to \Cof$.

  Then, the object $\GenGlueProb_\McU^\Id(\tMcU)$ is defined as the pullback of
  the unit $\eta \colon \Cof \times \McU \to \bP_\iota(\McU)$ of
  $\iota^* \dashv \iota_*$ pulled back along the map
  $\bP_\iota(\tgt) \colon \bP_\iota(\HIso_\McU^\Id(\tMcU)) \to \bP_\iota(\McU)$.
  One also defines the map
  $\bP_\iota(\tgt)^*(\Cof \times \tMcU) \twoheadrightarrow
  \GenGlueProb_\McU^\Id(\tMcU)$ as the pullback of
  $\Cof \times \tMcU \twoheadrightarrow \Cof \times \McU$ along
  $\bP_\iota(\tgt)$.
  This construction is summarised by the following iterated pullbacks over $\Cof$.
  \begin{equation*}
    % https://q.uiver.app/#q=WzAsNyxbMCwyLCJcXGJQX1xcaW90YShcXEhJc29fXFxNY1VeXFxJZChcXHRNY1UpKSJdLFsyLDIsIlxcYlBfXFxpb3RhKFxcTWNVKSJdLFsyLDEsIlxcQ29mIFxcdGltZXMgXFxNY1UiXSxbMCwxLCJcXEdlbkdsdWVQcm9iX1xcTWNVXlxcSWQoXFx0TWNVKSJdLFsyLDAsIlxcQ29mIFxcdGltZXMgXFx0TWNVIl0sWzAsMCwiXFxiUF9cXGlvdGEoXFx0Z3QpXiooXFxDb2YgXFx0aW1lcyBcXHRNY1UpIl0sWzEsMywiXFxDb2YiXSxbMCwxLCJcXGJQX1xcaW90YShcXHRndCkiXSxbMiwxLCJcXGV0YSJdLFszLDBdLFszLDJdLFszLDEsIiIsMSx7InN0eWxlIjp7Im5hbWUiOiJjb3JuZXIifX1dLFs0LDIsIiIsMCx7InN0eWxlIjp7ImhlYWQiOnsibmFtZSI6ImVwaSJ9fX1dLFs1LDMsIiIsMSx7InN0eWxlIjp7ImhlYWQiOnsibmFtZSI6ImVwaSJ9fX1dLFs1LDRdLFs1LDIsIiIsMSx7InN0eWxlIjp7Im5hbWUiOiJjb3JuZXIifX1dLFsxLDZdLFswLDZdXQ==&macro_url=https%3A%2F%2Fgist.githubusercontent.com%2Flim495062%2F61b94af9ef95c1c7b0763c937de29c2b%2Fraw%2F254aa652b96f7d1c6dece2f3eba929431bb0adfe%2Facmhwmacros.sty
    \begin{tikzcd}[cramped, row sep=small, column sep=small]
      {\bP_\iota(\tgt)^*(\Cof \times \tMcU)} && {\Cof \times \tMcU} \\
      {\GenGlueProb_\McU^\Id(\tMcU)} && {\Cof \times \McU} \\
      {\bP_\iota(\HIso_\McU^\Id(\tMcU))} && {\bP_\iota(\McU)} \\
      & \Cof
      \arrow[from=1-1, to=1-3]
      \arrow[two heads, from=1-1, to=2-1]
      \arrow[two heads, from=1-3, to=2-3]
      \arrow[from=2-1, to=2-3]
      \arrow[from=2-1, to=3-1]
      \arrow["\eta", from=2-3, to=3-3]
      \arrow["{\bP_\iota(\tgt)}", from=3-1, to=3-3]
      \arrow[from=3-1, to=4-2]
      \arrow[from=3-3, to=4-2]
      \arrow["\lrcorner"{anchor=center, pos=0.15, scale=1.5}, draw=none, from=1-1, to=2-3]
      \arrow["\lrcorner"{anchor=center, pos=0.15, scale=1.5}, draw=none, from=2-1, to=3-3]
    \end{tikzcd}
  \end{equation*}
\end{construction}

\begin{lemma}\label{lem:gen-glue-prob}
  Let $\Id \colon \tMcU \times_\McU \tMcU \to \McU$ be a pre-$\Id$-type
  structure on a universal map $\pi \colon \tMcU \to \McU$.

  Then, in the context of \Cref{constr:gen-glue-prob},
  \begin{enumerate}
    \item\label{itm:gen-glue-prob-rep} The object $\GenGlueProb_\McU^\Id(\tMcU)$ represents the presheaf
    taking each $B$ to a tuple consisting of an $\iota$-cofibration
    $\partial B \hookrightarrow B$, a pair of $\pi$-fibrations
    $\ul{E} \twoheadrightarrow B$ and
    $\partial{E} \twoheadrightarrow \partial{B}$, along with an $\Id$-homotopy
    isomorphism $f \colon \partial{E} \to \partial\ul{E}$ over $\partial{B}$,
    where $\partial\ul{E} \twoheadrightarrow \partial{B}$ is the pullback of
    $\ul{E} \twoheadrightarrow B$ along $\partial B \hookrightarrow B$, just
    like in \Cref{constr:glue-op}.
    \item\label{itm:gen-glue-prob-pb} The pullback of
    $\bP_\iota(\tgt)^*(\Cof \times \tMcU) \twoheadrightarrow
    \GenGlueProb_\McU^\Id(\tMcU)$ along
    $\iota \colon \partial\Cof \hookrightarrow \Cof$ is the map
    $\partial\Cof \times \tgt^*\tMcU \twoheadrightarrow \partial\Cof \times
    \HIso_\McU^\Id(\tMcU)$ obtained from \Cref{constr:gen-path-htpy}.
  \end{enumerate}
\end{lemma}
\begin{proof}
  The first part is again by representability, which is observed as follows.

  The object $\bP_\iota(\HIso_\McU^\Id(\tMcU))$ represents the presheaf taking
  each $B$ to an $\iota$-cofibration $\partial B \hookrightarrow B$ with a
  choice of an $\Id$-homotopy isomorphism $\partial E \to \partial \ul{E}$
  between two $\pi$-fibrant objects over $\partial B$.
  The map
  $\bP_\iota(\tgt) \colon \bP_\iota(\HIso_\McU^\Id(\tMcU)) \to \bP_\iota(\McU)$
  takes each such tuple and returns the $\iota$-cofibration
  $\partial B \hookrightarrow B$ along with the $\pi$-fibration
  $\partial \ul{E} \twoheadrightarrow \partial{B}$ in the codomain of the
  $\Id$-homotopy isomorphism.

  The object $\Cof \times \McU$ represents the presheaf taking each $B$ to an
  $\iota$-cofibration $\partial B \hookrightarrow B$ and a $\pi$-fibration
  $E \twoheadrightarrow B$ over $B$.
  The unit $\eta \colon \Cof \times \McU \to \bP_\iota(\McU)$ takes each such
  ($\iota$-cofibration, $\pi$-fibration)-pair and returns the pullback of the
  $\pi$-fibration along the $\iota$-cofibration.

  Therefore, the result for the first part follows.

  For the second part, we note that by the counit part of
  \Cref{lem:mono-pushfoward-counit-iso}, the naturality square of the counit of
  $\iota^* \dashv \iota_*$ is always an isomorphism.
  Thus, pulling back the iterated pullbacks of \Cref{constr:gen-glue-prob} over
  $\Cof$ along $\partial\Cof \hookrightarrow \Cof$ to be over $\partial\Cof$ and
  then composing with the naturality pullback of the $\iota^* \dashv \iota_*$
  counit gives the following iterated pullback.
  \begin{equation*}
    % https://q.uiver.app/#q=WzAsOCxbMCwyLCJcXGlvdGFeKlxcYlBfXFxpb3RhKFxcSElzb19cXE1jVV5cXElkKFxcdE1jVSkpIl0sWzEsMiwiXFxpb3RhXipcXGJQX1xcaW90YShcXE1jVSkiXSxbMSwxLCJcXHBhcnRpYWxcXENvZiBcXHRpbWVzIFxcTWNVIl0sWzAsMSwiXFxpb3RhXipcXEdlbkdsdWVQcm9iX1xcTWNVXlxcSWQoXFx0TWNVKSJdLFsxLDAsIlxccGFydGlhbFxcQ29mIFxcdGltZXMgXFx0TWNVIl0sWzAsMCwiXFxpb3RhXipcXGJQX1xcaW90YShcXHRndCleKihcXENvZiBcXHRpbWVzIFxcdE1jVSkiXSxbMCwzLCJcXHBhcnRpYWxcXENvZiBcXHRpbWVzIFxcSElzb19cXE1jVV5cXElkKFxcdE1jVSkiXSxbMSwzLCJcXHBhcnRpYWxcXENvZiBcXHRpbWVzIFxcTWNVIl0sWzAsMSwiXFxpb3RhXipcXGJQX1xcaW90YShcXHRndCkiLDJdLFsyLDEsIlxcaW90YV4qXFxldGEiXSxbMywwXSxbMywyXSxbMywxLCIiLDEseyJzdHlsZSI6eyJuYW1lIjoiY29ybmVyIn19XSxbNCwyLCIiLDAseyJzdHlsZSI6eyJoZWFkIjp7Im5hbWUiOiJlcGkifX19XSxbNSwzLCIiLDEseyJzdHlsZSI6eyJoZWFkIjp7Im5hbWUiOiJlcGkifX19XSxbNSw0XSxbNSwyLCIiLDEseyJzdHlsZSI6eyJuYW1lIjoiY29ybmVyIn19XSxbMCw2LCJcXGV2IiwyXSxbMSw3LCJcXGV2Il0sWzYsNywiXFxwYXJ0aWFsXFxDb2YgXFx0aW1lcyBcXHRndCIsMl0sWzIsNywiIiwyLHsiY3VydmUiOi0zLCJsZXZlbCI6Miwic3R5bGUiOnsiaGVhZCI6eyJuYW1lIjoibm9uZSJ9fX1dLFswLDcsIiIsMix7InN0eWxlIjp7Im5hbWUiOiJjb3JuZXIifX1dXQ==&macro_url=https%3A%2F%2Fgist.githubusercontent.com%2Flim495062%2F61b94af9ef95c1c7b0763c937de29c2b%2Fraw%2F254aa652b96f7d1c6dece2f3eba929431bb0adfe%2Facmhwmacros.sty
    \begin{tikzcd}[cramped, row sep=small]
      {\iota^*\bP_\iota(\tgt)^*(\Cof \times \tMcU)} & {\partial\Cof \times \tMcU} \\
      {\iota^*\GenGlueProb_\McU^\Id(\tMcU)} & {\partial\Cof \times \McU} \\
      {\iota^*\bP_\iota(\HIso_\McU^\Id(\tMcU))} & {\iota^*\bP_\iota(\McU)} \\
      {\partial\Cof \times \HIso_\McU^\Id(\tMcU)} & {\partial\Cof \times \McU}
      \arrow[from=1-1, to=1-2]
      \arrow[two heads, from=1-1, to=2-1]
      \arrow[two heads, from=1-2, to=2-2]
      \arrow[from=2-1, to=2-2]
      \arrow[from=2-1, to=3-1]
      \arrow["{\iota^*\eta}", from=2-2, to=3-2]
      \arrow[curve={height=-30pt}, equals, from=2-2, to=4-2]
      \arrow["{\iota^*\bP_\iota(\tgt)}"', from=3-1, to=3-2]
      \arrow["\ev"', from=3-1, to=4-1]
      \arrow["\ev", from=3-2, to=4-2]
      \arrow["{\partial\Cof \times \tgt}"', from=4-1, to=4-2]
      \arrow["\lrcorner"{anchor=center, pos=0.15, scale=1.5}, draw=none, from=1-1, to=2-2]
      \arrow["\lrcorner"{anchor=center, pos=0.15, scale=1.5}, draw=none, from=2-1, to=3-2]
      \arrow["\lrcorner"{anchor=center, pos=0.15, scale=1.5}, draw=none, from=3-1, to=4-2]
    \end{tikzcd}
  \end{equation*}
  Because $\ev \cdot \iota^*\eta = \id$ by the triangle laws, it thus follows
  that
  $\iota^*\GenGlueProb_\McU^\Id(\tMcU) \cong \partial\Cof \times
  \HIso_\McU^\Id(\tMcU)$ and that
  $\iota^*\bP_\iota(\tgt)^*(\Cof \times \tMcU) \to
  \iota^*\GenGlueProb_\McU^\Id(\tMcU)$ is the map
  $\partial\Cof \times \tgt^*\tMcU \twoheadrightarrow \partial\Cof \times
  \HIso_\McU^\Id(\tMcU)$ obtained from \Cref{constr:gen-path-htpy}.
\end{proof}

One may now apply the $\Glue$-type construction of \Cref{constr:glue-op} to the
generic homotopy isomorphism extension problem, constructed as follows.

\begin{construction}\label{constr:gen-glue}
  Let $\Id \colon \tMcU \times_\McU \tMcU \to \McU$ be a pre-$\Id$-type
  structure on $\tMcU \to \McU$.
  By the second part of \Cref{lem:gen-glue-prob}, one has the solid arrows in
  the following diagram making the front face to be a pullback
  \begin{equation*}
    % https://q.uiver.app/#q=WzAsNixbMCwyLCJcXHBhcnRpYWxcXENvZiBcXHRpbWVzIFxcSElzb19cXE1jVV5cXElkKFxcdE1jVSkiXSxbMiwyLCJcXEdlbkdsdWVQcm9iX1xcTWNVXlxcSWQoXFx0TWNVKSJdLFszLDEsIlxcYlBfXFxpb3RhKFxcdGd0KV4qKFxcQ29mIFxcdGltZXMgXFx0TWNVKSJdLFsxLDEsIlxccGFydGlhbFxcQ29mIFxcdGltZXMgXFx0Z3ReKlxcdE1jVSJdLFswLDAsIlxccGFydGlhbFxcQ29mIFxcdGltZXMgXFxzcmNeKlxcdE1jVSJdLFsyLDAsIlxcR2x1ZShcXEdlbkhJc29fXFxNY1VeXFxJZChcXHRNY1UpKSJdLFswLDEsIiIsMCx7InN0eWxlIjp7InRhaWwiOnsibmFtZSI6Imhvb2siLCJzaWRlIjoidG9wIn19fV0sWzIsMSwiIiwyLHsic3R5bGUiOnsiaGVhZCI6eyJuYW1lIjoiZXBpIn19fV0sWzMsMCwiIiwwLHsic3R5bGUiOnsiaGVhZCI6eyJuYW1lIjoiZXBpIn19fV0sWzMsMl0sWzMsMSwiIiwxLHsic3R5bGUiOnsibmFtZSI6ImNvcm5lciJ9fV0sWzQsMCwiIiwxLHsic3R5bGUiOnsiaGVhZCI6eyJuYW1lIjoiZXBpIn19fV0sWzQsMywiXFxwYXJ0aWFsXFxDb2YgXFx0aW1lcyBcXEdlbkhJc29fXFxNY1VeXFxJZChcXHRNY1UpIiwxXSxbNSwyLCJcXHVuZ2x1ZShcXEdlbkhJc29fXFxNY1VeXFxJZChcXHRNY1UpKSIsMSx7InN0eWxlIjp7ImJvZHkiOnsibmFtZSI6ImRhc2hlZCJ9fX1dLFs1LDEsIiIsMSx7InN0eWxlIjp7ImJvZHkiOnsibmFtZSI6ImRhc2hlZCJ9fX1dLFs0LDUsIiIsMSx7InN0eWxlIjp7ImJvZHkiOnsibmFtZSI6ImRhc2hlZCJ9fX1dXQ==
    \begin{tikzcd}[cramped, column sep=small]
      {\partial\Cof \times \src^*\tMcU} && {\Glue(\GenHIso_\McU^\Id(\tMcU))} \\
      & {\partial\Cof \times \tgt^*\tMcU} && {\bP_\iota(\tgt)^*(\Cof \times \tMcU)} \\
      {\partial\Cof \times \HIso_\McU^\Id(\tMcU)} && {\GenGlueProb_\McU^\Id(\tMcU)}
      \arrow[dashed, from=1-1, to=1-3]
      \arrow["{\partial\Cof \times \GenHIso_\McU^\Id(\tMcU)}"{description}, from=1-1, to=2-2]
      \arrow[two heads, from=1-1, to=3-1]
      \arrow["{\unglue(\GenHIso_\McU^\Id(\tMcU))}"{description}, dashed, from=1-3, to=2-4]
      \arrow[dashed, from=1-3, to=3-3]
      \arrow[two heads, from=2-2, to=3-1]
      \arrow[two heads, from=2-4, to=3-3]
      \arrow[hook, from=3-1, to=3-3]
      \arrow[from=2-2, to=2-4, crossing over]
      \arrow["\lrcorner"{anchor=center, pos=0.15, scale=1.5}, draw=none, from=2-2, to=3-3]
    \end{tikzcd}
  \end{equation*}
  so by Glue-unglue of \Cref{constr:glue-op}, one obtains the dashed object
  \begin{equation*}
    \Glue(\GenHIso_\McU^\Id(\tMcU)) \to \GenGlueProb_\McU^\Id(\tMcU) \in
    \sfrac{\bC}{\GenGlueProb_\McU^\Id(\tMcU)}
  \end{equation*}
  and dashed map
  \begin{equation*}
    \unglue(\GenHIso_\McU^\Id(\tMcU)) \colon
    \Glue(\GenHIso_\McU^\Id(\tMcU)) \to \bP_\iota(\tgt)^*(\Cof \times \tMcU)
  \end{equation*}
  as in the triangle on the right side of the above diagram.
\end{construction}

The above $\Glue$ construction admits the following universality.

\begin{proposition}\label{prop:gen-glue}
  Let $\Id \colon \tMcU \times_\McU \tMcU \to \McU$ be a pre-$\Id$-type
  structure on $\tMcU \to \McU$.

  Take an object $B$ along with an $\iota$-cofibration
  $\partial B \hookrightarrow B$, a pair of $\pi$-fibrations
  $\ul{E} \twoheadrightarrow B$ and
  $\partial{E} \twoheadrightarrow \partial{B}$, along with an $\Id$-homotopy
  isomorphism $f \colon \partial{E} \to \partial\ul{E}$ over $\partial{B}$,
  where $\partial\ul{E} \twoheadrightarrow \partial{B}$ is the pullback of
  $\ul{E} \twoheadrightarrow B$ along $\partial B \hookrightarrow B$, just like
  in \Cref{constr:glue-op}.

  Then, there is a map $B \to \GenGlueProb_\McU^\Id(\tMcU)$ such that the
  associated pullback functor
  $\sfrac{\bC}{\GenGlueProb_\McU^\Id(\tMcU)} \to \sfrac{\bC}{B}$ sends the
  diagram on the left below from \Cref{constr:gen-glue}, viewed as a diagram in
  $\sfrac{\bC}{\GenGlueProb_\McU^\Id(\tMcU)}$, to the diagram on the right below
  from \Cref{constr:glue-op}, viewed as a diagram in $\sfrac{\bC}{B}$.
  \begin{equation*}
    % https://q.uiver.app/#q=WzAsNSxbMSwyLCJcXHBhcnRpYWxcXENvZiBcXHRpbWVzIFxcSElzb19cXE1jVV5cXElkKFxcdE1jVSkiXSxbMiwyLCJcXEdlbkdsdWVQcm9iX1xcTWNVXlxcSWQoXFx0TWNVKSJdLFsyLDEsIlxcYlBfXFxpb3RhKFxcdGd0KV4qKFxcQ29mIFxcdGltZXMgXFx0TWNVKSJdLFsxLDEsIlxccGFydGlhbFxcQ29mIFxcdGltZXMgXFx0Z3ReKlxcdE1jVSJdLFswLDAsIlxccGFydGlhbFxcQ29mIFxcdGltZXMgXFxzcmNeKlxcdE1jVSJdLFswLDEsIiIsMCx7InN0eWxlIjp7InRhaWwiOnsibmFtZSI6Imhvb2siLCJzaWRlIjoidG9wIn19fV0sWzIsMSwiIiwyLHsic3R5bGUiOnsiaGVhZCI6eyJuYW1lIjoiZXBpIn19fV0sWzMsMCwiIiwwLHsic3R5bGUiOnsiaGVhZCI6eyJuYW1lIjoiZXBpIn19fV0sWzMsMiwiIiwyLHsic3R5bGUiOnsidGFpbCI6eyJuYW1lIjoiaG9vayIsInNpZGUiOiJ0b3AifX19XSxbMywxLCIiLDEseyJzdHlsZSI6eyJuYW1lIjoiY29ybmVyIn19XSxbNCwwLCIiLDEseyJzdHlsZSI6eyJoZWFkIjp7Im5hbWUiOiJlcGkifX19XSxbNCwzXV0=
    \left(\begin{tikzcd}[cramped, row sep=small, column sep=small]
      {\partial\Cof \times \src^*\tMcU} \\
      & {\partial\Cof \times \tgt^*\tMcU} & {\bP_\iota(\tgt)^*(\Cof \times \tMcU)} \\
      & {\partial\Cof \times \HIso_\McU^\Id(\tMcU)} & {\GenGlueProb_\McU^\Id(\tMcU)}
      \arrow[from=1-1, to=2-2]
      \arrow[two heads, from=1-1, to=3-2]
      \arrow[hook, from=2-2, to=2-3]
      \arrow[two heads, from=2-2, to=3-2]
      \arrow["\lrcorner"{anchor=center, pos=0.15, scale=1.5}, draw=none, from=2-2, to=3-3]
      \arrow[two heads, from=2-3, to=3-3]
      \arrow[hook, from=3-2, to=3-3]
    \end{tikzcd}\right)
    \mapsto
    % https://q.uiver.app/#q=WzAsNSxbMiwyLCJCIl0sWzEsMiwiXFxwYXJ0aWFse0J9Il0sWzIsMSwiXFx1bHtFfSJdLFsxLDEsIlxccGFydGlhbFxcdWx7RX0iXSxbMCwwLCJcXHBhcnRpYWx7RX0iXSxbMiwwLCIiLDEseyJzdHlsZSI6eyJoZWFkIjp7Im5hbWUiOiJlcGkifX19XSxbMSwwLCIiLDEseyJzdHlsZSI6eyJ0YWlsIjp7Im5hbWUiOiJob29rIiwic2lkZSI6InRvcCJ9fX1dLFs0LDEsIiIsMSx7InN0eWxlIjp7ImhlYWQiOnsibmFtZSI6ImVwaSJ9fX1dLFszLDEsIiIsMSx7InN0eWxlIjp7ImhlYWQiOnsibmFtZSI6ImVwaSJ9fX1dLFs0LDNdLFszLDIsIiIsMCx7InN0eWxlIjp7InRhaWwiOnsibmFtZSI6Imhvb2siLCJzaWRlIjoidG9wIn19fV0sWzMsMCwiIiwxLHsic3R5bGUiOnsibmFtZSI6ImNvcm5lciJ9fV1d
    \left(\begin{tikzcd}[cramped, row sep=small, column sep=small]
      {\partial{E}} \\
      & {\partial\ul{E}} & {\ul{E}} \\
      & {\partial{B}} & B
      \arrow[from=1-1, to=2-2]
      \arrow[two heads, from=1-1, to=3-2]
      \arrow[hook, from=2-2, to=2-3]
      \arrow[two heads, from=2-2, to=3-2]
      \arrow["\lrcorner"{anchor=center, pos=0.15, scale=1.5}, draw=none, from=2-2, to=3-3]
      \arrow[two heads, from=2-3, to=3-3]
      \arrow[hook, from=3-2, to=3-3]
    \end{tikzcd}\right)
  \end{equation*}
\end{proposition}
\begin{proof}
  The required map $B \to \GenGlueProb_\McU^\Id(\tMcU)$ is by representability
  of $\GenGlueProb_\McU^\Id(\tMcU)$ from the first part of
  \Cref{lem:gen-glue-prob}.
  Verifying that the image of the diagram of \Cref{constr:gen-glue} ends up
  being the diagram from \Cref{constr:glue-op} is also straightforward by the
  second part of \Cref{lem:gen-glue-prob} and the second part of
  \Cref{lem:gen-hiso}.
\end{proof}

Thus, with the $\Glue$-type structured defined as below, the map
$\Glue_{i,\ul{E}}(f) \to B$ from \Cref{constr:glue-op} is a $\pi$-fibration
whenever $f$ is an $\Id$-homotopy isomorphism.
\begin{definition}\label{def:glue-type}
  Let $\Id \colon \tMcU \times_\McU \tMcU \to \McU$ be a pre-$\Id$-type
  structure on $\tMcU \to \McU$.

  A \emph{$\Glue$-type structure} on $\tMcU \to \McU $for this pre-$\Id$-type
  structure $\Id \colon \tMcU \times_\McU \tMcU \to \McU$ is a pair of dashed
  maps $(\Glue, \glue)$ making up a pullback as follows, where the left map is
  from \Cref{constr:gen-glue}.
  \begin{equation*}
    % https://q.uiver.app/#q=WzAsNCxbMCwxLCJcXEhJc29fXFxNY1VeXFxJZChcXHRNY1UpIFxcdGltZXMgXFxDb2YiXSxbMCwwLCJcXEdsdWVfe1xcSElzb19cXE1jVV5cXElkKFxcdE1jVSkgXFx0aW1lcyBcXGlvdGF9KFxcR2VuSElzb19cXE1jVV5cXElkKFxcdE1jVSkpIl0sWzEsMSwiXFxNY1UiXSxbMSwwLCJcXHRNY1UiXSxbMSwwXSxbMywyXSxbMCwyLCJcXEdsdWUiLDIseyJzdHlsZSI6eyJib2R5Ijp7Im5hbWUiOiJkYXNoZWQifX19XSxbMSwzLCJcXGdsdWUiLDAseyJzdHlsZSI6eyJib2R5Ijp7Im5hbWUiOiJkYXNoZWQifX19XSxbMSwyLCIiLDEseyJzdHlsZSI6eyJuYW1lIjoiY29ybmVyIn19XV0=
    \begin{tikzcd}[cramped]
      {\Glue(\GenHIso_\McU^\Id(\tMcU))} & \tMcU \\
      {\GenGlueProb_\McU^\Id(\tMcU)} & \McU
      \arrow["\glue", dashed, from=1-1, to=1-2]
      \arrow[from=1-1, to=2-1]
      \arrow["\lrcorner"{anchor=center, pos=0.15, scale=1.5}, draw=none, from=1-1, to=2-2]
      \arrow[from=1-2, to=2-2]
      \arrow["\Glue"', dashed, from=2-1, to=2-2]
    \end{tikzcd}
  \end{equation*}
\end{definition}

As an immediate consequence, we obtain the following.
\begin{corollary}\label{cor:glue-type}
  Let $\Id \colon \tMcU \times_\McU \tMcU \to \McU$ be a pre-$\Id$-type
  structure on $\tMcU \to \McU$.
  Suppose $\tMcU \to \McU$ is equipped with a $\Glue$-type structure for this
  pre-$\Id$-type structure.

  Take an object $B$ along with an $\iota$-cofibration
  $i \colon \partial B \hookrightarrow B$, a pair of $\pi$-fibrations
  $\ul{E} \twoheadrightarrow B$ and
  $\partial{E} \twoheadrightarrow \partial{B}$, along with an $\Id$-homotopy
  isomorphism $f \colon \partial{E} \to \partial\ul{E}$ over $\partial{B}$,
  where $\partial\ul{E} \twoheadrightarrow \partial{B}$ is the pullback of
  $\ul{E} \twoheadrightarrow B$ along $\partial B \hookrightarrow B$, just like
  in \Cref{constr:glue-op}.

  Then, the object
  \begin{equation*}
    \Glue_{i,\ul{E}}(f) \to B \in \sfrac{\bC}{B}
  \end{equation*}
  is $\pi$-fibrant.
\end{corollary}
\begin{proof}
  Immediate by \Cref{prop:gen-glue} and the Beck-Chevalley condition.
\end{proof}

Under the presence of a $\Glue$-type structure for $\Path$-homotopy
equivalences, one can then repeat the proof of univalence as seen in
\cite[Section 7.2]{cchm15}.

\begin{theorem}\label{thm:cchm-cubical-univalence}
  Suppose $\partial\Cof \hookrightarrow \Cof$ admits a truth-structure relative
  to $\tMcU \to \McU$ under which $\bI$ has a disjoint cofibrant endpoint structure.
  Further assume
  \begin{enumerate}
    \item $\bI$ is equipped with a $\min$-structure.
    \item The above data is equipped with a $\mbbo{0}$-biased filling structure.
    \item The ambient universe $\tMcU \to \McU$ is equipped with
    $\Path$-,$\Sigma$,-$\Pi$-type structures $\Path,\Sigma,\Pi$.
    \item The internal universe $\tMcU_0 \to \McU_0$ is closed under
    $\Path$-,$\Sigma$,-$\Path$-type structures respectively denoted
    $\Path_0$,$\Sigma_0$,$\Pi_0$.
    \item The internal universe $\tMcU_0 \to \McU_0$ is equipped with a
    $\Glue$-type structure for the $\Path$-type structure.
  \end{enumerate}
  Then, the (internal universe, pre-$\Id$-type)-pair $(\pi_0,\Path_0)$ is book
  $\pi$-univalent.
\end{theorem}
\begin{proof}
  By the first part of \Cref{lem:gen-hiso}, the map
  $\partial\Cof \times \GenHIso_\McU^\Path(\tMcU)$ of \Cref{constr:gen-glue} is
  a $\Path$-homotopy isomorphism, so the map $\unglue(\GenHIso_\McU^\Id(\tMcU))$
  there is a $\Path$-homotopy isomorphism by \cite[Theorem 9]{cchm15}.
  This means that the $\Glue$-type structure induces a weak $\Path$-homotopy
  equivalence extension structure of \Cref{def:weak-hiso-ext-op}.
  One can now conclude by \Cref{thm:weak-cubical-univalence}.
\end{proof}

%%% Local Variables:
%%% TeX-master: "./main.tex"
%%% TeX-engine: default
%%% End:

\section{Naive Cubical Type Theory}\label{sec:naive-ctt}
Assembling the various definitions and properties from the previous section, we
are now able to precisely formulate a version of cubical type theory.
This version of cubical type theory is formulated in such a way that we can
establish a model of it in elegant Reedy presheaf model categories, a special
case of which includes that of simplicial sets, equipped with the classical
model structure.

\begin{definition}\label{def:naive-ctt}
  A \emph{naive cubical type theory} (NCTT) structure on a locally
  cartesian closed category $\bC$ consists of three universal maps
  \begin{align*}
    \pi\colon\tMcU \to \McU && \iota\colon\partial\Cof \hookrightarrow \Cof && \bI \to 1
  \end{align*}
  where $\iota$ is a monomorphism, along with
  \begin{itemize}
    \item[nctt-truth] \label{itm:nctt-truth} A truth structure on
    $\partial\Cof \hookrightarrow \Cof$ relative to to $\tMcU \to \McU$
    (\Cref{def:cof-truth-lattice}).
    \item[nctt-cof-sec] \label{itm:nctt-cof-sec} An $\iota$-cofibrant structure on
    ambient universe $\pi$-sections (\Cref{def:cof-sect}).
    \item[nctt-int] \label{itm:nctt-int} A disjoint $\iota$-cofibrant endpoint
    structure on $\bI$ whose endpoints are $\mbbo{0},\mbbo{1}$
    (\Cref{def:interval-disjoint-endpt}).
    \item[nctt-min] \label{itm:nctt-min} A $\min$-structure on $\bI$
    (\Cref{def:interval-min-max}).
    \item[nctt-fib] \label{itm:nctt-fib} $\Unit,\Sigma,\Pi$- and
    heterogeneous $\Path$-type structures on $\tMcU \to \McU$
    (\Cref{def:unit-type,def:sigma-type,def:pi-type,def:htr-path-type}).
    \item[nctt-fill] \label{itm:nctt-fill} A $\mbbo{0}$-biased filling structure
    for the ambient external universe $\tMcU \to \McU$
    (\Cref{def:biased-filling}).
    \item[nctt-univ] \label{itm:nctt-univ} An internal universe
    $\tMcU_0 \to \McU_0$ of $\tMcU \to \McU$ closed under $\Unit,\Sigma,\Pi$-
    and heterogeneous $\Path$-type structures
    (\Cref{def:int-univ-sigma-pi-id,def:htr-path-type}).
    \item[nctt-hext] \label{itm:nctt-hext} A strong homotopy equivalence
    extension structure on the internal universe $\tMcU_0 \to \McU_0$
    (\Cref{def:strong-hiso-ext-op}).
  \end{itemize}

  A \emph{weakly computational filling} structure on such an NCTT structure
  consists of
  \begin{itemize}
    \item[nctt-fill-wkcomp] \label{itm:nctt-fill-wkcomp} Stable propositional
    computation structures for the chosen $\mbbo{0}$-filling structure of
    \Cref{itm:nctt-fill} with respect to the chosen $\Sigma,\Pi,\Path$-type
    structures of \Cref{itm:nctt-fib}
    (\Cref{def:Sigma-fill-comp,def:Pi-fill-comp,def:Path-fill-comp}).
  \end{itemize}
  and an NCTT structure is said to have \emph{definitionally computational filling} when
  \begin{itemize}
    \item[nctt-fill-defcomp] \label{itm:nctt-fill-defcomp} The $\mbbo{0}$-filling
    structure of \Cref{itm:nctt-fill} computes definitionally with respect to
    the chosen $\Sigma,\Pi,\Path$-type structures of \Cref{itm:nctt-fib}
    (\Cref{def:Sigma-fill-comp,def:Pi-fill-comp,def:Path-fill-comp}).
  \end{itemize}
\end{definition}

Because our formulation causes much heavy pain, some explanation is needed to
justify the consistency of the points \Cref{itm:nctt-hext,itm:nctt-fill-wkcomp},
which both rely on $\Id$-types.

\begin{theorem}\label{thm:nctt-Path-Id}
  Let $\bC$ be equipped with an NCTT structure with external ambient universe
  $\tMcU \to \McU$ and internal universe $\tMcU_0 \to \McU_0$.

  Then, the NCTT structure canonically gives rise to $\Id$-type structures on
  $\tMcU \to \McU$ and $\tMcU_0 \to \McU_0$.
\end{theorem}
\begin{proof}
  By \Cref{lem:int-univ-truth}, the truth structure on
  $\partial\Cof \hookrightarrow \Cof$ relative to the ambient universe
  $\tMcU \to \McU$ from \Cref{itm:nctt-truth} is also a truth structure relative
  to the internal universe $\tMcU_0 \to \McU_0$.
  Similarly, by \Cref{lem:int-univ-cof-sect}, sections to internal universe
  $\pi_0$-fibrations are equipped with a canonical $\iota$-cofibration
  structure.

  Applying \Cref{prop:htr-path-type-path} on the $\Unit$-type and heterogeneous
  $\Path$-type structures of $\tMcU \to \McU$ and $\tMcU_0 \to \McU_0$ given by
  \Cref{itm:nctt-univ,itm:nctt-fib} produces $\Path$-type structures on each of
  these universal maps.

  Hence, we have observed that
  \begin{enumerate}[nolistsep]
    \item $\partial\Cof \hookrightarrow \Cof$ has a truth structure relative to the internal universe $\tMcU_0 \to \McU_0$
    \item The internal universe $\tMcU_0 \to \McU_0$ is equipped with a $\Path$-type structure
    \item Sections to internal universe $\pi_0$-fibrations are equipped with an
    $\iota$-cofibration structure
  \end{enumerate}
  Therefore, in combination with \Cref{itm:nctt-int,itm:nctt-min,itm:nctt-univ},
  applying \Cref{thm:Path-Id} shows that the induced $\Path$-type structure of
  the internal universe $\tMcU_0 \to \McU_0$ further induces an $\Id$-type
  structure on it.

  Similarly, the induced $\Path$-type structure of the ambient external universe
  $\tMcU \to \McU$ further induces an $\Id$-type structure on it.
\end{proof}

\begin{remark}\label{rmk:nctt-hiso-ext-op}
  In \Cref{itm:nctt-hext} we require a homotopy equivalence extension structure
  on the internal universe, as per \Cref{def:strong-hiso-ext-op}.
  However, \Cref{def:strong-hiso-ext-op} requires an $\Id$-type structure on the
  internal universe to even speak of the $\Id$-type homotopy extension
  operation.
  By this, we mean to take the $\Id$-type to be that given by the heterogeneous
  $\Path$-type structure, as described in \Cref{thm:nctt-Path-Id}.
\end{remark}

\begin{remark}\label{rmk:nctt-wkcomp}
  In a similar vein, in \Cref{itm:nctt-fill-wkcomp} the propositional
  computation structures for filling as formulated in
  \Cref{def:Sigma-fill-comp,def:Pi-fill-comp,def:Path-fill-comp} are all
  relative to some particular $\Id$-type structure on the ambient universe
  $\tMcU \to \McU$.
  Implicitly, we take the $\Id$-type structure to be that given by the
  heterogeneous $\Path$-type structure, as described in \Cref{thm:nctt-Path-Id}.
\end{remark}

One is also now ready to observe that the internal universe of NCTT is
pointed univalent.

\begin{theorem}\label{thm:nctt-univalence}
  Let $\bC$ be equipped with an NCTT structure whose ambient universe is
  $\pi\colon\tMcU \to \McU$ and internal universe is $\pi_0\colon\tMcU_0 \to \McU_0$.

  Then, the canonical $\Id$-type structure $\Id_0$ induced by the heterogeneous
  $\Path$-type structure on the internal universe $\pi_0$ from
  \Cref{thm:nctt-Path-Id} gives rise to a pointed $\pi$-univalence structure on
  $(\pi_0,\Id_0)$, in the sense of \Cref{def:axm-univalence}.
\end{theorem}
\begin{proof}
  Immediate by \Cref{thm:ptd-cubical-univalence,rmk:nctt-hiso-ext-op,def:naive-ctt}.
\end{proof}

\subsection{Related Axiomations of Cubical Type Theory}\label{subsec:ctt-comparison}
\subsubsection{Reformulation as a CwR}\label{subsubsec:cwr}
Although we have phrased the axioms in \Cref{def:naive-ctt} in terms of
universal maps in a locally cartesian closed category $\bC$, we expect it is
possible, though likely quite effortful, to rephrase this axiomatisation in the
framework of category with representable maps (CwR) of \citeauthor{uem23}
\cite{uem23} by taking the universal maps as the (generating) representable
maps.

The CwR describing various axioms of \Cref{def:naive-ctt} can be constructed
individually by way of free generation and then assembled together by a
colimiting process using the results provided by \citeauthor{jel24}
\cite{jel24}, as observed in \cite{axm-univalence}, for example.
Then, a category equipped with a structure of naive cubical type theory becomes
a equivalent to a CwR map from this CwR into a locally cartesian closed category
(i.e. a model of the CwR for NCTT in functorial semantics).
After doing so, we may then take as definition of naive cubical type theory
itself (as opposed to a \emph{structure} of such on a category as formally
defined in \Cref{def:naive-ctt}) to be this universal CwR.
We do not make this statement formal in this paper and leave it as future
work.

\subsubsection{Relation to CCHM Cubical Type Theory}
The naive cubical type theory differs from CCHM cubical type theory is several aspects, as outlined below:
\begin{enumerate}
  \item In \Cref{itm:nctt-fill} of \Cref{def:naive-ctt}, we have taken the
  filling structure to be a primitive, whereas in \cite{cchm15,ang+21} the
  filling structure is derived from the composition structure.
  The reason for this is two-fold:
  \begin{enumerate}
    \item As observed in \Cref{rmk:biased-filling}, the filling operation
    axiomatises the fact that the pushout-product of a cofibration with an
    interval endpoint inclusion is a trivial cofibration.
    Categorically, this is much more natural than the composition operation.
    \item The filling operation seems to have more practicality, as observed
    from the derivation of $\Id$-types from $\Path$-types in \Cref{thm:Path-Id},
    for example.
  \end{enumerate}
  \item The computational rules for filling
  \Cref{itm:nctt-fill-wkcomp,itm:nctt-fill-defcomp} are optional extras to the
  structure of naive cubical type theory.
  In contrast, the definitional versions \Cref{itm:nctt-fill-defcomp} are
  required by the cubical type theory formulations of \cite{cchm15,ang+21}.
  In fact, it was shown in \cite{chs22} that \cite{cchm15} cubical type theory
  with definitional composition replaced with non-computational filling still
  enjoys homotopical canonicity.
  \item The axiom \Cref{itm:nctt-cof-sec} that sections to fibrations are
  cofibrantions is an axiom of naive cubical type theory not present in the
  formulation of \cite{cchm15,ang+21}, which amounts to saying that terms are
  cofibrant.
  With this, we are able to derive the $\Id$-types from the $\Path$-types as in
  \Cref{thm:Path-Id}.
  Although we expect that this makes type checking undecidable, it appears that
  such undecidability is unavoidable if one wishes to take $\Path$-types as
  $\Id$-types, as analysed by \cite{swa18}.
  \item In \Cref{itm:nctt-hext}, motivated by Voevodsky's original proof of simplicial
  univalence \cite{kl21} we have required the strong homotopy equivalence
  extension structure (\Cref{def:strong-hiso-ext-op}).
  In contrast, \cite{cchm15} requires a $\Glue$-type structure
  (\Cref{def:glue-type}) for the $\Path$-type structure.
  Although the $\Glue$-type structure gives book univalence as recalled in
  \Cref{thm:cchm-cubical-univalence}, we were unable to obtain pointed
  univalence (\Cref{def:axm-univalence}) from it.
  We have opted for this modification for the following reasons:
  \begin{enumerate}
    \item As argued in \cite{axm-univalence}, pointed univalence is a
    strengthening of book univalence that provides an induction principle for
    propositional equivalences similar to $\MsJ$-elimination for $\Id$-types.
    \item Such a strong homotopy equivalence extension structure is supported by
    a large class of Quillen model categories in such a way that they induce an
    underling $\Glue$-type structure, as will be observed in
    \Cref{thm:heqv-ext}.
  \end{enumerate}
\end{enumerate}

\section{Equivalence Extension Properties} \label{sec:equiv-extension}
For purposes of supporting univalence in cubical type theory, we need to
show homotopy equivalence extension properties in certain model categories.
Throughout we work in a locally cartesian closed category $\bC$ equipped with a
model structure which we put various assumptions on at various points.

\subsection{Weak Equivalence Extension Property}\label{subsec:wk-eqv-ext}
In this part, we show that under some conditions, weak equivalences between
fibrations extend along cofibrations using the Glue-unglue operation in \Cref{constr:glue-op}.
Results established in this part are already well-known from \cite[Proposition
5.1]{sat17} and \cite[Theorem 3.4.1]{kl21}, amongst many others.
For the sake of completeness and ease of reference, we reproduce them here in
general cartesian closed model categories.

We start by noting that $\Glue\relax$ has the following composition property.
\begin{lemma}\label{lem:ext-comp}
  Let there be a mono $i \colon \partial B \hookrightarrow B$ along with an
  object $\underline{E} \to B \in \sfrac{\bC}{B}$ whose pullback under
  $i$ is $\partial\underline{E} \to \partial B$ as in \Cref{constr:glue-op}.

  Fix composable maps
  $\partial\overline{E} \xrightarrow{g} \partial E \xrightarrow{f}
  \partial\ul{E}$.
  Then, $\Glue_{i,\ul{E}}(fg) \to \ul{E}$ factors as
  $\Glue_{i,\Glue_{i,\ul{E}}(f)}(g) \to \Glue_{i,\ul{E}}(f) \to \ul{E}$ as
  below.
  % $\Ext_{}$
  % $\Ext_{\Ext__{\underline{E}}(\partial E)}(\partial\overline{E}) \to \Ext_{(\underline{E}, \partial\underline{E})} \to \underline{E}$, so in particular
  % $\Ext_{(\Ext_{(\underline{E}, \partial\underline{E})}(\partial E), \partial
  %   E)}(\partial\overline{E}) \cong \Ext_{(\underline{E},
  %   \partial\underline{E})}(\partial\overline{E})$.
  %
  \begin{equation*}
    \begin{tikzcd}[cramped, column sep=small]
      {\partial\overline{E}} &&&& {\Glue_{i,\ul{E}}(fg)} & {i_*\partial\overline{E}} \\
      {\partial\overline{E}} &&&& {\Glue_{i,\Glue_{i,\ul{E}}(f)}(g)} & {i_*\partial\overline{E}} \\
      && {\partial E} &&& {\Glue_{i,\ul{E}}(f)} & {i_*\partial E} \\
      &&&& {\partial \underline E} && {\underline{E}} & {i_*(\partial\underline{E})} \\
      &&& {\partial B} && B
      \arrow[dashed, from=1-1, to=1-5]
      \arrow[equals, from=1-1, to=2-1]
      \arrow[from=1-5, to=1-6]
      \arrow[equals, from=1-5, to=2-5]
      \arrow[equals, from=1-6, to=2-6]
      \arrow[curve={height=-6pt}, from=1-6, to=4-8]
      \arrow[dashed, from=2-1, to=2-5]
      \arrow[from=2-1, to=3-3]
      \arrow[from=2-1, to=5-4]
      \arrow[from=2-5, to=2-6]
      \arrow[dashed, from=2-5, to=5-6]
      \arrow[from=2-6, to=3-7]
      \arrow[from=3-3, to=4-5]
      \arrow[from=3-3, to=5-4]
      \arrow[from=3-6, to=3-7]
      \arrow[dashed, from=3-6, to=5-6]
      \arrow[from=3-7, to=4-8]
      \arrow[from=3-7, to=5-6]
      \arrow[from=4-5, to=5-4]
      \arrow["\eta"{description}, from=4-7, to=4-8]
      \arrow[from=4-7, to=5-6]
      \arrow[from=4-8, to=5-6]
      \arrow["i"', hook, from=5-4, to=5-6]
      \arrow[crossing over, dashed, from=2-5, to=3-6]
      \arrow[crossing over, dashed, from=3-3, to=3-6]
      \arrow[crossing over, dashed, from=3-6, to=4-7]
      \arrow[dashed, from=1-5, to=4-7, curve={height=-6pt}]
      \arrow[crossing over, hook, from=4-5, to=4-7]
      \arrow[crossing over, dashed, from=1-1, to=4-5]
      \arrow["\lrcorner"{anchor=center, pos=0.15, scale=1.5, rotate=45}, draw=none, from=2-5, to=3-7]
      \arrow["\lrcorner"{anchor=center, pos=0.15, scale=1.5, rotate=45}, draw=none, from=3-6, to=4-8]
      \arrow["\lrcorner"{anchor=center, pos=0.15, scale=1.5, rotate=0}, draw=none, from=4-5, to=5-6]
    \end{tikzcd}
  \end{equation*}
\end{lemma}
\begin{proof}
  It suffices to observe that the top row of the pullback on the left
  \begin{center}
    \begin{minipage}{0.45\linewidth}
      \begin{equation*}
        % https://q.uiver.app/#q=WzAsNCxbMCwxLCJcXHVuZGVybGluZXtFfSJdLFsxLDEsImlfKihcXHBhcnRpYWxcXHVuZGVybGluZXtFfSkiXSxbMSwwLCJpXypcXHBhcnRpYWwgRSJdLFswLDAsIlxcRXh0X3soXFx1bmRlcmxpbmV7RX0sIFxccGFydGlhbFxcdW5kZXJsaW5le0V9KX0oXFxwYXJ0aWFsIEUpIl0sWzAsMSwiXFxldGEiLDJdLFsyLDFdLFszLDAsIiIsMSx7InN0eWxlIjp7ImJvZHkiOnsibmFtZSI6ImRhc2hlZCJ9fX1dLFszLDIsIlxcZXRhIl0sWzMsMSwiIiwxLHsic3R5bGUiOnsibmFtZSI6ImNvcm5lciJ9fV1d
        \begin{tikzcd}[cramped]
          {\Glue_{i,\ul{E}}(f)} & {i_*\partial E} \\
          {\underline{E}} & {i_*(\partial\underline{E})}
          \arrow["", from=1-1, to=1-2]
          \arrow[dashed, from=1-1, to=2-1]
          \arrow["\lrcorner"{anchor=center, pos=0.05, scale=1.5}, draw=none, from=1-1, to=2-2]
          \arrow[from=1-2, to=2-2]
          \arrow["\eta"', from=2-1, to=2-2]
        \end{tikzcd}
      \end{equation*}
    \end{minipage}
    \begin{minipage}{0.45\linewidth}
      \begin{equation*}
        % https://q.uiver.app/#q=WzAsNCxbMCwxLCJpXipcXHVuZGVybGluZXtFfSJdLFsxLDEsIlxccGFydGlhbFxcdW5kZXJsaW5le0V9Il0sWzEsMCwiXFxwYXJ0aWFsIEUiXSxbMCwwLCJpXipcXEV4dF97KFxcdW5kZXJsaW5le0V9LCBcXHBhcnRpYWxcXHVuZGVybGluZXtFfSl9KFxccGFydGlhbCBFKSJdLFswLDEsIj0iLDJdLFsyLDFdLFszLDAsIiIsMSx7InN0eWxlIjp7ImJvZHkiOnsibmFtZSI6ImRhc2hlZCJ9fX1dLFszLDIsIj0iXV0=
        \begin{tikzcd}[cramped]
          {i^*\Glue_{i,\underline{E}}(f)} & {\partial E} \\
          {i^*\underline{E}} & {\partial\underline{E}}
          \arrow["{=}", from=1-1, to=1-2]
          \arrow[dashed, from=1-1, to=2-1]
          \arrow[from=1-2, to=2-2]
          \arrow["{=}"', from=2-1, to=2-2]
        \end{tikzcd}
      \end{equation*}
    \end{minipage}
  \end{center}
  is isomorphic to the unit of the adjunction $i^* \dashv i_*$ at
  $\Glue_{i,\partial\underline{E}}(f)$.
  This is indeed the case because by \Cref{lem:ext-retract}, transposing this
  square gives the square on the right, where the top row is the identity.
\end{proof}

With the goal of having $\Glue$-type structures (\Cref{def:glue-type}), we next
observe some conditions for $\Glue\relax$ take values in fibrations.
\begin{lemma}\label{lem:ext-preserve-fib}
  Let there be a mono $i \colon \partial B \hookrightarrow B$ along with an
  object $\underline{E} \to B \in \sfrac{\bC}{B}$ whose pullback under $i$ is
  $\partial\underline{E} \to \partial B$ as in \Cref{constr:glue-op}.
  Assume that $\bC$ is equipped with a model structure under which
  $\underline{E} \twoheadrightarrow B$ is a fibration.
  Then,
  \begin{enumerate}
    \item \label{itm:ext-preserve-fib-tf} If
    $\partial E \twoheadrightarrow \partial\underline{E}$ is a trivial fibration
    and the pullback functor
    $i^* \colon \sfrac{\bC}{B} \to \sfrac{\bC}{\partial B}$ preserves
    cofibrations then $\Glue_{i,\underline{E}}(\partial{E}) \to \underline{E}$
    is a trivial fibration.
    \item \label{itm:ext-preserve-fib-retract} If
    $\partial E \to \partial\underline{E}$ admits a retract over $\partial B$
    then the composite
    $\Glue_{i,\underline{E}}(\partial E) \to \underline{E} \twoheadrightarrow B$
    is a fibration.
    \item \label{itm:ext-wk-eqv-fib} In particular, if the monos are precisely
    the cofibrations and $\partial E \xrightarrow{\sim} \partial\underline{E}$
    is a weak equivalence over $\partial B$ then
    $\Glue_{i,\underline{E}}(\partial E) \to B$ is a fibration.
  \end{enumerate}
\end{lemma}
\begin{proof}
  In \Cref{itm:ext-preserve-fib-tf} where
  $\partial E \twoheadrightarrow \partial\underline{E}$ is a trivial fibration,
  because the pullback $i^*$ preserves cofibrations, its right adjoint $i_*$
  preserves trivial fibrations, so
  $i_*(\partial E) \twoheadrightarrow i_*(\partial\underline{E})$ is a trivial
  fibration, and hence its pullback
  $\Glue_{i,\underline{E}}(\partial E) \twoheadrightarrow \underline{E}$ is a
  trivial fibration as well.

  In \Cref{itm:ext-preserve-fib-retract} where
  $\partial E \to \partial\underline{E}$ admits a retract over $\partial B$ as
  on the left, functoriality of $\Glue$ ensures that
  $\Glue_{i,\underline{E}}(\partial E) \to
  \Glue_{i,\underline{E}}(\partial\underline{E}) = \underline{E}$ admits a
  retract over $B$ as on the right.
  \begin{center}
    \begin{minipage}{0.45\linewidth}
      \begin{equation*}
        % https://q.uiver.app/#q=WzAsNixbMCwwLCJcXHBhcnRpYWwgRSJdLFswLDEsIlxccGFydGlhbCBCIl0sWzEsMCwiXFxwYXJ0aWFsXFx1bmRlcmxpbmV7RX0iXSxbMSwxLCJcXHBhcnRpYWwgQiJdLFsyLDEsIlxccGFydGlhbCBCIl0sWzIsMCwiXFxwYXJ0aWFsIEUiXSxbMSwzLCI9IiwyLHsic3R5bGUiOnsidGFpbCI6eyJuYW1lIjoiYXJyb3doZWFkIn19fV0sWzMsNCwiPSIsMix7InN0eWxlIjp7InRhaWwiOnsibmFtZSI6ImFycm93aGVhZCJ9fX1dLFswLDEsIiIsMix7InN0eWxlIjp7ImhlYWQiOnsibmFtZSI6ImVwaSJ9fX1dLFswLDJdLFsyLDMsIiIsMCx7InN0eWxlIjp7ImhlYWQiOnsibmFtZSI6ImVwaSJ9fX1dLFsyLDUsIiIsMCx7InN0eWxlIjp7ImJvZHkiOnsibmFtZSI6ImRhc2hlZCJ9fX1dLFs1LDQsIiIsMCx7InN0eWxlIjp7ImhlYWQiOnsibmFtZSI6ImVwaSJ9fX1dLFswLDUsIj0iLDAseyJjdXJ2ZSI6LTJ9XV0=
        \begin{tikzcd}[cramped]
          {\partial E} & {\partial\underline{E}} & {\partial E} \\
          {\partial B} & {\partial B} & {\partial B}
          \arrow[from=1-1, to=1-2]
          \arrow["{=}", curve={height=-12pt}, from=1-1, to=1-3]
          \arrow[two heads, from=1-1, to=2-1]
          \arrow[dashed, from=1-2, to=1-3]
          \arrow[two heads, from=1-2, to=2-2]
          \arrow[two heads, from=1-3, to=2-3]
          \arrow["{=}"', tail reversed, from=2-1, to=2-2]
          \arrow["{=}"', tail reversed, from=2-2, to=2-3]
        \end{tikzcd}
      \end{equation*}
    \end{minipage}
    \begin{minipage}{0.45\linewidth}
      % https://q.uiver.app/#q=WzAsNixbMCwwLCJcXEV4dF9pKFxccGFydGlhbCBFKSJdLFswLDEsIkIiXSxbMSwwLCJcXHVuZGVybGluZXtFfSJdLFsxLDEsIkIiXSxbMiwxLCJCIl0sWzIsMCwiXFxFeHRfaShcXHBhcnRpYWwgRSkiXSxbMSwzLCI9IiwyLHsic3R5bGUiOnsidGFpbCI6eyJuYW1lIjoiYXJyb3doZWFkIn19fV0sWzMsNCwiPSIsMix7InN0eWxlIjp7InRhaWwiOnsibmFtZSI6ImFycm93aGVhZCJ9fX1dLFswLDEsIiIsMix7InN0eWxlIjp7ImhlYWQiOnsibmFtZSI6ImVwaSJ9fX1dLFswLDJdLFsyLDMsIiIsMCx7InN0eWxlIjp7ImhlYWQiOnsibmFtZSI6ImVwaSJ9fX1dLFsyLDUsIiIsMCx7InN0eWxlIjp7ImJvZHkiOnsibmFtZSI6ImRhc2hlZCJ9fX1dLFswLDUsIj0iLDAseyJjdXJ2ZSI6LTJ9XSxbNSw0LCIiLDAseyJzdHlsZSI6eyJoZWFkIjp7Im5hbWUiOiJlcGkifX19XV0=
      \begin{tikzcd}[cramped]
        {\Glue_{i,\underline{E}}(\partial E)} & {\underline{E}} & {\Glue_{i,\underline{E}}(\partial E)} \\
        B & B & B
        \arrow[from=1-1, to=1-2]
        \arrow["{=}", curve={height=-12pt}, from=1-1, to=1-3]
        \arrow[two heads, from=1-1, to=2-1]
        \arrow[dashed, from=1-2, to=1-3]
        \arrow[two heads, from=1-2, to=2-2]
        \arrow[two heads, from=1-3, to=2-3]
        \arrow["{=}"', tail reversed, from=2-1, to=2-2]
        \arrow["{=}"', tail reversed, from=2-2, to=2-3]
      \end{tikzcd}
    \end{minipage}
  \end{center}
  Because $\underline{E} \twoheadrightarrow B$ is assumed to be a fibration,
  $\Glue_{i,\underline{E}}(\partial E) \twoheadrightarrow B$ is a fibration as
  well.

  To conclude, for \Cref{itm:ext-wk-eqv-fib}, we assume that the monos are
  precisely cofibrations and that
  $\partial E \xrightarrow{\sim} \partial\underline{E}$ is a weak equivalence.
  Factoring $\partial E \xrightarrow{\sim} \partial\underline{E}$ into a trivial
  cofibration followed by a trivial fibration
  \begin{equation*}
    \begin{tikzcd}[cramped]
      \partial E \ar[r, hook, "\sim"] & \partial E' \ar[r, two heads, "\sim"] & \partial\underline{E}
    \end{tikzcd}
  \end{equation*}
  each of which is over $\partial B$, we note that
  $\Glue_{i,\underline{E}}(\partial E') \to B$ is a fibration by
  \Cref{itm:ext-preserve-fib-tf} as monos are stable under pullback.
  And since
  $\begin{tikzcd}[cramped]\partial E \ar[r, hook, "\sim"] & \partial
    E'\end{tikzcd}$ lifts against
  $\begin{tikzcd}[cramped]
    \partial E \ar[r, two heads] & \partial B
  \end{tikzcd}$, it admits a retract over $\partial B$, so by
  \Cref{itm:ext-preserve-fib-retract}, it follows that
  $\Glue_{i,\Glue_{i,\underline{E}}(\partial E')}(\partial E) \to B$ is a fibration.
  By \Cref{lem:ext-comp},
  $\Glue_{i,\Glue_{i,\underline{E}}(\partial E')}(\partial E) \cong
  \Glue_{i,\ul{E}}(\partial E)$, so the result follows.
\end{proof}

Next, we observe a weak form of homotopy section extension property of
$\Glue\relax$.
\begin{lemma}\label{lem:ext-htpy-sect-wk}
  Let $i \colon \partial B \hookrightarrow B$ be a mono and
  $\underline{E} \to B \in \sfrac{\bC}{B}$ an object in the slice category whose
  pullback under $i$ is $\partial\underline{E} \to \partial B$ as in
  \Cref{constr:glue-op}.
  Assume $\bC$ is equipped with a model structure in which
  \begin{enumerate}
    \item The cofibrations are the monos.
    \item The pushout-product of a (cofibration, trivial
    cofibration)-pair is a trivial cofibration.
    \item $\underline{E} \twoheadrightarrow B$ is a fibration.
    \item $I$ is a good, but not necessarily very good, cylinder object for the
    terminal object, with the two endpoint inclusions
    $\set{0},\set{1} \rightrightarrows I$.
  \end{enumerate}

  Fix a map $p \colon \partial E \to \partial\underline{E}$ over $\partial B$
  and suppose it has a homotopy section as on the left below
  $\partial s \colon \partial\underline{E} \to \partial E$ via a homotopy
  $\partial H \colon \partial\underline{E} \times I \to E$ over $\partial B$ as
  on the right below.
  \begin{center}
    \begin{minipage}{0.45\linewidth}
      \begin{equation*}
        % https://q.uiver.app/#q=WzAsMyxbMSwyLCJcXHBhcnRpYWwgQiJdLFsyLDEsIlxccGFydGlhbFxcdW5kZXJsaW5le0V9Il0sWzAsMCwiXFxwYXJ0aWFsIEUiXSxbMSwwLCIiLDAseyJzdHlsZSI6eyJoZWFkIjp7Im5hbWUiOiJlcGkifX19XSxbMiwxLCJwIiwyLHsib2Zmc2V0IjoyfV0sWzIsMF0sWzEsMiwiXFxwYXJ0aWFsIHMiLDIseyJvZmZzZXQiOjJ9XV0=
        \begin{tikzcd}[cramped, row sep=small, column sep=small]
          {\partial E} \\
          && {\partial\underline{E}} \\
          & {\partial B}
          \arrow["p"', shift right=1, from=1-1, to=2-3]
          \arrow[from=1-1, to=3-2]
          \arrow["{\partial s}"', shift right=1, from=2-3, to=1-1]
          \arrow[two heads, from=2-3, to=3-2]
        \end{tikzcd}
      \end{equation*}
    \end{minipage}
    \begin{minipage}{0.45\linewidth}
      % https://q.uiver.app/#q=WzAsNixbMywxLCJcXHBhcnRpYWxcXHVuZGVybGluZXtFfSJdLFszLDMsIlxccGFydGlhbCBCIl0sWzEsMywiXFxwYXJ0aWFsXFx1bmRlcmxpbmV7RX0iXSxbMSwxLCJcXHBhcnRpYWxcXHVuZGVybGluZXtFfSBcXHRpbWVzIEkiXSxbMCwyLCJcXHBhcnRpYWxcXHVuZGVybGluZXtFfSBcXHRpbWVzIFxcc2V0ezF9Il0sWzAsMCwiXFxwYXJ0aWFsXFx1bmRlcmxpbmV7RX0gXFx0aW1lcyBcXHNldHswfSJdLFszLDIsIlxccHJvaiIsMV0sWzMsMCwiXFxwYXJ0aWFsIEgiLDFdLFs0LDAsInAgXFxjZG90IFxccGFydGlhbCBzIiwxXSxbNSwwLCJcXGlkIiwxXSxbNSwzLCIiLDEseyJzdHlsZSI6eyJ0YWlsIjp7Im5hbWUiOiJob29rIiwic2lkZSI6InRvcCJ9fX1dLFs0LDMsIiIsMSx7InN0eWxlIjp7InRhaWwiOnsibmFtZSI6Imhvb2siLCJzaWRlIjoidG9wIn19fV0sWzIsMV0sWzAsMV1d
      \begin{tikzcd}[cramped, row sep=small, column sep=small]
        {\partial\underline{E} \times \set{0}} \\
        & {\partial\underline{E} \times I} && {\partial\underline{E}} \\
        {\partial\underline{E} \times \set{1}} \\
        & {\partial\underline{E}} && {\partial B}
        \arrow[hook, from=1-1, to=2-2]
        \arrow["\id"{description}, from=1-1, to=2-4]
        \arrow["\proj"{description}, from=2-2, to=4-2]
        \arrow[from=2-4, to=4-4]
        \arrow[hook, from=3-1, to=2-2]
        \arrow["{p \cdot \partial s}"{description}, from=3-1, to=2-4, crossing over]
        \arrow[from=4-2, to=4-4]
        \arrow["{\partial H}"{description}, from=2-2, to=2-4]
      \end{tikzcd}
    \end{minipage}
  \end{center}
  Then,
  $\unglue_{i,\underline{E}}(p) \colon \Glue_{i,\underline{E}}(\partial E) \to
  \underline{E}$ also admits a homotopy section
  $s \colon \underline{E} \to \Glue_{i,\underline{E}}(\partial E)$ as on the
  left below, not necessarily extending $\partial s$, with a homotopy of the
  composite $H \colon \id_{\underline{E}} \simeq ps$ that does extend
  $\partial H$, as on the right below.
  \begin{center}
    \begin{minipage}{0.45\linewidth}
      \begin{equation*}
        % https://q.uiver.app/#q=WzAsNixbMywxLCJcXHVuZGVybGluZXtFfSJdLFszLDMsIkIiXSxbMSwzLCJcXHVuZGVybGluZXtFfSJdLFsxLDEsIlxcdW5kZXJsaW5le0V9IFxcdGltZXMgSSJdLFswLDIsIlxcdW5kZXJsaW5le0V9IFxcdGltZXMgXFxzZXR7MX0iXSxbMCwwLCJcXHVuZGVybGluZXtFfSBcXHRpbWVzIFxcc2V0ezB9Il0sWzMsMiwiXFxwcm9qIiwxXSxbMywwLCJIIiwxLHsic3R5bGUiOnsiYm9keSI6eyJuYW1lIjoiZGFzaGVkIn19fV0sWzQsMCwicHMiLDEseyJzdHlsZSI6eyJib2R5Ijp7Im5hbWUiOiJkYXNoZWQifX19XSxbNSwwLCJcXGlkIiwxXSxbNSwzLCIiLDEseyJzdHlsZSI6eyJ0YWlsIjp7Im5hbWUiOiJob29rIiwic2lkZSI6InRvcCJ9fX1dLFs0LDMsIiIsMSx7InN0eWxlIjp7InRhaWwiOnsibmFtZSI6Imhvb2siLCJzaWRlIjoidG9wIn19fV0sWzAsMV0sWzIsMV1d
        \begin{tikzcd}[cramped, row sep=small, column sep=small]
          {\underline{E} \times \set{0}} \\
          & {\underline{E} \times I} && {\underline{E}} \\
          {\underline{E} \times \set{1}} \\
          & {\underline{E}} && B
          \arrow[hook, from=1-1, to=2-2]
          \arrow["\id"{description}, from=1-1, to=2-4]
          \arrow["\proj"{description}, from=2-2, to=4-2]
          \arrow[from=2-4, to=4-4]
          \arrow[hook, from=3-1, to=2-2]
          \arrow["ps"{description}, dashed, from=3-1, to=2-4, crossing over]
          \arrow["H"{description}, dashed, from=2-2, to=2-4]
          \arrow[from=4-2, to=4-4]
        \end{tikzcd}
      \end{equation*}
    \end{minipage}
    \begin{minipage}{0.45\linewidth}
      \begin{equation*}
        % https://q.uiver.app/#q=WzAsNCxbMCwwLCJcXHBhcnRpYWxcXHVuZGVybGluZXtFfSBcXHRpbWVzIEkiXSxbMCwxLCJcXHVuZGVybGluZXtFfSBcXHRpbWVzIEkiXSxbMSwxLCJcXHVuZGVybGluZXtFfSJdLFsxLDAsIlxccGFydGlhbFxcdW5kZXJsaW5le0V9Il0sWzAsMSwiIiwwLHsic3R5bGUiOnsidGFpbCI6eyJuYW1lIjoiaG9vayIsInNpZGUiOiJ0b3AifX19XSxbMSwyLCJIIiwyXSxbMywyLCIiLDAseyJzdHlsZSI6eyJ0YWlsIjp7Im5hbWUiOiJob29rIiwic2lkZSI6InRvcCJ9fX1dLFswLDMsIlxccGFydGlhbCBIIl1d
        \begin{tikzcd}[cramped, row sep=small, column sep=small]
          {\partial\underline{E} \times I} & {\partial\underline{E}} \\
          {\underline{E} \times I} & {\underline{E}}
          \arrow["{\partial H}", from=1-1, to=1-2]
          \arrow[hook, from=1-1, to=2-1]
          \arrow[hook, from=1-2, to=2-2]
          \arrow["H"', from=2-1, to=2-2]
        \end{tikzcd}
      \end{equation*}
    \end{minipage}
  \end{center}
\end{lemma}
\begin{proof}
  We proceed by constructing the homotopy extension
  $H \colon \underline{E} \times I \to \underline{E}$ and then showing the
  existence of the homotopy section $s$ by proving
  $H|_{\set{1}} \colon \underline{E} \to \underline{E}$ factors via
  $\unglue_{i,\underline{E}}(p) \colon \Glue_{i,\underline{E}}(\partial E) \to
  \underline{E}$.
  The required homotopy extension $H$ is given by the solution to the
  following lifting problem.
  \begin{equation*}
    % https://q.uiver.app/#q=WzAsNSxbMCwwLCJcXHBhcnRpYWxcXHVuZGVybGluZXtFfSBcXHRpbWVzIEkgXFxjdXAgXFx1bmRlcmxpbmV7RX0gXFx0aW1lcyBcXHNldHswfSJdLFswLDEsIlxcdW5kZXJsaW5le0V9IFxcdGltZXMgSSJdLFsxLDEsIlxcdW5kZXJsaW5le0V9Il0sWzIsMSwiQiJdLFsyLDAsIlxcdW5kZXJsaW5le0V9Il0sWzAsMSwiXFxzaW0iLDIseyJzdHlsZSI6eyJ0YWlsIjp7Im5hbWUiOiJob29rIiwic2lkZSI6InRvcCJ9fX1dLFsxLDIsIlxccHJvaiIsMl0sWzIsMywiIiwwLHsic3R5bGUiOnsiaGVhZCI6eyJuYW1lIjoiZXBpIn19fV0sWzQsMywiIiwwLHsic3R5bGUiOnsiaGVhZCI6eyJuYW1lIjoiZXBpIn19fV0sWzAsNCwiXFxiZWdpbntwbWF0cml4fSBcXHBhcnRpYWxcXHVuZGVybGluZXtFfSBcXHRpbWVzIEkgXFx4cmlnaHRhcnJvd3tcXHBhcnRpYWwgSH0gXFxwYXJ0aWFsXFx1bmRlcmxpbmV7RX0gXFxob29rcmlnaHRhcnJvdyBcXHVuZGVybGluZXtFfSBcXFxcICBcXHVuZGVybGluZXtFfSBcXHRpbWVzIFxcc2V0ezB9IFxceHJpZ2h0YXJyb3d7XFxpZFxccmVsYXh9IFxcdW5kZXJsaW5le0V9IFxcZW5ke3BtYXRyaXh9Il0sWzEsNCwiSCIsMSx7InN0eWxlIjp7ImJvZHkiOnsibmFtZSI6ImRhc2hlZCJ9fX1dXQ==
    \begin{tikzcd}[cramped]
      {\partial\underline{E} \times I \cup \underline{E} \times \set{0}} && {\underline{E}} \\
      {\underline{E} \times I} & {\underline{E}} & B
      \arrow["\begin{array}{c} \begin{pmatrix} \partial\underline{E} \times I \xrightarrow{\partial H} \partial\underline{E} \hookrightarrow \underline{E} \\  \underline{E} \times \set{0} \xrightarrow{\id\relax} \underline{E} \end{pmatrix} \end{array}", from=1-1, to=1-3]
      \arrow["\sim"', hook, from=1-1, to=2-1]
      \arrow[two heads, from=1-3, to=2-3]
      \arrow["H"{description}, dashed, from=2-1, to=1-3]
      \arrow["\proj"', from=2-1, to=2-2]
      \arrow[two heads, from=2-2, to=2-3]
    \end{tikzcd}
  \end{equation*}
  We must check that the map in the top row actually exists and this defines a
  valid lifting problem.
  Once we verify this, the solution $H$ exists because $E \twoheadrightarrow B$
  is assumed to be a fibration and the pushout-product of the (cofibration,
  trivial cofibration)-pair
  $(\partial\underline{E} \hookrightarrow \underline{E}, \set{0} \hookrightarrow
  I)$ remains a trivial cofibration by assumption.

  First, we note that the map in the top row actually exists because
  $\partial\underline{E} \times I \xrightarrow{\partial H} \partial\underline{E}
  \hookrightarrow \underline{E}$ when restricted to
  $\partial\underline{E} \times \set{0}$ is the map
  $\partial\underline{E} \xrightarrow{=} \partial\underline{E} \hookrightarrow
  \underline{E}$ since
  $\partial H \colon \id_{\partial \underline E} \simeq p \cdot \partial s$.

  To check that this defines a valid lifting problem, it suffices to check
  commutativity separately in the restrictions to
  $\partial\underline{E} \times I$ and $\underline{E} \times \set{0}$.
  Commutativity in the $\underline{E} \times \set{0}$ component follows since
  the inclusion
  $\underline{E} \times \set{0} \hookrightarrow \underline{E} \times I$ is a
  section of the projection $\underline{E} \times I \to
  \underline{E}$. % and $\Ext\relax$ is an extension by
  % \Cref{lem:ext-retract}.

  We now observe commutativity in the $\partial\underline{E} \times I$ component
  by the following diagram, which commutes because the homotopy $\partial H$ is
  fibred over $\partial B$.
  \begin{equation*}
    % https://q.uiver.app/#q=WzAsOCxbMCwwLCJcXHBhcnRpYWxcXHVuZGVybGluZXtFfSBcXHRpbWVzIEkiXSxbMCwyLCJcXHBhcnRpYWxcXHVuZGVybGluZXtFfSBcXHRpbWVzIEkiXSxbMSwyLCJcXHBhcnRpYWxcXHVuZGVybGluZXtFfSJdLFszLDIsIkIiXSxbMywwLCJcXHVuZGVybGluZXtFfSJdLFsxLDAsIlxccGFydGlhbFxcdW5kZXJsaW5le0V9Il0sWzEsMSwiXFxwYXJ0aWFsXFx1bmRlcmxpbmV7RX0iXSxbMiwxLCJcXHBhcnRpYWwgQiJdLFswLDEsIiIsMix7InN0eWxlIjp7InRhaWwiOnsibmFtZSI6Imhvb2siLCJzaWRlIjoidG9wIn19fV0sWzEsMiwiXFxwcm9qIiwyXSxbMiwzXSxbNCwzLCIiLDAseyJzdHlsZSI6eyJoZWFkIjp7Im5hbWUiOiJlcGkifX19XSxbNSw0LCIiLDAseyJzdHlsZSI6eyJ0YWlsIjp7Im5hbWUiOiJob29rIiwic2lkZSI6InRvcCJ9fX1dLFswLDUsIlxccGFydGlhbCBIIl0sWzUsN10sWzYsN10sWzAsNiwiXFxwcm9qIiwxXSxbNiwyLCIiLDEseyJzdHlsZSI6eyJ0YWlsIjp7Im5hbWUiOiJob29rIiwic2lkZSI6InRvcCJ9fX1dLFs3LDMsIiIsMSx7InN0eWxlIjp7InRhaWwiOnsibmFtZSI6Imhvb2siLCJzaWRlIjoidG9wIn19fV1d
    \begin{tikzcd}[cramped]
      {\partial\underline{E} \times I} & {\partial\underline{E}} && {\underline{E}} \\
      & {\partial\underline{E}} & {\partial B} \\
      {\underline{E} \times I} & {\underline{E}} && B
      \arrow["{\partial H}", from=1-1, to=1-2]
      \arrow["\proj"{description}, from=1-1, to=2-2]
      \arrow[hook, from=1-1, to=3-1]
      \arrow[hook, from=1-2, to=1-4]
      \arrow[from=1-2, to=2-3]
      \arrow[two heads, from=1-4, to=3-4]
      \arrow[from=2-2, to=2-3]
      \arrow[hook, from=2-2, to=3-2]
      \arrow[hook, from=2-3, to=3-4]
      \arrow["\proj"', from=3-1, to=3-2]
      \arrow[from=3-2, to=3-4]
    \end{tikzcd}
  \end{equation*}

  By construction
  $H \colon \id_{\underline{E}} = H|_{\set{0}} \simeq H|_{\set{1}}$, and
  post-composing $H$ with $\underline{E} \twoheadrightarrow B$ factors through
  $\proj\relax \colon \underline{E} \times I \to \underline{E}$, so that
  $H$ is a homotopy $H \colon \id_{\underline{E}} \simeq H|_{\set{1}}$ fibred over $B$.
  Moreover, when restricted along
  $\partial E \times I \hookrightarrow E \times I$, the construction of $H$ also
  guarantees that it agrees with
  $\partial\underline{E} \times I \xrightarrow{\partial H} \partial\underline{E}
  \hookrightarrow \underline{E}$ so that $H$ actually extends $\partial H$.

  To conclude, we check that
  $H|_{\set{1}} = (\underline{E} \times \set{1} \hookrightarrow \underline{E}
  \times I \xrightarrow{H} \underline{E})$ factors through
  $\unglue_{i,\underline{E}}(p) \colon \Glue_{i,\underline{E}}(\partial E) \to
  \underline{E}$.
  To do so, we note that $H$ extends $\partial H$ so in particular
  $i^*(H|_{\set{1}}) = (\partial H)|_{\set{1}}$, which means we have the diagram
  as on the left over $\partial B$.
  Taking its transpose over $i^* \dashv i_*$ then gives the outer edge on the
  right, which must factor through the pullback
  $\eta^*i_*p \cong \unglue_{i,\underline{E}}(p)$.
  \begin{center}
    \begin{minipage}{0.45\linewidth}
      \begin{equation*}
        % https://q.uiver.app/#q=WzAsNCxbMSwwLCJcXHBhcnRpYWwgRSJdLFsxLDEsIlxccGFydGlhbFxcdW5kZXJsaW5le0V9Il0sWzAsMSwiaV4qXFx1bmRlcmxpbmV7RX0iXSxbMCwwLCJpXipcXHVuZGVybGluZXtFfSJdLFswLDEsInAiXSxbMiwxLCI9IiwyLHsic3R5bGUiOnsidGFpbCI6eyJuYW1lIjoiYXJyb3doZWFkIn19fV0sWzMsMCwicyIsMCx7InN0eWxlIjp7InRhaWwiOnsibmFtZSI6ImFycm93aGVhZCJ9fX1dLFszLDIsImleKihIfF97XFxzZXR7MX19KSA9IChcXHBhcnRpYWwgSCl8X3tcXHNldHsxfX0iLDJdXQ==
        \begin{tikzcd}[cramped]
          {i^*\underline{E} = \partial\underline{E}} & {\partial E} \\
          {i^*\underline{E}} & {\partial\underline{E}}
          \arrow["\partial s", from=1-1, to=1-2]
          \arrow["{i^*(H|_{\set{1}}) = (\partial H)|_{\set{1}}}"', from=1-1, to=2-1]
          \arrow["p", from=1-2, to=2-2]
          \arrow["{=}"', tail reversed, from=2-1, to=2-2]
        \end{tikzcd}
      \end{equation*}
    \end{minipage}
    \begin{minipage}{0.45\linewidth}
      \begin{equation*}
        % https://q.uiver.app/#q=WzAsNSxbMiwxLCJpXyooXFxwYXJ0aWFsIEUpIl0sWzIsMiwiaV8qKFxccGFydGlhbFxcdW5kZXJsaW5le0V9KSJdLFsxLDIsIlxcdW5kZXJsaW5le0V9Il0sWzAsMCwiXFx1bmRlcmxpbmV7RX0iXSxbMSwxLCJcXEV4dChcXHBhcnRpYWwgRSkiXSxbMCwxLCJpXypwIl0sWzIsMSwiXFxldGEiLDIseyJzdHlsZSI6eyJ0YWlsIjp7Im5hbWUiOiJhcnJvd2hlYWQifX19XSxbMywwLCJzXlxcZGFnZ2VyIiwwLHsiY3VydmUiOi0yLCJzdHlsZSI6eyJ0YWlsIjp7Im5hbWUiOiJhcnJvd2hlYWQifX19XSxbMywyLCJIfF97XFxzZXR7MX19IiwyLHsiY3VydmUiOjJ9XSxbNCwyLCJcXEV4dChwKSIsMV0sWzQsMF0sWzQsMSwiIiwxLHsic3R5bGUiOnsibmFtZSI6ImNvcm5lciJ9fV0sWzMsNCwicyIsMSx7InN0eWxlIjp7ImJvZHkiOnsibmFtZSI6ImRhc2hlZCJ9fX1dXQ==
        \begin{tikzcd}[cramped]
          {\underline{E}} \\
          & {\Glue_{i,\underline{E}}(\partial E)} & {i_*(\partial E)} \\
          & {\underline{E}} & {i_*(\partial\underline{E})}
          \arrow["{s}"{description}, dashed, from=1-1, to=2-2]
          \arrow["{(\partial s)^\dagger}", curve={height=-12pt}, from=1-1, to=2-3]
          \arrow["{H|_{\set{1}}}"', curve={height=12pt}, from=1-1, to=3-2]
          \arrow[from=2-2, to=2-3]
          \arrow["{\unglue_{i,\underline{E}}(p)}"{description}, from=2-2, to=3-2]
          \arrow["\lrcorner"{anchor=center, pos=0.15, scale=1.5}, draw=none, from=2-2, to=3-3]
          \arrow["{i_*p}", from=2-3, to=3-3]
          \arrow["\eta"', from=3-2, to=3-3]
        \end{tikzcd}
      \end{equation*}
    \end{minipage}
  \end{center}
  We take the this factorisation
  $s \colon E \to \Ext_{\underline{E}}(\partial E)$ to be the homotopy section
  to $\Ext_{\underline{E}}(p)$.
\end{proof}

As consequences, we obtain the following equivalence extension properties
extending weak equivalences along monomorphisms.
\begin{proposition}[{\cite[Theorem 3.4.1]{kl21}}]\label{prop:ext-we}
  Let there be a mono $i \colon \partial B \hookrightarrow B$ along with an
  object $\underline{E} \to B \in \sfrac{\bC}{B}$ whose pullback under
  $i$ is $\partial\underline{E} \to \partial B$ as in \Cref{constr:glue-op}.
  Assume $\bC$ is equipped with a model structure under which:
  \begin{enumerate}
    \item The monomorphisms and the cofibrations coincide.
    \item The pushout-product of a (cofibration, trivial cofibration)-pair is a
    trivial cofibration.
    %
    % \item Homotopy equivalences are also weak equivalences.
    %
    \item $\underline{E} \twoheadrightarrow B$ is a fibration.
  \end{enumerate}
  Then for any fibrant object
  $\partial E \twoheadrightarrow \partial B \in \sfrac{\bC}{\partial B}$ in the
  slice model structure, if
  \begin{equation*}
    w \colon \partial E \xrightarrow{\sim} \partial\underline{E}
  \end{equation*}
  is a weak equivalence, the extension
  \begin{equation*}
    \Glue_{i,\underline{E}}(\partial E)
    \xrightarrow{\unglue_{i,\underline{E}}(w)}
    \underline{E}
  \end{equation*}
  which we recall is constructed as first taking the pushforward along $i$ and
  then the pullback along $\eta$
  \begin{equation*}
    % https://q.uiver.app/#q=WzAsOCxbMSwyLCJcXHBhcnRpYWwgQiJdLFszLDIsIkIiXSxbMiwxLCJcXHBhcnRpYWwgXFx1bmRlcmxpbmUgRSJdLFswLDAsIlxccGFydGlhbCBFIl0sWzQsMSwiXFx1bmRlcmxpbmV7RX0iXSxbNSwxLCJpXyooXFxwYXJ0aWFsXFx1bmRlcmxpbmV7RX0pIl0sWzQsMCwiaV8qXFxwYXJ0aWFsIEUiXSxbMywwLCJcXEV4dChcXHBhcnRpYWwgRSkiXSxbMCwxLCJpIiwyLHsic3R5bGUiOnsidGFpbCI6eyJuYW1lIjoiaG9vayIsInNpZGUiOiJ0b3AifX19XSxbNCwxLCJwIiwxXSxbMiwwLCJcXHBhcnRpYWwgcCIsMV0sWzIsNCwiIiwyLHsiY29sb3VyIjpbMCw2MCw2MF0sInN0eWxlIjp7InRhaWwiOnsibmFtZSI6Imhvb2siLCJzaWRlIjoidG9wIn19fV0sWzIsMSwiIiwxLHsic3R5bGUiOnsibmFtZSI6ImNvcm5lciJ9fV0sWzMsMl0sWzMsMF0sWzQsNSwiXFxldGEiLDFdLFs1LDFdLFs2LDVdLFs2LDFdLFs3LDQsIiIsMSx7ImNvbG91ciI6WzAsNjAsNjBdLCJzdHlsZSI6eyJib2R5Ijp7Im5hbWUiOiJkYXNoZWQifX19XSxbNyw2XSxbNyw1LCIiLDEseyJzdHlsZSI6eyJuYW1lIjoiY29ybmVyIn19XSxbNywxLCIiLDEseyJzdHlsZSI6eyJib2R5Ijp7Im5hbWUiOiJkYXNoZWQifX19XSxbMyw3LCIiLDEseyJzdHlsZSI6eyJib2R5Ijp7Im5hbWUiOiJkYXNoZWQifX19XV0=
    \begin{tikzcd}[cramped]
      {\partial E} &&& {\Glue_{i,\ul{E}}(\partial E)} & {i_*(\partial E)} \\
      && {\partial \underline E} && {\underline{E}} & {i_*(\partial\underline{E})} \\
      & {\partial B} && B
      \arrow[dashed, from=1-1, to=1-4]
      \arrow[from=1-1, to=2-3, "w"', "{\sim}"]
      \arrow[from=1-1, to=3-2]
      \arrow[from=1-4, to=1-5]
      \arrow[dashed, from=1-4, to=3-4]
      \arrow[from=1-5, to=2-6]
      \arrow[from=1-5, to=3-4]
      \arrow[from=2-3, to=3-2]
      \arrow["\eta"{description}, from=2-5, to=2-6]
      \arrow[from=2-5, to=3-4]
      \arrow[from=2-6, to=3-4]
      \arrow["i"', hook, from=3-2, to=3-4]
      \arrow[crossing over, dashed, from=1-4, to=2-5]
      \arrow[crossing over, hook, from=2-3, to=2-5]
      \arrow["\lrcorner"{anchor=center, pos=0.15, scale=1.5, rotate=45}, draw=none, from=1-4, to=2-6]
      \arrow["\lrcorner"{anchor=center, pos=0.15, scale=1.5}, draw=none, from=2-3, to=3-4]
    \end{tikzcd}
  \end{equation*}
  is a weak equivalence between fibrant objects in the slice
  model structure $\sfrac{\bC}{B}$.
\end{proposition}
\begin{proof}
  By \Cref{lem:ext-preserve-fib},
  $\Glue_{i,\underline{E}}(\partial E) \twoheadrightarrow B$ is a fibration because
  $\underline{E} \twoheadrightarrow B$ is a fibration and
  $w \colon \partial E \to \partial\underline{E}$ is a weak equivalence.

  It remains to check $w$ is a weak equivalence.
  Using \Cref{lem:ext-comp}, factoring $w$ into a trivial cofibration followed
  by a trivial fibration, it suffices to show separately $w$ is a weak
  equivalence assuming $w$ is either a trivial cofibration or a trivial
  fibration.
  If $w$ is a trivial fibration then by
  \Cref{lem:ext-preserve-fib}\Cref{itm:ext-preserve-fib-tf}, it follows that
  $\unglue_{i,\underline{E}}(w)$ is also a trivial fibration and hence a weak equivalence.

  It remains to consider the case where $w$ is a trivial cofibration.
  Because the monomorphisms and the cofibrations coincide,
  $w \colon \partial E \to \partial\underline{E}$ is a
  deformation retract over $\partial B$ with retract
  $r \colon \partial \underline{E} \to \partial E$, which is also a homotopy
  section over $\partial B$.
  The goal is to show that
  $\unglue_{i,\underline{E}}(w) \colon \Ext_{\underline{E}}(\partial E) \to
  \underline{E}$ is a homotopy equivalence so that because the monos are the
  cofibrations, it becomes a homotopy equivalence between two cofibrant-fibrant
  objects in the slice $\sfrac{\bC}{B}$ and thus a weak equivalence.
  Functoriality of $\Ext\relax$ means that
  \begin{equation*}
    \Glue_{i,\underline{E}}(\partial\underline{E}) = \underline{E}
    \xrightarrow{\unglue_{i,\underline{E}}(r)}
    \Ext_{\underline{E}}(\partial E)
  \end{equation*}
  is also a precise retract of $\unglue_{i,\underline{E}}(w)$.
  It therefore remains to check that $\unglue_{i,\underline{E}}(w)$ admits a homotopy
  section over $B$.
  But $r \colon \partial \underline{E} \to \partial E$ is also a homotopy
  section of $w \colon \partial E \to \partial\underline{E}$ over $\partial B$,
  so the result follows by \Cref{lem:ext-htpy-sect-wk}.
\end{proof}

%%% Local Variables:
%%% TeX-master: "./main.tex"
%%% TeX-engine: default
%%% End:

\subsection{Homotopy Equivalence Extension Property}\label{subsec:htpy-eqv-ext}
We now show an analogous result to \Cref{prop:ext-we}, but for homotopy
equivalences.
Throughout this part, we work under the following assumption.
\begin{assumption}\label{asm:htpy-eqv-ext}
  Assume locally cartesian closed $\bC$ is equipped with a model structure in
  which:
  \begin{enumerate}
    \item The cofibrations are precisely the monomorphisms.
    \item The pushout-product of a (cofibration, trivial cofibration)-pair is
    still a trivial cofibration.
    \item $I$ is a good, but not necessarily very good, cylinder object for the
    terminal object with the two disjoint endpoint inclusions
    $\set{0}, \set{1} \rightrightarrows I$.
  \end{enumerate}
\end{assumption}

\subsubsection{Extending Homotopy Retractions} \label{subsubsec:htpy-retract-ext}
We first establish a homotopy retraction extension property, which is the
equivalent of \cite[Lemma 3.3.5]{kl21}.

The problem setup fixes the following solid maps
\begin{equation}\label{eqn:hrep}\tag{\textsc{hrep}}
  % https://q.uiver.app/#q=WzAsNixbMSwyLCJcXHBhcnRpYWwgQiJdLFszLDIsIkIiXSxbMiwxLCJcXHBhcnRpYWwgXFx1bmRlcmxpbmUgRSJdLFswLDAsIlxccGFydGlhbCBFIl0sWzQsMSwiXFx1bmRlcmxpbmV7RX0iXSxbMiwwLCJFIl0sWzAsMSwiaSIsMix7InN0eWxlIjp7InRhaWwiOnsibmFtZSI6Imhvb2siLCJzaWRlIjoidG9wIn19fV0sWzIsMCwiIiwxLHsic3R5bGUiOnsiaGVhZCI6eyJuYW1lIjoiZXBpIn19fV0sWzIsNCwiIiwyLHsiY29sb3VyIjpbMCw2MCw2MF0sInN0eWxlIjp7InRhaWwiOnsibmFtZSI6Imhvb2siLCJzaWRlIjoidG9wIn19fV0sWzIsMSwiIiwxLHsic3R5bGUiOnsibmFtZSI6ImNvcm5lciJ9fV0sWzMsMiwiXFxwYXJ0aWFsIHciLDEseyJjdXJ2ZSI6MX1dLFszLDAsIiIsMCx7InN0eWxlIjp7ImhlYWQiOnsibmFtZSI6ImVwaSJ9fX1dLFs1LDQsInciLDEseyJjdXJ2ZSI6MSwiY29sb3VyIjpbMCw2MCw2MF19XSxbNSwxLCIiLDEseyJzdHlsZSI6eyJoZWFkIjp7Im5hbWUiOiJlcGkifX19XSxbMyw1XSxbMywxLCIiLDEseyJzdHlsZSI6eyJuYW1lIjoiY29ybmVyIn19XSxbNCwxLCIiLDEseyJzdHlsZSI6eyJoZWFkIjp7Im5hbWUiOiJlcGkifX19XSxbMiwzLCJcXHBhcnRpYWwgcyIsMSx7ImN1cnZlIjoxfV0sWzQsNSwicyIsMSx7ImN1cnZlIjoxLCJzdHlsZSI6eyJib2R5Ijp7Im5hbWUiOiJkYXNoZWQifX19XV0=
  \begin{tikzcd}[cramped]
    {\partial E} && E \\
    && {\partial \underline E} && {\underline{E}} \\
    & {\partial B} && B
    \arrow[from=1-1, to=1-3]
    \arrow["{\partial w}"{description}, curve={height=6pt}, from=1-1, to=2-3]
    \arrow[two heads, from=1-1, to=3-2]
    \arrow["w"{description}, curve={height=6pt}, from=1-3, to=2-5]
    \arrow[two heads, from=1-3, to=3-4]
    \arrow["{\partial r}"{description}, curve={height=6pt}, from=2-3, to=1-1]
    \arrow[two heads, from=2-3, to=3-2]
    \arrow["r"{description}, curve={height=6pt}, dashed, from=2-5, to=1-3]
    \arrow[two heads, from=2-5, to=3-4]
    \arrow["i"', hook, from=3-2, to=3-4]
    \arrow[crossing over, hook, from=2-3, to=2-5]
    \arrow["\lrcorner"{anchor=center, pos=0.15, scale=1.5}, draw=none, from=1-1, to=3-4]
    \arrow["\lrcorner"{anchor=center, pos=0.15, scale=1.5}, draw=none, from=2-3, to=3-4]
  \end{tikzcd}
\end{equation}
where
\begin{enumerate}
  \item $i$ is a cofibration.
  \item $E, \underline{E}$ are fibrant over $B$.
  \item $w \colon E \xrightarrow{\sim} \underline{E}$ is a weak equivalence over $B$.
  \item $\partial E, \partial\underline{E}, \partial w$ are the pullbacks of $E, \underline{E}, w$.
  \item $\partial r$ is a homotopy retraction of $\partial w$ over $\partial B$.
  \item $\partial H \colon \partial r \cdot \partial w \simeq \id_{\partial E}$,
  as depicted below, is a right homotopy over $\partial B$ via the fibred path
  object $P^I_{\partial B}(\partial E)$.
  \begin{equation*}
    % https://q.uiver.app/#q=WzAsOCxbMywxLCJbSSwgXFxwYXJ0aWFsIEVdIl0sWzMsMywiW0ksIFxccGFydGlhbCBCXSJdLFsyLDMsIlxccGFydGlhbCBCIl0sWzIsMSwiUF97XFxwYXJ0aWFsIEJ9KFxccGFydGlhbCBFKSJdLFs0LDIsIltcXHNldHsxfSwgXFxwYXJ0aWFsIEVdIl0sWzQsMCwiW1xcc2V0ezB9LCBcXHBhcnRpYWwgRV0iXSxbMCwxLCJcXHBhcnRpYWwgRSJdLFsyLDAsIlxccGFydGlhbFxcdWx7RX0iXSxbMiwxXSxbMCwxXSxbMywyXSxbMywwXSxbMCw0XSxbMywxLCIiLDEseyJzdHlsZSI6eyJuYW1lIjoiY29ybmVyIn19XSxbNiwzLCJIIiwxXSxbMCw1XSxbNiwyXSxbNiw0LCI9IiwxLHsibGFiZWxfcG9zaXRpb24iOjYwfV0sWzYsNywiXFxwYXJ0aWFsIHciXSxbNyw1LCJcXHBhcnRpYWwgciJdXQ==
    \begin{tikzcd}[cramped, row sep=small, column sep=small]
      && {\partial\ul{E}} && {[\set{0}, \partial E]} \\
      {\partial E} && {P_{\partial B}^I(\partial E)} & {[I, \partial E]} \\
      &&&& {[\set{1}, \partial E]} \\
      && {\partial B} & {[I, \partial B]}
      \arrow["{\partial r}", from=1-3, to=1-5]
      \arrow["{\partial w}", from=2-1, to=1-3]
      \arrow[from=2-1, to=4-3]
      \arrow[from=2-3, to=2-4]
      \arrow[from=2-3, to=4-3]
      \arrow[from=2-4, to=1-5]
      \arrow[from=2-4, to=3-5]
      \arrow[from=2-4, to=4-4]
      \arrow[from=4-3, to=4-4]
      \arrow["{=}"{description, pos=0.6}, from=2-1, to=3-5, crossing over]
      \arrow["\lrcorner"{anchor=center, pos=0.15, scale=1.5}, draw=none, from=2-3, to=4-4]
      \arrow["\partial H"{description}, from=2-1, to=2-3]
    \end{tikzcd}
  \end{equation*}
\end{enumerate}
Like in \Cref{lem:ext-htpy-sect-wk}, the goal is to extend $\partial H$ to a
homotopy $H \colon rw \simeq \id_{E}$ over $B$ for some homotopy retraction $r$
of $w$.
But the difference this time is that the homotopy retraction $r$ chosen must
also extend $\partial r$.

The solution proceeds by setting up a suitable lifting problem.
For this we need to verify the candidate left map is indeed in the left class,
for which we need the following fact about adhesive categories, whose role in
the equivalence extension property was observed by \cite{sat17}.
\begin{definition}
  A category is \emph{adhesive} if all of the following conditions hold:
  \begin{enumerate}[nolistsep]
    \item pushouts along monos exist;
    \item such pushout squares are pullback squares; and
    \item pullbacks preserve pushouts.
  \end{enumerate}
\end{definition}
\begin{lemma}\label{lem:adh-pushout-disjoint-subobj}
  Let $\bC$ be an adhesive category.
  Whenever one has objects and maps in $\bC$ as follows (i.e.
  $\partial_1X \cap \partial_0X = 0$ as subobjects of $X$), the composition of
  the entire bottom row is a mono.
  \begin{equation*}
    % https://q.uiver.app/#q=WzAsNixbMiwxLCJYIFxcY3VwIFkiXSxbMiwwLCJZIl0sWzEsMSwiWCJdLFsxLDAsIlxccGFydGlhbF8wWCJdLFswLDEsIlxccGFydGlhbF8xWCJdLFswLDAsIjAiXSxbMSwwLCIiLDEseyJzdHlsZSI6eyJ0YWlsIjp7Im5hbWUiOiJob29rIiwic2lkZSI6InRvcCJ9fX1dLFsyLDBdLFszLDIsIiIsMSx7InN0eWxlIjp7InRhaWwiOnsibmFtZSI6Imhvb2siLCJzaWRlIjoidG9wIn19fV0sWzMsMV0sWzUsNCwiIiwwLHsic3R5bGUiOnsidGFpbCI6eyJuYW1lIjoiaG9vayIsInNpZGUiOiJ0b3AifX19XSxbNSwzLCIiLDAseyJzdHlsZSI6eyJ0YWlsIjp7Im5hbWUiOiJob29rIiwic2lkZSI6InRvcCJ9fX1dLFs0LDIsIiIsMCx7InN0eWxlIjp7InRhaWwiOnsibmFtZSI6Imhvb2siLCJzaWRlIjoidG9wIn19fV0sWzAsMywiIiwwLHsic3R5bGUiOnsibmFtZSI6ImNvcm5lciJ9fV0sWzQsMCwiIiwwLHsiY3VydmUiOjIsInN0eWxlIjp7InRhaWwiOnsibmFtZSI6Imhvb2siLCJzaWRlIjoidG9wIn0sImJvZHkiOnsibmFtZSI6ImRhc2hlZCJ9fX1dXQ==
    \begin{tikzcd}[cramped]
      0 & {\partial_0X} & Y \\
      {\partial_1X} & X & {X \cup Y}
      \arrow[hook, from=1-1, to=1-2]
      \arrow[hook, from=1-1, to=2-1]
      \arrow[from=1-2, to=1-3]
      \arrow[hook, from=1-2, to=2-2]
      \arrow[hook, from=1-3, to=2-3]
      \arrow[hook, from=2-1, to=2-2]
      \arrow[curve={height=12pt}, dashed, hook, from=2-1, to=2-3]
      \arrow[from=2-2, to=2-3]
      \arrow["\lrcorner"{anchor=center, pos=0.15, scale=1.5}, draw=none, from=1-1, to=2-2]
      \arrow["\lrcorner"{anchor=center, pos=0.15, scale=1.5, rotate=180}, draw=none, from=2-3, to=1-2]
    \end{tikzcd}
  \end{equation*}
\end{lemma}
\begin{proof}
  By the pushout lemma, in the left diagram below, the right square in the back face
  is a pushout (i.e.
  $(\partial_1X \sqcup \partial_0X) \cup_{\partial_0 X} Y \cong \partial_1X
  \sqcup Y$).
  \begin{center}
    \begin{minipage}{0.45\linewidth}
      \begin{equation*}
        % https://q.uiver.app/#q=WzAsOCxbMywyLCJYIFxcY3VwIFkiXSxbMiwwLCJZIl0sWzIsMiwiWCJdLFsxLDAsIlxccGFydGlhbF8wWCJdLFswLDEsIlxccGFydGlhbF8xWCJdLFswLDAsIjAiXSxbMSwxLCJcXHBhcnRpYWxfMVggXFxzcWN1cCBcXHBhcnRpYWxfMFgiXSxbMiwxLCJcXHBhcnRpYWxfMVggXFxzcWN1cCBZIl0sWzEsMCwiIiwxLHsic3R5bGUiOnsidGFpbCI6eyJuYW1lIjoiaG9vayIsInNpZGUiOiJ0b3AifX19XSxbMiwwXSxbMywyLCIiLDEseyJzdHlsZSI6eyJ0YWlsIjp7Im5hbWUiOiJob29rIiwic2lkZSI6InRvcCJ9fX1dLFszLDFdLFs1LDQsIiIsMCx7InN0eWxlIjp7InRhaWwiOnsibmFtZSI6Imhvb2siLCJzaWRlIjoidG9wIn19fV0sWzUsMywiIiwwLHsic3R5bGUiOnsidGFpbCI6eyJuYW1lIjoiaG9vayIsInNpZGUiOiJ0b3AifX19XSxbNiwyLCIiLDAseyJzdHlsZSI6eyJ0YWlsIjp7Im5hbWUiOiJob29rIiwic2lkZSI6InRvcCJ9fX1dLFs0LDYsIiIsMCx7InN0eWxlIjp7InRhaWwiOnsibmFtZSI6Imhvb2siLCJzaWRlIjoidG9wIn19fV0sWzYsN10sWzEsNywiIiwwLHsic3R5bGUiOnsidGFpbCI6eyJuYW1lIjoiaG9vayIsInNpZGUiOiJ0b3AifX19XSxbNywwLCIiLDAseyJzdHlsZSI6eyJ0YWlsIjp7Im5hbWUiOiJob29rIiwic2lkZSI6InRvcCJ9fX1dLFs0LDIsIiIsMCx7InN0eWxlIjp7InRhaWwiOnsibmFtZSI6Imhvb2siLCJzaWRlIjoidG9wIn19fV0sWzMsNiwiIiwwLHsic3R5bGUiOnsidGFpbCI6eyJuYW1lIjoiaG9vayIsInNpZGUiOiJ0b3AifX19XSxbNiw1LCIiLDAseyJzdHlsZSI6eyJuYW1lIjoiY29ybmVyIn19XSxbNywzLCIiLDAseyJzdHlsZSI6eyJuYW1lIjoiY29ybmVyIn19XSxbMCwzLCIiLDAseyJzdHlsZSI6eyJuYW1lIjoiY29ybmVyIn19XSxbMCw2LCIiLDAseyJzdHlsZSI6eyJuYW1lIjoiY29ybmVyIn19XV0=
        \begin{tikzcd}[cramped]
          0 & {\partial_0X} & Y \\
          {\partial_1X} & {\partial_1X \sqcup \partial_0X} & {\partial_1X \sqcup Y} \\
          && X & {X \cup Y}
          \arrow[hook, from=1-1, to=1-2]
          \arrow[hook, from=1-1, to=2-1]
          \arrow[from=1-2, to=1-3]
          \arrow[hook, from=1-2, to=2-2]
          \arrow[hook, from=1-3, to=2-3]
          \arrow[hook, from=1-3, to=3-4]
          \arrow[hook, from=2-1, to=2-2]
          \arrow[hook, from=2-1, to=3-3]
          \arrow[from=2-2, to=2-3]
          \arrow[hook, from=2-2, to=3-3]
          \arrow[from=2-3, to=3-4]
          \arrow[from=3-3, to=3-4]
          \arrow["\lrcorner"{anchor=center, pos=0.15, scale=1.5, rotate=180}, draw=none, from=2-2, to=1-1]
          \arrow["\lrcorner"{anchor=center, pos=0.15, scale=1.5, rotate=180}, draw=none, from=2-3, to=1-2]
          \arrow["\lrcorner"{anchor=center, pos=0.15, scale=1.5, rotate=180}, draw=none, from=3-4, to=1-2]
          \arrow["\lrcorner"{anchor=center, pos=0.15, scale=1.5, rotate=225}, draw=none, from=3-4, to=2-2]
          \arrow[hook, from=1-2, to=3-3, crossing over]
        \end{tikzcd}
      \end{equation*}
    \end{minipage}
    \begin{minipage}{0.45\linewidth}
      \begin{equation*}
        % https://q.uiver.app/#q=WzAsNCxbMSwxLCJYIFxcY3VwIFkiXSxbMCwxLCJYIl0sWzAsMCwiXFxwYXJ0aWFsXzFYIFxcc3FjdXAgXFxwYXJ0aWFsXzBYIl0sWzEsMCwiXFxwYXJ0aWFsXzFYIFxcc3FjdXAgWSJdLFsxLDBdLFsyLDEsIiIsMCx7InN0eWxlIjp7InRhaWwiOnsibmFtZSI6Imhvb2siLCJzaWRlIjoidG9wIn19fV0sWzIsM10sWzMsMCwiIiwwLHsic3R5bGUiOnsidGFpbCI6eyJuYW1lIjoiaG9vayIsInNpZGUiOiJ0b3AifX19XSxbMCwyLCIiLDAseyJzdHlsZSI6eyJuYW1lIjoiY29ybmVyIn19XV0=
        \begin{tikzcd}[cramped]
          {\partial_1X \sqcup \partial_0X} & {\partial_1X \sqcup Y} \\
          X & {X \cup Y}
          \arrow[from=1-1, to=1-2]
          \arrow[hook, from=1-1, to=2-1]
          \arrow[hook, from=1-2, to=2-2]
          \arrow[from=2-1, to=2-2]
          \arrow["\lrcorner"{anchor=center, pos=0.15, scale=1.5, rotate=180}, draw=none, from=2-2, to=1-1]
        \end{tikzcd}
      \end{equation*}
    \end{minipage}
  \end{center}
  Therefore, by the pushout lemma once again, the square on the bottom face is a
  pushout (because the slanted face is a pushout by assumption).
  By adhesiveness, the map $\partial_1X \sqcup \partial_0X \hookrightarrow X$ is
  a mono.
  By adhesiveness again, pushouts preserve monos, so the right row
  $\partial_1X \sqcup Y \hookrightarrow X \cup Y$ of the diagram on the right
  above is a mono.
  Because $\partial_1X \hookrightarrow \partial_1X \sqcup Y$ is a mono by
  adhesiveness, and the map $\partial_1X \hookrightarrow X \to X \cup Y$ factors
  as
  $\partial_1X \hookrightarrow \partial_1X \sqcup Y \hookrightarrow X \cup Y$,
  the result follows.
\end{proof}
We are now ready to verify the candidate left map is indeed in the left class.
\begin{lemma}\label{lem:htpy-retract-ext-tc}
  The dashed map below
  $\partial\underline{E} \times \set{0} \cup (\partial E \times I \cup E \times
  \set{1}) \to \underline{E} \times \set{0} \cup E \times I$, viewed as the
  iterated pushout of the second row in the back face into the pushout in the
  front face induced by the component inclusions
  $\partial\ul{E} \times \set{0} \hookrightarrow \ul{E} \times \set{0}$ and
  $\partial E \times I \cup E \times \set{1} \hookrightarrow E \times I$ is a
  trivial cofibration.
  \begin{equation*}
    \begin{tikzcd}[cramped, column sep=small, row sep=small]
      && {\partial E \times \set{1}} & {E \times \set{1}} \\
      \\
      {\partial E \times \set{0}} && {\partial E \times I} & {\partial E \times I \cup E \times \set{1}} \\
      \\
      {\partial\underline{E} \times \set{0}} && {\partial\underline{E} \times \set{0} \cup \partial E \times I} & {\partial\underline{E} \times \set{0} \cup \partial E \times I \cup E \times \set{1}} \\
      &&& {E \times \set{0}} &&& {E \times I} \\
      \\
      &&& {\underline{E} \times \set{0}} &&& {\underline{E}\times \set{0} \cup E \times I}
      \arrow[hook, from=1-3, to=1-4]
      \arrow["\sim"{description}, hook, from=1-3, to=3-3]
      \arrow["\sim"{description}, hook, from=1-4, to=3-4]
      \arrow["\sim"{description}, hook, from=1-4, to=6-7]
      \arrow["\sim"{description}, hook, from=3-1, to=3-3]
      \arrow["{\partial w}"', "{\sim}", from=3-1, to=5-1]
      \arrow[hook, from=3-3, to=3-4]
      \arrow["\sim"{description}, from=3-3, to=5-3]
      \arrow["\sim"{description}, from=3-4, to=5-4]
      \arrow["\sim"{description, pos=0.3}, hook, from=3-4, to=6-7]
      \arrow["\sim"{description}, hook, from=5-1, to=5-3]
      \arrow[hook, from=5-1, to=8-4]
      \arrow[hook, from=5-3, to=5-4]
      \arrow[dashed, from=5-4, to=8-7]
      \arrow["w"', "\sim", from=6-4, to=8-4]
      \arrow["\sim"{description}, from=6-7, to=8-7]
      \arrow["\sim"{description}, hook, from=8-4, to=8-7]
      \arrow[from=5-3, to=8-7, dotted]
      \arrow[dotted, from=1-4, to=8-7]
      \arrow[hook, from=3-1, to=6-4]
      \arrow[crossing over, "\sim"{description}, hook, from=6-4, to=6-7]
      \arrow["\lrcorner"{anchor=center, pos=0.15, scale=1.5, rotate=180}, draw=none, from=3-4, to=1-3]
      \arrow["\lrcorner"{anchor=center, pos=0.15, scale=1.5, rotate=180}, draw=none, from=5-3, to=3-1]
      \arrow["\lrcorner"{anchor=center, pos=0.15, scale=1.5, rotate=180}, draw=none, from=5-4, to=3-3]
      \arrow["\lrcorner"{anchor=center, pos=0.15, scale=1.5, rotate=180}, draw=none, from=8-7, to=6-4]
    \end{tikzcd}
  \end{equation*}
\end{lemma}
\begin{proof}
  Because $\partial w$ and $w$ are both weak eqivalences and all objects are
  cofibrant by assumption, $\bC$ is left proper and so their pushouts
  $\partial E \times I \cup E \times \set{1} \to \partial \underline{E} \times
  \set{0} \cup (\partial E \times I \cup E \times \set{1})$ and
  $E \times I \to \underline{E} \times \set{0} \cup E \times I$ given by the
  rightmost map in the second row of the back face and in the rightmost map in
  the front face are weak equivalences as well.
  And because $\partial E \times I \cup E \times \set{1}\to E \times I$ is the pushout-product
  of the (cofibration, trivial cofibration)-pair
  $(\partial E \hookrightarrow E, \set{1} \hookrightarrow I)$, it is also a weak
  equivalence.
  Thus, by the 2-out-of-3 property applied to the square on the right face, so
  is the dashed map
  $\partial\underline{E} \times \set{0} \cup \partial E \times I \cup E \times
  \set{1} \to \underline{E} \times \set{0} \cup E \times I$ which we are
  interested in.

  It remains to check that the dashed map above is a monomorphism.
  For this, it suffices to see that the two dotted maps above are monos, at
  which point the dashed map is a mono because the right two squares on the back
  face are pushouts and $\bC$ is adhesive.

  We observe that the dotted map on the bottom face is a mono because
  adhesiveness of $\bC$ implies that that dotted map is exactly the connecting
  map of the pullback
  $(\partial\underline{E} \times \set{0} \cup \partial E \times I) \cong
  i^*(\underline{E} \times \set{0} \cup E \times I) \hookrightarrow
  \underline{E} \times \set{0} \cup E \times I$.
  To observe that the dotted map on the right face is a mono, we apply
  \Cref{lem:adh-pushout-disjoint-subobj} to the pushout on the front face,
  noting that the subobjects
  $E \times \set{0}, E \times \set{1} \rightrightarrows E \times I$ are
  disjoint.
\end{proof}

We are now ready to reproduce the proof of homotopy retraction extension
property from \cite[Lemma 3.3.5]{kl21}.

\begin{proposition}[{\cite[Lemma 3.3.5]{kl21}}]\label{prop:hrep}
  As from \Cref{eqn:hrep}, given a homotopy equivalence $w \colon E \to \ul{E}$
  over $B$ whose restriction to over $\partial B$ is
  $\partial w \colon \partial E \to \partial\ul{E}$ with homotopy retraction
  $\partial r \colon \partial\ul{E} \to \partial E$ witnessed by homotopy
  $\partial H \colon \partial r \cdot \partial w \simeq \id_{\partial E}$, one
  can extend this homotopy data to a homotopy section $r \colon \ul{E} \to E$
  over $B$ witnessed by homotopy $H \colon r \cdot w \simeq \id_{E}$.
\end{proposition}
\begin{proof}
  We transpose the right homotopy
  $\partial H \colon \partial E \to P_{\partial B}^I(\partial E)$ to a left homotopy
  $(\partial H)^\dagger \colon \partial E \times I \to \partial E$ over $\partial B$.
  Now, we would like to extend $\partial r$ and $\partial H^\dagger$ to a
  homotopy retraction $r$ with left homotopy $H^\dagger \colon E \times I \to E$
  over $B$ so that taking the transpose of $H^\dagger$ gives a right homotopy
  $H \colon E \to [I,E]$ extending $\partial H$.
  This can be accomplished because by \Cref{lem:htpy-retract-ext-tc}, the left
  map as below is a trivial cofibration and by assumption, the right map below
  is a fibration.
  Thus, the extension of $\partial r$ and $\partial H^\dagger$ is accomplished.
  \begin{equation*}
    % https://q.uiver.app/#q=WzAsNSxbMCwwLCJcXHBhcnRpYWxcXHVuZGVybGluZXtFfSBcXHRpbWVzIFxcc2V0ezB9IFxcY3VwIFxccGFydGlhbCBFIFxcdGltZXMgSSBcXGN1cCBFIFxcdGltZXMgXFxzZXR7MX0iXSxbMCwyLCJcXHVuZGVybGluZXtFfSBcXHRpbWVzIFxcc2V0ezB9IFxcY3VwIEUgXFx0aW1lcyBJIl0sWzMsMiwiQiJdLFszLDAsIkUiXSxbMiwyLCJFIl0sWzAsMSwiXFxiZWdpbntwbWF0cml4fSBcXHBhcnRpYWxcXHVuZGVybGluZXtFfSBcXHRpbWVzIFxcc2V0ezB9IFxcaG9va3JpZ2h0YXJyb3cgXFx1bmRlcmxpbmV7RX0gXFx0aW1lcyBcXHNldHswfSBcXFxcIFxccGFydGlhbCBFIFxcdGltZXMgSSBcXGhvb2tyaWdodGFycm93IEUgXFx0aW1lcyBJIFxcXFwgRSBcXHRpbWVzIFxcc2V0ezF9IFxcaG9va3JpZ2h0YXJyb3cgRSBcXHRpbWVzIEkgXFxlbmR7cG1hdHJpeH0iLDIseyJzdHlsZSI6eyJ0YWlsIjp7Im5hbWUiOiJob29rIiwic2lkZSI6InRvcCJ9fX1dLFszLDIsIiIsMix7InN0eWxlIjp7ImhlYWQiOnsibmFtZSI6ImVwaSJ9fX1dLFs0LDIsIiIsMCx7InN0eWxlIjp7ImhlYWQiOnsibmFtZSI6ImVwaSJ9fX1dLFsxLDQsIlxcYmVnaW57cG1hdHJpeH0gXFx1bmRlcmxpbmV7RX0gXFx4cmlnaHRhcnJvd3t3fSBFIFxcXFwgRSBcXHRpbWVzIEkgXFx4cmlnaHRhcnJvd3tcXHByb2p9IEUgXFxlbmR7cG1hdHJpeH0iLDJdLFsxLDMsIihyLCBIXlxcZGFnZ2VyKSIsMSx7InN0eWxlIjp7ImJvZHkiOnsibmFtZSI6ImRhc2hlZCJ9fX1dLFswLDMsIlxcYmVnaW57cG1hdHJpeH0gXFxwYXJ0aWFsXFx1bmRlcmxpbmV7RX0gXFx0aW1lcyBcXHNldHswfSA9IFxccGFydGlhbFxcdW5kZXJsaW5le0V9IFxceHJpZ2h0YXJyb3d7XFxwYXJ0aWFsIHJ9IFxccGFydGlhbCBFIFxcaG9va3JpZ2h0YXJyb3cgRSBcXFxcIFxccGFydGlhbCBFIFxcdGltZXMgSSBcXHhyaWdodGFycm93e0h9IFxccGFydGlhbCBFIFxcaG9va3JpZ2h0YXJyb3cgRSBcXFxcIEUgXFx0aW1lcyBcXHNldHsxfSBcXHhyaWdodGFycm93ez19IEUgXFxlbmR7cG1hdHJpeH0iXV0=&macro_url=https%3A%2F%2Fgist.githubusercontent.com%2Flim495062%2F61b94af9ef95c1c7b0763c937de29c2b%2Fraw%2F456d405748eab3250184512b5240467b6083b2b4%2Facmhwmacros.sty
    \begin{tikzcd}[cramped]
      {\partial\underline{E} \times \set{0} \cup \partial E \times I \cup E \times \set{1}} &&& E \\
      \\
      {\underline{E} \times \set{0} \cup E \times I} && E & B
      \arrow["\begin{array}{c} \begin{pmatrix} \partial\underline{E} \times \set{0} \cong \partial\underline{E} \xrightarrow{\partial r} \partial E \hookrightarrow E \\ \partial E \times I \xrightarrow{H^\dagger} \partial E \hookrightarrow E \\ E \times \set{1} \xrightarrow{\cong} E \end{pmatrix} \end{array}", from=1-1, to=1-4]
      \arrow["\begin{array}{c} \begin{pmatrix} \partial\underline{E} \times \set{0} \hookrightarrow \underline{E} \times \set{0} \\ \partial E \times I \hookrightarrow E \times I \\ E \times \set{1} \hookrightarrow E \times I \end{pmatrix} \end{array}"', hook, from=1-1, to=3-1]
      \arrow[two heads, from=1-4, to=3-4]
      \arrow["{(r, H^\dagger)}"{description}, dashed, from=3-1, to=1-4]
      \arrow["\begin{array}{c} \begin{pmatrix} \underline{E} \xrightarrow{w} E \\ E \times I \xrightarrow{\proj} E \end{pmatrix} \end{array}"', from=3-1, to=3-3]
      \arrow[two heads, from=3-3, to=3-4]
    \end{tikzcd}
  \end{equation*}
\end{proof}

\subsubsection{Extending Homotopy Sections} \label{subsubsec:htpy-sect-ext}
We next establish a homotopy section extension property similar to \cite[Lemma
3.3.6]{kl21}.
Similar to before, our problem setup fixes the following solid maps
\begin{equation}\label{eqn:hsep}\tag{\textsc{hsep}}
  % https://q.uiver.app/#q=WzAsNixbMSwyLCJcXHBhcnRpYWwgQiJdLFszLDIsIkIiXSxbMiwxLCJcXHBhcnRpYWwgXFx1bmRlcmxpbmUgRSJdLFswLDAsIlxccGFydGlhbCBFIl0sWzQsMSwiXFx1bmRlcmxpbmV7RX0iXSxbMiwwLCJFIl0sWzAsMSwiaSIsMix7InN0eWxlIjp7InRhaWwiOnsibmFtZSI6Imhvb2siLCJzaWRlIjoidG9wIn19fV0sWzIsMCwiIiwxLHsic3R5bGUiOnsiaGVhZCI6eyJuYW1lIjoiZXBpIn19fV0sWzIsNCwiIiwyLHsiY29sb3VyIjpbMCw2MCw2MF0sInN0eWxlIjp7InRhaWwiOnsibmFtZSI6Imhvb2siLCJzaWRlIjoidG9wIn19fV0sWzIsMSwiIiwxLHsic3R5bGUiOnsibmFtZSI6ImNvcm5lciJ9fV0sWzMsMiwiXFxwYXJ0aWFsIHciLDEseyJjdXJ2ZSI6MX1dLFszLDAsIiIsMCx7InN0eWxlIjp7ImhlYWQiOnsibmFtZSI6ImVwaSJ9fX1dLFs1LDQsInciLDEseyJjdXJ2ZSI6MSwiY29sb3VyIjpbMCw2MCw2MF19XSxbNSwxLCIiLDEseyJzdHlsZSI6eyJoZWFkIjp7Im5hbWUiOiJlcGkifX19XSxbMyw1XSxbMywxLCIiLDEseyJzdHlsZSI6eyJuYW1lIjoiY29ybmVyIn19XSxbNCwxLCIiLDEseyJzdHlsZSI6eyJoZWFkIjp7Im5hbWUiOiJlcGkifX19XSxbMiwzLCJcXHBhcnRpYWwgcyIsMSx7ImN1cnZlIjoxfV0sWzQsNSwicyIsMSx7ImN1cnZlIjoxLCJzdHlsZSI6eyJib2R5Ijp7Im5hbWUiOiJkYXNoZWQifX19XV0=
  \begin{tikzcd}[cramped]
    {\partial E} && E \\
    && {\partial \ul{E}} && {\ul{E}} \\
    & {\partial B} && B
    \arrow[from=1-1, to=1-3]
    \arrow["{\partial w}"{description}, curve={height=6pt}, from=1-1, to=2-3]
    \arrow[two heads, from=1-1, to=3-2]
    \arrow[two heads, from=1-3, to=3-4]
    \arrow["{\partial s}"{description}, curve={height=6pt}, from=2-3, to=1-1]
    \arrow[two heads, from=2-3, to=3-2]
    \arrow["s"{description}, curve={height=6pt}, dashed, from=2-5, to=1-3]
    \arrow[two heads, from=2-5, to=3-4]
    \arrow["i"', hook, from=3-2, to=3-4]
    \arrow[crossing over, hook, from=2-3, to=2-5]
    \arrow["w"{description}, curve={height=6pt}, from=1-3, to=2-5]
    \arrow["\lrcorner"{anchor=center, pos=0.15, scale=1.5}, draw=none, from=1-1, to=3-4]
    \arrow["\lrcorner"{anchor=center, pos=0.15, scale=1.5}, draw=none, from=2-3, to=3-4]
  \end{tikzcd}
\end{equation}
where
\begin{enumerate}
  \item $i$ is a cofibration.
  \item $E, \ul{E}$ are fibrant over $B$.
  \item $w \colon E \xrightarrow{\sim} \ul{E}$ is a weak equivalence over $B$.
  \item $\partial E, \partial\ul{E}, \partial w$ are the pullbacks of $E, \ul{E}, w$.
  \item $\partial s$ is a homotopy section of $\partial w$ over $\partial B$.
  \item
  $\partial H \colon \partial w \cdot \partial s \simeq \id_{\partial\ul{E}}$,
  as observed in the diagram below, is a right homotopy over $\partial B$ via
  the fibred path object $P^I_{\partial B}(\partial\ul{E})$.
  \begin{equation*}
    % https://q.uiver.app/#q=WzAsOCxbMywxLCJbSSwgXFxwYXJ0aWFsXFx1dWx7RX1dIl0sWzMsMywiW0ksIFxccGFydGlhbCBCXSJdLFsyLDMsIlxccGFydGlhbCBCIl0sWzIsMSwiUF97XFxwYXJ0aWFsIEJ9KFxccGFydGlhbFxcdXVse0V9KSJdLFs0LDIsIltcXHNldHsxfSwgXFxwYXJ0aWFsXFx1dWx7RX1dIl0sWzQsMCwiW1xcc2V0ezB9LCBcXHBhcnRpYWxcXHV1bHtFfV0iXSxbMCwxLCJcXHBhcnRpYWxcXHV1bHtFfSJdLFsyLDAsIlxccGFydGlhbCBFIl0sWzIsMV0sWzAsMV0sWzMsMl0sWzMsMF0sWzAsNF0sWzMsMSwiIiwxLHsic3R5bGUiOnsibmFtZSI6ImNvcm5lciJ9fV0sWzYsMywiSCIsMV0sWzAsNV0sWzYsMl0sWzYsNCwiPSIsMSx7ImxhYmVsX3Bvc2l0aW9uIjo2MH1dLFs2LDcsIlxccGFydGlhbCBzIl0sWzcsNSwiXFxwYXJ0aWFsIHciXV0=
    \begin{tikzcd}[cramped, row sep=small, column sep=small]
      && {\partial E} && {[\set{0}, \partial\ul{E}]} \\
      {\partial\ul{E}} && {P^I_{\partial B}(\partial\ul{E})} & {[I, \partial\ul{E}]} \\
      &&&& {[\set{1}, \partial\ul{E}]} \\
      && {\partial B} & {[I, \partial B]}
      \arrow["{\partial w}", from=1-3, to=1-5]
      \arrow["{\partial s}", from=2-1, to=1-3]
      \arrow[from=2-1, to=4-3]
      \arrow[from=2-3, to=2-4]
      \arrow[from=2-3, to=4-3]
      \arrow[from=2-4, to=1-5]
      \arrow[from=2-4, to=3-5]
      \arrow[from=2-4, to=4-4]
      \arrow[from=4-3, to=4-4]
      \arrow["{=}"{description, pos=0.6}, from=2-1, to=3-5, crossing over]
      \arrow["\lrcorner"{anchor=center, pos=0.125, scale=1.5}, draw=none, from=2-3, to=4-4]
      \arrow["{\partial H}"{description}, from=2-1, to=2-3]
    \end{tikzcd}
  \end{equation*}
\end{enumerate}
The goal is to extend both $\partial s$ to a homotopy section $s$ of $w$ with
homotopy $H \colon ws \simeq \id_{\ul{E}}$.

As before, this is done by solving a suitable lifting problem, which is defined
using the fibred mapping path objects.

\begin{definition}\label{def:mapping-path-obj}
  Let $f \colon X \to Y \in \sfrac{\bC}{C}$ be a map over $C$.
  The \emph{fibred mapping path object} of $f$ is given by the pullback of the
  endpoint evaluation $\ev_0 \colon P_C^\bI(Y) \to Y$ along $f \colon X \to Y$,
  as follows.
  \begin{equation*}
    % https://q.uiver.app/#q=WzAsNixbMCwyLCJYIl0sWzEsMiwiWSJdLFsxLDEsIlkgXFx0aW1lc19DWSJdLFsxLDAsIlBfQ15JKFkpIl0sWzAsMSwiWCBcXHRpbWVzX0MgWSJdLFswLDAsIlBfQ15JKGYpIl0sWzMsMiwiXFxldl9cXHBhcnRpYWwiXSxbMiwxLCJcXHByb2pfMSJdLFswLDEsImYiLDJdLFs0LDBdLFs0LDJdLFs0LDEsIiIsMSx7InN0eWxlIjp7Im5hbWUiOiJjb3JuZXIifX1dLFs1LDRdLFs1LDNdLFs1LDIsIiIsMSx7InN0eWxlIjp7Im5hbWUiOiJjb3JuZXIifX1dXQ==
    \begin{tikzcd}[cramped]
      {P_C^I(f)} & {P_C^I(Y)} \\
      {X \times_C Y} & {Y \times_CY} \\
      X & Y
      \arrow[from=1-1, to=1-2]
      \arrow[from=1-1, to=2-1]
      \arrow["{\ev_\partial}", from=1-2, to=2-2]
      \arrow[from=2-1, to=2-2]
      \arrow[from=2-1, to=3-1]
      \arrow["{\proj_1}", from=2-2, to=3-2]
      \arrow["f"', from=3-1, to=3-2]
      \arrow["\lrcorner"{anchor=center, pos=0.05, scale=1.5}, draw=none, from=1-1, to=2-2]
      \arrow["\lrcorner"{anchor=center, pos=0.05, scale=1.5}, draw=none, from=2-1, to=3-2]
    \end{tikzcd}
  \end{equation*}
  It is equipped with two endpoint evaluation maps respectively denoted
  $\ev_0 \colon P_C^I(f) \to X$ and $\ev_1 \colon P_C^I(f) \to Y$.
\end{definition}

To show the homotopy section extension property, we set up a suitable lifting
problem lifting against the endpoint evaluation map of the fibred mapping path
object, so we must check that it lives in the right class.
\begin{lemma}\label{lem:mapping-path-obj-tf}
  Suppose that $f \colon X \to Y$ is a map of fibrant objects over $C$.
  Then, the endpoint evaluation map $\ev_1 \colon P_C^I(f) \to Y$ is a
  fibration.
  It is furthermore a trivial fibration if $f$ is a weak equivalence and the
  model structure of $\bC$ is right proper.
\end{lemma}
\begin{proof}
  We first note that $\ev_1 \colon P_C^I(f) \to Y$ factors via
  $P_C^I(f) \to X \times_C Y \to Y$ as in the curved left map as follows.
  Because $X$ is fibrant, the projection $X \times_C Y \to Y$ is fibrant, and
  $P_C^I(f) \to X \times_C Y$ arises as a pullback of the fibration
  $\ev_\partial \colon P_C^I(Y) \to Y \times_C Y$, it is also fibrant.
  Hence, $\ev_1 \colon P_C^I(f) \to X \times_C Y \to Y$ is the composition of
  two fibrations, thus proving fibrancy.
  \begin{equation*}
    % https://q.uiver.app/#q=WzAsMTAsWzAsNCwiWCJdLFsyLDQsIlkiXSxbMiwyLCJZIFxcdGltZXNfQ1kiXSxbMiwwLCJQX0NeSShZKSJdLFswLDIsIlggXFx0aW1lc19DIFkiXSxbMCwwLCJQX0NeSShmKSJdLFszLDMsIlkiXSxbMSwzLCJZIl0sWzMsNSwiQyJdLFsxLDUsIkMiXSxbMywyLCJcXGV2X1xccGFydGlhbCIsMSx7InN0eWxlIjp7ImhlYWQiOnsibmFtZSI6ImVwaSJ9fX1dLFsyLDEsIlxccHJval8xIiwxLHsibGFiZWxfcG9zaXRpb24iOjcwLCJzdHlsZSI6eyJoZWFkIjp7Im5hbWUiOiJlcGkifX19XSxbMCwxLCJmIiwyLHsibGFiZWxfcG9zaXRpb24iOjcwfV0sWzQsMCwiXFxwcm9qXzEiLDIseyJsYWJlbF9wb3NpdGlvbiI6NzAsInN0eWxlIjp7ImhlYWQiOnsibmFtZSI6ImVwaSJ9fX1dLFs0LDJdLFs0LDEsIiIsMSx7InN0eWxlIjp7Im5hbWUiOiJjb3JuZXIifX1dLFs1LDQsIiIsMSx7InN0eWxlIjp7ImhlYWQiOnsibmFtZSI6ImVwaSJ9fX1dLFs1LDNdLFs1LDIsIiIsMSx7InN0eWxlIjp7Im5hbWUiOiJjb3JuZXIifX1dLFsyLDYsIlxccHJval8yIiwxLHsic3R5bGUiOnsiaGVhZCI6eyJuYW1lIjoiZXBpIn19fV0sWzMsNiwiXFxldl8xIiwwLHsiY3VydmUiOi0yfV0sWzcsNiwiPSIsMSx7ImxhYmVsX3Bvc2l0aW9uIjozMH1dLFs0LDcsIiIsMSx7InN0eWxlIjp7ImhlYWQiOnsibmFtZSI6ImVwaSJ9fX1dLFs5LDgsIj0iLDJdLFs2LDgsIiIsMix7InN0eWxlIjp7ImhlYWQiOnsibmFtZSI6ImVwaSJ9fX1dLFsxLDgsIiIsMix7InN0eWxlIjp7ImhlYWQiOnsibmFtZSI6ImVwaSJ9fX1dLFsyLDgsIiIsMix7InN0eWxlIjp7Im5hbWUiOiJjb3JuZXIifX1dLFs3LDksIiIsMSx7InN0eWxlIjp7ImhlYWQiOnsibmFtZSI6ImVwaSJ9fX1dLFswLDksIiIsMSx7InN0eWxlIjp7ImhlYWQiOnsibmFtZSI6ImVwaSJ9fX1dLFs0LDksIiIsMSx7InN0eWxlIjp7Im5hbWUiOiJjb3JuZXIifX1dLFs1LDcsIlxcZXZfMSIsMCx7ImN1cnZlIjotMn1dXQ==
    \begin{tikzcd}[cramped, row sep=small, column sep=small]
      {P_C^I(f)} && {P_C^I(Y)} \\
      \\
      {X \times_C Y} && {Y \times_CY} \\
      & Y && Y \\
      X && Y \\
      & C && C
      \arrow[from=1-1, to=1-3]
      \arrow[two heads, from=1-1, to=3-1]
      \arrow["{\ev_\partial}"{description}, two heads, from=1-3, to=3-3]
      \arrow["{\ev_1}", curve={height=-12pt}, from=1-3, to=4-4]
      \arrow[from=3-1, to=3-3]
      \arrow[two heads, from=3-1, to=4-2]
      \arrow["{\proj_1}"'{pos=0.7}, two heads, from=3-1, to=5-1]
      \arrow["{\proj_2}"{description}, two heads, from=3-3, to=4-4]
      \arrow["{\proj_1}"{description, pos=0.7}, two heads, from=3-3, to=5-3]
      \arrow["{=}"{description, pos=0.3}, from=4-2, to=4-4]
      \arrow[two heads, from=4-4, to=6-4]
      \arrow["f"'{pos=0.7}, from=5-1, to=5-3]
      \arrow[two heads, from=5-1, to=6-2]
      \arrow[two heads, from=5-3, to=6-4]
      \arrow["{=}"', from=6-2, to=6-4]
      \arrow[two heads, from=4-2, to=6-2, crossing over]
      \arrow["{\ev_1}", curve={height=-12pt}, from=1-1, to=4-2, crossing over]
      \arrow["\lrcorner"{anchor=center, pos=0.15, scale=1.5}, draw=none, from=1-1, to=3-3]
      \arrow["\lrcorner"{anchor=center, pos=0.15, scale=1.5}, draw=none, from=3-1, to=5-3]
      \arrow["\lrcorner"{anchor=center, pos=0.15, scale=1.5, rotate=-45}, draw=none, from=3-1, to=6-2]
      \arrow["\lrcorner"{anchor=center, pos=0.15, scale=1.5, rotate=-45}, draw=none, from=3-3, to=6-4]
    \end{tikzcd}
  \end{equation*}

  When $f$ is a weak equivalence in a right proper model category, the map
  $P_C^I(f) \to P_C^I(Y)$ is a weak equivalence as it occurs as the pullback of
  $f$ along a fibration.
  But because the endpoint evaluation map $P_C^I(Y) \to Y$ is a weak
  equivalence, the curved top face in the diagram shows that by 2-out-of-3, the
  map $\ev_1 \colon P_C^I(f) \to Y$ is also a weak equivalence.
\end{proof}

\begin{remark}
  In \Cref{lem:mapping-path-obj-tf}, we have showed separately that the endpoint
  evaluation map of the fibred mapping path object of a weak equivalence is both
  a fibration and a weak equivalence.
  And to show weak equivalence, we have required right properness.
  We note that even without the right properness condition, we can still show
  that $\ev_1 \colon P_C^I(f) \to Y$ is a trivial fibration when $f$ is a weak
  equivalence by manually constructing solutions to lifting problems against all
  cofibrations.
\end{remark}

\begin{proposition}\label{prop:hsep}
  As from \Cref{eqn:hsep}, given a homotopy equivalence $w \colon E \to \ul{E}$
  over $B$ whose restriction to over $\partial B$ is
  $\partial w \colon \partial E \to \partial\ul{E}$ with homotopy section
  $\partial r \colon \partial\ul{E} \to \partial E$ witnessed by homotopy
  $\partial H \colon \partial w \cdot \partial s \simeq \id_{\partial\ul{E}}$,
  one can extend this homotopy data to a homotopy section
  $s \colon \ul{E} \to E$ over $B$ witnessed by homotopy
  $H \colon w \cdot s \simeq \id_{\ul{E}}$, provided $\bC$ is right proper.
\end{proposition}
\begin{proof}
  By \Cref{lem:mapping-path-obj-tf}, the map $\ev_1 \colon P_B^I(w) \to \ul{E}$
  is a trivial fibration.
  Because $\partial{H} \colon \partial w \cdot \partial s \simeq \id$, we have a
  square as below giving rise to a map
  $(\partial s, \partial H) \colon \partial\ul{E} \to
  P_{\partial{B}}^I(\partial{w})$ into the fibred mapping path object of
  $\partial{w}$, which is the rebase of the fibred mapping path object of $w$
  under $\partial{B} \hookrightarrow B$.
  And because $\ev_1 \cdot H = \id_{\partial\ul{E}}$, it follows that we have
  the following commutative square.
  \begin{equation*}
    % https://q.uiver.app/#q=WzAsNixbMCwwLCJcXHBhcnRpYWxcXHVse0V9Il0sWzEsMCwiUF97XFxwYXJ0aWFsIEJ9XkkoXFxwYXJ0aWFse3d9KSJdLFsyLDAsIlBfe0J9XkkodykiXSxbMiwyLCJcXHVse0V9Il0sWzEsMSwiXFxwYXJ0aWFsXFx1bHtFfSJdLFswLDIsIlxcdWx7RX0iXSxbMSwyLCIiLDAseyJzdHlsZSI6eyJ0YWlsIjp7Im5hbWUiOiJob29rIiwic2lkZSI6InRvcCJ9fX1dLFs0LDMsIiIsMCx7InN0eWxlIjp7InRhaWwiOnsibmFtZSI6Imhvb2siLCJzaWRlIjoidG9wIn19fV0sWzIsMywiXFxldl8xIiwwLHsic3R5bGUiOnsiaGVhZCI6eyJuYW1lIjoiZXBpIn19fV0sWzEsNF0sWzAsMSwiKFxccGFydGlhbHtzfSwgXFxwYXJ0aWFse0h9KSJdLFswLDUsIiIsMix7InN0eWxlIjp7InRhaWwiOnsibmFtZSI6Imhvb2siLCJzaWRlIjoidG9wIn19fV0sWzUsMywiPSIsMl0sWzAsNCwiPSIsMV1d&macro_url=https%3A%2F%2Fgist.githubusercontent.com%2Flim495062%2F61b94af9ef95c1c7b0763c937de29c2b%2Fraw%2F456d405748eab3250184512b5240467b6083b2b4%2Facmhwmacros.sty
    \begin{tikzcd}[cramped, column sep=small]
      {\partial\ul{E}} & {P_{\partial B}^I(\partial{w})} & {P_{B}^I(w)} \\
      & {\partial\ul{E}} \\
      {\ul{E}} && {\ul{E}}
      \arrow["{(\partial{s}, \partial{H})}", from=1-1, to=1-2]
      \arrow["{=}"{description}, from=1-1, to=2-2]
      \arrow[hook, from=1-1, to=3-1]
      \arrow[hook, from=1-2, to=1-3]
      \arrow[from=1-2, to=2-2]
      \arrow["{\ev_1}", "\sim"', two heads, from=1-3, to=3-3]
      \arrow[hook, from=2-2, to=3-3]
      \arrow["{=}"', from=3-1, to=3-3]
    \end{tikzcd}
  \end{equation*}
  This living problem thus admits a solution
  $(s,H) \colon \ul{E} \to E \times_{\ul{E}} P^I_B(\partial\ul{E})$, which
  extends $(\partial s, \partial H)$ by construction.
\end{proof}

\subsubsection{Extending Homotopy Equivalences} \label{subsubsec:htpy-eqv-ext}
We are now ready to prove the result of extending homotopy equivalences.

\begin{theorem}\label{thm:heqv-ext}
  Let $\bC$ be equipped with a model structure in which:
  \begin{enumerate}
    \item The cofibrations are precisely the monomorphisms.
    \item The pushout-product of a (cofibration, trivial cofibration)-pair is
    still a trivial cofibration.
    \item $I$ is a good, but not necessarily very good, cylinder object for the
    terminal object with the two disjoint endpoint inclusions
    $\set{0},\set{1} \rightrightarrows I$.
    \item Right properness holds
  \end{enumerate}
  Then, if there is a situation as below
  \begin{equation*}
    % https://q.uiver.app/#q=WzAsNixbMSwyLCJcXHBhcnRpYWwgQiJdLFszLDIsIkIiXSxbMiwxLCJcXHBhcnRpYWwgXFx1bmRlcmxpbmUgRSJdLFswLDAsIlxccGFydGlhbCBFIl0sWzQsMSwiXFx1bmRlcmxpbmV7RX0iXSxbMiwwLCJFIl0sWzAsMSwiaSIsMix7InN0eWxlIjp7InRhaWwiOnsibmFtZSI6Imhvb2siLCJzaWRlIjoidG9wIn19fV0sWzIsMCwiIiwxLHsic3R5bGUiOnsiaGVhZCI6eyJuYW1lIjoiZXBpIn19fV0sWzIsNCwiIiwyLHsiY29sb3VyIjpbMCw2MCw2MF0sInN0eWxlIjp7InRhaWwiOnsibmFtZSI6Imhvb2siLCJzaWRlIjoidG9wIn0sImJvZHkiOnsibmFtZSI6ImRhc2hlZCJ9fX1dLFsyLDEsIiIsMSx7InN0eWxlIjp7Im5hbWUiOiJjb3JuZXIifX1dLFszLDIsIlxccGFydGlhbCB3IiwxXSxbMywwLCIiLDAseyJzdHlsZSI6eyJoZWFkIjp7Im5hbWUiOiJlcGkifX19XSxbNSwxLCIiLDEseyJzdHlsZSI6eyJoZWFkIjp7Im5hbWUiOiJlcGkifX19XSxbMyw1XSxbMywxLCIiLDEseyJzdHlsZSI6eyJuYW1lIjoiY29ybmVyIn19XSxbNCwxLCIiLDEseyJzdHlsZSI6eyJib2R5Ijp7Im5hbWUiOiJkYXNoZWQifSwiaGVhZCI6eyJuYW1lIjoiZXBpIn19fV0sWzUsNCwidyIsMSx7InN0eWxlIjp7ImJvZHkiOnsibmFtZSI6ImRhc2hlZCJ9fX1dLFsyLDMsIlxccGFydGlhbCBzIiwxLHsib2Zmc2V0IjotM31dLFsyLDMsIlxccGFydGlhbCByIiwxLHsib2Zmc2V0IjozfV0sWzQsNSwicyIsMSx7Im9mZnNldCI6LTMsInN0eWxlIjp7ImJvZHkiOnsibmFtZSI6ImRhc2hlZCJ9fX1dLFs0LDUsInIiLDEseyJvZmZzZXQiOjMsInN0eWxlIjp7ImJvZHkiOnsibmFtZSI6ImRhc2hlZCJ9fX1dXQ==
    \begin{tikzcd}[cramped]
      {\partial E} && E \\
      && {\partial \underline E} && {\underline{E}} \\
      & {\partial B} && B
      \arrow[from=1-1, to=1-3]
      \arrow["{\partial w}"{description}, from=1-1, to=2-3]
      \arrow[two heads, from=1-1, to=3-2]
      \arrow["w"{description}, dashed, from=1-3, to=2-5]
      \arrow[two heads, from=2-3, to=3-2]
      \arrow["s"{description}, shift left=3, dashed, from=2-5, to=1-3]
      \arrow["r"{description}, shift right=3, dashed, from=2-5, to=1-3]
      \arrow[dashed, two heads, from=2-5, to=3-4]
      \arrow["i"', hook, from=3-2, to=3-4]
      \arrow[two heads, from=1-3, to=3-4]
      \arrow[crossing over, "{\partial s}"{description}, shift left=3, from=2-3, to=1-1]
      \arrow[crossing over, "{\partial r}"{description}, shift right=3, from=2-3, to=1-1]
      \arrow[crossing over, dashed, hook, from=2-3, to=2-5]
      \arrow["\lrcorner"{anchor=center, pos=0.15, scale=1.5}, draw=none, from=2-3, to=3-4]
      \arrow["\lrcorner"{anchor=center, pos=0.15, scale=1.5}, draw=none, from=1-1, to=3-4]
    \end{tikzcd}
  \end{equation*}
  in which
  \begin{enumerate}
    \item $i$ is a cofibration.
    \item $\partial E$ and $\partial\ul{E}$ are fibrant over $\partial B$.
    \item $\partial w$ has a homotopy retraction $\partial r$ with a homotopy
    $\partial H_r \colon \partial r \cdot \partial w \simeq \id_{\partial E}$
    via $P^I_{\partial B}(\partial E)$.
    \item $\partial w$ has a homotopy section $\partial s$ with a homotopy
    $\partial H_s \colon \id_{\partial\ul{E}} \simeq \partial w \cdot \partial
    s$ via $P^I_{\partial B}(\partial\ul{E})$.
  \end{enumerate}
  Then, one can extend
  \begin{enumerate}
    \item $\partial r$ to a homotopy retraction of
    $\unglue_{i,\ul{E}}(w)$ and $\partial H_r$ to a homotopy
    $H_r \colon r \cdot \unglue_{i,\ul\partial{E}}(w) \simeq \id_{E}$ via
    $P^I_{B}(E)$.
    \item $\partial s$ to a homotopy section of $\unglue_{i,\ul\partial{E}}(w)$
    and $\partial H_s$ to a homotopy
    $H_s \colon \id_{\ul{E}} \simeq \unglue_{i,\ul{E}}(w) \cdot s$ via
    $P^I_{B}(\ul{E})$.
  \end{enumerate}
\end{theorem}
\begin{proof}
  Because all cofibrations are monos, in the slice over $\partial B$, the map
  $w$ is a homotopy equivalence between cofibrant-fibrant objects, so it is a
  homotopy weak equivalence.
  Thus, \Cref{prop:ext-we} applies to show $\unglue_{i,\ul{E}}(w)$ is a weak
  equivalence.
  Now, the extended inverses and homotopies follow from
  \Cref{prop:hrep,prop:hsep}.
\end{proof}

\subsection{Related work}\label{subsec:ext-comparison}
The weak equivalence extension property of \Cref{subsec:wk-eqv-ext} has been
shown in various other works such as \cite[Proposition 5.1]{sat17} and
\cite[Proposition 116]{awo23a}.
These extend weak equivalence along a cofibration.
This is the content of \Cref{subsec:wk-eqv-ext}.
However, they do not handle the homotopy data.
That is, even if all objects are cofibrant, just extending the weak equivalence
is not enough to show that the map $\tgt \colon \HIso_\McU^\Id(\tMcU) \to \McU$
of \Cref{constr:hiso} is a trivial fibration when instantiated to model category
settings.
This is because the object $\HIso_\McU^\Id(\tMcU)$ represents homotopy
equivalences paired with homotopy data (i.e. the homotopies giving witness to
the homotopy sections and retractions), so to show that $\tgt \colon \HIso_\McU^\Id(\tMcU) \to \McU$ is a trivial fibration, one must also extend the homotopy data along cofibrations.

This is what \Cref{subsec:htpy-eqv-ext} shows.
The only other work that we are aware of that shows this homotopy equivalence
extension property is the original simplicial model of \cite{kl21}, but only for
the setting of simplicial sets.
In \Cref{subsec:htpy-eqv-ext}, we provide a generalisation of the homotopy
equivalence extension property to a larger class of model categories (namely
those satisfying \Cref{asm:htpy-eqv-ext}).

%%% Local Variables:
%%% TeX-master: "./main.tex"
%%% TeX-engine: default
%%% End:

\section{Models in Presheaf Model Categories}\label{sec:presheaf-models}
With naive cubical type theory axiomatised in \Cref{def:naive-ctt} and the
homotopy isomorphism extension operation justified in \Cref{thm:heqv-ext}, we
are now ready to construct various models of naive cubical type theory.

\subsection{Classifying Fibrations and Cofibrations}
We would like to show that certain presheaf model structures $\widehat{\bA}$ are
universe models of naive cubical type theory where the type-theoretic fibrations
agree with the (small) model-theoretic fibrations.
In order to do this, we must find a generic small fibration $\tMcU \to \McU$ and
choose a strictly family of pullbacks for it.
This is done by using \cite[Proposition 10]{awo23b} to construct
Hofmann-Streicher universes classifying small fibrations.
To do so, we must construct a classification device for fibration structures.

\begin{construction}\label{constr:fib-data}
  For each $E \to B \in \widehat{\bA}$, set
  $\Fib_B(E) \hookrightarrow B \in \widehat{\bA}$ to be the sub-presheaf of $B$
  consisting its complexes $b$ along which the pullback $E \to B$ is a
  fibration.
  That is, for each $a \in \bA$
  \begin{align*}
    \Fib_B(E)_a \coloneqq \set{a \xrightarrow{b} B ~|~ b^*E \twoheadrightarrow a \text{ is a fibration} }
    \hookrightarrow B_a
  \end{align*}
  whose action on maps $a' \to a$ is defined by composition
  $(a \xrightarrow{b} B) \mapsto (a' \to a \xrightarrow{b} B)$.
  By stability of fibrations under pullback, this indeed defines a presheaf.
\end{construction}

To use \cite[Proposition 10]{awo23b}, we need to first ensure that this
construction of the fibration data given by $\Fib$ is stable under pullback, in
the following sense.

\begin{lemma}\label{lem:fib-data-pb}
  Let there be a map $E \to B \in \widehat{\bA}$ along with a map
  $f \colon X \to B$.
  Then, one has the following pullback
  \begin{equation*}
    % https://q.uiver.app/#q=WzAsNCxbMCwwLCJcXEZpYl9BKGZeKkUpIl0sWzEsMCwiXFxGaWJfQihFKSJdLFsxLDEsIkIiXSxbMCwxLCJBIl0sWzEsMiwiIiwwLHsic3R5bGUiOnsidGFpbCI6eyJuYW1lIjoiaG9vayIsInNpZGUiOiJ0b3AifX19XSxbMywyXSxbMCwzLCIiLDIseyJzdHlsZSI6eyJ0YWlsIjp7Im5hbWUiOiJob29rIiwic2lkZSI6InRvcCJ9fX1dLFswLDFdLFswLDIsIiIsMSx7InN0eWxlIjp7Im5hbWUiOiJjb3JuZXIifX1dXQ==
    \begin{tikzcd}[cramped]
      {\Fib_X(f^*E)} & {\Fib_B(E)} \\
      X & B
      \arrow[from=1-1, to=1-2]
      \arrow[hook, from=1-1, to=2-1]
      \arrow["\lrcorner"{anchor=center, pos=0.05, scale=1.5}, draw=none, from=1-1, to=2-2]
      \arrow[hook, from=1-2, to=2-2]
      \arrow[from=2-1, to=2-2, "f"']
    \end{tikzcd}
  \end{equation*}
\end{lemma}
\begin{proof}
  Evaluating at each $a \in \bA$, it suffices to show that the maps
  $x \colon a \to X$ such that $x^*f^*E \to a$ is a fibration is in bijective
  correspondence with those maps $b \colon a \to B$ factoring through $f \colon X \to B$
  such that $b^*E \to a$ is a fibration.
  This is exactly by functoriality of the pullback.
\end{proof}

Now, we obtain the Hoffman-Streicher universe as follows.
\begin{lemma}\label{lem:fib-class}
  Suppose $\widehat{\bA}$ is equipped with a cofibrantly generated model
  structure in which the generating trivial cofibrations all have representable codomains.
  Let $\kappa$ be an inaccessible cardinal.

  Then, there exists a $\kappa$-small fibration
  $\tMcU_\kappa \twoheadrightarrow \McU_\kappa$ classifying $\kappa$-small
  fibrations.
  That is, $\tMcU_\kappa \twoheadrightarrow \McU_\kappa$ is equipped with a
  strictly functorial choice of pullbacks and moreover if
  $E \twoheadrightarrow B$ is a $\kappa$-small fibration then there exists
  dashed maps making $E \twoheadrightarrow B$ a pullback of
  $\tMcU_\kappa \twoheadrightarrow \McU_\kappa$.
  \begin{equation*}
    % https://q.uiver.app/#q=WzAsNCxbMSwwLCJcXHRNY1VfXFxrYXBwYSJdLFsxLDEsIlxcTWNVX1xca2FwcGEiXSxbMCwxLCJCIl0sWzAsMCwiRSJdLFswLDEsIiIsMCx7InN0eWxlIjp7ImhlYWQiOnsibmFtZSI6ImVwaSJ9fX1dLFsyLDEsIiIsMix7InN0eWxlIjp7ImJvZHkiOnsibmFtZSI6ImRhc2hlZCJ9fX1dLFszLDAsIiIsMCx7InN0eWxlIjp7ImJvZHkiOnsibmFtZSI6ImRhc2hlZCJ9fX1dLFszLDIsIiIsMix7InN0eWxlIjp7ImhlYWQiOnsibmFtZSI6ImVwaSJ9fX1dLFszLDEsIiIsMSx7InN0eWxlIjp7Im5hbWUiOiJjb3JuZXIifX1dXQ==
    \begin{tikzcd}[cramped]
      E & {\tMcU_\kappa} \\
      B & {\McU_\kappa}
      \arrow[dashed, from=1-1, to=1-2]
      \arrow[two heads, from=1-1, to=2-1]
      \arrow["\lrcorner"{anchor=center, pos=0.15, scale=1.5}, draw=none, from=1-1, to=2-2]
      \arrow[two heads, from=1-2, to=2-2]
      \arrow[dashed, from=2-1, to=2-2]
    \end{tikzcd}
  \end{equation*}
\end{lemma}
\begin{proof}
  By \Cref{lem:fib-data-pb} and \cite[Proposition 10]{awo23b}, it suffices to
  show that for any map $E \to B$, the inclusion $\Fib_B(E) \hookrightarrow B$
  admits a section (i.e. is an isomorphism) precisely when $E \twoheadrightarrow B$ is a
  fibration.

  By pullback-stability of fibrations, if $E \twoheadrightarrow B$ is a
  fibration then $\Fib_B(E) \cong B$.
  Conversely, suppose $\Fib_B(E) \cong B$ so that the goal is to show $E \to B$
  is a fibration.
  To this end, pick any generating trivial cofibration so that it must be of the
  form $X \to a$, with $a$ representable.
  Then, each lifting problem of $X \to a$ against $E \to B$ as given by the
  entire rectangle below can be solved by solving the lifting problem with
  $E \to B$ pulled back to $a$, since this pullback is a fibration as
  $\Fib_B(E)_a = \widehat{\bA}(a, B)$ by assumption, as observed in the left
  square below.
  \begin{equation*}
    % https://q.uiver.app/#q=WzAsNixbMiwwLCJFIl0sWzIsMSwiQiJdLFsxLDAsIlxcYnVsbGV0Il0sWzEsMSwiYSJdLFswLDAsIlgiXSxbMCwxLCJhIl0sWzMsMSwiIiwwLHsic3R5bGUiOnsiYm9keSI6eyJuYW1lIjoiZGFzaGVkIn19fV0sWzAsMV0sWzIsMCwiIiwyLHsic3R5bGUiOnsiYm9keSI6eyJuYW1lIjoiZGFzaGVkIn19fV0sWzIsMywiIiwwLHsic3R5bGUiOnsiaGVhZCI6eyJuYW1lIjoiZXBpIn19fV0sWzIsMSwiIiwxLHsic3R5bGUiOnsibmFtZSI6ImNvcm5lciJ9fV0sWzUsMywiPSIsMl0sWzQsNSwiIiwyLHsic3R5bGUiOnsidGFpbCI6eyJuYW1lIjoibW9ubyJ9fX1dLFs0LDIsIiIsMix7InN0eWxlIjp7ImJvZHkiOnsibmFtZSI6ImRhc2hlZCJ9fX1dLFs1LDIsIiIsMCx7InN0eWxlIjp7ImJvZHkiOnsibmFtZSI6ImRvdHRlZCJ9fX1dLFs0LDAsIiIsMCx7ImN1cnZlIjotMiwic3R5bGUiOnsiYm9keSI6eyJuYW1lIjoiZGFzaGVkIn19fV1d
    \begin{tikzcd}[cramped]
      X & \bullet & E \\
      a & a & B
      \arrow[dashed, from=1-1, to=1-2]
      \arrow[curve={height=-12pt}, dashed, from=1-1, to=1-3]
      \arrow[hook, from=1-1, to=2-1]
      \arrow[dashed, from=1-2, to=1-3]
      \arrow[two heads, from=1-2, to=2-2]
      \arrow["\lrcorner"{anchor=center, pos=0.15, scale=1.5}, draw=none, from=1-2, to=2-3]
      \arrow[from=1-3, to=2-3]
      \arrow[dotted, from=2-1, to=1-2]
      \arrow["{=}"', from=2-1, to=2-2]
      \arrow[dashed, from=2-2, to=2-3]
    \end{tikzcd}
  \end{equation*}
\end{proof}

\subsection{Elegant Reedy Presheaf Models}
We now show that various presheaf model categories are models of naive cubical
type theory with propositionally computational filling.
In particular, the classical model structure on simplicial sets and several
model structures on various cubical sets are such examples.

\begin{proposition}\label{prop:psh-models}
  Suppose that $\widehat{\bA}$ is the presheaf category of a category $\bA$
  equipped with a cofibrantly-generated model structure such that
  \begin{enumerate}
    \item[mod-cof]\label{itm:mod-cof} The cofibrations are exactly the monomorphisms.
    \item[mod-gtc]\label{itm:mod-gtc} The generating trivial cofibrations have representable codomains.
    \item[mod-pp]\label{itm:mod-pp} The pushout-product of a (trivial
    cofibration, cofibration)-pair is a trivial cofibration.
    \item[mod-int]\label{itm:mod-int} $I$ is a good, but not necessarily very good, cylinder object for the
    terminal object with a $\min$-structure.
    \item[mod-rp]\label{itm:mod-rp} The model structure is right proper.
    \item[mod-fe]\label{itm:mod-fe} If $\alpha$ is an inaccessible cardinal then
    $\alpha$-small fibrations extend along trivial cofibrations.
  \end{enumerate}

  Then, assuming the existence of two arbitrarily large inaccessible cardinals
  $\kappa < \lambda$, this presheaf model category $\widehat{\bA}$ has the
  structure $(\tMcU \to \McU, \partial\Cof \hookrightarrow \Cof, \bI)$ of naive
  cubical type theory with weakly computational filling structure
  (\Cref{def:naive-ctt}) in which
  \begin{enumerate}[]
    \item The external ambient universe $(\tMcU \to \McU)$-fibrations are the
    $\lambda$-small fibrations and the internal universe
    $(\tMcU_0 \to \McU_0)$-fibrations are the $\kappa$-small fibrations.
    \item $\partial\Cof \hookrightarrow \Cof$ is the truth map into the
    subobject classifier.
    \item $\bI$ is the chosen cylinder object $I$ for the terminal object with a
    $\min$-structure as above.
    \item The strong homotopy isomorphism extension operation
    (\Cref{itm:nctt-hext}) induces a $\Glue$-type structure
    (\Cref{def:glue-type}).
  \end{enumerate}
\end{proposition}
\begin{proof}
  Let $\kappa < \lambda$ be two inaccessible cardinals.

  By \Cref{lem:fib-class}, we can find a classifier
  $\tMcU_\lambda \to \McU_\lambda$ for $\lambda$-small fibrations, which we take
  to be the ambient universe $\tMcU \to \McU$ for the structure of NCTT.
  Informed by \Cref{itm:mod-cof}, we take the cofibration classifier as the
  subobject classifier $\partial\Omega \hookrightarrow \Omega$ and informed by
  \Cref{itm:mod-int}, we take the given cylinder object $I$ as the interval
  $\bI$.
  We verify each point of \Cref{def:naive-ctt}.

  \begin{enumerate}
    \item[\Cref{itm:nctt-truth}] This follows by \Cref{lem:omega-truth}.
    \item[\Cref{itm:nctt-cof-sec}] Cofibrancy of sections is clear by the
    universal property of the subobject classifier.
    \item[\Cref{itm:nctt-int}] Likewise, cofibrancy of the interval endpoints is
    clear by universal property of the subobject classifier.
    \item[\Cref{itm:nctt-min}]   By \Cref{itm:mod-int}, the $\min$-structure required for $\bI$
    is assumed to exist.
    \item[\Cref{itm:nctt-fib}] We show each logical construction separately.
    \begin{itemize}
      \item[$\Unit$-type] Fibrations contain all isomorphisms, so
      $\tMcU_\lambda \to \McU_\lambda$ admits a $\Unit$-type structure.
      \item[$\Sigma$-type] Because $(\tMcU_\lambda \to \McU_\lambda)$-fibrations
      are $\lambda$-small fibrations by construction and $\lambda$-small
      fibrations are closed under composition, universality of
      $\tMcU_\lambda \to \McU_\lambda$ gives rise to a $\Sigma$-type structure
      on it.
      \item[$\Pi$-type] For $\Pi$-types, the result follows by \cite[Lemma
      1.5]{axm-univalence}, and the fact that $\lambda$-small fibrations are
      closed under pushforwards along $\lambda$-small fibrations by right
      properness \Cref{itm:mod-rp} and the inaccessibility of $\lambda$.
      \item[Het. $\Path$-type] For heterogeneous $\Path$-type structures, we
      note that by \cite[Theorem 5.9]{struct-lift} applied to
      \Cref{itm:mod-pp,itm:mod-int}, the endpoint evaluation map of the fibred
      path object
      \begin{equation*}
        \hP_{(\pi_\lambda)_*(\tMcU_\lambda \times\McU_\lambda)}^I\left(
          \ev^*(\tMcU_\lambda \times \tMcU_\lambda)
          \to
          (\pi_\lambda)^*(\pi_\lambda)_*(\tMcU_\lambda \times \McU_\lambda)
        \right)
      \end{equation*}
      is a fibration which is also $\lambda$-small by inaccessibility of $\lambda$.
      Thus by universality of $\tMcU_\lambda \to \McU_\lambda$ once more,
      $\tMcU_\lambda \to \McU_\lambda$ admits a heterogeneous $\Path$-type
      structure.
    \end{itemize}
    \item[\Cref{itm:nctt-fill}] We must show the existence of $\delta$-biased
    filling structures for $\tMcU_\lambda \to \tMcU_\lambda$.
    For this, we use \cite[Theorem 1.10]{struct-lift}.
    In particular, by \Cref{ex:omega-pushout-product}, the left map
    $(\Omega \vee \txtis_{\delta})^*\partial\Omega \hookrightarrow \Omega \times
    \bI$ of \Cref{def:biased-filling} is the pushout-product
    $(\partial\Omega \hookrightarrow \Omega) \ltimes (\set{\delta} \hookrightarrow
    I)$.
    By assumption \Cref{itm:mod-pp}, this pushout-product is a trivial
    cofibration, so the result follows by \cite[Theorem 1.10]{struct-lift}.
    \item[\Cref{itm:nctt-fill-wkcomp}] This follows by \cite[Theorem
    6.5]{struct-lift} applied to \Cref{itm:mod-int,itm:mod-pp} and the
    definition of $\Id$-types using $\Path$-types as per
    \Cref{thm:nctt-Path-Id}.
    \item[\Cref{itm:nctt-univ}] Applying \Cref{lem:fib-class} again, we can find
    a classifier $\tMcU_\kappa \to \McU_\kappa$ for $\kappa$-small fibrations,
    which we take to be the internal universe.
    Repeating the argument used to check \Cref{itm:nctt-fib} for $\kappa$, we
    see that $\tMcU_\kappa \to \McU_\kappa$ is also equipped with these logical
    structures.
    Furthermore, we note that $\McU_\kappa \to 1$ is $\lambda$-small and
    furthermore fibrant \Cref{itm:mod-fe} so $\McU_\kappa \to 1$ occurs as a
    pullback of $\tMcU_\lambda \to \McU_\lambda$.
    To ensure that these logical structures of the internal universe are
    compatible with those of the external universe, we note that
    $\McU_\kappa \hookrightarrow \McU_\lambda$ is a sub-presheaf so we can realign
    the logical structures of $\McU_\lambda$ using \cite[Proposition 6]{awo23b}.
    \item[\Cref{itm:nctt-hext}] The map
    $\tgt \colon \HIso^{\Id_0}_{\tMcU_\kappa}(\tMcU_\kappa) \to \McU_\kappa$ of
    \Cref{def:strong-hiso-ext-op} is a trivial fibration by
    \Cref{constr:hiso,thm:heqv-ext}, so \cite[Theorem 1.10]{struct-lift} applied
    to \Cref{itm:mod-cof,itm:mod-pp,itm:mod-int} produces the required strong
    homotopy isomorphism extension structure inducing a $\Glue$-type structure,
    as claimed.
  \end{enumerate}
\end{proof}

The conditions for \Cref{prop:psh-models} is readily verified for the following
Reedy presheaf model categories.
\begin{theorem}\label{thm:simplicial-model}
  Assuming the existence of two arbitrarily large inaccessible cardinals
  $\kappa<\lambda$, the classical model structure on simplicial sets $\sSet$
  induces the structure of naive cubical type theory
  $(\tMcU \to \McU, \partial\Cof\hookrightarrow\Cof,\bI)$ with weakly
  computational filling structure (\Cref{def:naive-ctt}) in which
  \begin{enumerate}[]
    \item The external ambient universe $(\tMcU \to \McU)$-fibrations are the
    $\lambda$-small Kan fibrations and the internal universe
    $(\tMcU_0 \to \McU_0)$-fibrations are the $\kappa$-small Kan fibrations.
    \item $\partial\Cof \hookrightarrow \Cof$ is the truth map into the
    subobject classifier.
    \item $\bI$ is the standard 1-simplex $\Delta^1$.
    \item The strong homotopy isomorphism extension operation
    (\Cref{itm:nctt-hext}) induces a $\Glue$-type structure
    (\Cref{def:glue-type}).
  \end{enumerate}
\end{theorem}
\begin{proof}
  It is well-known that all assumptions of \Cref{prop:psh-models} are satisfied
  by $\sSet$, with the only possible non-standard point being \Cref{itm:mod-fe},
  which is proved in \cite[Theorem 2.2.1]{kl21} and \cite[Proposition
  2.21]{cis14} using the theory of minimal fibrations and constructively in
  \cite[Corollary 4.2.7]{gss22}.
\end{proof}

Asides from the simplicial model, we also can apply \Cref{prop:psh-models} to
build some cubical models of naive cubical type theory.
\begin{theorem}\label{thm:cis-cubical-model}
  The categories of
  \begin{enumerate}[]
    \item Cubical sets with only faces and degeneracies
    \item Cubical sets with faces, degeneracies plus connections
  \end{enumerate}
  from \cite[Examples 1.5 and 1.6]{cis14} both carry model structures giving
  rise to models $(\tMcU\to\McU,\partial\Cof\hookrightarrow\Cof,\bI)$ of naive
  cubical type theory with weak computational filling (\Cref{def:naive-ctt})
  such that if $\kappa<\lambda$ are two arbitrarily large inaccessible cardinals
  then
  \begin{enumerate}[]
    \item The external ambient universe $(\tMcU \to \McU)$-fibrations are the
    $\lambda$-small fibrations and the internal universe
    $(\tMcU_0 \to \McU_0)$-fibrations are the $\kappa$-small fibrations.
    \item $\partial\Cof \hookrightarrow \Cof$ is the truth map into the
    subobject classifier.
    \item $\bI$ is the standard 1-cube $\square^1$.
    \item The strong homotopy isomorphism extension operation
    (\Cref{itm:nctt-hext}) induces a $\Glue$-type structure
    (\Cref{def:glue-type}).
  \end{enumerate}
\end{theorem}
\begin{proof}
  By \cite[Examples 1.5 and 1.6]{cis14} , these two categories both admit the
  Grothendieck model structure, which is such that
  \begin{enumerate}[]
    \item[\Cref{itm:mod-cof}] The cofibrations are exactly the monomorphisms \cite[Paragraph
    1.2]{cis14}.
    \item[\Cref{itm:mod-gtc}] The generating cofibrations are the boundary inclusions of the
    standard $n$-cubes \cite[Examples 1.5 and 1.6]{cis14}
    \begin{equation*}
      \partial\square_n \hookrightarrow \square_n
    \end{equation*}
    \item[\Cref{itm:mod-pp}] The pushout-product of a (trivial cofibration,
    cofibration)-pair is a trivial cofibration because of \cite[Lemma
    8.4.36]{cis06} and the fact that the generating trivial cofibrations are the
    horn inclusions of the standard $n$-cubes
    \begin{equation*}
      \sqcap^{i,e}_n \hookrightarrow \square_n
    \end{equation*}
    \item[\Cref{itm:mod-int}] The standard 1-cube $\square_1$ is observed to be a good, but not
    necessarily very good, cylinder object for the terminal object with a
    $\min$-structure.
    \item[\Cref{itm:mod-rp}] The model structure is proper \cite[Paragraph 1.2]{cis14}.
    \item[\Cref{itm:mod-fe}] Small fibrations extend along trivial cofibrations
    \cite[Proposition 2.21]{cis14}.
  \end{enumerate}
  Therefore, the result follows from \Cref{prop:psh-models}.
\end{proof}

\subsection{Cartesian Cubical Model}
We have previously showed in \Cref{prop:psh-models} the required
conditions on model structures on presheaf categories for them to carry
the structure of naive cubical type theory.
Examples of these are observed in
\Cref{thm:simplicial-model,thm:cis-cubical-model}, which happen to be all elegant
Reedy presheaf categories.
We now use the work of \cite{awo23a} to show that the Cartesian cubes is an
example of a non-Reedy presheaf model category equipped with a model structure
that also carries the structure of naive cubical type theory.

We begin by recalling the model structure constructed in \cite{awo23a}.
\begin{definition}[{\cite[Definition 2]{awo23a}}]\label{def:cart-cubes}
  The \emph{Cartesian cube category $\Cube$} consists of objects finite sets of
  the form
  \begin{equation*}
    [n] \coloneqq \set{\bot,1,...,n,\top}
  \end{equation*}
  with two distinguished points $\bot,\top$ with maps
  $f \colon [m] \to [n] \in \Cube$ the functions
  $f \colon \set{\bot,1,...,n,\top} \to \set{1,...,m}$ (i.e. going in the other
  way and preserving the endpoints).
  The category of \emph{Cartesian cubical sets} is the presheaf category
  \begin{align*}
    \cSet \coloneqq \Set^{\Cube^\op}
  \end{align*}
  and we write $\square^- \colon \Cube \hookrightarrow \Set^{\Cube^\op}$ for the
  Yoneda embedding.

  We write $I$ for $I \coloneqq \square^1$, the standard 1-cube.
\end{definition}

\begin{definition}[{\cite[Definition 29, Equation
    (43)]{awo23a}}]\label{def:cart-cubes-gen-tc}
  A \emph{generating unbiased trivial cofibration} of $\cSet$ is a
  pushout-product
  \begin{equation*}
    (C \hookrightarrow Z) \ltimes_I (I \xrightarrow{\Delta} I \times I)
    \colon C \times I \cup Z \hookrightarrow Z \times I
  \end{equation*}
  in the slice over $I$, viewed as a map in $\cSet$ by forgetting the codomain
  $I$, of a mono $C \hookrightarrow Z \in \sfrac{\cSet}{I}$ with the diagonal map
  $\Delta \colon I \to I \times I \in \sfrac{\cSet}{I}$.
\end{definition}

\begin{theorem}[{\cite[Theorem 124]{awo23a}}]\label{thm:awo-model}
  There is a cofibrantly-generated model structure on $\cSet$ where the
  cofibrations are exactly the monos and the generating trivial cofibrations the
  generating unbiased trivial cofibrations.
  \def\endingmark{\qedsymbol}
\end{theorem}

Most of the actual work for the model structure of \cite[Theorem 124]{awo23a},
henceforth referred to as the \emph{Cartesian cubical model structure}, to assemble into a
model of naive cubical type theory is already done in \cite{awo23a}.
%
% Since \cite[Proposition 116]{awo23a} showed only that weak equivalences between
% fibrant objects extend along cofibrations, whereas the homotopy equivalence
% extension structure, as formulated in \Cref{def:hiso-ext-op}, requires one to
% extend the data comprising a homotopy equivalence (i.e. the left and right
% homotopy inverses as well as the homotopies to identities), we need to apply the
% result \Cref{prop:heqv-ext} of extending homotopy equivalences along
% cofibrations to the Cartesian cubical model structure on $\cSet$.
% %
% For this, we need to ensure the premise of \Cref{prop:heqv-ext} holds, which is
% just amounting to checking the following.
%
\begin{lemma}\label{lem:cart-cube-pushout-product}
  In the Cartesian cubical model structure on $\cSet$, the pushout-product of a
  (cofibration, trivial cofibration)-pair is still a trivial cofibration.
\end{lemma}
\begin{proof}
  By cofibrant generation, it suffices to check that the generating trivial
  cofibrations (i.e.  generating unbiased trivial cofibrations) are closed under
  pushout-product with cofibrations (i.e. monos).

  Let $\partial V \hookrightarrow V$ be a mono over $I$ and
  $\partial U \hookrightarrow U$ be just any mono.
  Thus, $U \times V \to V \to I$ is also an object over $I$, which means that the
  pushout-product
  \begin{equation*}
    (\partial U \hookrightarrow U) \ltimes (\partial V \hookrightarrow V) =
    (\partial U \times V \cup U \times \partial V \hookrightarrow U \times V)
  \end{equation*}
  is a mono over $I$.
  Then, by associativity of the pushout via a 3-by-3 argument, one calculates
  that
  \begin{equation*}
    (\partial U \hookrightarrow U) \ltimes ((\partial V \hookrightarrow V)
    \ltimes_I (I \xrightarrow{\Delta} I \times I)) \cong ((\partial U
    \hookrightarrow U) \ltimes (\partial V \hookrightarrow V)) \ltimes_I (I
    \xrightarrow{\Delta} I \times I)
  \end{equation*}
\end{proof}

\begin{theorem}\label{thm:cart-cube-model}
  Assuming the existence of two arbitrarily large inaccessible cardinals
  $\kappa<\lambda$ the Cartesian cubical model structure on $\cSet$ induces a naive cubical
  type theory structure
  $(\tMcU \to \McU, \partial\Cof \hookrightarrow \Cof, \bI)$ of naive cubical
  type theory with weakly computational filling structure in which
  \begin{enumerate}
    \item $(\tMcU \to \McU)$-fibrations are the $\lambda$-small fibrations and
    the $(\tMcU_0 \to \McU_0)$-fibrations are the $\kappa$-small fibrations.
    \item $\partial\Cof \hookrightarrow \Cof$ is the truth map into the
    subobject classifier.
    \item $\bI$ is taken to be the standard 1-cube $I = \square^1$.
    \item The strong homotopy isomorphism extension operation
    (\Cref{itm:nctt-hext}) induces a $\Glue$-type structure
    (\Cref{def:glue-type}).
  \end{enumerate}
\end{theorem}
\begin{proof}
  Let $\kappa<\lambda$ be two inaccessible cardinals.
  Then, by \cite[Proposition 91]{awo23a}, there is a classifier
  $\tMcU_\lambda \to \McU_\lambda$ for $\lambda$-small fibrations, which we take
  to be the ambient universe $\tMcU \to \McU$ for the structure of NCTT.
  Then, as before in \Cref{prop:psh-models}, we take the universe of
  cofibrations to be the truth map for the subobject classifier and the interval
  as the standard 1-cube $\square^1$
  One proceeds to verify each point of \Cref{def:naive-ctt}.

  \begin{enumerate}
    \item[\Cref{itm:nctt-truth}] This follows by \Cref{lem:omega-truth}.
    \item[\Cref{itm:nctt-cof-sec}] Cofibrancy of sections is clear by the
    universal property of the subobject classifier.
    \item[\Cref{itm:nctt-int}] Likewise, cofibrancy of the interval endpoints is
    clear by universal property of the subobject classifier.
    \item[\Cref{itm:nctt-min}] The $\min$-structure on $\square^1$ is obvious.
    \item[\Cref{itm:nctt-fib}] We show each logical construction separately.
    \begin{itemize}
      \item[$\Unit$-type] Fibrations contain all isomorphisms, so
      $\tMcU_\lambda \to \McU_\lambda$ admits a $\Unit$-type structure.
      \item[$\Sigma$-type] Because $(\tMcU_\lambda \to \McU_\lambda)$-fibrations
      are $\lambda$-small fibrations by construction and $\lambda$-small
      fibrations are closed under composition, universality of
      $\tMcU_\lambda \to \McU_\lambda$ gives rise to a $\Sigma$-type structure
      on it.
      \item[$\Pi$-type] For $\Pi$-types, the result follows by \cite[Lemma
      1.5]{axm-univalence}, and the fact that $\lambda$-small fibrations are
      closed under pushforwards along $\lambda$-small fibrations by right
      properness as a result of \cite[Corollary 123]{awo23a} and the
      inaccessibility of $\lambda$.
      \item[Het. $\Path$-type] For heterogeneous $\Path$-type structures, we
      note that by \cite[Theorem 5.9]{struct-lift} applied to $\square^1$ and
      \Cref{lem:cart-cube-pushout-product}, the endpoint evaluation map of the
      fibred path object
      \begin{equation*}
        \hP_{(\pi_\lambda)_*(\tMcU_\lambda \times\McU_\lambda)}^{\square^1}\left(
          \ev^*(\tMcU_\lambda \times \tMcU_\lambda)
          \to
          (\pi_\lambda)^*(\pi_\lambda)_*(\tMcU_\lambda \times \McU_\lambda)
        \right)
      \end{equation*}
      is a fibration which is also $\lambda$-small by inaccessibility of $\lambda$.
      Thus by universality of $\tMcU_\lambda \to \McU_\lambda$ once more,
      $\tMcU_\lambda \to \McU_\lambda$ admits a heterogeneous $\Path$-type
      structure.
    \end{itemize}
    \item[\Cref{itm:nctt-fill}] We must show the existence of $\delta$-biased
    filling structures for $\tMcU_\lambda \to \tMcU_\lambda$.
    For this, we use \cite[Theorem 1.10]{struct-lift}.
    In particular, by \Cref{ex:omega-pushout-product}, the left map
    $(\Omega \vee \txtis_{\delta})^*\partial\Omega \hookrightarrow \Omega \times
    \bI$ of \Cref{def:biased-filling} is the pushout-product
    $(\partial\Omega \hookrightarrow \Omega) \ltimes (\set{\delta} \hookrightarrow
    \square^1)$.
    By \Cref{lem:cart-cube-pushout-product}, this pushout-product is a trivial
    cofibration, so the result follows by \cite[Theorem 1.10]{struct-lift}.
    \item[\Cref{itm:nctt-fill-wkcomp}] The reasoning is the same as for
    \Cref{itm:nctt-fill}, except we apply \cite[Theorem 6.5]{struct-lift}
    instead of \cite[Theorem 1.10]{struct-lift}.
    This checks that the computation rules for filling hold propositionally.
    \item[\Cref{itm:nctt-univ}]
    Applying \cite[Proposition 91]{awo23a} again, we can find a classifier
    $\tMcU_\kappa \to \McU_\kappa$ for $\kappa$-small fibrations, which we take
    to be the internal universe.
    Repeating the argument used to check \Cref{itm:nctt-fib} for $\lambda$, we
    see that $\tMcU_\kappa \to \McU_\kappa$ is also equipped with these logical
    structures.
    Furthermore, we note that $\McU_\kappa \to 1$ is $\lambda$-small and
    furthermore fibrant because of \cite[Corollary 123]{awo23a} shows that small
    fibrations extend along trivial cofibrations so $\McU_\kappa \to 1$ occurs
    as a pullback of $\tMcU_\lambda \to \McU_\lambda$.
    To ensure that these logical structures of the internal universe are
    compatible with those of the external universe, we note that
    $\McU_\kappa \hookrightarrow \McU_\lambda$ is a sub-presheaf so we can realign
    the logical structures of $\McU_\lambda$ using \cite[Proposition 6]{awo23b}.
    \item[\Cref{itm:nctt-hext}] The map
    $\tgt \colon \HIso^{\Id_0}_{\tMcU_\kappa}(\tMcU_\kappa) \to \McU_\kappa$ of
    \Cref{def:strong-hiso-ext-op} is a trivial fibration by
    \Cref{constr:hiso,thm:heqv-ext}, so the result follows by \cite[Theorem
    1.10]{struct-lift} applied to
    \Cref{lem:cart-cube-pushout-product,thm:awo-model}, which also produces the
    required $\Glue$-type structure.
  \end{enumerate}
  %
  % By \cite[Proposition 91]{awo23a}, for each $\kappa$, there exists a universal
  % $\kappa$-small fibration $\McU_\kappa \to \McU_\kappa$.
  % %
  % Furthermore, $\McU_\kappa$ is fibrant as \cite[Corollary 123]{awo23a} shows
  % that fibrations extend along trivial cofibrations and admits $\Pi$-types as
  % \cite[Corollary 73]{awo23a} shows that pushforwards of fibrations along
  % fibrations is still a fibration.
  % %
  % Therefore, one may replicate the proof of \Cref{prop:psh-models} to
  % obtain a naive cubical type theory with weakly computational filling from the
  % Cartesian cubical model structure on $\cSet$.
  %
  %
  % To make this internal universe closed under the $\Unit,\Sigma,\Pi$- and
  % heterogeneous $\Path$-type structures, we use the realignment property from
  % \cite[Lemma 96]{awo23a}.
  %
\end{proof}

%%% Local Variables:
%%% TeX-master: "./main.tex"
%%% TeX-engine: default
%%% End:

\printbibliography

\end{document}